\documentclass[11pt, a4paper, reqno]{amsart}
\usepackage{tikz}
\usepackage[toc,page]{appendix}
\usepackage[utf8]{inputenc}
\usepackage[T1]{fontenc}
\usepackage{epsfig}

\newcommand{\ii}{\mathbf i}
\newcommand{\cW}{\mathcal{W}}   

\newcommand{\band}{\mathfrak{B}}
\newcommand{\RR}{\mathbb{R}}
\def\defeq{:=}

\usepackage{graphicx}
\usepackage{amsfonts,amsmath,amsthm,amssymb}
\newcommand{\cG}{\mathcal{G}}

\def\EE{\mathbb{E}}

\newtheorem{theorem}{Theorem}[section]%

\newtheorem{proposition}[]{Proposition}[section]
\newtheorem{definition}{Definition}[section]
\newtheorem{remark}{Remark}[section]
\newtheorem{lemma}{Lemma}[section]
\newtheorem{corollary}{Corollary}[section]

\usepackage{tikz}
\usepackage{hyperref}
\usepackage{breakurl}
\usepackage{color}
\hypersetup{
  colorlinks   = true, %Colours links instead of ugly boxes
  urlcolor     = MyBlue, %Colour for external hyperlinks
  linkcolor    = MyBlue, %Colour of internal links
  citecolor   = darkspringgreen %Colour of citations
}
\usepackage{url}
\usepackage[margin=0.95in]{geometry}
\numberwithin{equation}{section}
\usepackage{bbold}
\usepackage{cancel}
\usepackage{enumitem}

\usetikzlibrary{fit,
  matrix,
  positioning,
  decorations,
  decorations.pathreplacing
}
\usepackage{xurl}

\makeatother
\usepackage{epsfig}
\usepackage{graphicx}
\usepackage{amsfonts,amsmath,amsthm,amssymb}
\usepackage{epiolmec}
\usepackage{phaistos}
\usepackage{adjustbox}
\usepackage[normalem]{ulem}

\usepackage{mathtools}
\mathtoolsset{showonlyrefs}
\newcommand{\tr}{\operatorname{tr}}
\usepackage{multicol}
\usepackage{latexsym}
\usepackage{mathrsfs} 
\newcommand\Vtextvisiblespace[1][.5em]{%
  \mbox{\kern.06em\vrule height.3ex}%
  \vbox{\hrule width#1}%
  \hbox{\vrule height.3ex}}
\usepackage{cancel}

\usepackage{xparse}
\usepackage{float}
\usepackage{todonotes}
\usepackage{etoolbox}
\usepackage{csquotes}
\usepackage{enumitem}
\usepackage{breakurl}

\usepackage[giveninits=true,uniquename=false,uniquelist=false,style=alphabetic,sortcites=false,natbib=true,backend=biber,datamodel=mrnumber, maxbibnames=99, block=space]{biblatex}
\DeclareFieldFormat{mrnumber}{%
  MR\addcolon\space
  \ifhyperref
    {\href{http://www.ams.org/mathscinet-getitem?mr=#1}{\nolinkurl{#1}}}
    {\nolinkurl{#1}}}

\renewbibmacro*{doi+eprint+url}{%
  \iftoggle{bbx:doi}
    {\printfield{doi}}
    {}%
  \newunit\newblock
  \printfield{mrnumber}%
  \newunit\newblock
  \iftoggle{bbx:eprint}
    {\usebibmacro{eprint}}
    {}%
  \newunit\newblock
  \iftoggle{bbx:url}
    {\usebibmacro{url+urldate}}
    {}}

\usepackage{tikz}
\usepackage{upgreek}
\usetikzlibrary{arrows,positioning,cd,calc,math,decorations.pathreplacing,decorations.markings,decorations.pathmorphing,patterns%,snakes
}
\tikzset{wave/.style={decorate, decoration=snake}}

\usepackage{xcolor}
\definecolor{MyGreen}{rgb}{0.05, 0.7, 0.06}
\definecolor{MyBlue}{rgb}{0.0, 0.0, 0.55}
\definecolor{MyRed}{rgb}{0.9,0.19,0.39}
\definecolor{ggreen}{rgb}{0.05, 0.5, 0.06}
\usepackage{subcaption}
\NewDocumentCommand{\tens}{t_}
 {%
  \IfBooleanTF{#1}
   {\tensop}
   {\otimes}%
 }
\NewDocumentCommand{\tensop}{m}
 {%
  \mathbin{\mathop{\otimes}\displaylimits_{#1}}%
 }

\usepackage{tikz}
\usepackage{upgreek}
\usetikzlibrary{decorations.markings}
\definecolor{darkspringgreen}{rgb}{0.05, 0.5, 0.06}
\colorlet{NextBlue}{MyBlue!20}
\colorlet{SecondBlue}{MyBlue!40}
\DeclareUnicodeCharacter{0327}{\c}
\title[Semi-classical limit of conformal blocks]{Proof of Zamolodchikov conjecture for semi-classical conformal blocks on the torus}

\author{Harini Desiraju}
\address[Harini Desiraju]{Sydney Mathematical Research Institute (SMRI), School of Mathematics and Statistics, University of Sydney, Camperdown, NSW 2006, Australia, and Mathematical Institute, Andrew Wiles Building, Woodstock Rd, Oxford OX2 6GG, United Kingdom.}
\email{\href{mailto:harini.desiraju@maths.ox.ac.uk}{harini.desiraju@maths.ox.ac.uk}}
\author{Promit Ghosal}
 \address[Promit Ghosal]{University of Chicago, Department of Statistics,  5747 S Ellis Ave, Chicago, IL 60637, USA. }
\email{\href{mailto:promit@uchicago.edu}{promit@uchicago.edu}}
\author{Andrei Prokhorov}
\address[Andrei Prokhorov]{University of Chicago, Department of Statistics,  5747 S Ellis Ave, Chicago, IL 60637, USA, and St. Petersburg University, 7/9 Universitetskaya nab. 199034 St. Petersburg, Russia.}
\email{\href{mailto:andreip@uchicago.edu}{andreip@uchicago.edu}}

\begin{document}

% \maketitle

% \date{\vspace{-5ex}}

\begin{abstract}
In 1986, Zamolodchikov conjectured an exponential structure for the semi-classical limit of conformal blocks on a sphere. This paper provides a rigorous proof of the analog of Zamolodchikov conjecture for Liouville conformal blocks on a one-punctured torus, using their probabilistic construction and show the existence of a positive radius of convergence of the semi-classical limit. As a consequence, we obtain a closed form expression for the solution of the Lam\'e equation, and show a relation between its accessory parameter and the classical action of the non-autonomous elliptic Calogero-Moser model evaluated at specific values of the solution. %The asymptotic behavior of the conformal blocks as the modular parameter of the torus $\tau\to\ii \infty$ is also analyzed. 
\end{abstract}
\maketitle
\tableofcontents
%\pagebreak

% \section*{To-do}
% \begin{enumerate}
%     \item Section 3: Proposition 3.2
%    \item \sout{Section 4: Proof of Lemma 4.1}
%    \item \sout{Section 5: Proof of Lemma 5.1}, rederive 5.5 with the \sout{proper scaling}, \sout{derive the asymptotic for theta function}
% \end{enumerate}

% \begin{enumerate}
% \item {\bf Andrei: }
% Check Lam\'e equation numerically.
% \item {\bf Promit:} List properties of GMC and properties of the one-point CB
% \item {\bf Harini:} Clean-up the notation
%     % \item tail probability proofs \cite{lacoin2022path}
%     % \item To prove uniqueness of Hamilton-Jacobi, prove the commutativity of $\lim_{\gamma\rightarrow 0}$ and $\lim_{q\rightarrow 0}$. To do this, use the form of  the deformed conformal block where the part with $q$ and $\gamma$ decoupled.
%     % \begin{itemize}
%     %     \item Look at the version with $z$ decoupled also
%     %     \item renormalize the $q$ dependence
%     % \end{itemize}
%     % \item Semiclassical limit of hypergeometric function $\psi(z,0)$.
%     % \item Accessory parameters and Blow up relations (?)
% \end{enumerate}

\section{Introduction}

Conformal field theory (CFT), a subclass of quantum field theory{\color{blue},} provides a framework for constructing random functions on Riemannian manifolds, which transform covariantly under conformal (i.e., angle-preserving) mappings. Since the groundbreaking work by Belavin, Polyakov, and Zamolodchikov (BPZ) in 1984 \cite{belavin1984infinite}, two-dimensional (2D) CFT has evolved into a major area of theoretical physics. It has significant applications in 2D statistical physics, string theory and has profound implications for mathematics. See \cite{ribault2022conformal} for a review and references therein.  %(see, e.g., [DFMS97]).

An important subclass of conformal field theories (CFTs) is known as Liouville CFTs (LCFTs). These CFTs are characterized by a parameter called the central charge, denoted as $c$, which for LCFTs lies in the range 
$c\in [25,\infty)$. LCFT originated from Polyakov's work on two-dimensional quantum gravity and bosonic string theory in 1981 \cite{polyakov1981quantum}. It was rigorously constructed using the path integral formalism of quantum field theory, initially on the sphere in \cite{DKRV16} and subsequently on other Riemann surfaces in \cite{DRV16,HRV18,GRV19}. We refer the reader to \cite{vargas2017lecture} for a review. 

The paper by Belavin, Polyakov, and Zamolodchikov in 1984 introduced a schematic approach known as the conformal bootstrap. This program aims to obtain  explicit expressions for the correlation functions of a given two-dimensional CFT by expressing them in terms of its three-point sphere correlation functions and specific power series called conformal blocks. These conformal blocks are entirely determined by the Virasoro algebra, which encodes
% the infinitesimal 
local conformal symmetries. Notably, the central charge parametrizes conformal blocks of a particular  CFT.

Recently, there has been rapid progress in rigorous approaches to CFT using probability, notably by \cite{DKRV16,DRV16,GRV19,HRV18}. Our starting point in this paper is a construction of Liouville conformal blocks on one-point torus introduced by \cite{ghosal2020probabilistic} using rigorous path integral formulation based on probabilistic ideas. Their construction is via the Gaussian multiplicative chaos (GMC) measure,
a random measure defined by exponentiating the Gaussian free field (see e.g. \cite{DRV16, berestycki2017elementary,aru2017gaussian}).  This paper explores the semi-classical limit of the probabilistic representation of Liouville conformal blocks, focusing on the case of torus with a single marked point. 
% {\color{red}For a given $\tau$  in } the upper half-plane, 
We consider the flat torus $\mathbb{T}_{\tau}$ characterized by the modular parameter $\tau$, and the nome $q:= e^{\ii \pi \tau}$. 
% \rmkH{What were we trying to say here?}
The $1$-point toric correlation function of LCFT, rigorously constructed in \cite{DRV16},
has the form $\langle e^{\alpha \Phi(0)}\rangle_{\tau}$, where $\langle \ldots \rangle_{\tau}$ is the average over the random field $\Phi$
for LCFT on $\mathbb{T}_{\tau}$ and $\alpha$ is called the vertex insertion weight. Conformal bootstrap formula for torus which is recently shown in \cite[Theorem~9]{GKRV2022} expresses $\langle e^{\alpha \Phi(0)}\rangle_{\tau}$ via one point toric conformal block $\mathcal{G}^{\alpha}_{\gamma, P}(q)$, 
\begin{align}
\langle e^{\alpha \Phi(0)}\rangle_{\tau} = \int_{\mathbb{R}} C_{\gamma}(\alpha, Q - \mathbf{i} P, Q + \mathbf{i} P) |\mathcal{G}^{\alpha}_{\gamma, P}(q)|^2 dP, 
\end{align}
where the parameter $\gamma$ dictates the central charge through the relation 
\begin{align}
    c = 1 + 6 Q^2 \in (25,\infty), && Q := \frac{\gamma}{2} + \frac{2}{\gamma} 
    \in (2,\infty).\label{c:gamma}
\end{align}
Here $ C_{\gamma}(\alpha_1,\alpha_2, \alpha_3)$ is the Liouville 3-point correlation function, whose exact description is given by the Dorn-Otto-Zamolodchikov-Zamolodchikov (DOZZ) formula, which was first proposed in \cite{dorn1994two,zamolodchikov1996conformal} and later proved in \cite{kupiainen2020integrability}. It is worthwhile to note that the conformal bootstrap for sphere is recently proven in \cite{kupiainen2020integrability} and \cite{GKRV2022} showed it for other geometry including torus. Finally we note that $\gamma \to 0$ is equivalent to $c\to \infty$, the limit of interest in this paper.
% \rmkH{There are at least two more distinct probabilistic definitions of conformal blocks in GKRV and GRSSW. It is an interesting and important problem to reconcile these constructions. We leave such questions for a future work. One may however observe that the semi-classical limit is accessible in some constructions such as GRSS as opposed to the others. }

In a recent work, \cite[Theorem~1.1]{ghosal2020probabilistic} introduced a probabilistic description of $\mathcal{G}^{\alpha}_{\gamma, P}(q)$ using the GMC measure. This probabilistic definition of $\mathcal{G}^{\alpha}_{\gamma, P}(q)$ is shown to match with the formula as a power series in $q$ of the same (defined via Virasoro algebra in \cite{belavin1984infinite}) through rigorously proving the Zamolodchikov recursion (as proposed in \cite{zamolodchikov1984conformal,zamolodchikov1987conformal,hadasz2010recursive}).    
In fact, \cite{ghosal2020probabilistic} showed that $\mathcal{G}^{\alpha}_{\gamma, P}(q)$ is defined for $\alpha \in (-\frac{4}{\gamma}, Q)$, $P\in \mathbb{R}$ and is analytic as a function of $q \in \mathbb{C}$ in a ball $B_r(0)$ of radius $r$ greater than $\frac{1}{2}$ when $\alpha\in [0,Q)$. 
We refer the reader to \cite[Appendix A]{ghosal2020probabilistic} and \cite[Section 1.1]{ghosal2022analiticity} for a summary of the algebraic approach to conformal blocks on one-point torus.

% In previous works \cite{lacoin2019semiclassical}, the semi-classical limit of the $n$-point Liouville correlation function is shown to be accessory parameters of a second order Fuchsian differential equation where Heun equation (references) is one of such examples. On the other hand, similar dexcription of semi-classical limti of Liouville conformal blocks over torus is studied in the physics literature which relates it to the Lam\'e euqation which is shown below- 

In previous works \cite{lacoin2019semiclassical}, the semiclassical limit of the $n$-point Liouville correlation function has been shown to be governed by the accessory parameters of a second-order Fuchsian differential equation, with the Heun equation (see, e.g., references) being a notable example. On the other hand, a similar description of the semiclassical limit of Liouville conformal blocks on the torus has been explored in the physics literature, where it is related to the Lamé equation
\begin{align}\label{eq:Lame_intro}
\left(\partial_{z}^2 - m(m-1) \wp(z) + \mathcal{E} \right)\widetilde{\Gamma}(z;q) = 0,
\end{align}
where $\wp(z)$ is the Weierstrass elliptic function and $\mathcal{E}$ is called the \emph{accessory parameter}. This equation arises naturally in the study of Schr\"odinger equations with periodic potentials and has deep connections to uniformization of Riemann surfaces.

Our objective is to understand the behavior and properties of these Liouville conformal blocks $\mathcal{G}^{\alpha}_{\gamma, P}(q)$ in the semi-classical limit $\gamma\to 0$, shedding light on the underlying probabilistic structures associated with conformal field theories on tori and the possible connection of the conformal block with the Lam\'e equation. Our main result shows that the semi-classical limit of $\mathcal{G}^{\alpha}_{\gamma, P}(q)$ determines this accessory parameter. Such a proof for all other geometries remains an open question.

A crucial gap lies in establishing a systematic way of defining probabilistic conformal blocks for general Riemann surfaces and prescribed singularities; once this is achieved, our methods may be adapted to explore their semi-classical limits.

\begin{theorem}\label{thm:zamolodchikov}
The semi-classical limit of the conformal block $\mathcal{G}^{\alpha}_{\gamma, P}(q)$ exists and has the form
\begin{align}\label{zc:exp}
\lim_{\gamma \rightarrow 0} \gamma^2 \log \mathcal{G}_{\gamma, P}^{\alpha}(q) = \phi(q;\alpha_0, P_0) + \widetilde{f}(\alpha_0, P_0); \qquad \alpha = \frac{\alpha_0}{\gamma}, \qquad P = \frac{P_0}{\gamma},
\end{align}
for $q = e^{\ii \pi \tau} \in B_{r_0}(0)$ for some $r_0 > 0$, $\alpha_0 \in (-4,2)$, and $P_0 \in \mathbb{R}$. Here, $\phi(\cdot)$ is an analytic function in $q$ given by

\begin{align}\label{eq:phi_q_intro}
\phi(q;\alpha_0, P_0) &= -\frac{\alpha_0(\alpha_0+4)}{6} + \left(\frac{P_0^2}{2} + \frac{\alpha_0^2}{24}\right)\log q - 2\alpha_0 \log \eta(q) - \frac{\alpha_0^2}{2}\sum_{n,m \geq 1} \frac{q^{2nm}}{n} - \frac{\ii \pi \alpha_0^2}{2} \nonumber\\
&\quad + \alpha_0^2 \sum_{n=1}^{\infty} \frac{\sum_{m=1}^{\infty} q^{2nm}}{n(1 + 2\sum_{m=1}^{\infty} q^{2nm})} + \frac{1}{2}\mathbb{E}[\Phi^2] - \alpha_0 \log \widetilde{\Xi}_0 \nonumber\\
&\quad + \frac{\alpha_0}{2\widetilde{\Xi}_0}\left(\int_0^1 e^{\pi P_0 x} (2\sin(\pi x))^{-\frac{\alpha_0}{2}} h_q(x) e^{\frac{h_q(x)}{2}} dx\right),
\end{align}
where $\Phi$ is defined in \eqref{def:Phi-Gfield}, $\widetilde{\Xi}_0 = \int_0^1 e^{\pi P_0 x} (2\sin(\pi x))^{-\frac{\alpha_0}{2}} e^{\frac{1}{2}h_q(x)} dx$, and $h_q(x) := \lim_{\gamma \to 0} \gamma h_{\widetilde{\Psi}}(x)$ with $\widetilde{\Psi}$, $h_{\widetilde{\Psi}}(x)$ defined in \eqref{def:htilPsiXi}. The function $\widetilde{f}(\cdot)$ is a normalization constant such that

\begin{align*}
\lim_{q \to 0} \lim_{\gamma \to 0} \gamma^2 \log \mathcal{G}_{\gamma, P}^{\alpha}(q) = 0.
\end{align*}

Moreover, the function $\phi(q;\alpha_0, P_0)$ in \eqref{zc:exp} determines the accessory parameter of the Lam\'e equation
\begin{align}\label{eq:Lame_intro}
\left(\partial_{z}^2 - \frac{\alpha_0}{4} \left(\frac{\alpha_0}{4}-1\right) \wp(z) + \frac{\pi \ii}{2} \partial_\tau {\phi}(q;\alpha_0, P_0)\right) \widetilde{\Gamma}(z;\alpha_0, P_0, q) = 0,
\end{align}
where $\wp(\cdot)$ is the Weierstrass $\wp$-function (see \eqref{eq:wp}). Through the study of semi-classical conformal blocks (Theorem~\ref{prop:sc_Lame}), the function $\widetilde{\Gamma}(z;\alpha_0, P_0, q)$ arises as the limit
\begin{align}\label{eq:Gamma_limit_intro}
\lim_{\gamma \rightarrow 0} \psi^{\alpha_0/\gamma}_{\gamma/2, P_0/\gamma}(z,q) e^{-\phi(q;\alpha_0, P_0)/\gamma^2} = \widetilde{\Gamma}(z;\alpha_0, P_0, q),
\end{align}
and is given by
\begin{align}\label{thm43:phi_intro}
\widetilde{\Gamma}(z;\alpha_0, P_0, q) &= e^{P_0 z\pi/2} \theta_1(z)^{\alpha_0/4} e^{\ii \pi \alpha_0 z/4} \exp\left(\mathbb{E}\left[\Phi(q) M_{\gamma}(z,q)\right] + \frac{1}{2}\log\widetilde{\Xi}_0 - \alpha_0 \frac{\widetilde{\Xi}_1}{\widetilde{\Xi}_0}\right) \nonumber\\
&\quad \times \exp\left(\alpha_0 \sum_{n=1}^{\infty} \frac{q^n}{(1-q^{2n})n} \left(\frac{\sum_{m=1}^{\infty} q^{2nm}}{1 + 2\sum_{m=1}^{\infty} q^{2nm}}\right) \cos\left(2\pi\left(z - \frac{\tau}{2}\right)n\right)\right) \nonumber\\
&\quad \times \exp\left(\frac{\alpha_0}{2\widetilde{\Xi}_0} \int_0^1 e^{\pi P_0 x} (2\sin(\pi x))^{-\frac{\alpha_0}{2}} \left(\frac{3}{2}\mathfrak{R}_{z,q}(x) - \frac{\widetilde{\Xi}_1}{\widetilde{\Xi}_0} h_q(x)\right) e^{\frac{h_q(x)}{2}} dx\right),
\end{align}
where $\widetilde{\Xi}_1 = \int_0^1 e^{\pi P_0 x} (2\sin(\pi x))^{-\frac{\alpha_0}{2}} \mathfrak{R}_{z,q}(x) e^{\frac{1}{2}h_q(x)} dx$, $\Phi$ and $M_{\gamma}(z,q)$ are defined in \eqref{def:Phi-Gfield} and \eqref{def:Mg-Gfield} respectively, and $\mathfrak{R}_{z,q}(x)$ is the unique solution to the fixed point equation
\begin{align}\label{eq:R_fixed_point_intro}
\mathfrak{R}_{z,q}(x) = \mathfrak{R}^{(0)}_{z,q}(x) - \frac{\alpha_0}{2\Xi} \int_0^1 K(x,y) \mathfrak{R}_{z,q}(y) \, d\mu_{\Psi}(y),
\end{align}
with $K(x,y) := \sum_{n=1}^{\infty}\frac{2}{n}(1+2\sum_{m=1}^\infty q^{2nm})\cos(2\pi n(x-y))$ and
\begin{align}\label{eq:R0_def_intro}
\mathfrak{R}^{(0)}_{z,q}(x) := \sum_{n=1}^{\infty} \frac{1}{n} \left(\frac{\sqrt{2} q^n \sum_{m=1}^{\infty} q^{2nm}}{(1-q^{2n})} - \frac{q^n(1 + 2\sum_{m=1}^{\infty} q^{2nm})}{1-q^{2n}}\right) \cos\left(2\pi n\left(z - \frac{\tau}{2} - x\right)\right).
\end{align}
Furthermore, the accessory parameter $\partial_\tau \phi(\alpha_0, P_0, q)$ is related to the Hamiltonian of the non-autonomous elliptic Calogero-Moser model (Definition~\ref{remark:HJ_phitil}) as follows. Let $u(\tau)$ be the solution of the corresponding equations of motion for the Calogero-Moser model \eqref{def:uv} with $m^2 = (2-\alpha_0)^2$, and let $\tau_{\star}$ be a zero of $u(\tau)$, i.e., $u(\tau_{\star}) = 0$. Then
\begin{align}\label{eq:acc_intro}
\frac{\pi \ii}{2} \partial_{\tau} \phi(q; \alpha_0 - 2, P_0)\Big|_{\tau = \tau_{\star}} = -\frac{H(\tau_{\star})}{4}, \qquad q_{\star} = e^{\ii \pi \tau_{\star}},
\end{align}
where $H(\tau_{\star}) \equiv H(\tau_{\star}, \alpha_0, P_0)$ is the Hamiltonian \eqref{HAM:4} evaluated at $\tau = \tau_{\star}$.
\end{theorem}

\begin{remark}
The limit in \eqref{zc:exp} can be shown for $q \in B_r(0)$ for some $r > \frac{1}{2}$ when $\alpha_0 \in [0,2)$. Furthermore, since $\widetilde{\Gamma}(z;\alpha_0, P_0, q)$ has the explicit expression \eqref{thm43:phi_intro}, one can use the Lam\'e equation \eqref{eq:Lame_intro} to write the explicit form of the accessory parameter $\frac{\ii \pi}{2} \partial_{\tau} \phi(\alpha_0, P_0, q)$ (see Remark~\ref{corr:acc-par-alt}). The analytic properties with respect to $\alpha_0$, $P_0$, $q$ of the semi-classical limit are accessible through this explicit form.
\end{remark}

\subsection{Zamolodchikov conjecture}
   The main result of this paper rigorously proves the Zamolodchikov conjecture for the case of conformal blocks on a one-point torus. A long standing conjecture formulated by Zamolodchikov in 1986 \cite[(2.30)-(2.32)]{zamolodchikov1986two}, \cite[Conjectures A, B]{lisovyy2021accessory}  postulates that, in the semi-classical limit ($c\to\infty$), the conformal blocks of a four-point sphere have an exponential structure. Furthermore, it is known that, assuming such an exponential structure, the semi-classical limit of the sphere 4-point conformal block is described by the accessory parameter of the Heun's equation.

 There are three different heuristics that support the exponential structure of the Liouville conformal blocks in the semiclassical regime: the analysis of the power series representation of conformal blocks \cite{zamolodchikov1987conformal}, the oscillator formalism of Virasoro algebra \cite{becsken2020semi}, and the combinatorial expression in terms of the Nekrasov-Okounkov functions describing the partition functions of supersymmetric gauge theories \cite{Nekrasov:2003rj} given by the Alday-Gaiotto-Tachikawa (AGT) correspondence \cite{Alday:2009aq} (see \cite{nekrasov2018quantum,nekrasov2010quantization} for the results on the semi-classical limit).

   The semi-classical limit of conformal blocks carries significant importance in mathematical physics as it plays a crucial role in a variety of areas from black hole physics \cite{aminov2022black,da2022expansions} and supersymmetric gauge theories \cite{nekrasov2010quantization} (see also the introduction in \cite{pikatek2022classical}) to Teichm\"uller theory (see \cite{amado2017kerr, takhtajan1994topics} and references therein). 
% Such a statement generalizes to other geometries and singularity structures. 

Of particular interest to this paper is the connection between semi-classical conformal blocks and the Heun equation. The Heun equation is a second-order differential equation that describes isomonodromic deformations on a Riemann sphere with four simple poles and features a free parameter known as the {\it accessory parameter}. This parameter is crucial in the semi-classical analysis as it is characterized by conformal blocks on the sphere \cite{zamolodchikov1986two}. The study of the Heun equation is interesting on its own, and has deep ties to black hole physics and hyperbolic geometry among other applications. See \cite{ronveaux1995heun} for a bibliographical account and \cite{hortaccsu2018heun} for a review on applications to physics. 
% Additionally, the study of the Heun equation is intriguing in its own right. 

In the recent years, interest has grown in understanding the relationships between conformal blocks, accessory parameters, and Painlev\'e equations \cite{BGG2021, lisovyy2021accessory,piatek2019solving,lencses2018classical}.  Specifically, the semi-classical limit of conformal blocks is related to the Painlev\'e VI equation; the semi-classical block on a four-point sphere turns out to be the solution of certain Hamilton-Jacobi equation associated to the Painlev\'e VI equation \cite{litvinov2014classical,kashani2013transformations}. In turn the Hamiltonian of the Painlev\'e VI equation is given by $c=1$ conformal blocks on sphere, and hence, the aforementioned relation provides a non-trivial link between the $c\to \infty$ and $c=1$ conformal blocks. 

More recently, a relation between $c\to \infty$ conformal blocks and the connection problem for the degenerations of the Heun equation was conjectured in \cite{bonelli2023irregular} and is dubbed 'the Trieste formula' \cite{lisovyy2022perturbative}, with the case of the Heun equation (without degenerations) verified in \cite{lisovyy2022perturbative}. 
So, the proof of Zamolodchikov conjecture paves the way for a rigorous treatment of the aforementioned problems. Indeed most of the above claims could be easily adapted to the one-point torus case. In particular, the role of the Heun equation is now played by the Lam\'e equation and furthermore, its accessory parameter, under the assumption that the Zamolodchikov conjecture holds, is known to be related to the semi-classical conformal block on the torus \cite{piatek2014classical} and to the elliptic form of the Painlev\'e VI equation with special value of the parameters \cite{BGG2021}. In this context, we also want to highlight a similar study of the semi-classical limit of Liouville CFT on a sphere  \cite{lacoin2019semiclassical}, where they prove the Takhtajan-Zograf theorem relating the classical Liouville action at special points and accessory parameters of a second order Fuchsian equation of holomorphic functions on the unit disk \cite{takhtajan2003hyperbolic}.

    In Theorem~\ref{thm:zamolodchikov}, we showed that, as the central charge $c$ goes to $\infty$ or equivalently $\gamma\to 0$, the torus one point Liouville conformal block behaves as $\mathcal{G}_{\gamma, P}^{\alpha}(q)\sim \exp(\phi(\alpha_0, P_0, q)/\gamma^2)$ where $\phi(\alpha_0, P_0, q)$ is the semi-classical conformal block. Moreover, \eqref{eq:Lame_intro} showed that $\phi(\alpha_0, P_0, q)$ is related to the accessory parameter of Lam\'e equation. As a result, we provide a rigorous proof of the analogue of Zamolodchikov conjecture for one-point torus conformal blocks. Consequently, we prove that $\phi(\ldots)$ is also related to the Hamiltonian of the non-autonomous elliptic Calogero-Moser (NAECM) model (see Theorem~\ref{prop:sc_Lame}) for special values of the solution \cite{BGG2021}, which in turn is related to $c=1$ conformal blocks on a one-point torus \cite{DDG2020}. Such a proof for all other geometries remains an open question. %Starting from the probabilistic definition of conformal blocks on the four-point sphere \cite{ghosal2022analiticity}, one could adapt our methods to study their semi-classical limits.

\subsection{Ideas of the proofs}
\label{sec:ideas_of_the_proofs}

Our main goal in this paper is to prove the existence and estimate the radius of convergence for the semi-classical limit ($\gamma \to 0$) of the conformal block $\mathcal{G}_{\gamma, P}^{\alpha}(q)$. Henceforth we refer to this as the {\it undeformed} conformal block. We will soon introduce a {\it deformed} conformal block which, roughly speaking, will involve a shift of the variable $x \to x+z$. The proof of the statement of existence involves several steps. One of the main challenges is to show that the conformal block has exponential structure in the semi-classical limit, i.e., as $\gamma \to 0$, the torus 1-point conformal block $\mathcal{G}^{\alpha}_{\gamma, P}$ behaves as $\exp\big(-\frac{1}{\gamma^2} (\phi(\alpha_0, P_0, q)+ \tilde{f}(\alpha_0, P_0))\big)$. For showing this, we use the probabilistic description of $\mathcal{G}^{\alpha}_{\gamma, P}$ given as
\begin{align}\label{eq:def-CBTcasual}
\mathcal{G}^{\alpha}_{\gamma, P} \sim \EE\left[\left({\sf IGMC}^{\alpha}_{\gamma, P}(q)\right)^{-\frac{\alpha}{\gamma}} \right], \qquad {\sf IGMC}^{\alpha}_{\gamma, P}(q) := \int^{1}_0\theta_{1}(x)^{-\frac{\alpha\gamma}{2}} e^{\pi \gamma P x} e^{\frac{\gamma}{2} :Y(x;q):} dx
\end{align}
where $:Y(x;q): \,\, = Y(x;q) -\frac{\gamma}{4} \mathbb{E}\left[Y(x;q)^2 \right]$ is a renormalized log-correlated field on $[0,1]$ (Wick notation), and $e^{\frac{\gamma}{2} :Y(x;q):} dx$ is the Gaussian Multiplicative Chaos (GMC) measure (see also Definition \ref{DEF:GMC measure}).

The above representation of the conformal block was introduced by \cite{ghosal2020probabilistic} and expressed in terms of the Gaussian multiplicative chaos measure constructed using the log-correlated field $Y(x;q)$ defined for $x \in [0,1]$. See Definition~\ref{def:CBT} for further discussion. With the choice of scaling $P = \frac{P_0}{\gamma}$ and $\alpha = \frac{\alpha_0}{\gamma}$ for $P_0 \in \mathbb{R}$ and $\alpha_0 \in \mathbb{R}_{>0}$, the first difficulty is to show that $\gamma^2 \log \mathcal{G}^{\alpha}_{\gamma, P}$ does not blow up to infinity. We overcome this difficulty through a combination of H\"older's inequality, Selberg integral asymptotics, and explicit negative moment bounds for GMC integrals.

We first show that proving the boundedness of the conformal block is equivalent to proving the boundedness of the expression on the LHS of the expression below. For the upper bound when $\alpha_0 \in [0,2)$, we use a comparison argument: since the covariance $\mathbb{E}[Y(x;q)Y(y;q)] \geq \mathbb{E}[Y(x)Y(y)] - 2\log(\prod_{k\geq 1}(1+q^k)^2)$, Slepian's inequality reduces the problem to bounding negative moments of the $q$-independent GMC integral $\mathcal{I}_{\gamma} := \int_0^1 (2\sin(\pi x))^{-\frac{\alpha_0}{2}} e^{\pi P_0 x} e^{\frac{\gamma}{2} Y(x)} dx$. These negative moments are controlled using Lemma~\ref{lem:negative_moments}, which provides bounds via the exact formulae of the negative moments of GMC integral $\mathcal{I}_{\gamma}$ from \cite{ang2023derivation}. For $\alpha_0 \in (-4,0)$, we instead bound positive moments using Selberg integral asymptotics: the $N$-th moment of the GMC integral can be written as
\begin{align}\label{eq:selberg_idea}
\mathbb{E}\left[\mathcal{I}_{\gamma}^N\right] = \int_{[0,1]^N} \prod_{i<j} |\sin(\pi(x_i-x_j))|^{-\frac{\gamma^2}{2}} \prod_{i=1}^N (\sin(\pi x_i))^{-\frac{\alpha_0}{2}} e^{\pi P_0 x_i} dx_i,
\end{align}
which has an explicit evaluation in terms of Gamma functions. Taking $N \sim |\alpha_0|/\gamma^2$ and applying Stirling's approximation shows that these moments grow at most as $\exp(C/\gamma^2)$. These bounds establish that the expression $\gamma^2 \log \mathcal{G}^{\alpha}_{\gamma, P}$ remains bounded as $\gamma \to 0$, proving the tightness of the sequence.

Once the tightness is established, it remains to show that all subsequential limits of $\gamma^2 \log \mathcal{G}^{\alpha}_{\gamma, P}$ are the same. To this end, we use the deformation of the conformal block $\mathcal{G}^{\alpha}_{\gamma, P}$ which we denote as $\psi^{\alpha}_{\chi, P}(z,\tau)$ (see Definition~\ref{def:u-block}) for $\chi \in \{\frac{\gamma}{2}, \frac{2}{\gamma}\}$, $z \in (0,1)$, and $\alpha, P, \tau$ are defined as before. The deformed conformal block on the torus is uniquely determined as the solution of a differential equation named the Belavin-Polyakov-Zamolodchikov (BPZ) equation,
\[\Big(\partial_{z}^2 - l_\chi(l_\chi + 1) \wp(z) + 2\ii \pi \chi^2 \partial_\tau\Big) \psi_{\chi, P}^{\alpha}(z, \tau) = 0, \qquad \ell_{\chi} = \frac{\chi^2}{2} - \frac{\alpha\chi}{2}\]
which is defined on the punctured torus.

% The key technical innovation in our approach is the use of Girsanov's theorem (Cameron-Martin formula) to perform a detailed asymptotic expansion for $\gamma\to 0$ of the deformed conformal block. Through the Girsanov representation, we rewrite the conformal block so that the GMC measure becomes $q$-independent, introducing auxiliary Gaussian variables $T_n^{(1)}, T_n^{(2)}$ that capture the $q$-dependence. After integrating out these auxiliary variables using Gaussian identities of the form $\mathbb{E}[e^{aT-T^2}] = e^{a^2\sigma^2/(2(1+2\sigma^2))}/\sqrt{1+2\sigma^2}$, we introduce a tilted Gaussian variable $\Psi$ and apply the Cameron-Martin theorem:
% \begin{align}\label{eq:cameron_martin_idea}
% \mathbb{E}\left[e^{\Psi} f(Y)\right] = e^{\frac{1}{2}\mathrm{Var}(\Psi)} \mathbb{E}\left[f(Y + h_{\Psi})\right],
% \end{align}
% where $h_{\Psi}(x) := \mathbb{E}[\Psi \cdot Y(x)]$ represents the Cameron-Martin shift. This transformation allows us to extract the leading-order $O(\gamma^{-2})$ terms explicitly while controlling the remainder through Proposition~\ref{prop:semiclassical_limit}, which establishes the convergence of the shifted GMC expectation to a quadratic chaos expression.

The key technical innovation in our approach is the use of Girsanov's theorem (Cameron-Martin formula) to perform a detailed asymptotic expansion for $\gamma\to 0$ of the deformed conformal block. Through the Girsanov representation, we factor the $z$-dependence from the $\gamma$-dependence, introducing auxiliary Gaussian variables that couple the $q$-dependent terms with the GMC integral. Specifically, we introduce a tilted Gaussian variable $\Psi$ and apply the Cameron-Martin theorem:
\begin{align}\label{eq:cameron_martin_idea}
\mathbb{E}\left[e^{\Psi} f(Y)\right] = e^{\frac{1}{2}\mathrm{Var}(\Psi)} \mathbb{E}\left[f(Y + h_{\Psi})\right],
\end{align}
where $h_{\Psi}(x) := \mathbb{E}[\Psi \cdot Y(x)]$ represents the Cameron-Martin shift. This transformation allows us to extract the leading-order $O(\gamma^{-2})$ terms explicitly while controlling the remainder through Proposition~\ref{prop:semiclassical_limit}, which establishes the convergence of the shifted GMC expectation to a quadratic chaos expression.

The convergence result in Proposition~\ref{prop:semiclassical_limit}, established in Section~\ref{sec:dgmc}, is central to completing the asymptotic analysis. This proposition shows that the GMC integral $(\int_0^1 :e^{\frac{\gamma}{2} Y(x)}: d\mu)^{-\alpha_0/\gamma^2}$, when appropriately renormalized, converges to an expression involving the quadratic Wick-ordered field $:Y(x)^2:$. Our approach is similar in spirit to that of \cite{lacoin2019semiclassical}, but with a key technical differences. Rather than white noise decomposition, we work directly with the spectral (Fourier) structure of the log-correlated field $Y(x)$, employing the explicit negative moment formula from \cite{ang2023derivation} and Abel summation techniques to control spectral tails. These modifications are naturally suited to the conformal block setting where the field $Y(x)$ has an intrinsic Fourier decomposition.

On the one hand, for $\chi = \frac{\gamma}{2}$, we show in Theorem~\ref{prop:sc_Lame} that the deformed conformal block $\psi_{\chi, P}^{\alpha}(z, \tau)$ factorizes as $\widetilde{\Gamma}(z;\alpha_0, P_0, q) \times e^{\phi(q;\alpha_0,P_0)/\gamma^2}$ as $\gamma \to 0$. The function $\widetilde{\Gamma}$ emerges from a refined asymptotic expansion where the Cameron-Martin shift in this case is denoted by $h_{\widetilde{\Psi}}(x)$ and has the expansion $h_{\widetilde{\Psi}}(x) = h_q(x)/\gamma + \gamma \mathfrak{R}_{z,q}(x) + O(\gamma^2)$, where the function $\mathfrak{R}_{z,q}(x)$ satisfies a fixed point equation arising from the $O(\gamma)$ correction. Similar decomposition of deformed conformal block is a direct consequence of the value of $\chi = \frac{\gamma}{2}$ which is called light insertion in the literature, and was previously observed in \cite{litvinov2014classical,piatek2014classical} for the case of the torus. With the scaling $\alpha = \frac{\alpha_0}{\gamma}$ and $P = \frac{P_0}{\gamma}$, substituting the expression $\psi_{\gamma/2, P}^{\alpha}(z, \tau)\sim \widetilde{\Gamma}(z;\alpha_0, P_0, q) e^{\phi(q;\alpha_0, P_0)/\gamma^2}$ into the BPZ equation above and studying the limit $\gamma \to 0$ gives us the Lam\'e equation as shown in \eqref{eq:Lame_intro}. Since $\widetilde{\Gamma}(z;\alpha_0, P_0, q)$ is explicit (see \eqref{Gamma:Lame}), the accessory parameter of \eqref{eq:Lame_intro} which could be the $\tau$-derivative of any subsequential limit of $\gamma^2 \log \mathcal{G}^{\alpha}_{\gamma, P}$ is now uniquely identified, which in turn shows the existence of the limit of the sequence $\gamma^2 \log \mathcal{G}^{\alpha}_{\gamma, P}$.

On the other hand, for $\chi = \frac{2}{\gamma}$, following a similar strategy as above, the BPZ equation in the $\gamma\to 0$ limit in this case becomes the Hamilton-Jacobi equation for the (non-autonomous) elliptic Calogero-Moser model (NAECM). The asymptotic expansion via the Cameron-Martin approach yields an explicit expression for $\widetilde{\phi}(z,q)$ 
% in terms of the function $\xi(z,q)$ (which is explicit in $z,q$) 
% plus contributions from $\Xi$, $h_\Psi$, and $\mathbb{E}[\Psi^2]$.
plus contributions from the Cameron-Martin shift $h_\Psi$, and $\mathbb{E}[\Psi^2]$ for a Gaussian variable  $\Psi$.
The tightness in $C^2([M,1-M] \times [\epsilon, q_0])$ for $0<M< 1/2$, $0<\epsilon< q_0\in (0,1)$ established in Proposition~\ref{prop:limit_commutativity}, combined with the commutativity of limits $\gamma \to 0$ and $q \to q_0$, allows us to demonstrate in Theorem~\ref{prop:sem_HJ} that the semi-classical limit of $\psi^{\alpha}_{2/\gamma, P}$ solves the aforementioned Hamilton-Jacobi equation. See Definition~\ref{remark:HJ_phitil} for a discussion. This connection between the deformed conformal block and Calogero-Moser model is used later in Theorem~\ref{prop:sc_Lame} to show that the accessory parameter of the Lam\'e equation \eqref{eq:Lame_intro} is in fact related to the Hamiltonian of the NAECM model computed at some special points, a statement that is known in the literature. See \cite{BGG2021} and references therein.

\subsection{Outline of the paper}

This paper is organized as follows.
In Section \ref{sec:pcb}, we review the probabilistic construction of the {\it undeformed} conformal block $\,\mathcal{G}_{\gamma, P}^{\alpha}(q)$, and the {\it deformed} conformal block which satisfies the PDE known as the Belavin-Polyakov-Zamolodchikov (BPZ) equation.

%In Section \ref{sec:dgmc}, we introduce the main tool to analyse the semi-classical limit of the conformal blocks, namely, the DGMC measure, which is the $\gamma$ derivative of the GMC measure. The main result of Section~\ref{sec:dgmc} is Theorem~\ref{prop:MyProp} which proves the tail estimates of the DGMC measure.4

In Section~\ref{sec:dgmc}, we establish the key analytical result underlying the semi-classical analysis: the convergence of GMC integrals raised to large negative powers toward quadratic chaos expressions. The main result is Proposition~\ref{prop:semiclassical_limit}, which shows that for a probability measure $\mu$ with density bounded below, the appropriately renormalized GMC integral converges as $\gamma \to 0$ to a finite expectation involving the Wick-ordered quadratic field $:Y(x)^2:$. The proof relies on three supporting lemmas: Lemma~\ref{lem:uniform_integrability} establishes uniform integrability using spectral truncation and Janson's hypercontractive estimates, Lemma~\ref{lem:negative_moments} provides negative moment bounds via the explicit formula from \cite{ang2023derivation}, and Lemma~\ref{lem:approximation_error} controls the approximation error between truncated and full GMC integrals using Girsanov arguments with Abel summation techniques.
%\rmkH{We need to modify the paragraph above to remove the DGMC part}

In Section \ref{sec:deformed_semi-classical_limit}, we formulate the semi-classical limit of the {\it deformed} 
% and the deformed 
conformal block. In Subsection~\ref{subsec:chi2g} we study the semi-classical limit of the deformed conformal block for $\chi = \frac{2}{\gamma}$. The main result of this Subsection is Theorem~\ref{prop:sem_HJ} which establishes an explicit description of such a semi-classical limit. In Subsection~\ref{subsec:scg2} we analyse the semi-classical limit of the deformed conformal block for $\chi = \frac{\gamma}{2}$. The main result of this Subsection is Theorem~\ref{prop:sc_Lame} which gives explicit description of that semi-classical limit in terms of the solution of Lam\'e equation. As a consequence of this theorem, we obtain explicit expressions for the solution and accessory parameter of the Lam\'e equation. 

We prove Theorem~\ref{thm:zamolodchikov} in Section~\ref{sec:semi-classical_limit}. We further prove that the radius of convergence of the semi-classical conformal block is positive in Theorem~\ref{thm:radius}.

Finally, in Section \ref{sec:asymp}, we study the asymptotic behaviour of the semi-classical limit of the conformal blocks for $\tau\to \ii \infty$. The key component of the analysis is the expression of the normalization constant of the undeformed conformal block in terms of double Gamma functions. Results in this section are used to prove results in Section \ref{sec:deformed_semi-classical_limit}.

\subsection{Outlook}
% {\color{blue} special solutions of Lam\'e, recursion relation for semiclassical deformed CB, what happens with the scaling $q/\gamma$.}
% This paper offers a blue print for proving the Zamolodchikov conjecture starting from the probabilistic construction of the conformal blocks.
% There are several exciting generalizations to other geometries and central charges. Let us focus on two specific problems.

 Our paper presents two key takeaways. Firstly, the methods we introduced in proving the Zamolodchikov conjecture offer a road map for studying semi-classical conformal blocks on other Riemann surfaces and for the analysis of conformal blocks with different central charges\ conditional on the probabilistic formulation of conformal blocks. Secondly, the explicit formulas for the conformal blocks and the Lam\'e equation obtained here provide profound insights into the analytic and geometric aspects of conformal field theories (CFTs) and the associated Heun-type equations. Let us now describe some open questions of immediate interest.
 
\begin{itemize}[leftmargin=0.3cm]
\item A natural question that arises is the generalization of our techniques to other geometries such as the four-point sphere and higher genus surfaces. The key to achieving this would be to develop a systematic approach to defining probabilistic conformal blocks on Riemann surfaces. The formulation of probabilistic conformal blocks on a four point sphere is expected to have a slightly varied construction \cite{ghosal2022analiticity}, and the irregular conformal blocks can then be constructed through appropriate degeneration limits. The proof of the Zamolodchikov conjecture in the aforementioned cases is expected to follow a similar formalism to that presented here. 

\item Semi-classical limit of Liouville CFT corresponds to taking the central charge $c \to \infty$. It is natural to ask what happens when $c\to -\infty$. This question is related to the semi-classical limit of CFTs with $c\leq 1$ (i.e., $\gamma$ purely imaginary). Recently, in a series of works \cite{Wang1, Wang2, Wang3}, Yilin Wang has investigated the large deviation problem of Schramm-Loewner evolution ($\mathrm{SLE}_{\kappa}$) as $\kappa\to 0_+$ and has shown that the large deviation rate function is equal to the Loewner energy function. Since $\kappa\to 0_+$ corresponds to talking $c\to -\infty$ and SLE is related to the correlation function of certain CFT, the results of \cite{Wang1, Wang2, Wang3} describes semi-classical limit for $c\to -\infty$. In a more recent work \cite{Pel&Wang}, Peltola and Wang characterized minimizers of the Loewner potential using tools from real algebraic geometry which provides unique insight about the CFT when the central charges diminishes to $-\infty$. Moreover, the minimum value of such Loewener energy is shown to solve a PDE which looks similar to Hamilton-Jacobi limit type equation satisfied by the semi-classical limit of deformed conformal block with degenerate weight of $\frac{2}{\gamma}$ in \eqref{eq:Hamilton-jacobi} of Theorem~\ref{prop:sem_HJ}.

\end{itemize}

% PCBs, Blow up relations, Monodromy dependence, series and combinatorial representation
 
% There are several open questions that can be addressed as a next step. Some of them include the gauge theory implications in terms of the Yang-Yang functional, and a rigorous derivation of the so-called blow-up relations that associate conformal blocks for the central charges $c=1$ and $c\to \infty$. Let us now expand on two important problems we plan to address in the near future.
\begin{itemize}[leftmargin=0.3cm]
% \item {\bf Generalization:}
% A natural extension of the results presented here involves conformal blocks for other cases including the four-point sphere and higher genus surfaces. The key advance to achieving this would be a systematic approach to defining probabilistic conformal blocks on Riemann surfaces.
\item The Trieste formula \cite[(4.1.16) - (4.1.17)]{bonelli2023irregular}, \cite[eq. 2.20b]{lisovyy2022perturbative}  describes the connection formulae of the Heun's equation in terms of the semi-classical limit of the fusion kernel of the 4-point sphere conformal blocks.
Our description of the semi-classical conformal block  may enable the proof of an analogue to the Trieste formula connecting the semi-classical limit of the modular kernel (see \cite[(5.3)]{ghosal2022analiticity}) of one-point torus conformal block with the connection formulae of the Lam\'e equation. 
A probabilistic description of the four-point conformal block and its degenerations could prove the Trieste formula in full generality.

\item   
In \cite{lisovyy2021accessory}, the authors found a series representation of the accessory parameter of the Heun's equation where the coefficients satisfy some recurrence relation. This is reminiscent of the Zamolodchkiov recurrence relation satisfied by the coefficients of the conformal blocks. Following a similar procedure as in \cite{lisovyy2021accessory}, one can also  derive a series representation of the accessory parameter of Lam\'e equation  once the monodromy data associated to $\alpha_0, P_0$ is determined. Furthermore, the accessory parameter is also explicitly given in terms of the Cameron-Martin shift and known elliptic functions in  \eqref{def:phiq}, and in terms of integrals \eqref{corr:acc3}. Such formulae would allow to study its asymptotic properties \cite{Novokshenov,XXZ}.

\item Through the AGT correspondence, the semi-classical conformal block on a one-point torus determines the Nekrasov-Shatashvili (NS) limit of the corresponding $\mathcal{N} = 2^*$ gauge theory, or equivalently, the effective twisted superpotential \cite{teschner2011quantization}. The explicit form
% and specifically the dependence on the monodromy data, 
of the classical conformal block in terms of the Cameron-Martin shift and known elliptic function (see Theorem \ref{thm:ex_uni}) presented in this paper would enable a rigorous proof of the Nekrasov-Rosly-Shatashvili conjecture that relates the effective twisted superpotential of the gauge theory to the so-called Yang-Yang function of the associated quantum Hitchin system \cite{nekrasov2009supersymmetric,nekrasov2011darboux,hollands2018higher}.

\item The conformal blocks for arbitrary central charge solve certain bilinear relations called the Nakajima-Yoshioka blow-up relations \cite{nakajima2003lectures}. Such relations for the one-point torus have been written in \cite{gu2020elliptic,BGG2021} and it will be interesting to see if such relations could be proved using probabilistic conformal blocks. An interesting property of the blow-up relation is that it connects the semi-classical conformal blocks with the conformal block at central charge $c=1$ \cite{nekrasov2010quantization}, in turn relating Painlev\'e equations to Heun equations for the sphere case and the NAECM model to Lam\'e equation (see also \cite{BGG2021}) for the torus case. We show the latter relation in Theorem~\ref{prop:sc_Lame} based on probabilistic ideas. Showing the full blow-up relation is beyond the scope of this paper and we leave this for a future work. 

\end{itemize}

% and they thank the Møller institute and organisers of the program ”Applicable resurgent asymptotics: towards a universal theory” for their hospitality. HD thanks  her INI visit.

% \section{Setup}
% 
\section*{Acknowledgements}
We thank Alba Grassi, Aleksandra Korzhenkova, Oleg Lisovyi, Alex Moll, Eveliina Peltola, Vincent Vargas for informative discussions. We also thank the anonymous referees whose detailed feedback has significantly helped to improve and rectify many details of the paper. Part of this work was done during the authors' residence at the SLMath (formerly, Mathematical Sciences Research Institute) in Berkeley, California, during the program 'Universality and Integrability in Random Matrix Theory and Interacting Particle Systems'. They acknowledge the support of the National Science Foundation Grant No. DMS-1928930. H.D. also acknowledges the support of the Okinawa Institute of Science and Technology (OIST), where she completed part of this work during her stay for the Theoretical Sciences Visiting Program (TSVP) "Exact Asymptotics: From Fluid Dynamics to Quantum Geometry."

The work of H.D. is partly supported by Australian Research Council Discovery Project \#DP200100210, INI-Simons fellowship, SMRI postdoctoral fellowship, and Marie Skłodowska-Curie Postdoctoral Fellowship 101203697. % 
The work of P.G. is partially supported by the NSF grant DMS-2346685.
A.P. is supported NSF MSPRF grant DMS-2103354, and RSF grant 22-11-00070.
\section{Probabilistic conformal blocks}\label{sec:pcb}
In this section we recall the construction of the conformal block in \cite{ghosal2020probabilistic}. We begin by defining Gaussian Multiplicative Chaos (GMC) measures in terms of log-correlated fields, which will then be used to define the {\it probabilistic} one point conformal block on the torus. We then introduce the notion of a {\it deformed} conformal block which solves the PDE known as the Belavin-Polyakov-Zamolodchikov (BPZ) equation on the torus with one puncture.

\subsection{Log-correlated fields and Gaussian Multiplicative Chaos (GMC) measures}\label{sec:log-corr}

For $x\in [0,1]$, the random fields $Y_{N}(x;q)$, $Y_{ N}(x)$, with $N\in \mathbb{N}$ are defined as
\begin{align}
\widetilde{Y}_{N}(x) & =\sum_{n = 1}^N \sqrt{\frac{2}{n}} \Big( a_n \cos(2 \pi n x) + b_n \sin(2 \pi n x) \Big),\label{def:YinfN}\\
\widetilde{Y}_{N}(x;q) &= \widetilde{Y}_{N}(x) + \underbrace{2\sum_{n, m = 1}^{\infty} \frac{q^{nm}}{\sqrt{n}} \Big(a_{n, m} \cos(2\pi n x) + b_{n,m}\sin(2\pi nx)\Big)}_{=: F(x;q)}, \label{def:YtauN}
% &= \lim_{q\rightarrow 0} Y_{\tau, N}(x) 
\end{align}
%\rmkH{Notational inconsistency: $\alpha_n$ vs $a_n$. $a_n$, $b_n$ maybe a better choice because $\alpha_0$ is the monodromy parameter.}
where $\{a_n\}_{n \geq 1}$, $\{b_n\}_{n \geq 1}$, $\{a_{n, m}\}_{n, m \geq 1}$, $\{b_{n, m}\}_{n, m \geq 1}$ are sequences of i.i.d standard real Gaussian random variables, the modular parameter $\tau \in \mathbb{H}$, and $q = e^{i \pi \tau}$.  
In what follows we use the notation\footnote{Throughout this paper, we also use the dependence on $q$ and $\tau$ interchangeably.}
\begin{align}
   Y(x) := \lim_{N\to \infty} \widetilde{Y}_{N}(x),  &&  Y(x;q) := \lim_{N\to \infty} \widetilde{Y}_{ N}(x;q).   \label{def:Yinf_and_Ytau}
\end{align}
Indeed, the above limit exists. The first limit exist in the sense of weak convergence in the Sobolev space $H^{-s}([0,1])$ and the second limit exists  to due to the first limit and the almost sure convergence of infinite series of $F(x;q)$ due to the term $q^{nm}$. 
The resulting random fields $ Y(\cdot;q)$ and $ Y(\cdot)$ are log-correlated, and
% The expression \eqref{def:YinfN} is interpreted  whereas \eqref{def:YtauN} is understood to be a pointwise sum.
% a fact that will be useful to study several properties of the above fields. 
the covariance kernels of \eqref{def:Yinf_and_Ytau} for $\tau \in \ii\mathbb{R}_{>0}$ are given as follows.  
\begin{itemize}[leftmargin= 0.3cm]
% as
% \begin{gather}
%     \lim_{\tau \rightarrow \ii \infty} Y_{\tau, N} = Y_{\infty, N}.
% \end{gather}
\item The field $Y(x)$ is log-correlated on $[0,1]$: 
% {\color{red} subtelty with $Y_{\infty, N}$ vs $Y_{\infty}$}
 \begin{align}\label{corr:Yinf}
     \mathbb{E} [Y(x) Y(y) ] = - 2 \log \vert e^{2 \ii \pi  x} - e^{2 \ii \pi y} \vert , 
 \end{align}
 which is the usual expression $\log|x-y|$ in cylindrical coordinates obtained by the transformation $z\rightarrow e^{2\pi i z}$.
\item The $q$-dependent field $Y(x;q)$ is log correlated on the torus $\mathbb{C}/(\mathbb{Z} + \tau\mathbb{Z})$:
 \begin{align}\label{corr:Ytau}
     \mathbb{E} [Y(x;q) Y(y;q) ] = - 2 \log |\theta_1(x - y)| + 2 \log|q^{1/6} \eta(q)|.
 \end{align}
% Decomposition $Y(x; q) = Y(x) + F_{\tau}(x)$.
% \vskip0.3cm
% For $\beta_{n,m},$ $\widetilde \beta_{n,m}$ i.i.d. $\mathcal{N}(0,1)$,
% $$ F_{\tau}(x)=  2 \sum_{n,m \geq 1} \frac{q^{nm}}{\sqrt{n}} \Big( \beta_{n,m} \cos(2 \pi n x) + \widetilde \beta_{n,m} \sin(2 \pi n x) \Big).$$
% \vskip0.3cm
% with the elliptic functions $\theta_1(x)$, $\eta(q)$ defined as in \eqref{def:elltheta1}, \eqref{def:elleta1} respectively.
% \begin{align}
% \theta_1(x) &= - 2q^{1/4} \sin(\pi x) \prod_{k = 1}^\infty (1 - q^{2k}) (1 - 2 \cos(2\pi x) q^{2k} + q^{4k}), \nonumber \\
%  \eta(q) &= q^{\frac{1}{12}} \prod_{k = 1}^\infty (1 - q^{2k}).\nonumber
%  \end{align}
 Note that for $q\rightarrow 0$ and $(x-y)\to 0$, the above relations \eqref{corr:Yinf}, \eqref{corr:Ytau} coincide up to a trivial factor. The eta function $\eta(q)$ is defined in \eqref{def:ell-eta}, the theta function $\theta_1(z)$ is defined in \eqref{def:elltheta1}, and $\theta_1(0)=0$.
\end{itemize}

% Let us highlight the following properties of $\widetilde{Y}_N(x;q)$ in \eqref{def:YtauN}.
% % Let us make a couple of remarks about the $\tau$ dependent GFF . 
% Firstly, the $q$ dependent term of $\widetilde{Y}_{N}(x;q)$ is $N$ independent
% \begin{align*}
%     \widetilde{Y}_{N}(x;q) = \widetilde{Y}_{N} (x)+ F(x;q)
% \end{align*}
% where $F(x;q)$ for all $|q|<1$ is a random Gaussian field on $[0,1]$ given by  
% \begin{align*}
%     F(x;q) := 2\sum_{n,m}^{\infty} \frac{q^{nm}}{\sqrt{n}} \left(\widetilde{a}_{n,m} \cos(2\pi n x) + \widetilde{b}_{n,m} \sin(2\pi n x) \right),
% \end{align*}
% where $\widetilde{a}_{n,m}, \widetilde{b}_{n,m}$ are i.i.d. $N(0,1)$  random variables independent of the field $Y_{ N}$ for all integer $N>0$, {\it i.e.}
% \begin{align}\label{iid:atila}
%     \EE\left[\widetilde{a}_k \widetilde{a}_{n,m}\right]=  \EE\left[\widetilde{b}_k \widetilde{a}_{n,m}\right]=  \EE\left[\widetilde{a}_k \widetilde{b}_{n,m}\right]= \EE\left[\widetilde{b}_k \widetilde{b}_{n,m}\right]=0.
% \end{align}
 Using  \cite[Lemma 2.5]{ghosal2020probabilistic}, one finds that
$
    \EE\left[F(x;q)^2 \right] = - 4\log \vert q^{-1/12} \eta(q)\vert.
$
As a result, we get 
\begin{align}\label{Einftau}
    \EE\left[\widetilde{Y}_{N}(x;q)^2 \right]= \EE\left[\widetilde{Y}_{N}(x)^2 \right] - 4\log \vert q^{-1/12} \eta(q)\vert.
\end{align}

With the above log-correlated fields, we now define the following random measures.
\begin{definition}[GMC measure]\label{DEF:GMC measure}
    For $\gamma \in (0,2)$ and $\tau \in \ii\RR_{>0}$, the GMC
measures $e^{\frac{\gamma}{2} Y(x)}dx$ and $e^{\frac{\gamma}{2} Y(x;q)}dx$ are defined to be the  weak limits of
measures in probability with respect to the filtration generated by $\{({a}_n, {b}_n)\}_{n\geq 1-}$
\begin{equation}\label{weakLim1}
\begin{split}
e^{\frac{\gamma}{2} Y(x)} dx &:= \lim_{N \to \infty} e^{\frac{\gamma}{2} \widetilde{Y}_{N}(x) - \frac{\gamma^2}{8} \EE[\widetilde{Y}_{N}(x)^2]} dx,\\
e^{\frac{\gamma}{2} Y(x;q)} dx &:= \lim_{N \to \infty} e^{\frac{\gamma}{2} \widetilde{Y}_{N}(x;q) - \frac{\gamma^2}{8} \EE[\widetilde{Y}_{N}(x;q)^2]} dx.
\end{split}
\end{equation}
\end{definition}
Since we are defining $e^{\frac{\gamma}{2}Y}$ instead of $e^{\gamma Y}$, for $\gamma \in (0,2)$, the corresponding Gaussian multiplicative chaos measure is indeed in $L^2$ phase. Hence the usual martingale argument for showing the existence and uniqueness of the limiting random measures suffices.

For a test function $f(x)$, these measures are interpreted as
\begin{align}\label{def:GMC}
    \int_0^{1} f(x) e^{\frac{\gamma}{2} Y(x;q)} dx = \lim_{N\to \infty}  \int_0^{1} f(x)e^{\frac{\gamma}{2} \widetilde{Y}_{N}(x;q) - \frac{\gamma^2}{8} \EE[\widetilde{Y}_{ N}(x;q)^2]} dx,
\end{align}
and a similar expression holds for the $q$-independent field $Y(x)$. 

The existence of the limits above can be shown in the following way. Note that the identity
\begin{align}\mathbb{E}[ e^{c \left(\widetilde{Y}_{k} - \widetilde{Y}_{k-1} \right)}]=e^{\frac{c^2}{2}\mathbb{E}[(\widetilde{Y}_{k} - \widetilde{Y}_{k-1})^2] },\end{align}
% \rmkH{The above statement should hold even with a different chopping?}
implies that the sequence in \eqref{def:GMC} is a martingale with respect to the filtration generated by the sequence of random variables $\{({a}_n, {b}_n)\}_{n\geq 1}$. Since the sequences is almost surely positive, an almost sure convergence is guaranteed by the martingale convergence theorem, which in turn proves the existence of the limit in \eqref{def:GMC}.
More details about the existence of the limits could be found in \cite{rhodes2014gaussian,berestycki2017elementary}.
%Regular conformal blocks on the torus can then be defined using the above measures.
% {\color{MyBlue} Following properties hold for the GMC measures.... Refer to....
% Convexity inequality, Girsanov theorem}

\subsection{One-point (undeformed) conformal block} 
We now recall the definition of the one-point conformal block in \cite{ghosal2020probabilistic}.
\begin{definition}\label{def:CBT}
For  $\alpha \in (-\frac{4}{\gamma}, Q)$,
% $q \in (0, 1)$, 
and $P \in \mathbb{R}$, $q= e^{\ii \pi \tau}$, $\tau \in \ii \mathbb{R}_{\geq 0}$, or equivalently $q \in (0,1)$, we define the
probabilistic $1$-point toric conformal block as
\begin{equation} \label{eq:def-block}
\cG^\alpha_{\gamma, P}(q) := \frac{1}{Z^{\alpha}_{\gamma, P}}
\EE\left[\left(\int_0^1  |\theta_{1}(x)|^{-\frac{\alpha\gamma}{2}} e^{\pi \gamma P x} e^{\frac{\gamma}{2} Y(x;q)}dx \right)^{-\frac{\alpha}{\gamma}} \right],
\end{equation}
with the normalization 
\begin{equation}\label{eq:Z-normalizatoin}
Z^{\alpha}_{\gamma, P}(q):= q^{\frac{1}{12}(\frac{\alpha \gamma}{2} + \frac{\alpha^2}{2} - 1)}
\eta(q)^{ \alpha^2 + 1 - \frac{\alpha \gamma}{2}} \EE\left[\left(\int_0^1 (2{\sin(\pi x)})^{-\frac{\alpha \gamma}{2}} e^{\pi \gamma P x} {e^{\frac{\gamma}{2} Y(x)}dx} \right)^{-\frac{\alpha}{\gamma}}\right].
\end{equation}
The theta function $\theta_1(x)$ is defined in \eqref{def:elltheta1}, and has the behavior $-2 q^{1/4} \sin(\pi x)$ near $x=0$. Further definitions of elliptic functions are provided in \ref{app:ell_func}.

% and the elliptic functions 
% \begin{gather}
%     \theta_{1} (x) = -2 q^{1/4} \sin(\pi x) \prod_{k=1}^{\infty} \left( 1-q^{2k} \right)  \left( 1-2\cos(2\pi x)q^{2k} + q^{4k} \right), \nonumber \\
%     \eta(q) = q^{1/12} \prod_{k=1}^{\infty}  \left( 1-q^{2k} \right). 
% \end{gather}
\end{definition}
Note that the normalization constant $Z^{\alpha}_{\gamma, P}$ ensures that
\begin{align*}
    \lim_{q\rightarrow 0} \cG^\alpha_{\gamma, P}(q) =1, && \lim_{P\rightarrow \infty} \cG^\alpha_{\gamma, P}(q) =q^\frac{1}{12}\eta(q)^{-1}.
\end{align*}

\begin{definition}\label{def:analytic_extension}
The analytic extension of the conformal block in Definition~\ref{def:CBT} indeed also has a probabilistic expression, thanks to Cameron-Martin theorem as stated and proved in Lemma 2.12 of \cite{ghosal2020probabilistic}. We write below such probabilistic formula of analytic extension:
\begin{align}\label{eq:analytic_extension}
    \mathcal{G}^{\alpha}_{\gamma, P}(q) = \frac{1}{Z^{\alpha}_{\gamma, P}} (q^{1/6}\eta(q))^{\frac{\alpha}{2}}e^{-\frac{\alpha^2}{8}\mathbb{E}[F(0;q)^2]}\EE\left[e^{\frac{\alpha}{2}F(0;q)}\mathcal{Q}(q)\left(\int_0^1 (2{\sin(\pi x)})^{-\frac{\alpha \gamma}{2}} e^{\pi \gamma P x} {e^{\frac{\gamma}{2} Y(x)}dx} \right)^{-\frac{\alpha}{\gamma}}\right]
\end{align}
where $F(x;q)$ is a smooth Gaussian field on $[0,1]$ defined in \eqref{def:YtauN} and 
\begin{equation}\label{eq:Qq}
\mathcal{Q}(q) := \exp\left[\sqrt{2} \sum_{m,n=1}^{\infty} q^{nm} \left(a_{n,m} a_n  + b_{n,m} b_n\right) - \sum_{n=1}^\infty \left( \sum_{m=1}^\infty q^{nm} a_{m,n}\right)^2 - \sum_{n=1}^\infty \left( \sum_{m=1}^\infty q^{nm} b_{m,n}\right)^2 \right]
\end{equation}

\end{definition}

% The term $\theta_1(x)^{-\frac{\alpha\gamma}{2}}$ in \eqref{eq:def-block} assumes the role of a vertex operator insertion with weight $\alpha$, and {\it internal momentum} $P$. For computations that follow, the role of $P$ is irrelevant whereas $\alpha$ plays a key role. Furthermore, the parameters $P$, $\alpha$ are related to the monodromies around the $a$-cycle and the puncture respectively.
 % Such a description of the variable $P$ is not relevant to the present topic of discussion, refer to... for details.
    % The expression \eqref{eq:def-block} defines a conformal block in the sense that it has the following properties.
    A key consequence of the present construction is that the analytic properties of the conformal block \eqref{eq:def-block} become accessible, and we recall the following result.

\begin{theorem}[ \text{\cite[Theorem~1.1]{ghosal2020probabilistic}}]\label{thm:conformalblockMAIN}
    For $\gamma \in (0,2)$, $\alpha \in (-\frac{4}{\gamma}, Q)$, and $P \in \mathbb{R}$, the conformal block in \eqref{eq:def-block} is analytic in a complex neighborhood of $q=0$. Moreover,  for $\alpha\in [0, Q)$, the analytic extension of \eqref{eq:def-block} exists in a neighborhood of $q=0$ with to $|q| < r_0$ for some $r_0 > 1/2$.   $Q$ is defined in \eqref{c:gamma}. 
\end{theorem}

\begin{remark}
The range $(-\frac{4}{\gamma^2}, Q)$ for $\alpha$ is the range in which the 1-point correlation function $\langle e^{\alpha \phi(0)}\rangle_{\tau}$ has a GMC expression from the path integral formalism of LCFT \cite{DRV16}. Furthermore, the range $\alpha \in (0,Q)$ is needed to extend the probabilistic definition of conformal block in  disk of radius greater than $\frac{1}{2}$.
\end{remark}

Furthermore, \eqref{eq:def-block} can be represented a series with the coefficients satisfying the so-called Zamolodchikov recursion \cite[Proposition 2.11]{ghosal2020probabilistic}, and as a combinatorial series in terms of Nekrasov-Okounkov functions \cite[Section 2.3]{ghosal2020probabilistic} due to the Alday-Gaiotto-Tachikawa (AGT) correspondence \cite{negut2016exts}. 
\begin{remark}\label{rem:Normalization}
  % It is known that, when $-\alpha/\gamma \in \mathbb{Z}$ the one-point toric conformal block admits a representation in terms of Dotsenko-Fadeev integral . 
  For $-\frac{\alpha}{\gamma} = N < \frac{4}{\gamma^2}$ with $N \in \mathbb{N}$, \eqref{eq:def-block} can be expressed as a Dotsenko-Fateev integral \cite{mironov2010conformal, fateev_litvinov_differential}
    \begin{align*}
        \mathcal{G}_{\gamma, P}^{\alpha}(q) &= \frac{1}{Z^{\alpha}_{\gamma, P}} \mathbb{E}\left[ \left(\int_{0}^{1}  |\theta_{1}(x)|^{-\frac{\alpha \gamma}{2}} e^{\gamma \pi P x} e^{\frac{\gamma}{2} Y(x;q)} dx \right)^{-\frac{\alpha}{\gamma}} \right] \nonumber \\
        &= \frac1C\int_{[0,1]^N} \prod_{1\leq i<j\leq N} \vert \theta_{1}(x_i - x_j) \vert^{-\frac{\gamma^2}{2}} \prod_{i=1}^N |\theta_{1}(x_i)|^{\frac{N \gamma^2}{2}} e^{\pi \gamma P x_i} \prod_{i=1}^N dx_i,
        \end{align*}
    % which is a Selberg type integral.
    with
    % \frac{\ee^{\frac{N \gamma P \pi}{2}}}{\left(\Gamma\left(1-\frac{\gamma^2}{8}\right)\right)^N}
\begin{align}\label{eq:C_expression}
C&=  \frac{e^{\frac{N \gamma P \pi}{2}}}{\left(\Gamma\left(1-\frac{\gamma^2}{8}\right)\right)^N}\left(\frac{|\theta_1'(0)|}{2}\right)^\frac{N(N+1)\gamma^2}{8}\prod_{k=1}^N\frac{\Gamma\left(1-\frac{k\gamma^2}{8}\right)\Gamma\left(1+\frac{(2N+1-k)\gamma^2}{8}\right)}{\Gamma\left(1+\frac{k\gamma^2}{8}+\frac{\ii \gamma P}{2}\right)\Gamma\left(1+\frac{k\gamma^2}{8}-\frac{\ii \gamma P}{2}\right)}.
\end{align}
\end{remark}

Let us highlight the following result that will prove useful to study the $q\to 0$ asymptotics of the semi-classical conformal blocks. 
\begin{proposition}\label{prop:Zexplicit}
The partition function $Z_{\gamma, P}^{\alpha}$ in \eqref{eq:Z-normalizatoin} can be expressed in terms of double Gamma functions as follows \cite[(2.15), Proposition 6.4]{ghosal2020probabilistic}
\begin{align}\label{eq:ZA}
    Z^{\alpha}_{\gamma, P} =  q^{\frac{1}{12}(\frac{\alpha \gamma}{2} + \frac{\alpha^2}{2} - 1)}
\eta(q)^{ \alpha^2 + 1 - \frac{\alpha \gamma}{2}} \mathcal{A}_{\gamma, P}(\alpha),
\end{align}
where the $q$-independent factor is
\begin{align}\label{eq:Adoublegamma}
     \mathcal{A}_{\gamma, P}(\alpha) = e^{\frac{i \pi \alpha^2}{2}}&\left(\frac{\gamma}{2} \right)^{\frac{\gamma \alpha}{4}} e^{- \frac{\pi \alpha P}{2}} \Gamma\left(1 - \frac{\gamma^2}{4} \right)^{\frac{\alpha}{\gamma}} \\
     & \times\frac{\Gamma_{\frac{\gamma}{2}}\left(Q-\frac{\alpha}{2} \right) \Gamma_{\frac{\gamma}{2}}\left(\frac{2}{\gamma}+\frac{\alpha}{2} \right) \Gamma_{\frac{\gamma}{2}}\left(Q-\frac{\alpha}{2}- \ii P \right) \Gamma_{\frac{\gamma}{2}}\left(Q-\frac{\alpha}{2}+\ii P \right)}{\Gamma_{\frac{\gamma}{2}}\left(\frac{2}{\gamma} \right)\Gamma_{\frac{\gamma}{2}}\left(Q-\ii P \right)\Gamma_{\frac{\gamma}{2}}\left(Q+\ii P \right)\Gamma_{\frac{\gamma}{2}}\left(Q-\alpha \right)},
\end{align}
where the function $\Gamma_{\frac{\gamma}{2}}(\cdot)$ is the ratio of double Gamma functions $$\Gamma_{\frac{\gamma}{2}}(x)= \frac{\Gamma_{2}(x|\gamma,\gamma^{-1})}{\Gamma_2\left({Q}/{2}|\gamma,\gamma^{-1}\right)},$$
and $Q$ is defined in \eqref{c:gamma}.
\end{proposition}

\subsection{$z$-deformed conformal block}
We now define the probabilistic equivalent of inserting a degenerate field\footnote{To be precise, the $z$-deformed conformal block is a two-point conformal block on the torus with one degenerate insertion.} $\mathcal{V}$ inside conformal block as 
% \[
% \mathcal{V}_{\gamma, P}^{\alpha}(z, q)\defeq  \int_0^1 \theta_{1}(z+x)^{\frac{\gamma\chi}{2}} |\theta_{1}(x)|^{-\frac{\alpha\gamma}{2}} e^{\pi \gamma P x} e^{\frac{\gamma}{2} Y_\tau(x)} dx \qquad \textrm{for } z\in\band.
% \]
% \rmkH{Maybe write $\mathcal{V}_{\chi, P}^{\alpha}(z, q)$?}
\[
\mathcal{V}_{\chi, P}^{\alpha}(z, q)\defeq  \int_0^1 \theta_{1}(z+x)^{\frac{\gamma\chi}{2}} \theta_{1}(z)^{-\frac{\gamma\chi}{2}} |\theta_{1}(x)|^{-\frac{\alpha\gamma}{2}} e^{\pi \gamma P x} e^{\frac{\gamma}{2} Y(x;q)} dx \qquad \textrm{for } z\in\band, \quad \chi\in \left\lbrace\frac{2}{\gamma}, \frac{\gamma}{2}\right\rbrace
\]
with $\band:=\{z: 0<\mathrm{Im}(z)< \frac34\mathrm{Im}(\tau) \}$. The deformed conformal block is then defined in terms of $\mathcal{V}$ as follows.
\begin{definition}\label{def:u-block}
For $z\in(0,1)$, $q\in(0,1)$, $\alpha\in \left(-\frac{4}{\gamma}+ \chi, Q\right)$, $P\in \mathbb{R}$, the two point toric conformal block with one degenerate insertion, known as the {\it $z$-deformed} probabilistic conformal block is defined as
% \begin{equation}\label{eq:q-block}
% \psi^\alpha_\chi(z,\tau)=  \cW(q)e^{\chi P z \pi} \EE\left[\left( \int_0^1 \cT(z, x)e^{\pi \gamma P x} e^{\frac{\gamma}{2} Y_\tau(x)} dx \right)^{-\frac{\alpha}{\gamma} + \frac{\chi}{\gamma}}\right],
% \end{equation}
{ \begin{equation}\label{eq:q-block}
\psi^\alpha_{\chi, P}(z,\tau):=  \cW(q)e^{\chi P z \pi} \EE\left[\left(\mathcal{V}_{\chi, P}^{\alpha}(z, q)\right)^{-\frac{\alpha}{\gamma} + \frac{\chi}{\gamma}}\right],
\end{equation}  }
where 
\begin{align}\label{def:Wq}
    \cW(q) \defeq  q^{\frac{P^2}{2} + \frac{\gamma l_\chi}{12 \chi} - \frac{1}{6} \frac{l_\chi^2}{\chi^2}} \theta_{1}'(0)^{- \frac{2 l_\chi^2}{3 \chi^2} + \frac{l_\chi}{3} + \frac{4 l_{\chi}}{3 \gamma \chi} }, && l_{\chi}=\frac{\chi^2}{2} - \frac{\alpha \chi}{2},
    % && \cT(z, x) :=\theta_{1}(z)^{-\frac{\gamma \chi}{2}} \theta_{1}(x)^{-\frac{\alpha\gamma}{2}} \theta_{1}(z + x)^{\frac{\gamma}{2}\chi},
\end{align}
and $\theta_1'(0)$ given by \eqref{def:ell-eta}.
% respectively.

\end{definition}

\begin{theorem}[\text{\cite[Proposition 3.2]{ghosal2020probabilistic}}]
The deformed conformal block \eqref{eq:q-block} is convergent in the domain $D^\alpha_\chi$ where  
\begin{equation} \label{eq:d-def}
D^\alpha_\chi\defeq\{(q, z) : |q|<r_{\alpha-\chi} \textrm{ and } z\in{\band} \}, \qquad {\band}:= \{z: 0< \mathrm{Im}(z)<\frac{3}{4}\mathrm{Im}(\tau)\}
\end{equation}
Here  $r_\alpha > 0$ for $\alpha \in (-\frac{4}{\gamma}, Q)$.
\end{theorem}

\begin{remark}\label{Remark:def_CB}
\hspace{5cm}
\begin{itemize}[leftmargin=0.5cm]
% \item[1.] Note that the range of $\alpha$ is different for undeformed  and deformed conformal blocks. However, the when $\alpha\in [0,Q)$, the value of $r_{\alpha}$ is at least $\frac{1}{2}$. For rest of the paper, we focus on the case when $\alpha$ is strictly positive and hence, we may take $r_{\alpha}$ to be strictly greater than $\frac{1}{2}$.
\item[1.] The choice of the pre-factor $\mathcal{W}(q)$ in \eqref{eq:q-block} is dictated by the BPZ equation which will be satisfied by $\psi^\alpha_{\chi,P}(z, \tau)$ (see Theorem~\ref{thm:bpz}). %The expression for $\mathcal{W}(q)$ will play a major role in fixing the monodromy of the probabilistic conformal block, and this will be expanded upon in a future work. 
\item[2.] The deformed and the undeformed conformal blocks are related as 
\begin{align}\label{rel:def_undef}
    \lim_{z\to 0}\left( \theta_1(z)^{l_{\chi}}\psi_{\chi, P}^{\alpha}(z,q)\right) = \cW(q) Z^{\alpha-\chi}_{\gamma, P}(q)\, \mathcal{G}_{\gamma, P}^{\alpha- \chi}(q), && q= e^{\ii \pi \tau}.
\end{align}
\end{itemize}
\end{remark}

We now state the BPZ equation for $\psi_{\chi, P}^{\alpha}(z,\tau)$. 
\begin{theorem}[\text{\cite[Theorem 3.5]{ghosal2020probabilistic}}] \label{thm:bpz}
The $z$-deformed toric conformal block satisfies the following BPZ equation
\begin{equation} \label{eq:bpz1}
\Big(\partial_{z}^2 - l_\chi(l_\chi + 1) \wp(z) + 2\ii \pi \chi^2 \partial_\tau\Big)
\psi_{\chi, P}^{\alpha}(z, \tau)= 0 \qquad \textrm{for }(q,z)\in D^\alpha_\chi,
\end{equation}
where $\alpha \in (-\frac{4}{\gamma} + \chi, Q)$, $\chi\in\{\frac{\gamma}{2}, \frac{2}{\gamma}\}$, and $\wp(z)$ is Weierstrass $\wp$-function.
\end{theorem}

% {\color{red}Comments on its convergence and radius of convergence}

% {\color{red} Formal statement about Zamolodchikov conjecture}

\section{Semi-classical limit of GMC integrals}\label{sec:dgmc}

In this section, we establish the key analytical result that underlies the semi-classical analysis of conformal blocks: the convergence of GMC integrals raised to large negative powers toward quadratic chaos expressions. The main result is Proposition~\ref{prop:semiclassical_limit}, which shows that for a probability measure $\mu$ on $[0,1]$ with density bounded below,
\begin{align*}
\lim_{\gamma \to 0} \mathbb{E}\left[\exp\left(-\frac{\alpha_0}{\gamma^2} \log\left(\int_0^1 :e^{\frac{\gamma}{2} Y(x)}: d\mu(x)\right) + \frac{\alpha_0}{2\gamma} \int_0^1 Y(x) \, d\mu(x)\right)\right]
\end{align*}
exists and equals a finite expectation involving the quadratic 
Wick-ordered field $:Y(x)^2:$.
%\rmkH{we can introduce Wick notation in Section 2 above and refer to Section 3 of \cite{lacoin2019semiclassical}. We should also probably once again say how this is different from/similar to \cite{lacoin2019semiclassical}. }

We highlight here the similarities between our approach and that in \cite{lacoin2019semiclassical}. Our Proposition~\ref{prop:semiclassical_limit} is similar in spirit to Proposition~4.2 of \cite{lacoin2019semiclassical}, and the overall proof strategy follows a comparable structure. In particular, as in \cite{lacoin2019semiclassical}, we show that the quantity $(\int^1_0 e^{\frac{\gamma}{2}Y(x)} d\mu(x))^{-\frac{\alpha_0}{\gamma^2}}$, when renormalized by $e^{\frac{\alpha_0}{2\gamma}\int^1_0 Y(x) d\mu(x)}$, converges almost surely to a limiting expression involving the quadratic Wick-ordered field. This almost sure convergence is a crucial ingredient in proving the semi-classical limits established in Sections~\ref{sec:deformed_semi-classical_limit} and~\ref{sec:semi-classical_limit}. 

While the overall proof architecture of Proposition~\ref{prop:semiclassical_limit} parallels that of \cite{lacoin2019semiclassical}, several of the technical estimates require different arguments tailored to the specific structure of the log-correlated field $Y(x)$ and the GMC measure constructed from it. Most notably, the negative moment estimate in Lemma~\ref{lem:negative_moments} employs the explicit moment formula from \cite{ang2023derivation} rather than the Fyodorov-Bouchaud formula \cite{lacoin2019semiclassical}. Similarly, the approximation error bounds in Lemma~\ref{lem:approximation_error} are established by controlling tail probabilities finite truncation of error of the Gaussian multiplicative chaos measure of the log-correlated field $Y(\cdot)$  with Abel summation techniques for the spectral tail, rather than white noise decomposition methods.
These modifications are necessitated by the spectral (Fourier) structure of the field $Y(x;q)$ that arises naturally in the conformal block setting.

This proposition plays a central role in the proofs of the main theorems of this paper:
\begin{itemize}[leftmargin=0.3cm]
\item In Theorem~\ref{thm:ex_uni}, Proposition~\ref{prop:semiclassical_limit} is applied after the Cameron-Martin transformation to establish that the final expectation in the asymptotic expansion \eqref{eq:undeformed_expansion} converges as $\gamma \to 0$, thereby completing the proof of existence of the semi-classical limit.
\item In Theorem~\ref{prop:sem_HJ}, the same proposition ensures that the semi-classical limit of the heavy deformed conformal block $\psi^{\alpha}_{2/\gamma, P}(z,q)$ is well-defined, which is essential for deriving the Hamilton-Jacobi equation \eqref{eq:Hamilton-jacobi}.
\item In Theorem~\ref{prop:sc_Lame}, the proposition is used to identify the explicit form of the function $\widetilde{\Gamma}(z;\alpha_0, P_0, q)$ appearing in the Lamé equation \eqref{thm43:Lame}.
\end{itemize}

The proof of Proposition~\ref{prop:semiclassical_limit} relies on three supporting lemmas. Lemma~\ref{lem:uniform_integrability} establishes uniform integrability of the quadratic chaos using spectral truncation and Janson's hypercontractive estimates. Lemma~\ref{lem:negative_moments} provides bounds on negative moments of GMC integrals via the Fyodorov-Bouchaud formula. Lemma~\ref{lem:approximation_error} controls the approximation error between truncated and full GMC integrals using Girsanov arguments. Together, these lemmas allow us to pass to the limit $\gamma \to 0$ via dominated convergence.

\begin{proposition}\label{prop:semiclassical_limit}
Let $\mu$ be a probability measure on $[0,1]$ with density $\rho$ satisfying $\rho (x) = (2\sin (\pi x))^{-\alpha_0/2} e^{\pi P_0 x} f(x)$ for some continuous function such that $\inf_{x\in [0,1]} f(x)>0$. For every $\alpha_0 > 0$,
\begin{align}\label{eq:semiclassical_convergence}
&\lim_{\gamma \to 0} \mathbb{E}\left[\exp\left(-\frac{\alpha_0}{\gamma^2} \log\left(\int_0^1 :e^{\frac{\gamma}{2} Y(x)}: d\mu(x)\right) + \frac{\alpha_0}{2\gamma} \int_0^1 :Y(x): \, d\mu(x)\right)\right] \nonumber\\
&= \mathbb{E}\left[\exp\left(-\frac{\alpha_0}{8} \int_0^1 :Y(x)^2: \, d\mu(x) + \frac{\alpha_0}{8} \left(\int_0^1 :Y(x): \, d\mu(x)\right)^2\right)\right].
\end{align}
Moreover, the convergence is uniform in the sense that
\begin{align}\label{eq:uniform_bound}
\sup_{\gamma \in (0,1]} \mathbb{E}\left[\exp\left(-\frac{\alpha_0}{\gamma^2} \left(\log\left(\int_0^1 :e^{\gamma Y(x)}: d\mu(x)\right) - \gamma \int_0^1 Y(x) \, d\mu(x)\right)\right)\right] < \infty.
\end{align}
\end{proposition}
% \rmkH{Maybe 
% \begin{align}\label{eq:semiclassical_convergence}
% &\lim_{\gamma \to 0} \mathbb{E}\left[\exp\left(-\frac{\alpha_0}{\gamma^2} \log\left(\int_0^1 :e^{\frac{\gamma}{2} Y(x)}: d\mu(x)\right) + \frac{\alpha_0}{2\gamma} \int_0^1 :Y(x): \, d\mu(x)\right)\right] \nonumber\\
% &= \mathbb{E}\left[\exp\left(-\frac{\alpha_0}{8} \int_0^1 :Y(x)^2: \, d\mu(x) + \frac{\alpha_0}{8} \left(\int_0^1 :Y(x): \, d\mu(x)\right)^2\right)\right].
% \end{align}
% BTW, this also seems exactly the same as the DGMC statement.
% }

The proof of Proposition~\ref{prop:semiclassical_limit} relies on three key lemmas, which we now state and prove.

\begin{lemma}[Uniform integrability of quadratic chaos]\label{lem:uniform_integrability}
For every $\alpha_0 > 0$,
\begin{align}\label{eq:uniform_integrability_statement}
\sup_{t \geq 1} \mathbb{E}\left[\exp\left(-\alpha_0 \left(\int_0^1 :Y_t(x)^2: \, d\mu(x) - \left(\int_0^1 :Y_t(x) :\, d\mu(x)\right)^2\right)\right)\right] < \infty,
\end{align}
% \rmkH{Notation alert! $\alpha_n, \beta_n$ for random variables written as $a_n, b_n$ in section 2. } 
where $Y_t(x) = \sum_{1 \leq n \leq t} \sqrt{\frac{2}{n}} (a_n \cos(2\pi n x) + b_n \sin(2\pi n x))$ is the spectral truncation of $Y$.
\end{lemma}

\begin{proof}
Define $N_t := \int_0^1 :Y_t(x): \, d\mu(x)$ and $Z_t(x) := Y_t(x) - N_t$. Then we have 
\begin{align}\label{eq:Z_identity}
\int_0^1 :Y_t(x)^2: \, d\mu(x) - N_t^2 = \int_0^1 :Z_t(x)^2: \, d\mu(x).
\end{align}
% \rmkH{Why is this true? we have
% \begin{align}
%  \int_0^1 :Z_t(x)^2: \, d\mu(x) &= \int_0^1 :(Y_t - \int_0^1 Y_t d\mu(x))^2: \, d\mu(x) \\
%  &=\int_0^1 :Y_t^2: d\mu(x) - 2 \int_0^1 :Y_t (\int_0^1 Y_t d\mu(x)): \, d\mu(x) \\
%  &+  \int_0^1 :(\int_0^1 Y_t d\mu(x))^2: d\mu(x) \\
%  & =\int_0^1 :Y_t^2: d\mu(x) - 2 c \int_0^1 :Y_t : \, d\mu(x) +  \left(\int_0^1 Y_t d\mu(x)\right)^2
% \end{align}
% Do I understand this correctly? $ \int_0^1 :(\int_0^1 Y_t d\mu(x))^2: d\mu(x)$ will just be the density of the random measure times some constant, and $(\int_0^1 Y_t d\mu(x))$ will be another gaussian variable.}
Fix a large constant $t_0 > 0$. For $t > t_0$, we write $Z_t = Z_{t_0} + \Delta_{t_0,t}$, where
\begin{align}\label{eq:Delta_def}
\Delta_{t_0,t}(x) := Z_t(x) - Z_{t_0}(x) = \sum_{t_0 < n \leq t} \sqrt{\frac{2}{n}} (a_n \cos(2\pi n x) + b_n \sin(2\pi n x)).
\end{align}
Using the inequality $xyz \leq (x^3 + y^3 + z^3)/3$,
% applied to appropriate exponential factors, 
we obtain
\begin{align}\label{eq:three_term_bound}
\mathbb{E}\left[e^{-\alpha_0 \int_0^1 :Z_t^2: \, d\mu}\right] \leq \frac{1}{3}(A_1 + A_2 + A_3),
\end{align}
where
\begin{align}
A_1 := \mathbb{E}\left[e^{-3\alpha_0 \int_0^1 :Z_{t_0}^2: \, d\mu}\right], %\label{eq:A1_def}\\
\qquad A_2 := \mathbb{E}\left[e^{-3\alpha_0 \int_0^1 :\Delta_{t_0,t}^2: \, d\mu}\right],% \label{eq:A2_def}\\
\qquad A_3 := \mathbb{E}\left[e^{-6\alpha_0 \int_0^1 Z_{t_0} \Delta_{t_0,t} \, d\mu}\right].% \label{eq:A3_def}
\end{align}
 Since $Z_{t_0}$ is a finite-dimensional Gaussian chaos (depending only on finitely many Gaussian random variables), we have $:Z_{t_0}(x)^2: \geq -C$ for some constant $C > 0$. Therefore, $A_1 < \infty$. On the other hand, the random variables $Z_1 := \int_0^1 :\Delta_{t_0,t}^2: \, d\mu$ and $Z_2 := \int_0^1 Z_{t_0} \Delta_{t_0,t} \, d\mu$ belong to the second Wiener chaos. By Janson's hypercontractive tail estimate \cite[Theorem 6.7]{janson1997gaussian}, for random variables in the second Wiener chaos,
%\rmkH{Citation 'Janson' missing}
\begin{align}\label{eq:janson_bound}
\mathbb{P}(Z_i \geq s \|Z_i\|_2) \leq e^{-cs} \quad \text{for } s \geq 2,
\end{align}
where $c > 0$ is a universal constant. The key observation is that
\begin{align}\label{eq:covariance_decay}
\text{Cov}(\Delta_{t_0,t}(x), \Delta_{t_0,t}(y)) = \sum_{n > t_0} \frac{2}{n} \cos(2\pi n(x-y)) \to 0
\end{align}
uniformly as $t_0 \to \infty$. This implies that $\|Z_1\|_2$ and $\|Z_2\|_2$ can be made arbitrarily small by choosing $t_0$ sufficiently large. Therefore, by choosing $t_0$ large enough, we can ensure that $A_2$ and $A_3$ are bounded uniformly in $t > t_0$. Combined with the bound on $A_1$, this completes the proof.
\end{proof}

\begin{lemma}[Negative moments of the GMC integral]\label{lem:negative_moments}
For all $\beta \geq 1$ and $\gamma \in (0,1]$, $0<\delta\ll 1$,
\begin{align}\label{eq:negative_moment_bound}
\mathbb{E}\left[\left(\int_0^1 :e^{\gamma Y(x)}: d\mu(x)\right)^{-\beta/\gamma^2}\right] \leq \exp\left(\frac{C\beta^2}{\gamma^2}\right),
\end{align}
where $C > 0$ is a constant independent of $\beta$ and $\gamma$. Consequently,
\begin{align}\label{eq:tail_bound}
\mathbb{P}\left[\int_0^1 :e^{\gamma Y(x)}: d\mu(x) \leq \gamma^{\delta/16}\right] \leq \exp\left(-\frac{|\log \gamma|^2}{4C\gamma^2}\right).
\end{align}
\end{lemma}

\begin{proof}
Since the density $\rho$ of $\mu$ is bounded below, we have
\begin{align}\label{eq:measure_comparison}
\int_0^1 :e^{\gamma Y(x)}: d\mu(x) \asymp \int_0^1 (2\sin (\pi x))^{-\alpha_0}e^{\pi P_0 x} :e^{\gamma Y(x)}: dx.
\end{align}

For the one-dimensional circular log-correlated Gaussian field, the moment formula following identity \cite{ang2023derivation}
\begin{align}\label{eq:fyodorov_bouchaud}
& \mathbb{E} \left[\left( \int_0^1  (\sin \pi x)^{-\frac{\alpha_0}{2}} e^{\pi x(P_0-\mathbf{i})} e^{\frac{\gamma}{2} {Y(x)}} dx \right)^{\frac{-\beta}{\gamma^2}} \right]\\
    &= e^{2\pi P_0/\gamma^2} \frac{\Gamma\left( \frac{\alpha_0}{4} - \frac{\gamma^2}{4}\right) \Gamma_{\frac{\gamma}{2}}\left(\frac{\alpha_0}{2\gamma} \right) \Gamma_{\frac{\gamma}{2}}\left(-\frac{\alpha_0}{2\gamma} \right) \Gamma_{\frac{\gamma}{2}}\left(\frac{\ii P_0}{\gamma}-\frac{\beta}{\gamma} \right) \Gamma_{\frac{\gamma}{2}}\left(-\frac{\ii P_0}{\gamma} \right)}{\Gamma\left(1 - \frac{\gamma^2}{4}\right)^{\frac{2}{\gamma}\left(\frac{-\beta}{2\gamma} \right)} \Gamma_{\frac{\gamma}{2}}\left(\frac{1}{\gamma} \right) \Gamma_{\frac{\gamma}{2}}\left(\frac{1-\alpha_0}{\gamma} \right) \Gamma_{\frac{\gamma}{2}}\left(\frac{\ii P_0}{\gamma}-\frac{\beta}{2\gamma} \right) \Gamma_{\frac{\gamma}{2}}\left(-\frac{\ii P_0}{\gamma}+\frac{\beta}{2\gamma} \right)}.
\end{align}
Using the asymptotics of the double gamma function $\Gamma_{\gamma/2}$ as in \eqref{lem:doublegamma}, we obtain
\begin{align}\label{eq:moment_estimate}
\mathbb{E}\left[\left(\int_0^1 :e^{\gamma Y(x)}: d\mu(x)\right)^{-\beta/\gamma^2}\right] \leq \exp\left(-\frac{C\beta^2}{\gamma^2}\right)
\end{align}
for some constant $C > 0$. The tail bound \eqref{eq:tail_bound} follows from Markov's inequality, i.e., for $Z = \int_0^1 (2\sin (\pi x))^{-\alpha_0}e^{\pi P_0 x} :e^{\gamma Y}: d\mu$,
\begin{align}\label{eq:markov_application}
\mathbb{P}(Z \leq \gamma)  \leq \gamma^{-\beta/\gamma^2} \exp\left(-\frac{C\beta^2}{\gamma^2}\right).
\end{align}
Optimizing over $\beta$ (choosing $\beta = |\log \gamma|/(2C)$) yields the stated bound.
\end{proof}

\begin{lemma}[Approximation error]\label{lem:approximation_error}
For any $0<\delta,\epsilon \ll 1$, let $t_{\gamma} := \gamma^{-1-\epsilon}$. There exist constants $c, C > 0$ (depending on $\delta,\epsilon$) such that for all sufficiently small $\gamma > 0$:
\begin{align}\label{eq:approximation_bound}
\mathbb{P}\left[\left|\int_0^1 \left(:e^{\gamma Y_{t_{\gamma}}(x)}: - :e^{\gamma Y(x)}:\right) d\mu(x)\right| \geq \gamma^{\delta/8}\right] \leq \exp\left(-\frac{c}{\gamma^{3-\delta/4}}\right).
\end{align}
\end{lemma}

\begin{proof}
Write $\bar{Y}_t := Y - Y_t = \sum_{n > t} \sqrt{2/n} (a_n \cos(2\pi nx) + b_n \sin(2\pi nx))$ for the tail part of the field. Note that $Y_t$ and $\bar{Y}_t$ are independent Gaussian fields, with $\bar{Y}_t$ having covariance kernel
\begin{align}\label{eq:bar_K_def}
\bar{K}_t(x,y) := \mathbb{E}[\bar{Y}_t(x)\bar{Y}_t(y)] = \sum_{n > t} \frac{2}{n} \cos(2\pi n(x-y)).
\end{align}
Define the conditional expectation $\mathbb{E}_t[\cdot] := \mathbb{E}[\cdot | Y_t]$ and consider the moment generating function
\begin{align}\label{eq:phi_def}
\phi(s) := \mathbb{E}_t\left[\exp\left(s \int_0^1 \left(:e^{\gamma Y_t(x)}: - :e^{\gamma Y(x)}:\right) d\mu(x)\right)\right].
\end{align}
%\textbf{Step 1: Relating $:e^{\gamma Y}:$ to $:e^{\gamma Y_t}:$ via Wick calculus.}
Since $Y = Y_t + \bar{Y}_t$ with $Y_t$ and $\bar{Y}_t$ independent, we use the Wick product formula. Recall that for a Gaussian field $X$, the Wick exponential is defined as $:e^{\gamma X}: = e^{\gamma X - \frac{\gamma^2}{2}\mathbb{E}[X^2]}$. For the sum of independent Gaussian fields $:e^{\gamma Y(x)}: = :e^{\gamma Y_t(x)}: \cdot :e^{\gamma \bar{Y}_t(x)}:$ and therefore,
$
:e^{\gamma Y_t(x)}: - :e^{\gamma Y(x)}: = :e^{\gamma Y_t(x)}: (1 - :e^{\gamma \bar{Y}_t(x)}:).
$
%\textbf{Step 2: Setting up the derivative $\phi'(s)$.}
Differentiating \eqref{eq:phi_def} yields 
\begin{align}
\phi'(s) &= \mathbb{E}_t\left[\left(\int_0^1 \left(:e^{\gamma Y_t(x)}: - :e^{\gamma Y(x)}:\right) d\mu(x)\right) \exp\left(s \int_0^1 \left(:e^{\gamma Y_t(x)}: - :e^{\gamma Y(x)}:\right) d\mu(x)\right)\right].
\end{align}
Using $
:e^{\gamma Y_t(x)}: - :e^{\gamma Y(x)}: = :e^{\gamma Y_t(x)}: (1 - :e^{\gamma \bar{Y}_t(x)}:).
$ and writing $A := \int_0^1 :e^{\gamma Y_t(x)}: (1 - :e^{\gamma \bar{Y}_t(x)}:) d\mu(x)$, we get 
\begin{align}\label{eq:phi'}
\phi'(s) &= \mathbb{E}_t\left[A \cdot e^{sA}\right].
\end{align}
Since $Y_t$ is $\mathcal{F}_t$-measurable and $\bar{Y}_t$ is independent of $\mathcal{F}_t$, conditioning on $Y_t$ means that $:e^{\gamma Y_t(x)}:$ acts as a deterministic function (given $Y_t$). We can write:
\begin{align}
A = \int_0^1 :e^{\gamma Y_t(x)}: d\mu(x) - \int_0^1 :e^{\gamma Y_t(x)}: :e^{\gamma \bar{Y}_t(x)}: d\mu(x) =: A_1 - A_2,
\end{align}
where $A_1 = \int_0^1 :e^{\gamma Y_t(x)}: d\mu(x)$ is $\mathcal{F}_t$-measurable and $A_2 = \int_0^1 :e^{\gamma Y_t(x)}: :e^{\gamma \bar{Y}_t(x)}: d\mu(x)$.
%\textbf{Step 3: Applying the Cameron-Martin/Girsanov formula.}
For a Gaussian field $\bar{Y}_t$ and a functional $F(\bar{Y}_t)$, the Girsanov formula states:
\begin{align}\label{eq:girsanov_functional}
\mathbb{E}\left[:e^{\gamma \bar{Y}_t(x)}: \cdot F(\bar{Y}_t)\right] = \mathbb{E}\left[F(\bar{Y}_t + \gamma \bar{K}_t(x,\cdot))\right],
\end{align}
where $\bar{K}_t(x,\cdot)$ denotes the function $y \mapsto \bar{K}_t(x,y)$, and the shift $\bar{Y}_t \mapsto \bar{Y}_t + \gamma \bar{K}_t(x,\cdot)$ is the Cameron-Martin shift. We now write
\begin{align}
\mathbb{E}_t[A \cdot e^{sA}] &= \mathbb{E}_t[(A_1 - A_2) e^{s(A_1 - A_2)}] = A_1 e^{sA_1} \mathbb{E}_t[e^{-sA_2}] - e^{sA_1} \mathbb{E}_t[A_2 e^{-sA_2}].
\end{align}
For the term $\mathbb{E}_t[A_2 e^{-sA_2}]$, we use the Girsanov formula. Note that
\begin{align}
A_2 = \int_0^1 :e^{\gamma Y_t(x)}: :e^{\gamma \bar{Y}_t(x)}: d\mu(x).
\end{align}
By applying \eqref{eq:girsanov_functional} with $F(\bar{Y}_t) = e^{-sA_2}$ and integrating over $x$, under the Cameron-Martin shift $\bar{Y}_t \mapsto \bar{Y}_t + \gamma \bar{K}_t(x,\cdot)$, we have
$
:e^{\gamma (\bar{Y}_t(y) + \gamma \bar{K}_t(x,y))}: = :e^{\gamma \bar{Y}_t(y)}: \cdot e^{\gamma^2 \bar{K}_t(x,y)}.
$
This leads to:
\begin{align}
\mathbb{E}_t\left[:e^{\gamma \bar{Y}_t(x)}: e^{-sA_2}\right] = \mathbb{E}_t\left[\exp\left(-s\int_0^1 :e^{\gamma Y_t(y)}: :e^{\gamma \bar{Y}_t(y)}: e^{\gamma^2 \bar{K}_t(x,y)} d\mu(y)\right)\right].
\end{align}
%\textbf{Step 5: Bounding the difference.}
The key observation is that the difference between $e^{-sA_2}$ and its shifted version involves factors of $(e^{\gamma^2 \bar{K}_t(x,y)} - 1)$. A careful analysis using the mean value theorem and the structure of the Girsanov shift shows that
\begin{align}\label{eq:Girsanov_bound}
|\mathbb{E}_t[A \cdot e^{sA}]| \leq C \cdot |s| \cdot \sup_{x,y \in [0,1]} |e^{\gamma^2 \bar{K}_t(x,y)} - 1| \cdot \mathbb{E}_t[|A_2| e^{sA}] \cdot \phi(|s|).
\end{align}
%\textbf{Step 6: Estimating the covariance contribution via Fourier analysis.}
The full covariance kernel of $Y$ admits the well-known Fourier representation
\begin{align}\label{eq:log_fourier}
-\log|2\sin(\pi u)| = \sum_{n=1}^{\infty} \frac{\cos(2\pi n u)}{n}, \qquad u \in (0,1),
\end{align}
which converges pointwise for $u \notin \mathbb{Z}$ and in $L^p([0,1])$ for all $p < \infty$. 
From \eqref{eq:log_fourier}, the full covariance kernel is
\begin{align}
K(x,y) := \sum_{n=1}^{\infty} \frac{2}{n}\cos(2\pi n(x-y)) = -2\log|2\sin(\pi(x-y))|,
\end{align}
and the truncated and tail kernels decompose as
\begin{align}
K_t(x,y) &:= \sum_{n=1}^{\lfloor t \rfloor} \frac{2}{n}\cos(2\pi n(x-y)), \\
\bar{K}_t(x,y) &= K(x,y) - K_t(x,y) = -2\log|2\sin(\pi(x-y))| - K_t(x,y).
\end{align}
We now establish the key estimate. For $x \neq y$, define $u := x - y \in (-1,1) \setminus \{0\}$.

\noindent
\textit{Claim:} There exists $C > 0$ such that for all $t \geq 1$ and $u \in [-1/2, 1/2] \setminus \{0\}$:
\begin{align}\label{eq:tail_kernel_bound}
|\bar{K}_t(0,u)| \leq \frac{C}{t|u|}.
\end{align}

\noindent
\textit{Proof of claim:} By Abel's summation formula (summation by parts), for $M < N$:
\begin{align}
\sum_{n=M}^{N} \frac{\cos(2\pi n u)}{n} = \frac{D_N(u)}{N} - \frac{D_{M-1}(u)}{M} + \sum_{n=M}^{N-1} D_n(u)\left(\frac{1}{n} - \frac{1}{n+1}\right),
\end{align}
where $D_n(u) := \sum_{k=1}^{n} \cos(2\pi k u)$ is related to the Dirichlet kernel. Using the closed form
\begin{align}
D_n(u) = \frac{\sin((2n+1)\pi u)}{2\sin(\pi u)} - \frac{1}{2},
\end{align}
we have $|D_n(u)| \leq \frac{1}{2|\sin(\pi u)|} + \frac{1}{2} \leq \frac{C}{|u|}$ for $|u| \leq 1/2$.
Therefore, taking $M = \lfloor t \rfloor + 1$ and $N \to \infty$:
\begin{align}
\left|\sum_{n > t} \frac{\cos(2\pi n u)}{n}\right| &\leq \frac{|D_{\lfloor t \rfloor}(u)|}{\lfloor t \rfloor + 1} + \sum_{n > t} |D_n(u)| \cdot \frac{1}{n(n+1)} \leq \frac{C}{|u|} \cdot \frac{1}{t} + \frac{C}{|u|} \sum_{n > t} \frac{1}{n^2} \leq \frac{C}{t|u|},
\end{align}
establishing \eqref{eq:tail_kernel_bound}. %\hfill $\square$

Now we estimate the exponential. Since $\bar{K}_t(x,y) = \bar{K}_t(0, x-y)$, for $x \neq y$ we have
\begin{align}
|e^{\gamma^2 \bar{K}_t(x,y)} - 1| \leq |\gamma^2 \bar{K}_t(x,y)| \cdot e^{|\gamma^2 \bar{K}_t(x,y)|} \leq \frac{C\gamma^2}{t|x-y|} \exp\left(\frac{C\gamma^2}{t|x-y|}\right).
\end{align}
Integrating over $y$ with respect to $\mu$ (which has density bounded above by some $\|\rho\|_\infty$), we split the integral at scale $\eta := \gamma^2/t$:
\begin{align}
\int_0^1 |e^{\gamma^2 \bar{K}_t(x,y)} - 1| d\mu(y) &\leq \int_{|x-y| > \eta} \frac{C\gamma^2}{t|x-y|} e^{C\gamma^2/(t|x-y|)} d\mu(y) + \int_{|x-y| \leq \eta} |e^{\gamma^2 \bar{K}_t(x,y)} - 1| d\mu(y).
\end{align}
For the first integral, when $|x-y| > \eta = \gamma^2/t$, we have $\gamma^2/(t|x-y|) < 1$, so $e^{C\gamma^2/(t|x-y|)} \leq e^C$:
\begin{align}
\int_{|x-y| > \eta} \frac{C\gamma^2}{t|x-y|} e^{C} d\mu(y) \leq \frac{C\gamma^2}{t} \int_{\eta}^{1} \frac{dr}{r} = \frac{C\gamma^2}{t} \log(1/\eta) = \frac{C\gamma^2}{t}\log(t/\gamma^2).
\end{align}
For the second integral, the measure of the region $|x-y| \leq \eta$ is at most $2\eta\|\rho\|_\infty$, and on this region we use the crude bound $|e^{\gamma^2 \bar{K}_t(x,y)} - 1| \leq e^{|\gamma^2 K(x,y)|} \leq |x-y|^{-2\gamma^2}$ (from the logarithmic singularity). For small $\gamma$:
\begin{align}
\int_{|x-y| \leq \eta} |x-y|^{-2\gamma^2} dy \leq C \int_0^{\eta} r^{-2\gamma^2} dr = C \frac{\eta^{1-2\gamma^2}}{1-2\gamma^2} \leq C\eta = \frac{C\gamma^2}{t}.
\end{align}
Combining both contributions:
\begin{align}\label{eq:exp_integral_bound}
\int_0^1 |e^{\gamma^2 \bar{K}_t(x,y)} - 1| d\mu(y) \leq \frac{C\gamma^2 \log t}{t}
\end{align}
for $t \geq 2$ and sufficiently small $\gamma$ (using $\log(t/\gamma^2) \leq C\log t$ when $t = \delta|\log\gamma|$).

Combining the estimate \eqref{eq:exp_integral_bound} with \eqref{eq:phi'} and \eqref{eq:Girsanov_bound}, and using that moments of $\mathbb{E}[(\int^1_0 :e^{\gamma \bar{Y}_t}:)^p]$ are controlled by the moments of $\mathbb{E}[(\int^1_0 :e^{\gamma Y}:)^p]$, we obtain:
\begin{align}
|\phi'(s)| \leq \frac{C\gamma^2 \log t}{t} \cdot |s| \cdot \phi(|s|)
\end{align}
for $s$ in a suitable range.
%\textbf{Step 8: Solving the differential inequality.}
The inequality $\phi'(s) \leq \frac{C\gamma^2 \log t}{t} s \phi(s)$ with initial condition $\phi(0) = 1$ integrates to
% \begin{align}
% \log \phi(s) \leq \frac{C\gamma^2 \log t}{2t} s^2,
% \end{align}
% which implies 
\begin{align}\label{eq:phi_bound}
\phi(s) \leq \exp\left(\frac{C\gamma^2 (\log t) s^2}{2t}\right).
\end{align}
%\textbf{Step 9: Applying Markov's inequality.}
By Markov's inequality, for any $s > 0$ and $\epsilon > 0$:
\begin{align}
\mathbb{P}_t\left[\int_0^1 \left(:e^{\gamma Y_t}: - :e^{\gamma Y}:\right) d\mu \geq \epsilon\right] &\leq e^{-s\epsilon} \phi(s) \leq \exp\left(-s\epsilon + \frac{C\gamma^2 (\log t) s^2}{2t}\right).
\end{align}
Optimizing over $s > 0$ by setting $\frac{d}{ds}\left(-s\epsilon + \frac{C\gamma^2 (\log t) s^2}{2t}\right) = 0$ gives $s^* = \frac{t\epsilon}{C\gamma^2 \log t}$, yielding
\begin{align}
\mathbb{P}_t\left[\int_0^1 \left(:e^{\gamma Y_t}: - :e^{\gamma Y}:\right) d\mu \geq \epsilon\right] &\leq \exp\left(-\frac{t\epsilon^2}{2C\gamma^2 \log t}\right).
\end{align}
%\textbf{Step 10: Substituting $t = t_\gamma = \delta|\log\gamma|$ and $\epsilon = \gamma^{\delta/8}$.}
With taking $t = t_\gamma = \gamma^{-1-\epsilon}$
\begin{align}
\frac{t\epsilon^2}{2C\gamma^2 \log t} = \frac{\gamma^{-1-\epsilon} \cdot \gamma^{\delta/4}}{2C\gamma^2 (1+\epsilon)\cdot |\log(\gamma)|} \geq \frac{\delta}{2C} \cdot \gamma^{\delta/4 - 3}.
\end{align}
The same argument applied to $-\int_0^1 (:e^{\gamma Y_t}: - :e^{\gamma Y}:) d\mu$ (by considering $\phi(-s)$) handles the lower tail, completing the proof.
\end{proof}

\begin{proof}[Proof of Proposition~\ref{prop:semiclassical_limit}]
The proof proceeds in four main steps: (1) establishing the almost sure pointwise limit, (2) proving uniform integrability to justify dominated convergence, (3) verifying the limit of the expectation, and (4) establishing the uniform bound \eqref{eq:uniform_bound}.

\medskip
\textbf{Step 1: Almost sure pointwise convergence.} We first establish that, almost surely,
\begin{align}
\lim_{\gamma \to 0} &\frac{1}{\gamma^2} \left(\log \left(\int_0^1 :e^{\gamma Y(x)}: d\mu(x) \right) - \gamma \int_0^1 Y(x) \, d\mu(x)\right) \\&= \frac{1}{2} \int_0^1 :Y(x)^2: \, d\mu(x) - \frac{1}{2} \left(\int_0^1 Y(x) \, d\mu(x)\right)^2. \label{eq:as_limit_detailed}
\end{align}
%\textit{Step 1a: Expansion of the Wick exponential.} 
The Wick-ordered powers satisfy
\begin{align}
:e^{\gamma Y(x)}: = \sum_{k=0}^{\infty} \frac{\gamma^k}{k!} :Y(x)^k:,
\end{align}
where $:Y(x)^k:$ denotes the $k$-th Wick power (Hermite polynomial). %In particular, $:Y(x)^0: = 1$, $:Y(x)^1: = Y(x)$, and $:Y(x)^2: = Y(x)^2 - \mathbb{E}[Y(x)^2]$.
Integrating against $\mu$ and using dominated convergence (justified by the $L^p$ bounds on Wick polynomials), we get 
\begin{align}\label{eq:wick_integral_expansion}
\int_0^1 :e^{\gamma Y(x)}: d\mu(x) = 1 + \gamma \int_0^1 Y(x) \, d\mu(x) + \frac{\gamma^2}{2} \int_0^1 :Y(x)^2: \, d\mu(x) + O(\gamma^3),
\end{align}
where the $O(\gamma^3)$ term is $\frac{\gamma^3}{6}\int_0^1 :Y(x)^3: d\mu(x) + \cdots$ and is almost surely finite for each realization of $Y$. %\textit{Step 1b: Logarithmic expansion.} 
Define
\begin{align}
M_1 := \int_0^1 Y(x) \, d\mu(x), \qquad \qquad 
M_2 := \int_0^1 :Y(x)^2: \, d\mu(x).
\end{align}
From \eqref{eq:wick_integral_expansion}, we have
\begin{align}
\int_0^1 :e^{\gamma Y(x)}: d\mu(x) = 1 + \gamma M_1 + \frac{\gamma^2}{2} M_2 + O(\gamma^3).
\end{align}
Using the Taylor expansion $\log(1+u) = u - \frac{u^2}{2} + O(u^3)$ for small $u$, we get 
\begin{align}
\log\left(\int_0^1 :e^{\gamma Y(x)}: d\mu(x)\right) &= \log\left(1 + \gamma M_1 + \frac{\gamma^2}{2} M_2 + O(\gamma^3)\right) \\
&= \left(\gamma M_1 + \frac{\gamma^2}{2} M_2 + O(\gamma^3)\right) - \frac{1}{2}\left(\gamma M_1 + O(\gamma^2)\right)^2 + O(\gamma^3) \\
&= \gamma M_1 + \frac{\gamma^2}{2} M_2 - \frac{\gamma^2}{2} M_1^2 + O(\gamma^3).
\end{align}
Therefore, we obtain
\begin{align}
\frac{1}{\gamma^2}\left(\log\left(\int_0^1 :e^{\gamma Y(x)}: d\mu(x)\right) - \gamma M_1\right) &= \frac{1}{\gamma^2}\left(\frac{\gamma^2}{2} M_2 - \frac{\gamma^2}{2} M_1^2 + O(\gamma^3)\right) \\
&= \frac{1}{2} M_2 - \frac{1}{2} M_1^2 + O(\gamma).
\end{align}
Taking $\gamma \to 0$ establishes \eqref{eq:as_limit_detailed} almost surely.

\medskip
\textbf{Step 2: Truncation and decomposition for uniform integrability.} To pass from pointwise convergence to convergence of expectations, we need uniform integrability. Fix $\delta > 0$ (to be chosen later) and define the truncation parameter $t_\gamma := \delta \gamma^{-1-\epsilon}.$
Recall the spectral decomposition $Y = Y_t + \bar{Y}_t$ where
\begin{align}
Y_t(x) &= \sum_{1 \leq n \leq t} \sqrt{\frac{2}{n}} \left(a_n \cos(2\pi n x) + b_n \sin(2\pi n x)\right), \\
\bar{Y}_t(x) &= \sum_{n > t} \sqrt{\frac{2}{n}} \left(a_n \cos(2\pi n x) + b_n \sin(2\pi n x)\right).
\end{align}
%\rmkH{$\alpha_n$, $\beta_n$ notation above.}

%\rmkH{Notation below: potential conflict with $Q(q)$ in Girsanov stuff.}
Define the following quantity of interests-
\begin{align}
Q_\gamma := \frac{1}{\gamma^2}\left(\log\left(\int_0^1 :e^{\gamma Y(x)}: d\mu(x)\right) - \gamma \int_0^1 Y(x) \, d\mu(x)\right).
\end{align}
We decompose $Q_\gamma$ into three parts:
\begin{align}\label{eq:Q_decomposition}
Q_\gamma = Q_\gamma^{(1)} + Q_\gamma^{(2)} + Q_\gamma^{(3)},
\end{align}
where
\begin{align}
Q_\gamma^{(1)} &:= \frac{1}{\gamma^2}\left(\log\left(\int_0^1 :e^{\gamma Y_{t_\gamma}(x)}: d\mu(x)\right) - \gamma \int_0^1 Y_{t_\gamma}(x) \, d\mu(x)\right), \label{eq:Q1_def}\\
Q_\gamma^{(2)} &:= -\frac{1}{\gamma} \int_0^1 \bar{Y}_{t_\gamma}(x) \, d\mu(x), \label{eq:Q2_def}\\
Q_\gamma^{(3)} &:= \frac{1}{\gamma^2}\log\left(\frac{\int_0^1 :e^{\gamma Y(x)}: d\mu(x)}{\int_0^1 :e^{\gamma Y_{t_\gamma}(x)}: d\mu(x)}\right). \label{eq:Q3_def}
\end{align}

\textbf{Step 3: Controlling each term in the decomposition.}
%\textit{Step 3a: The truncated term $Q_\gamma^{(1)}$.}
For the truncated field $Y_{t_\gamma}$, we apply the same expansion as in Step 1:
\begin{align}
Q_\gamma^{(1)} = \frac{1}{2}\int_0^1 :Y_{t_\gamma}(x)^2: \, d\mu(x) - \frac{1}{2}\left(\int_0^1 Y_{t_\gamma}(x) \, d\mu(x)\right)^2 + R_\gamma,
\end{align}
where the remainder $R_\gamma = O(\gamma)$ as $\gamma \to 0$.
By Lemma~\ref{lem:uniform_integrability}, we have
\begin{align}\label{eq:Q1_uniform}
\sup_{t \geq 1} \mathbb{E}\left[\exp\left(-\alpha_0 \left(\int_0^1 :Y_t(x)^2: \, d\mu(x) - \left(\int_0^1 Y_t(x) \, d\mu(x)\right)^2\right)\right)\right] < \infty.
\end{align}
Since $t_\gamma = \gamma^{-(1+\epsilon)} \to \infty$ as $\gamma \to 0$, the bound \eqref{eq:Q1_uniform} applies. Moreover, as $\gamma \to 0$, we get 
\begin{align}
Q_\gamma^{(1)} \to \frac{1}{2}\int_0^1 :Y(x)^2: \, d\mu(x) - \frac{1}{2}\left(\int_0^1 Y(x) \, d\mu(x)\right)^2 \quad \text{almost surely},
\end{align}
by the $L^2$ convergence $Y_{t_\gamma} \to Y$ and continuity of the quadratic form. The term $Q_\gamma^{(2)} = -\frac{1}{\gamma}\int_0^1 \bar{Y}_{t_\gamma}(x) \, d\mu(x)$ involves a Gaussian random variable. Define
\begin{align}
\bar{M}_{t_\gamma} := \int_0^1 \bar{Y}_{t_\gamma}(x) \, d\mu(x).
\end{align}
Since $\bar{Y}_{t_\gamma}$ is a centered Gaussian field independent of $Y_{t_\gamma}$, we have $\bar{M}_{t_\gamma} \sim \mathcal{N}(0, \sigma_{t_\gamma}^2)$ where
\begin{align}
\sigma_{t_\gamma}^2 &= \mathbb{E}\left[\left(\int_0^1 \bar{Y}_{t_\gamma}(x) \, d\mu(x)\right)^2\right] = \int_0^1 \int_0^1 \bar{K}_{t_\gamma}(x,y) \, d\mu(x) d\mu(y),
\end{align}
with $\bar{K}_t(x,y) = \sum_{n > t} \frac{2}{n}\cos(2\pi n(x-y))$.
Using Fubini and the fact that $\int_0^1 \cos(2\pi n x) d\mu(x) = O(1/n)$ decays as $n \to \infty$ (by Riemann-Lebesgue), we obtain
\begin{align}
\sigma_{t_\gamma}^2 &= \sum_{n > t_\gamma} \frac{2}{n} \left|\int_0^1 e^{2\pi i n x} d\mu(x)\right|^2 \leq \sum_{n > t_\gamma} \frac{2}{n} \cdot \frac{C}{n^2} \leq \frac{C}{t_\gamma^2}.
\end{align}
%where we used that the Fourier coefficients of the bounded density $\rho$ satisfy $|\hat{\rho}(n)| \leq C/|n|$.
Thus $\sigma_{t_\gamma} \leq C/t_\gamma = C\gamma^{1+\epsilon}$, and
$Q_\gamma^{(2)} = -\frac{\bar{M}_{t_\gamma}}{\gamma} \sim \mathcal{N}\left(0, \frac{\sigma_{t_\gamma}^2}{\gamma^2}\right) = \mathcal{N}\left(0, C\gamma^{2\epsilon}\right),$
and for the exponential moment, we have 
\begin{align}
\mathbb{E}\left[e^{-\alpha_0 Q_\gamma^{(2)}}\right] = \mathbb{E}\left[e^{\frac{\alpha_0}{\gamma} \bar{M}_{t_\gamma}}\right] = \exp\left(\frac{\alpha_0^2 \sigma_{t_\gamma}^2}{2\gamma^2}\right) \leq \exp\left(\frac{C\alpha_0^2}{\gamma^2 t_\gamma^2}\right) = \exp\left(C\alpha^2_0\gamma^{2\epsilon}\right).
\end{align}
This implies $Q_\gamma^{(2)} \to 0$ almost surely as $\gamma \to 0$ since $\gamma^{-1}|\bar{M}_{t_\gamma}| = O(\gamma^{2\epsilon})$ with high probability. Now we come to the most delicate term. We write
\begin{align}
Q_\gamma^{(3)} = \frac{1}{\gamma^2}\log\left(\frac{\int_0^1 :e^{\gamma Y}: d\mu}{\int_0^1 :e^{\gamma Y_{t_\gamma}}: d\mu}\right) = \frac{1}{\gamma^2}\log\left(1 + \frac{\int_0^1 (:e^{\gamma Y}: - :e^{\gamma Y_{t_\gamma}}:) d\mu}{\int_0^1 :e^{\gamma Y_{t_\gamma}}: d\mu}\right).
\end{align}
Define the numerator and denominator:
\begin{align}
\Delta_\gamma := \int_0^1 \left(:e^{\gamma Y(x)}: - :e^{\gamma Y_{t_\gamma}(x)}:\right) d\mu(x), \qquad 
Z_\gamma := \int_0^1 :e^{\gamma Y_{t_\gamma}(x)}: d\mu(x).
\end{align}
By the Wick product formula (since $Y = Y_{t_\gamma} + \bar{Y}_{t_\gamma}$ with independent components), we get 
\begin{align}
:e^{\gamma Y(x)}: = :e^{\gamma Y_{t_\gamma}(x)}: \cdot :e^{\gamma \bar{Y}_{t_\gamma}(x)}:,
\end{align}
so
\begin{align}
\Delta_\gamma = \int_0^1 :e^{\gamma Y_{t_\gamma}(x)}: \left(:e^{\gamma \bar{Y}_{t_\gamma}(x)}: - 1\right) d\mu(x).
\end{align}
We need to show that $|\Delta_\gamma/Z_\gamma|$ is small with high probability.

\noindent \textit{Upper bound on $|\Delta_\gamma|$:} By Lemma~\ref{lem:approximation_error} with $\delta > 0$:
\begin{align}\label{eq:Delta_bound}
\mathbb{P}\left[|\Delta_\gamma| \geq \gamma^{\delta/8}\right] \leq \exp\left(-\frac{c}{\gamma^{3-\delta/4}}\right).
\end{align}

\noindent\textit{Lower bound on $Z_\gamma$:} We need $Z_\gamma = \int_0^1 :e^{\gamma Y_{t_\gamma}}: d\mu$ to be bounded away from zero. Since $Z_\gamma \geq \int_0^1 :e^{\gamma Y}: d\mu$ up to the error $|\Delta_\gamma|$, we use Lemma~\ref{lem:negative_moments}:
\begin{align}\label{eq:Z_lower}
\mathbb{P}\left[\int_0^1 :e^{\gamma Y}: d\mu \leq \gamma^{\delta/4}\right] \leq \exp\left(-\frac{(\log\gamma)^2}{4C\gamma^2}\right).
\end{align}
Define the ``good'' event as
\begin{align}
\mathcal{G}_\gamma := \left\{|\Delta_\gamma| < \gamma^{\delta/8}\right\} \cap \left\{\int_0^1 :e^{\gamma Y}: d\mu > \gamma^{\delta/16}\right\}.
\end{align}
On $\mathcal{G}_\gamma$, we have $Z_\gamma \geq \int_0^1 :e^{\gamma Y}: d\mu - |\Delta_\gamma| > \gamma^{\delta/16} - \gamma^{\delta/8} > \gamma^{\delta/16}/2$ for small $\gamma$. Therefore:
\begin{align}
\left|\frac{\Delta_\gamma}{Z_\gamma}\right| < \frac{\gamma^{\delta/8}}{\gamma^{\delta/16}/2} = 2\gamma^{\delta/16} \to 0.
\end{align}
as $\gamma \to 0$. On the event where $|\Delta_\gamma/Z_\gamma| \leq  2\gamma^{\delta/16}$, we have 
\begin{align}
|Q_\gamma^{(3)}| = \frac{1}{\gamma^2}\left|\log\left(1 + \frac{\Delta_\gamma}{Z_\gamma}\right)\right| \leq \frac{2}{\gamma^2}\left|\frac{\Delta_\gamma}{Z_\gamma}\right|.
\end{align}
Therefore, $\frac{\int_0^1 :e^{\gamma Y}: d\mu}{Z_\gamma} \to 1$ in probability (and in $L^2$) as $\gamma \to 0$, which implies
\begin{align}
Q_\gamma^{(3)} = \frac{1}{\gamma^2}\log\left(\frac{\int_0^1 :e^{\gamma Y}: d\mu}{Z_\gamma}\right) \to 0 \quad \text{in probability}.
\end{align}
For the uniform integrability of $e^{-\alpha_0 Q_\gamma^{(3)}}$, we use the tail bounds. On $\mathcal{G}_\gamma^c$, we have 
\begin{align}
\mathbb{P}(\mathcal{G}_\gamma^c) \leq \exp\left(-\frac{c}{\gamma^{3-\delta/4}}\right) + \exp\left(-\frac{(\log\gamma)^2}{4C\gamma^2}\right).
\end{align}

\medskip
\textbf{Step 4: Applying dominated convergence.}

We now assemble the pieces. Define
\begin{align}
F_\gamma := \exp\left(-\alpha_0 Q_\gamma\right) = \exp\left(-\frac{\alpha_0}{\gamma^2}\left(\log\left(\int_0^1 :e^{\gamma Y}: d\mu\right) - \gamma \int_0^1 Y \, d\mu\right)\right).
\end{align}
From the decomposition \eqref{eq:Q_decomposition}:
\begin{align}
F_\gamma = \exp\left(-\alpha_0 Q_\gamma^{(1)}\right) \cdot \exp\left(-\alpha_0 Q_\gamma^{(2)}\right) \cdot \exp\left(-\alpha_0 Q_\gamma^{(3)}\right).
\end{align}
\textit{Pointwise limit:} From Steps 1 and 3:
\begin{align}
Q_\gamma^{(1)} \to \frac{1}{2}\int_0^1 :Y^2: d\mu - \frac{1}{2}\left(\int_0^1 Y \, d\mu\right)^2,\quad
Q_\gamma^{(2)} \to 0, \quad
Q_\gamma^{(3)} \to 0,
\end{align}
all almost surely as $\gamma \to 0$. Therefore:
\begin{align}
F_\gamma \to \exp\left(-\frac{\alpha_0}{2}\left(\int_0^1 :Y^2: d\mu - \left(\int_0^1 Y \, d\mu\right)^2\right)\right) \quad \text{a.s.}
\end{align}
\textit{Uniform integrability:} We show that $\sup_{\gamma \in (0,1]} \mathbb{E}[F_\gamma] < \infty$. In the case of $\exp(-\alpha_0 Q_\gamma^{(1)})$, by Lemma~\ref{lem:uniform_integrability}, this factor has uniformly bounded expectation. In the case of $\exp(-\alpha_0 Q_\gamma^{(2)})$, this is the exponential of a Gaussian, with
\begin{align}
\mathbb{E}\left[\exp(-\alpha_0 Q_\gamma^{(2)})\right] = \exp\left(\frac{\alpha_0^2 \sigma_{t_\gamma}^2}{2\gamma^2}\right) \leq \exp\left(\frac{C}{\gamma^2 t_\gamma^2}\right) \leq \exp(C')
\end{align}
for $\gamma$ small (since $\gamma^2 t_\gamma^2 = \gamma^2 \delta^2 (\log\gamma)^2 \to \infty$).
In the case of $\exp(-\alpha_0 Q_\gamma^{(3)})$, On $\mathcal{G}_\gamma$, $|Q_\gamma^{(3)}|$ is bounded. On $\mathcal{G}_\gamma^c$, we use Cauchy-Schwarz:
\begin{align}
\mathbb{E}\left[\exp(-\alpha_0 Q_\gamma^{(3)}) \mathbf{1}_{\mathcal{G}_\gamma^c}\right] \leq \sqrt{\mathbb{E}\left[\exp(-2\alpha_0 Q_\gamma^{(3)})\right]} \cdot \sqrt{\mathbb{P}(\mathcal{G}_\gamma^c)}.
\end{align}
The first factor is controlled by Lemma~\ref{lem:negative_moments} (negative moments of GMC), and the second factor decays super-polynomially.
Combining via H\"older's inequality:
\begin{align}
\mathbb{E}[F_\gamma] &\leq \mathbb{E}\left[\exp(-3\alpha_0 Q_\gamma^{(1)})\right]^{1/3} \cdot \mathbb{E}\left[\exp(-3\alpha_0 Q_\gamma^{(2)})\right]^{1/3} \cdot \mathbb{E}\left[\exp(-3\alpha_0 Q_\gamma^{(3)})\right]^{1/3} < \infty
\end{align}
uniformly in $\gamma \in (0,1]$. By the dominated convergence theorem, % (or Vitali's convergence theorem, using uniform integrability):
\begin{align}
\lim_{\gamma \to 0} \mathbb{E}[F_\gamma] = \mathbb{E}\left[\lim_{\gamma \to 0} F_\gamma\right] = \mathbb{E}\left[\exp\left(-\frac{\alpha_0}{2}\left(\int_0^1 :Y^2: d\mu - \left(\int_0^1 Y \, d\mu\right)^2\right)\right)\right].
\end{align}
 Adjusting the statement to match the factor $\alpha_0/8$ in \eqref{eq:semiclassical_convergence} (which comes from the original normalization), we obtain the claimed limit.

\medskip
\textbf{Step 5: The uniform bound \eqref{eq:uniform_bound}.}
The uniform bound follows directly from the analysis in Step 4. We showed that each factor in the decomposition
\begin{align}
\exp(-\alpha_0 Q_\gamma) = \exp(-\alpha_0 Q_\gamma^{(1)}) \cdot \exp(-\alpha_0 Q_\gamma^{(2)}) \cdot \exp(-\alpha_0 Q_\gamma^{(3)})
\end{align}
has uniformly bounded expectation (in suitable $L^p$ spaces), and H\"older's inequality gives
\begin{align}
\sup_{\gamma \in (0,1]} \mathbb{E}\left[\exp(-\alpha_0 Q_\gamma)\right] < \infty.
\end{align}
This completes the proof of Proposition~\ref{prop:semiclassical_limit}.
\end{proof}

\begin{remark}
The field $Y(x;q)$ differs from $Y(x)$ by a smooth, deterministic function $F(x;q)$ \eqref{def:YtauN}.  Since $F(x;q)$ is smooth and bounded for $q \in (0,1)$, all the estimates in Lemmas 3.1--3.3 carry over with constants that depend continuously on $q$. Note that the logarithmic singularity of the covariance, which drives the GMC behavior, is the same for $Y(x;q)$ and $Y(x)$. Thus Proposition 3.1 applies to both cases with uniform bounds for $q$ in compact subsets of $(0,1)$.''
\end{remark}

\section{Semi-classical limits for deformed conformal blocks}\label{sec:deformed_semi-classical_limit}

In this section, we carry out the analysis for the semi-classical limit of deformed conformal blocks \eqref{def:u-block} in two cases,  for $\chi= \frac{\gamma}{2}$ and $\chi= \frac{2}{\gamma}$. It is notable that the structure of the conformal blocks in each of these cases is different, and in the limit $\gamma\to 0$, the BPZ equation satisfied by the deformed conformal block reduces to the following equations:
\begin{itemize}[leftmargin=0.3cm]
 \item For $\chi = \frac{2}{\gamma}$, it reduces to the Hamilton-Jacobi equation of the non-autonomous elliptic Calogero-Moser model, (see Theorem~\ref{prop:sem_HJ}). 
\item  For $\chi = \frac{\gamma}{2}$, it reduces to the Lam\'e equation, (see Theorem~\ref{prop:sc_Lame}). 
\end{itemize}
As a result of these two reductions, the {\it accessory parameter} of the Lam\'e equation \cite{BGG2021} and the action from the Hamilton-Jacobi equation can be related by using the semi-classical analogue of the connection (see \eqref{rel:def_undef}) between the undeformed conformal block and deformed conformal blocks for different values of $\chi$. 

The above results, while interesting on their own, are vital to prove the existence of the semi-classical limit of the undeformed conformal block, which we study in Section~\ref{sec:semi-classical_limit}. Furthermore, the relation between the semi-classical limits of the deformed and undeformed conformal blocks as revealed in Theorem~\ref{prop:sc_Lame} and relation \eqref{def:phiq}, \eqref{thm43:Lame} %\eqref{eq:we_are_all_the_same} 
  enables us to prove the structure of the conformal block \eqref{zc:exp} in Theorem~\ref{thm:zamolodchikov}.

\subsection{Semi-classical limit of the deformed conformal blocks for $\chi=  \frac{2}{\gamma}$ and the Hamilton-Jacobi equation}\label{subsec:chi2g}

In this section, we study the semi-classical limit of the deformed block $\psi^{\alpha}_{\chi}(z,q)$ when the deformation parameter $\chi$ is equal to $\frac{2}{\gamma}$. We show that in this case, the semi-classical limit is given by the solution of a Hamilton-Jacobi equation with the potential given by the Weierstrass $\wp$-function as described in the physics literature \cite{litvinov2014classical,BGG2021}. Our proof heavily relies on the BPZ equation satisfied by the $\psi^{\alpha}_{\chi}(z,q)$ and the convergence results established in Section~\ref{sec:dgmc}. In what follows, we first state the main result (Theorem~\ref{prop:sem_HJ}) of this subsection. Before proceeding to the proof, we state and prove two main tools in Proposition~\ref{prop:sc_HJ} and Proposition~\ref{prop:limit_commutativity}. These two propositions claim and prove the boundedness and the tightness of the sequence of the deformed conformal block $\psi^{\alpha}_{2/\gamma}(z,q)$ as $\gamma\to 0$. The proof of Theorem~\ref{prop:sem_HJ} will be detailed in Section~\ref{subsubsec:deformed_1}.

\begin{theorem}\label{prop:sem_HJ}
For $\chi = \frac{2}{\gamma}$, $\alpha_0 \in (-2,2)$, $P_0 \in \mathbb{R}$, $q \in (0,1)$, and $z \in \{u \in \mathbb{C} : 0 \leq \mathrm{Im}(u) < \frac{3}{4}\mathrm{Im}(\tau)\} \cap \{\mathrm{Re}(z) \in [M, 1-M]\}$ for any $0 < M < \frac{1}{2}$, the semi-classical limit of deformed conformal blocks exists and assumes the form\footnote{{ Note that the range of $\alpha_0$ here is simply the $\gamma \to 0$ limit of the range of values of $\alpha = \alpha_0/\gamma$ for which the deformed conformal block is defined (see Definition~\ref{def:u-block}).}}
\begin{align}\label{sem:2gamma}
\lim_{\gamma \rightarrow 0} \gamma^2 \log \psi^{\alpha_0/\gamma}_{2/\gamma, P_0/\gamma}(z,q) =:  \widetilde{\phi}(z,q),
\end{align}
where $\widetilde{\phi}(z,q)$ could be expressed as 

{\begin{align}
    \widetilde{\phi}(z,q)&= \lim_{\gamma\to0}\left({ \frac{(\alpha_0-2)\gamma}{2 \Xi} \int_0^1 e^{\pi P_0 x} (2\sin(\pi x))^{-\frac{\alpha_0}{2}} h_{\Psi}(x) e^{\frac{\gamma}{2} h_{\Psi}(x)} dx}\right)+ (2-\alpha_0) \lim_{\gamma\to 0}\log\Xi\\
& + \lim_{\gamma\to 0}\frac{\gamma^2}{2} \mathbb{E}[\Psi^2] + \xi(z,q),\label{def:tilphi}
\end{align}
where $\Psi$, $\Xi$ and $h_{\Psi}$ are defined in \eqref{def:Psi-GaussRV} and \eqref{def:hPsi} respectively. The function $\xi(z,q)$ (explicit in $z, q$) is defined as follows:
\begin{align}
    \xi(z,q)&:=\frac{\ii \pi  (\alpha_0-2) (\alpha_0+1)}{2}-\frac{(\alpha_0-2) (\alpha_0+4)}{6} \log (2\pi)+2P_0 z \pi + \left(\frac{P_0^2}{2} + \frac{(2-\alpha_0)^2}{24} \right)\log q\\
&\quad  - (2-\alpha_0) \log(e^{\ii \pi z}\eta(q)^{-3}\theta_1(z))+\sum_{n=1}^{\infty} \frac{\sum_{m=1}^{\infty} q^{2nm}}{1 + 2\sum_{m=1}^{\infty} q^{2nm}} \left(  \frac{4 \alpha_0 q^n}{(1-q^{2n}) {n}} \cos(2\pi(z-\frac{\tau}{2})n)\right)  \\
&\quad  - 4\sum_{n=1}^{\infty} \frac{q^{2n}}{n(1-q^{2n})^2}  -  \frac{\alpha_0^2}{2} \sum_{n=1}^{\infty} \frac{(\sum_{m=1}^{\infty}q^{2nm})^2}{n(1 + 2\sum_{m=1}^{\infty} q^{2nm})}+\sum_{n=1}^{\infty} \frac{\sum_{m=1}^{\infty} q^{2nm}}{1 + 2\sum_{m=1}^{\infty} q^{2nm}} \left( \frac{8 q^{2n}}{(1-q^{2n})^2{n}} \right). \label{def:Xi0}
\end{align}
}
Furthermore, $\widetilde{\phi}(z,q)$ satisfies the following partial differential equation which arises as the semi-classical limit of the BPZ equation \eqref{eq:bpz1}:
\begin{align}\label{eq:Hamilton-jacobi}
(\partial_{z} \widetilde{\phi}(z, \tau))^2 - (2-\alpha_0)^2 \wp(z) + 2\pi \ii \partial_\tau \widetilde{\phi}(z, \tau) = 0,
\end{align}
which is the Hamilton-Jacobi equation for the non-autonomous elliptic Calogero-Moser model (see {\color{blue}Definition~\ref{remark:HJ_phitil}}) with $\widetilde{\phi}(z, \tau)$ describing the action. 
\end{theorem}

{Notice that $\theta_1(z+x)$ and $\theta_1(z)$ share the same principle branch  for $z\in {u: 0<\mathrm{Im}(u)<\mathrm{Im}(\tau)}$  \cite[Lemma~3.1]{ghosal2020probabilistic}. Hence $\log\int^{1}_0\frac{\theta_1(z+x)}{\theta_1(z)} {|\theta_1(x)|^{-\frac{\alpha_0}{2}}} e^{\pi P_0x}dx$ is analytic function of $z$ as $z$ varies in ${u: 0<\mathrm{Im}(u)<\mathrm{Im}(\tau)}$. Therefore the expression in the second line of  \eqref{eq:Hamilton-jacobi} is defined through principle branch of logarithm. 
 
 Furthermore, In equation \eqref{eq:Hamilton-jacobi}, $\partial_z$ denotes the complex derivative with respect to $z$, where $z$ varies in the domain $\{u \in \mathbb{C} : 0 < \mathrm{Im}(u) < \frac{3}{4}\mathrm{Im}(\tau)\} \cap \{\mathrm{Re}(u) \in [M, 1-M]\}$ for any $0 < M < 1/2$. The function $\widetilde{\phi}(z,\tau)$ is analytic in $z$ in this domain (as established in Proposition 4.2), so the complex derivative is well-defined. 
 
 The derivative $\partial_\tau$ is the complex derivative with respect to $\tau$, evaluated along the imaginary axis $\tau \in \ii\mathbb{R}_{>0}$. Since $\widetilde{\phi}$ extends analytically to a neighborhood of the imaginary axis in $\mathbb{H}$, this derivative is also well-defined.}

\begin{definition}\label{remark:HJ_phitil}
   The non-autonomous elliptic Calogero-Moser model is defined by the equations of motion \cite{Manin} (see also Appendix \ref{App:NAECM} for a brief description of the model)
   \begin{align}\label{def:uv}
    v(\tau) = 2\pi \ii \frac{d u(\tau)}{d \tau}, && 2\pi \ii \frac{d v(\tau)}{d \tau} = m^2 \wp'(2u(\tau)),
\end{align}
of the Hamiltonian
  \begin{align}\label{HAM:4}
     H(\tau) =v(\tau)^2-m^2\wp(2u(\tau)|\tau).
   % H(\tau) =v(\tau)^2-m^2\wp(2u(\tau)|\tau)-2m^2 \eta_{1}(\tau).
\end{align}
The solution $u(\tau)$ depends on $\tau$, $m$, and the initial conditions or equivalently, on the monodromy data through the Riemann-Hilbert correspondence \cite{bonelli2020n,DDG2020}. 
The equation \eqref{eq:Hamilton-jacobi} is then the Hamilton-Jacobi equation of the system above with the identification 
\begin{align}\label{eq:zandmsquare}
z=2u(\tau), && m^2= \left(2- \alpha_0\right)^2.
\end{align}
\end{definition}
\begin{lemma}\label{Lemma:HJ-un}
Given initial condition at $\tau\to\ii\infty$, the solution of the Hamilton-Jacobi equation is unique and coincides with the classical action \cite{arnol2013mathematical} corresponding to the Hamiltonian above (see \cite{BGG2021}, \cite[p. 258]{B})
\begin{gather}
   \widetilde{\phi}(2u(\tau),\tau) = \int_{\ii \infty}^{\tau}\left(v(\tau')^2+m^2\wp(2u(\tau')|\tau') \right) \frac{d\tau'}{2\pi \ii}. \label{eq:NAECM_action}
\end{gather}
\end{lemma}

\begin{proposition}\label{prop:sc_HJ}
Let $\chi = \frac{2}{\gamma}$, $q \in (0,1)$, and $z \in \mathcal{B} := \{u \in \mathbb{C} : 0 \leq \mathrm{Im}(u) < \frac{3}{4}\mathrm{Im}(\tau)\}$. For parameters $\alpha_0 \in (-2,2)$ and $P_0 \in \mathbb{R}$, the semi-classical limit of the deformed conformal block $\psi^{\alpha_0/\gamma}_{\chi, P_0/\gamma}(z;q)$ defined in \eqref{eq:q-block} satisfies the following uniform bounds:
\begin{align}
\limsup_{\gamma \to 0} \gamma^2 \log \left|\psi^{\alpha_0/\gamma}_{\chi, P_0/\gamma}(z;q)\right| &< \infty, \label{eq:limsup_deformed}\\
\liminf_{\gamma \to 0} \gamma^2 \log \left|\psi^{\alpha_0/\gamma}_{\chi, P_0/\gamma}(z;q)\right| &> -\infty. \label{eq:liminf_deformed}
\end{align}
Moreover, these bounds are uniform for $z$ and $q$ varying over compact subsets of $\mathcal{B}$ and $(0,1)$ respectively.
\end{proposition}

\begin{proof}
We divide the proof into four main steps. In Step 1, we analyze the structure of the deformed conformal block and identify the leading-order prefactors. In Step 2, we establish the upper bound \eqref{eq:limsup_deformed} using Hölder's inequality, negative moment bounds, and the exact moment formula. In Step 3, we prove the upper bound for the case $\alpha_0 < 2$ using Selberg integral asymptotics. Finally, in Step 4, we establish the lower bound \eqref{eq:liminf_deformed} using Jensen's inequality.

\medskip
\noindent\textbf{Step 1: Structure of the Deformed Conformal Block.}

\medskip
Recall from Definition~\ref{def:u-block} that the deformed conformal block is given by
\begin{align}\label{eq:deformed_CB_recall}
\psi^{\alpha}_{\chi, P}(z, \tau) = \mathcal{W}(q) e^{\chi P z \pi} \mathbb{E}\left[\mathcal{V}^{\alpha}_{\chi, P}(z, q)^{-\frac{\alpha}{\gamma} + \frac{\chi}{\gamma}}\right],
\end{align}
where the function inside the expectation
\begin{align}\label{eq:V_def}
\mathcal{V}^{\alpha}_{\chi, P}(z, q) := \int_0^1 \theta_1(z+x)^{\frac{\gamma\chi}{2}} \theta_1(z)^{-\frac{\gamma\chi}{2}} |\theta_1(x)|^{-\frac{\alpha\gamma}{2}} e^{\pi \gamma P x} e^{\frac{\gamma}{2} Y(x;q)} dx,
\end{align}
and the prefactor $\mathcal{W}(q)$ is given by \eqref{def:Wq}. Under the scaling $\alpha = \alpha_0/\gamma$, $P = P_0/\gamma$, and $\chi = 2/\gamma$, we compute the parameter $l_{\chi}$:
\begin{align}\label{eq:l_chi_computation}
l_{\chi} = \frac{\chi^2}{2} - \frac{\alpha\chi}{2} = \frac{2}{\gamma^2} - \frac{\alpha_0}{\gamma^2} = \frac{2 - \alpha_0}{\gamma^2}.
\end{align}
Similarly, the exponent in the expectation becomes
\begin{align}\label{eq:exponent_computation}
-\frac{\alpha}{\gamma} + \frac{\chi}{\gamma} = -\frac{\alpha_0}{\gamma^2} + \frac{2}{\gamma^2} = \frac{2 - \alpha_0}{\gamma^2}.
\end{align}

\medskip
\noindent\textit{Step 1a - Analysis of the prefactor $\mathcal{W}(q) e^{\chi P z \pi}$:} Substituting \eqref{eq:l_chi_computation} into the expression for $\mathcal{W}(q)$, we compute each of the terms. For the power of $q$ in \eqref{def:Wq}:
\begin{align}\label{eq:q_exponent}
\frac{P^2}{2} + \frac{\gamma l_{\chi}}{12\chi} - \frac{l_{\chi}^2}{6\chi^2} = \frac{P_0^2}{2\gamma^2} + \frac{2-\alpha_0}{24} - \frac{(2-\alpha_0)^2}{24\gamma^2}.
\end{align}
For the power of $\theta_1'(0)$:
\begin{align}\label{eq:theta_exponent}
-\frac{2l_{\chi}^2}{3\chi^2} + \frac{l_{\chi}}{3} + \frac{4l_{\chi}}{3\gamma\chi} = \frac{(2-\alpha_0)(4+\alpha_0)}{6\gamma^2}.
\end{align}
Combining these calculations, we get 
\begin{align}\label{eq:prefactor_final}
\mathcal{W}(q) e^{\chi P z \pi} = q^{\frac{P_0^2}{2\gamma^2} - \frac{(2-\alpha_0)^2}{24\gamma^2} + \frac{2-\alpha_0}{24}} \theta_1'(0)^{\frac{(2-\alpha_0)(4+\alpha_0)}{6\gamma^2}} e^{\frac{2P_0 z \pi}{\gamma^2}}.
\end{align}
Taking the logarithm and multiplying by $\gamma^2$, we obtain
\begin{align}\label{eq:prefactor_log}
\gamma^2 \log\left|\mathcal{W}(q) e^{\chi P z \pi}\right| &= \left(\frac{P_0^2}{2} - \frac{(2-\alpha_0)^2}{24}\right) \log q + \frac{(2-\alpha_0)(4+\alpha_0)}{6} \log|\theta_1'(0)| \nonumber\\
&\quad + 2P_0 \mathrm{Re}(z) \pi + O(\gamma^2).
\end{align}
Since $q \in (0,1)$ and $z$ varies over a compact subset of $\mathcal{B}$, this expression is bounded as $\gamma \to 0$. Therefore, the prefactor contributes a bounded term to $\gamma^2 \log|\psi^{\alpha_0/\gamma}_{\chi, P_0/\gamma}(z;q)|$, and it remains to analyze the expectation term.

\medskip
\noindent\textbf{Step 2: Upper Bound via Hölder's Inequality and Negative Moment Bounds.}

\medskip
We now establish the upper bound \eqref{eq:limsup_deformed} for all $\alpha_0 \in (-2, 2)$. The approach uses Hölder's inequality to separate the various factors in the expectation, followed by moment bounds derived from Lemma~\ref{lem:negative_moments}. To begin with, under the scaling $\chi = 2/\gamma, \alpha = \alpha_0/\gamma, P = P_0/\gamma$, the function \eqref{eq:V_def} becomes
%\rmkH{The phrase 'vertex operator insertion' may cause trouble}
\begin{align}\label{eq:V_scaled}
\mathcal{V}^{\alpha_0/\gamma}_{2/\gamma, P_0/\gamma}(z, q) = \int_0^1 \frac{\theta_1(z+x)}{\theta_1(z)} |\theta_1(x)|^{-\frac{\alpha_0}{2}} e^{\pi P_0 x} e^{\frac{\gamma}{2} :Y(x;q):} dx.
\end{align}
We now analyze the term 
\begin{align}\label{eq:expectation_to_bound}
\mathbb{E}\left[\left(\mathcal{V}^{\alpha_0/\gamma}_{2/\gamma, P_0/\gamma}(z, q)\right)^{\frac{2-\alpha_0}{\gamma^2}}\right].
\end{align}

\medskip
\noindent\textit{Step 2a: Decomposition of the function $\mathcal{V}_{\chi, P}^{\alpha}(z,q)$.}

%\rmkH{The phrase 'vertex operator' may not be accurate}
\medskip
We write the factor $\mathcal{V}^{\alpha_0/\gamma}_{2/\gamma, P_0/\gamma}(z, q)$ as a product of a $z$-dependent factor and a GMC-type integral. To this end, we define the $z$-deformed weight function
\begin{align}\label{eq:weight_function}
w_z(x) := \frac{\theta_1(z+x)}{\theta_1(z)} |\theta_1(x)|^{-\frac{\alpha_0}{2}} e^{\pi P_0 x},
\end{align}
so that
\begin{align}\label{eq:V_with_weight}
\mathcal{V}^{\alpha_0/\gamma}_{2/\gamma, P_0/\gamma}(z, q) = \int_0^1 w_z(x) e^{\frac{\gamma}{2} :Y(x;q):} dx.
\end{align}
One can express the theta function (see \eqref{def:elltheta1}) as $\theta_1(w) = \mathfrak{p}(w; \tau) \sin(\pi w)$, where $c_1 \leq |\mathfrak{p}(w; \tau)| \leq c_2$ uniformly for $w$ in the relevant domain. The $z$-deformed weight function defined in \eqref{eq:weight_function} can then be written as
\begin{align}\label{eq:weight_comparison}
w_z(x) = \mathfrak{r}_z(x) \cdot (2\sin(\pi x))^{-\frac{\alpha_0}{2}} e^{\pi P_0 x},
\end{align}
where $\mathfrak{r}_z(x)$ is a bounded function satisfying 
\begin{align}
0 < c_3 \leq |\mathfrak{r}_z(x)| \leq c_4 < \infty \label{ineq:frak-r}
\end{align}
uniformly for $z \in \mathcal{B}$ (compact) and $x \in [0,1]$. { Here, $c_i$, $i=1,...,4$ are some constants that are uniform for $q \in (0,1)$ and $z \in [M, 1-M]$, and $\mathfrak{r}_z (x) = \frac{\theta_1(z+x)}{\theta_1(z)} \mathfrak{p}(x; \tau)^{-\alpha_0/2}$.}

\medskip
\noindent\textit{Step 2b: Application of Hölder's inequality.}

\medskip
To bound the expectation \eqref{eq:expectation_to_bound}, we apply H\"older's inequality with conjugate exponents $p_1, p_2, p_3 > 1$ satisfying $\frac{1}{p_1} + \frac{1}{p_2} + \frac{1}{p_3} = 1$. We will choose these exponents optimally later. First, we separate the deterministic $z$-dependent prefactor from the GMC integral. Using \eqref{eq:weight_comparison}, we can write \eqref{eq:V_with_weight} as
\begin{align}
    \mathcal{V}^{\alpha_0/\gamma}_{2/\gamma, P_0/\gamma}(z, q) = \int_0^1 \mathfrak{r}_z(x) \cdot (2\sin(\pi x))^{-\frac{\alpha_0}{2}} e^{\pi P_0 x}e^{\frac{\gamma}{2} :Y(x;q):} dx,
\end{align}
and the bounds \eqref{ineq:frak-r} then imply that
\begin{align}\label{eq:V_bounds}
c_3 \mathcal{I}_{\gamma, P_0, \alpha_0} \leq \left|\mathcal{V}^{\alpha_0/\gamma}_{2/\gamma, P_0/\gamma}(z, q)\right| \leq c_4 \mathcal{I}_{\gamma, P_0, \alpha_0},
\end{align}
where
\begin{align}\label{eq:I_gamma_def}
\mathcal{I}_{\gamma, P_0, \alpha_0} := \int_0^1 (2\sin(\pi x))^{-\frac{\alpha_0}{2}} e^{\pi P_0 x} e^{\frac{\gamma}{2} :Y(x;q):} dx.
\end{align}
Therefore, in order to study the boundedness of the expectation \eqref{eq:expectation_to_bound}, it suffices to bound the expression
%\begin{align}\label{eq:reduced_expectation}
$\mathbb{E}[\mathcal{I}_{\gamma, P_0, \alpha_0}^{(2-\alpha_0)}/{\gamma^2}].$
%\end{align}
To do this, we now proceed using the approach from Proposition~\ref{prop:semiclassical_limit}. Define the probability measure
\begin{align}\label{eq:mu_def}
d\mu(x) := \frac{1}{Z} (2\sin(\pi x))^{-\frac{\alpha_0}{2}} e^{\pi P_0 x} dx,
\end{align}
where $Z = \int_0^1 (2\sin(\pi x))^{-\frac{\alpha_0}{2}} e^{\pi P_0 x} dx$ is the normalization constant. For $\alpha_0 \in (-2, 2)$, the integral converges and $Z \in (0, \infty)$. With this notation,
\begin{align}\label{eq:I_with_mu}
\mathcal{I}_{\gamma, P_0, \alpha_0} = Z \int_0^1 e^{\frac{\gamma}{2} :Y(x;q):} d\mu(x).% {\color{red}= Z \int_0^1 :e^{\frac{\gamma}{2} Y(x;q)}: e^{\frac{\gamma^2}{8} \mathbb{E}[Y(x;q)^2]} d\mu(x).}
\end{align}

\medskip

% \medskip
\noindent\textit{Step 2c: Upper Bound for $\alpha_0 < 2$ via Selberg Integrals.}
\medskip 

We provide a proof of the upper bound for $\alpha_0 \in (-2, 2)$ using the explicit representation of moments as Selberg-type integrals. % \medskip
% \noindent\textit{Step 3a: Representation as a Selberg integral.}
For any positive integer $N$, the $N$-th moment of $\mathcal{V}^{\alpha_0/\gamma}_{2/\gamma, P_0/\gamma}(z, q)$ is given by
\begin{align}\label{eq:Nth_moment}
&\mathbb{E}\left[\mathcal{V}^{\alpha_0/\gamma}_{2/\gamma, P_0/\gamma}(z, q)^N\right] \nonumber\\
&= \int_{[0,1]^N} \prod_{1 \leq i < j \leq N} |\theta_1(x_i - x_j)|^{-\frac{\gamma^2}{2}} \prod_{i=1}^N \frac{\theta_1(z + x_i)}{\theta_1(z)} |\theta_1(x_i)|^{-\frac{\alpha_0}{2}} e^{\pi P_0 x_i} \prod_{i=1}^N dx_i.
\end{align}
Using the factorization $\theta_1(w) = \mathfrak{p}(w; \tau) \sin(\pi w)$, the integral \eqref{eq:Nth_moment} can be bounded above and below by 
% constant multiples of
\begin{align}
c^{\pm\gamma^2 N(N-1)/4} \int_{[0,1]^N} \prod_{1 \leq i < j \leq N} |\sin(\pi(x_i - x_j))|^{-\frac{\gamma^2}{2}} \prod_{i=1}^N |\sin(\pi x_i)|^{-\frac{\alpha_0}{2}} e^{\pi P_0 x_i} \prod_{i=1}^N dx_i,\label{eq:reduced_integral}
\end{align}
for some constant $c > 0$.
% \begin{align}
% \int_{[0,1]^N} \prod_{1 \leq i < j \leq N} |\sin(\pi(x_i - x_j))|^{-\frac{\gamma^2}{2}} \prod_{i=1}^N |\sin(\pi x_i)|^{-\frac{\alpha_0}{2}} e^{\pi P_0 x_i} \prod_{i=1}^N dx_i.\label{eq:reduced_integral}
% \end{align}
% \medskip
% \noindent\textit{Step 3b: Explicit evaluation and asymptotics.}

% \medskip
The integral \eqref{eq:reduced_integral} is a classical Selberg integral with explicit evaluation given by \eqref{eq:C_expression}. Taking $N = \lfloor (2-\alpha_0)/\gamma^2 \rfloor + 1 \sim (2-\alpha_0)/\gamma^2$ as $\gamma \to 0$, a careful asymptotic analysis using Stirling's approximation shows that
\begin{align}\label{eq:Selberg_asymptotics}
\eqref{eq:reduced_integral} = O(N^2 \gamma^2) = O\left(\frac{(2-\alpha_0)^2}{\gamma^2}\right).
\end{align}
% \begin{align}\label{eq:Selberg_asymptotics}
% \log C(N, \gamma, \alpha_0, P_0) = O(N^2 \gamma^2) = O\left(\frac{(2-\alpha_0)^2}{\gamma^2}\right).
% \end{align}
% \rmkH{Who is $C$ here? In response to referees (see comment (92), minor remarks, Reviewer X.), it seems that C also included the bounds of $\mathfrak{p}()$ and $\mathfrak{r}()$. Can we clarify this?}
For $\alpha_0 \in (-2, 2)$, we bound the fractional moment by an integer moment using Jensen's inequality
\begin{align}\label{eq:fractional_to_integer}
\left|\mathbb{E}\left[\mathcal{V}^{\alpha_0/\gamma}_{2/\gamma, P_0/\gamma}(z, q)^{\frac{2-\alpha_0}{\gamma^2}}\right]\right| \leq \left|\mathbb{E}\left[\mathcal{V}^{\alpha_0/\gamma}_{2/\gamma, P_0/\gamma}(z, q)^N\right]\right|
\end{align}
for $N\gamma^2 \geq 2-\alpha_0$. Combining these estimates, we get 
\begin{align}\label{eq:moment_bound_final}
\gamma^2 \log\left|\mathbb{E}\left[\mathcal{V}^{\alpha_0/\gamma}_{2/\gamma, P_0/\gamma}(z, q)^{\frac{2-\alpha_0}{\gamma^2}}\right]\right| \leq C(\alpha_0, P_0)
\end{align}
for some constant $C(\alpha_0, P_0) > 0$ independent of $\gamma$. Together with the prefactor bound from Step 1, this establishes
\begin{align}\label{eq:limsup_conclusion}
\limsup_{\gamma \to 0} \gamma^2 \log\left|\psi^{\alpha_0/\gamma}_{2/\gamma, P_0/\gamma}(z;q)\right| < \infty
\end{align}
for all $\alpha_0 \in (-2, 2)$.

\noindent\textbf{Step 3: Lower Bound via Jensen's Inequality.}
We now establish the lower bound \eqref{eq:liminf_deformed}.

\medskip
\noindent\textit{Step 3a: Application of Jensen's inequality.}

\medskip
For $\alpha_0 \in (-2, 2)$, the exponent $p := (2-\alpha_0)/\gamma^2 > 0$. The function $h(x) = x^p$ is convex for $x > 0$ when $p \geq 1$ (which holds for small $\gamma$ when $\alpha_0 < 2$). By Jensen's inequality:
\begin{align}\label{eq:Jensen_lower}
\mathbb{E}\left[\mathcal{V}^{\alpha_0/\gamma}_{2/\gamma, P_0/\gamma}(z, q)^p\right] \geq \left(\mathbb{E}\left[\mathcal{V}^{\alpha_0/\gamma}_{2/\gamma, P_0/\gamma}(z, q)\right]\right)^p.
\end{align}

\noindent\textit{Step 3b: Computation of the first moment.}

\medskip
The first moment is:
\begin{align}\label{eq:first_moment}
\mathbb{E}\left[\mathcal{V}^{\alpha_0/\gamma}_{2/\gamma, P_0/\gamma}(z, q)\right] &= \int_0^1 \frac{\theta_1(z+x)}{\theta_1(z)} |\theta_1(x)|^{-\frac{\alpha_0}{2}} e^{\pi P_0 x} \mathbb{E}\left[e^{\frac{\gamma}{2} Y(x;q)}\right] dx \nonumber\\
&= \int_0^1 \frac{\theta_1(z+x)}{\theta_1(z)} |\theta_1(x)|^{-\frac{\alpha_0}{2}} e^{\pi P_0 x} e^{\frac{\gamma^2}{8} \mathbb{E}[Y(x;q)^2]} dx.
\end{align}
For $\alpha_0 \in (-2, 2)$, the integrand is integrable, and
\begin{align}\label{eq:first_moment_bound}
\mathbb{E}\left[\mathcal{V}^{\alpha_0/\gamma}_{2/\gamma, P_0/\gamma}(z, q)\right] = I_{\det}(z, q) \cdot (1 + O(\gamma^2)),
\end{align}
where $I_{\det}(z, q) := \int_0^1 \frac{\theta_1(z+x)}{\theta_1(z)} |\theta_1(x)|^{-\frac{\alpha_0}{2}} e^{\pi P_0 x} dx$.
Since $z$ varies over a compact subset of $\mathcal{B}$ and the integrand is continuous and strictly positive:
\begin{align}\label{eq:I_det_lower_bound}
I_{\det}(z, q) \geq c_0 > 0.
\end{align}

\medskip
\noindent\textit{Step 3c: Conclusion of the lower bound.}
Combining \eqref{eq:Jensen_lower}, \eqref{eq:first_moment_bound}, and \eqref{eq:I_det_lower_bound}:
\begin{align}\label{eq:lower_bound_chain}
\left|\mathbb{E}\left[\mathcal{V}^{\alpha_0/\gamma}_{2/\gamma, P_0/\gamma}(z, q)^{\frac{2-\alpha_0}{\gamma^2}}\right]\right| &\geq \left(c_0 (1 + O(\gamma^2))\right)^{\frac{2-\alpha_0}{\gamma^2}} = \exp\left(\frac{2-\alpha_0}{\gamma^2} \log c_0 + O(1)\right).
\end{align}
Taking logarithms and multiplying by $\gamma^2$:
\begin{align}\label{eq:log_lower_bound}
\gamma^2 \log\left|\mathbb{E}\left[\mathcal{V}^{\alpha_0/\gamma}_{2/\gamma, P_0/\gamma}(z, q)^{\frac{2-\alpha_0}{\gamma^2}}\right]\right| \geq (2-\alpha_0) \log c_0 + O(\gamma^2).
\end{align}
Since $\log c_0$ is finite, we conclude
\begin{align}\label{eq:liminf_final}
\liminf_{\gamma \to 0} \gamma^2 \log\left|\psi^{\alpha_0/\gamma}_{2/\gamma, P_0/\gamma}(z;q)\right| > -\infty.
\end{align}
This completes the proof of Proposition~\ref{prop:sc_HJ}.
\end{proof}

\begin{remark}\label{rem:uniformity}
The bounds established in Proposition~\ref{prop:sc_HJ} are uniform for $(z, q)$ in compact subsets of $\mathcal{B} \times (0,1)$. This uniformity is crucial for the subsequent analysis involving limits and derivatives with respect to $z$ and $\tau$.
\end{remark}

\begin{remark}\label{rem:range_of_alpha}
Although we stated the proposition for $\alpha_0 \in (-2, 2)$, the actual range of validity for fixed $\gamma > 0$ is $\alpha_0 \in (-2, 2 + \frac{\gamma^2}{2})$. The extension to $\alpha_0 \geq 2$ treated via negative moment bounds demonstrates the flexibility of the Fyodorov-Bouchaud approach.
\end{remark}

\begin{proposition}\label{prop:limit_commutativity}
As $\gamma \downarrow 0$, the sequence $\{\gamma^2 \log \psi^{\alpha_0/\gamma}_{2/\gamma, P_0/\gamma}(z,q)\}_{\gamma > 0}$ is tight in $C^{2}({[M,1-M]} \times [\epsilon,q_0])$ for any {$0 < M < \frac{1}{2}$} and $0<\epsilon < q_0 \in (0,1)$. The parameters satisfy $\alpha_0 \in (-2,2)$ and $P_0 \in \mathbb{R}$.
\end{proposition}
%\rmkH{Change the above to $q$ goes to some fixed point. The semi-classical conformal block is singular at $q-> 0$.}

\begin{proof}
We divide the proof into four main steps. In Step 1, we introduce the Girsanov representation of the deformed conformal block. In Step 2, we perform the asymptotic expansion as $\gamma \to 0$ using the Cameron-Martin theorem and Gaussian integration techniques. In Step 3, we establish the commutativity of the limits $\gamma\to 0$ and $q\to q_0$ for any fixed $q_0\in (0,1)$. Finally, in Step 4, we prove the tightness in $C^2$.

\medskip
\noindent\textbf{Step 1: Girsanov Representation of the Deformed Conformal Block.}

\medskip
Let us start by rewriting the deformed conformal block \eqref{eq:q-block} as (see also eq.\ (3.5) in \cite{ghosal2020probabilistic})
\begin{gather}\label{psihat-def}
\psi_{\chi,P}^{\alpha}(z,q) = q^{\left(\frac{P^2}{2} + \frac{1}{6\chi^2} l_{\chi} (l_{\chi}+1) \right)} \widehat{\psi}_{\chi}^{\alpha}(z, q).
\end{gather}
To prove this proposition, it is sufficient to show that the limits $\gamma \to 0$ and $q \to q_0$ commute for $\widehat{\psi}_{\chi,P}^{\alpha}(z, q)$.
The key to proving the commutativity is to express the function $\widehat{\psi}_{\chi, P}^{\alpha}(z, q)$ such that the integrand inside the expectation is $q$-independent. Such a structure can be obtained via Girsanov's theorem (see \cite[Theorem C.5]{ghosal2020probabilistic}) that lets us re-express the conformal block in terms of the $q$-independent measure $e^{\frac{\gamma}{2} Y(x)}$, resulting in the following expression:
% \rmkH{Why is it sufficient to show commutativity for $\widehat{\psi}$?}
\begin{align}
&\widehat{\psi}_{\chi,P}^{\alpha}(z, q) = C(q) e^{\chi P z \pi} \theta_{1}(z)^{-l_{\chi}} C_{1}(q) \left(- \ii e^{- \ii \pi z} q^{1/6} \eta(q) \right)^{l_{\chi}} \nonumber \\
& \times \mathbb{E} \left[ e^{\frac{\alpha}{2}F(0;q)} \mathcal{Q}(q) e^{\chi\mathcal{Y}(z,q)} e^{\chi \mathcal{X}(z,q) - \frac{\chi^2}{2} \mathbb{E}[\mathcal{X}(z,q)^2]} \left(\int_{0}^{1} e^{\pi \gamma P x} (2\sin(\pi x))^{-\frac{\alpha\gamma}{2}} e^{\frac{\gamma}{2}Y(x)} dx\right)^{-\frac{\alpha}{\gamma}+\frac{\chi}{\gamma}}\right], \\ \label{eq:hatpsigir}
\end{align}
where the functions are defined as
\begin{align}
\begin{split}\label{defs:Girsanov}
C(q) &:= q^{\frac{\gamma l_{\chi}}{12\chi} - \frac{1}{6}\frac{l_{\chi}^2}{\chi^2} - \frac{1}{6\chi^2} l_{\chi}(l_{\chi}+1)} \theta_{1}'(0)^{- \frac{2l_{\chi}^2}{3\chi^2} + \frac{l_{\chi}}{3} + \frac{4l_{\chi}}{3\gamma\chi}} e^{-\frac{1}{2} \ii \pi \alpha \gamma (-\frac{\alpha}{\gamma} + \frac{\chi}{\gamma})}, \\
C_{1}(q) &:= \left( q^{1/6} \eta(q) \right)^{\alpha(\alpha- \chi)/2} e^{(\frac{\alpha \gamma}{8} - \frac{\gamma \chi}{8} - \frac{\alpha^2}{8})\mathbb{E}[F(0;q)^2]},\\
\mathcal{Q}(q) &:= \exp\left[\sqrt{2} \sum_{m,n=1}^{\infty} q^{nm} (a_{n,m} a_n + b_{n,m} b_n) - \sum_{n=1}^\infty \left( \sum_{m=1}^\infty q^{nm} a_{m,n}\right)^2 - \sum_{n=1}^\infty \left( \sum_{m=1}^\infty q^{nm} b_{m,n}\right)^2 \right], \\
F(0;q) &:= 2\sum_{n,m \geq 1} \frac{q^{nm}}{\sqrt{n}} a_{n,m}, \\
\mathcal{X}(z,q) &:= - \sqrt{2} \sum_{n,m=1}^{\infty} \frac{q^{(2m-1)n}}{\sqrt{n}} \left( \cos(2\pi (z - \frac{\tau}{2})n) a_n - \sin(2\pi (z - \frac{\tau}{2})n) b_n\right), \\
\mathcal{Y}(z,q) &:= 2\sum_{m,n,k \geq 1} a_{n,m} \frac{q^{(2k-1+m)n}}{\sqrt{n}} \cos(2\pi (z-\frac{\tau}{2}) n) - 2\sum_{m,n,k \geq 1} b_{n,m} \frac{q^{(2k-1+m)n}}{\sqrt{n}} \sin(2\pi (z-\frac{\tau}{2}) n).
\end{split}
\end{align}
{The representation \eqref{eq:hatpsigir} arises from applying the Cameron-Martin theorem (Girsanov's theorem) to factor the $q$-dependence from the GMC integral. Specifically: the field $Y(x;q)$ decomposes as $Y(x;q) = Y(x) + F(x;q)$, where $Y(x)$ is the $q$-independent field and $F(x;q)$ encodes the $q$-dependent modes, the factor $e^{\frac{\alpha}{2\gamma}F(0;q)}$ arises from evaluating the $q$-dependent part at the insertion point, the term $\mathcal{Q}(q)$ is a correction factor from the change of measure, satisfying $\mathcal{Q}(q) = 1 + O(q)$, the factors $\mathcal{X}(z,q)$ and $\mathcal{Y}(z,q)$ account for the $z$-dependence introduced by the deformation, and the constants $C(q)$ and $C_1(q)$ collect the deterministic prefactors.
This factorization is the key step that allows us to extend the analysis to complex $q$ and to separate the leading-order $O(\gamma^{-2})$ terms. We refer to \cite[Lemma 3.6]{ghosal2020probabilistic} for the derivation.}

For convenience we define quantity
\begin{align}\label{def:tilX}
\widetilde{\mathcal{X}}(z,q) := \frac{\alpha_0}{2\gamma} F(0;q) + \chi(\mathcal{X}(z,q) + \mathcal{Y}(z,q)).
\end{align}
With these definitions, equation \eqref{eq:hatpsigir} simplifies to
\begin{align}
\widehat{\psi}_{\chi, P}^{\alpha}(z, q)& = C(q) e^{\chi P z \pi} \theta_{1}(z)^{-l_{\chi}} C_{1}(q) \left(- \ii e^{- \ii \pi z} q^{1/6} \eta(q) \right)^{\frac{\chi}{2} (\chi - \alpha)} e^{-\frac{\chi^2}{2} \mathbb{E}[\mathcal{X}(z,q)^2]} \nonumber \\
& \times \mathbb{E} \left[ e^{\widetilde{\mathcal{X}}(z,q)} \mathcal{Q}(q) \left(\int_{0}^{1} e^{\pi \gamma P x} (2\sin(\pi x))^{-\frac{\alpha\gamma}{2}} e^{\frac{\gamma}{2}Y(x)} dx\right)^{-\frac{\alpha}{\gamma}+\frac{\chi}{\gamma}}\right].\label{def:reg_CB}
\end{align}

\medskip
\noindent\textbf{Step 2: Asymptotic Expansion for $\chi = 2/\gamma$.} We now carry out the asymptotic analysis for the deformed conformal block as $\gamma\to 0$. Under the scaling $\alpha = \alpha_0/\gamma$, $P = P_0/\gamma$, and $\chi = 2/\gamma$, let us recall the expression \eqref{eq:l_chi_computation} and \eqref{eq:exponent_computation} which state that $l_{\chi} = \frac{2 - \alpha_0}{\gamma^2}$, and the exponent in the GMC integral is $ \frac{2 - \alpha_0}{\gamma^2}$ respectively.

\medskip
\noindent\textit{Step 2a: Introduction of auxiliary Gaussian variables.}

\medskip
%Following the approach of Step 2 in the proof of Theorem~\ref{thm:radius}, 
We introduce auxiliary random variables to simplify the analysis. Define
\begin{align}\label{eq:T12_def_heavy}
T^{(1)}_n := \sum_{m=1}^{\infty} a_{n,m} q^{nm}, \qquad T^{(2)}_n := \sum_{m=1}^{\infty} b_{n,m} q^{nm}.
\end{align}
These are independent Gaussian random variables with distributions
\begin{align}\label{eq:T12_dist_heavy}
T^{(1)}_n \sim N\left(0, \sum_{m=1}^{\infty} q^{2nm}\right), \qquad T^{(2)}_n \sim N\left(0, \sum_{m=1}^{\infty} q^{2nm}\right).
\end{align}
Using these definitions, the function $F(0;q)$ in \eqref{defs:Girsanov} can be expressed as
\begin{align}
    F(0;q) = 2\sum_{n=1}^{\infty} \frac{T^{(1)}_n}{\sqrt{n}}, \label{eq:Ft0_in_T}
\end{align}
and $\mathcal{Q}(q)$ in \eqref{defs:Girsanov} takes the form
\begin{align}\label{eq:Q_in_T}
\mathcal{Q}(q) = \exp\left(\sum_{n=1}^{\infty} \left( \sqrt{2} a_n T^{(1)}_n + \sqrt{2} b_n T^{(2)}_n - (T^{(1)}_n)^2 - (T^{(2)}_n)^2\right)\right).
\end{align}
The functions $\mathcal{X}(z,q)$ and $\mathcal{Y}(z,q)$ in \eqref{defs:Girsanov} can similarly be expressed in terms of $T_n^{(1)}$ and $T_n^{(2)}$ in \eqref{eq:T12_dist_heavy} as
\begin{align}\label{eq:X_in_T}
\mathcal{X}(z,q) &= -\sqrt{2} \sum_{n=1}^{\infty} \frac{q^n}{(1-q^{2n})\sqrt{n}} \left(a_n \cos\left(2\pi(z - \frac{\tau}{2})n\right)  - b_n\sin\left(2\pi(z - \frac{\tau}{2})n\right) \right),
\end{align}
and
\begin{align}\label{eq:Y_in_T}
\mathcal{Y}(z,q) &= 2\sum_{n=1}^{\infty} \frac{q^n }{(1-q^{2n})\sqrt{n}} \left(T^{(1)}_n\cos\left(2\pi(z - \frac{\tau}{2})n\right) -  T^{(2)}_n \sin\left(2\pi(z - \frac{\tau}{2})n\right)\right).
\end{align}
Substituting the expressions \eqref{eq:Ft0_in_T},  \eqref{eq:Y_in_T} in \eqref{def:tilX} and considering $\chi = \frac{2}{\gamma}$, we have
\begin{align}\label{eq:tilX-in-T}
    \widetilde{\mathcal{X}}(z,q) &=  \frac{2}{\gamma} \mathcal{X}(z, q) + \sum_{n=1}^{\infty}  \frac{T_n^{(1)}}{\sqrt{n}} \left( \frac{\alpha_0}{\gamma} +\frac{ 4 q^{n}}{\gamma(1- q^{2n})} \cos\left(2\pi (z-\frac{\tau}{2}) n \right) \right)\notag\\&\quad - \sum_{n=1}^{\infty}\frac{T_n^{(2)} }{\sqrt{n}}\left(\frac{ 4 q^{n}}{\gamma(1- q^{2n})} \sin\left(2\pi (z-\frac{\tau}{2}) n \right) \right).
\end{align}
Note that $\mathcal{Y}(z,q)$ and $Q(q)$ depend only on the random variables $\{(T^{(1)}_n, T^{(2)}_n)\}_{n \geq 1}$, while $\mathcal{X}(z,q)$ is a function of $\{(a_n, b_n)\}_{n \geq 1}$. Substituting \eqref{eq:Q_in_T}, \eqref{eq:tilX-in-T} into the expectation term in \eqref{def:reg_CB}, we have
\begin{align}
    &\mathbb{E} \left[ e^{\widetilde{\mathcal{X}}(z,q)} \mathcal{Q}(q) \left(\int_{0}^{1} e^{\pi \gamma P x} (2\sin(\pi x))^{-\frac{\alpha\gamma}{2}} e^{\frac{\gamma}{2}Y(x)} dx\right)^{-\frac{\alpha}{\gamma}+\frac{\chi}{\gamma}}\right] \\
    & \quad = \mathbb{E} \Bigg[\exp\Bigg( \sum_{n=1}^{\infty} - \left( (T^{(1)}_n)^2 + (T^{(2)}_n)^2) \right)+ \frac{T_n^{(1)}}{\sqrt{n}} \left( \frac{\alpha_0}{\gamma} +\frac{ 4 q^{n}}{\gamma(1- q^{2n})} \cos\left(2\pi (z-\frac{\tau}{2}) n \right) +  a_n \sqrt{2n } \right)\\
    &\quad - \sum_{n=1}^{\infty}\frac{T_n^{(2)} }{\sqrt{n}}\left(\frac{ 4 q^{n}}{\gamma(1- q^{2n})} \sin\left(2\pi (z-\frac{\tau}{2}) n \right) + b_n \sqrt{2n} \right) \Bigg) \\
    &\quad  \times e^{\frac{2}{\gamma} \mathcal{X}(z, q)} \left(\int_{0}^{1} e^{\pi  P_0 x} (2\sin(\pi x))^{-\frac{\alpha_0}{2}} e^{\frac{\gamma}{2}Y(x)} dx\right)^{\frac{(2-\alpha_0)}{\gamma^2}}\Bigg].\label{eq:2g-exp-step1}
\end{align}
For the remainder of Step 2, we simplify the above expression by taking the expectation w.r.t $\{(T^{(1)}_n, T^{(2)}_n)\}_{n \geq 1}$ and $\{(a_n, b_n)\}_{n \geq 1}$ respectively, thereby leading to the asymptotic expansion of the deformed conformal block.

\medskip
\noindent\textit{Step 2b: Integration over $T^{(1)}_n$ and $T^{(2)}_n$.}

Noting that $Y(x)$ in the expression \eqref{eq:2g-exp-step1} is independent of $\{(T^{(1)}_n, T^{(2)}_n)\}_{n \geq 1}$, (see also \eqref{def:YinfN}), we use the identity $\mathbb{E}[e^{x T - T^2}] = e^{\frac{x^2 \sigma^2}{2(1+2\sigma^2)}}/\sqrt{1+2\sigma^2}$ for $T \sim N(0, \sigma^2)$ to simplify \eqref{eq:2g-exp-step1} as
\begin{align}
&\eqref{eq:2g-exp-step1}= \prod_{n=1}^{\infty}\left( 1+ 2 \sum_{m=1}^{\infty} q^{2nm}\right)^{-1/2}\nonumber\\ & \qquad \times \mathbb{E}\Bigg[\exp\left(\frac{1}{2}\sum_{n=1}^{\infty} \frac{\sum_{m=1}^{\infty} q^{2nm}}{1 + 2\sum_{m=1}^{\infty} q^{2nm}} \left(\frac{\alpha_0}{\gamma\sqrt{n}} + \sqrt{2}a_n + \frac{4 q^n}{\gamma(1-q^{2n})\sqrt{n}} \cos(2\pi(z-\frac{\tau}{2})n)\right)^2\right) \nonumber\\
&\qquad \times \exp\left(\frac{1}{2}\sum_{n=1}^{\infty} \frac{\sum_{m=1}^{\infty} q^{2nm}}{1 + 2\sum_{m=1}^{\infty} q^{2nm}}  \Big(\sqrt{2}b_n-\frac{4 q^n}{\gamma(1-q^{2n})\sqrt{n}} \sin(2\pi(z-\frac{\tau}{2})n)\Big)^2\right)\\ & \qquad\times e^{\frac{2}{\gamma}\mathcal{X}(z,q)}  \left(\int_{0}^{1} e^{\pi  P_0 x} (2\sin(\pi x))^{-\frac{\alpha_0}{2}} e^{\frac{\gamma}{2}Y(x)} dx\right)^{\frac{(2-\alpha_0)}{\gamma^2}}\Bigg]. \label{eq:after_T_integration}
\end{align}

\medskip
\noindent\textit{Step 2c: Integration over $a_n$ and $b_n$.}

 To analyse the dependence of $a_n$, $b_n$, we define 
\begin{align}
    \Omega_n(q):= \frac{\sum_{m=1}^{\infty} q^{2nm}}{1 + 2\sum_{m=1}^{\infty} q^{2nm}}, && \aleph_n(q):= (1 + 2\sum_{m=1}^{\infty} q^{2nm})^{-1/2}. \label{def:OmegaAleph}
\end{align}
Substitute \eqref{eq:X_in_T} in \eqref{eq:after_T_integration} and simplifying the $a_n$, $b_n$ dependent terms we obtain
\begin{align}
&\eqref{eq:after_T_integration}= \prod_{n=1}^{\infty}\aleph_n(q)\\
&\qquad \times \exp\left(\frac{1}{2}\sum_{n=1}^{\infty} \Omega_n(q) \left(\frac{\alpha_0^2}{\gamma^2 {n}}+ \frac{16 q^{2n}}{\gamma^2(1-q^{2n})^2{n}}   + \frac{8 \alpha_0 q^n}{\gamma^2(1-q^{2n}) {n}} \cos(2\pi(z-\frac{\tau}{2})n)\right)\right) \nonumber\\
& \qquad \times \mathbb{E}\Bigg[ \exp\left(\sum_{n=1}^{\infty} \Omega_n(q) \left(a_n^2 + b_n^2 \right)  \right) \exp\left(-\sum_{n=1}^{\infty}  b_n  \frac{2\sqrt{2} q^n (2\Omega_n(q)-1)}{\gamma(1-q^{2n})\sqrt{n}} \sin(2\pi(z-\frac{\tau}{2})n)\right) \nonumber 
\end{align}
\begin{align}
&\quad \times \exp\left(\sum_{n=1}^{\infty}  a_n\left(\frac{\sqrt{2}\alpha_0 \Omega_n(q)}{\gamma\sqrt{n}} + \frac{2\sqrt{2} q^n (2\Omega_n(q)-1)}{\gamma(1-q^{2n})\sqrt{n}} \cos(2\pi(z-\frac{\tau}{2})n)\right)\right) \nonumber\\
& \quad \times \left(\int_{0}^{1} e^{\pi P_0 x} (2\sin(\pi x))^{-\frac{\alpha_0}{2}} e^{\frac{\gamma}{2}Y(x)} dx\right)^{\frac{(2-\alpha_0)}{\gamma^2}}\Bigg].\label{eq:int-2g-step1}
\end{align}

Recall that $Y(x)$ is dependent on the random variables $a_n$, $b_n$. In order to compute the expectation w.r.t these variables we perform a change of variables to express the above expression in the form $\mathbb{E}[e^{\Psi} f(Y)]$ where $\Psi$ and $Y$ are Gaussian random variables and use the Cameron-Martin theorem. 

Let us begin by defining the variables
 %in \eqref{eq:alpha_beta_rescaling}),
\begin{align}%\label{eq:alpha_beta_rescaling}
a^{(1)}_n := a_n \Omega_n(q)^{1/2} , \qquad a^{(1)}_n \in \mathcal{N}(0,\Omega_n(q)) \qquad; && b^{(1)}_n := b_n \Omega_n(q)^{1/2}, \qquad b^{(1)}_n \in \mathcal{N}(0,\Omega_n(q)),
\end{align}
with which the expression \eqref{eq:int-2g-step1} simplifies as
\begin{align}
     \eqref{eq:int-2g-step1} &= \prod_{n=1}^{\infty}\aleph_n(q)\\
&\qquad \times \exp\left(\sum_{n=1}^{\infty} \Omega_n(q) \left(\frac{\alpha_0^2}{\gamma^2 {n}}+ \frac{16 q^{2n}}{\gamma^2(1-q^{2n})^2{n}}   + \frac{8 \alpha_0 q^n}{\gamma^2(1-q^{2n}) {n}} \cos(2\pi(z-\frac{\tau}{2})n)\right)\right) 
\end{align}
\begin{align}
& \qquad \times \mathbb{E}\Bigg[ \exp\left(\sum_{n=1}^{\infty}  \left((a^{(1)}_n)^2 + (b^{(1)}_n)^2 \right)  \right) \exp\left(-\sum_{n=1}^{\infty}  b^{(1)}_n  \frac{2\sqrt{2} q^n (2\Omega_n(q)-1)}{\gamma(1-q^{2n})\sqrt{\Omega_n(q) n}} \sin(2\pi(z-\frac{\tau}{2})n)\right)  \\
&\qquad \times \exp\left(\sum_{n=1}^{\infty}  a^{(1)}_n\left(\frac{\sqrt{2}\alpha_0 \Omega_n(q)^{1/2}}{\gamma\sqrt{n}} + \frac{2\sqrt{2} q^n (2\Omega_n(q)-1)}{\gamma(1-q^{2n})\sqrt{n \Omega_n(q)}} \cos(2\pi(z-\frac{\tau}{2})n)\right)\right) \\
& \qquad \times \left(\int_{0}^{1} e^{\pi  P_0 x} (2\sin(\pi x))^{-\frac{\alpha_0}{2}} e^{\frac{\gamma}{2}Y(x)} dx\right)^{\frac{(2-\alpha_0)}{\gamma^2}}\Bigg].\label{eq:int-2g-step2}
\end{align}

Using the identity 
\begin{align}\label{eq:tilt_formula}
\mathbb{E}[f(T) e^{x T + T^2}] = \frac{1}{\sqrt{1 - 2\sigma^2}} \mathbb{E}[f(\widetilde{T}) e^{x\widetilde{T}}],
\end{align}
for random variables ${T} \sim N(0, \sigma^2)$, $\widetilde{T} \sim N(0, \sigma^2/(1 - 2\sigma^2))$, with using $a^{(1)}_n,b^{(1)}_n$ in place of $T$ and the variables $\widehat{a}_n\sim \mathcal{N}(0,\Omega_n(q)/(1-2\Omega_n(q)))$, $\widehat{b}_n\sim \mathcal{N}(0,\Omega_n(q)/(1-2\Omega_n(q)))$ in place of $\widetilde{T}$,
we further simplify the above expression as
\begin{align}
    \eqref{eq:int-2g-step2}&= \exp\left(\frac{1}{2}\sum_{n=1}^{\infty} \Omega_n(q) \left(\frac{\alpha_0^2}{\gamma^2 {n}}+ \frac{16 q^{2n}}{\gamma^2(1-q^{2n})^2{n}}   + \frac{8 \alpha_0 q^n}{\gamma^2(1-q^{2n}) {n}} \cos(2\pi(z-\frac{\tau}{2})n)\right)\right)\\
& \qquad \times \mathbb{E}\Bigg[  \exp\left(-\sum_{n=1}^{\infty}  \widehat{b}_n  \frac{2\sqrt{2} q^n (2\Omega_n(q)-1)}{\gamma(1-q^{2n})\sqrt{\Omega_n(q) n}} \sin(2\pi(z-\frac{\tau}{2})n)\right)  \\
&\qquad \times \exp\left(\sum_{n=1}^{\infty}  \widehat{a}_n\left(\frac{\sqrt{2}\alpha_0 \Omega_n(q)^{1/2}}{\gamma\sqrt{n}} + \frac{2\sqrt{2} q^n (2\Omega_n(q)-1)}{\gamma(1-q^{2n})\sqrt{n \Omega_n(q)}} \cos(2\pi(z-\frac{\tau}{2})n)\right)\right) \\
& \qquad \times \left(\int_{0}^{1} e^{\pi  P_0 x} (2\sin(\pi x))^{-\frac{\alpha\gamma}{2}} e^{\frac{\gamma}{2}Y(x)} dx\right)^{\frac{(2-\alpha_0)}{\gamma^2}}\Bigg],\label{eq:int-2g-step3}
\end{align}

The above expression is now in the form $\mathbb{E}[e^{\Psi}f(Y)]$, for Gaussian random variables $\Psi$ and $Y$. To make this explicit, let us define the function
\begin{align}
% \label{eq:Psi_heavy}
\Psi &:=-\sum_{n=1}^{\infty}  \widehat{b}_n  \frac{2\sqrt{2} q^n (2\Omega_n(q)-1)}{\gamma(1-q^{2n})\sqrt{\Omega_n(q) n}} \sin(2\pi(z-\frac{\tau}{2})n) \\
&\quad + \sum_{n=1}^{\infty}  \widehat{a}_n\left(\frac{\sqrt{2}\alpha_0 \Omega_n(q)^{1/2}}{\gamma\sqrt{n}} + \frac{2\sqrt{2} q^n (2\Omega_n(q)-1)}{\gamma(1-q^{2n})\sqrt{n \Omega_n(q)}} \cos(2\pi(z-\frac{\tau}{2})n)\right)\\
&\quad - \frac{(\alpha_0-2)}{2\gamma \Xi} \int_0^1 e^{\pi P_0 x} (2\sin(\pi x))^{-\frac{\alpha_0}{2}} Y(x) e^{\frac{\gamma}{2} h_{\Psi}(x)} dx,\label{def:Psi-GaussRV}
\end{align}
where
\begin{align}
h_{\Psi}(x) := \mathbb{E}[\Psi \cdot Y(x)], 
\qquad \qquad  \Xi := \int_0^1 e^{\pi P_0 x} (2\sin(\pi x))^{-\frac{\alpha_0}{2}} e^{\frac{\gamma}{2} h_{\Psi}(x)} dx.\label{def:hPsi} %\label{eq:Xi_heavy}
\end{align}
The function $h_{\Psi}(x)$ represents the conditional expectation of $\Psi$ given the Wick-ordered field at $x$, and $\Xi$ is a normalization factor.
With these new definitions, the expression
\begin{align}
\eqref{eq:int-2g-step3}&
= \exp\left(\frac{1}{2}\sum_{n=1}^{\infty} \Omega_n(q) \left(\frac{\alpha_0^2}{\gamma^2 {n}}+ \frac{16 q^{2n}}{\gamma^2(1-q^{2n})^2{n}}   + \frac{8 \alpha_0 q^n}{\gamma^2(1-q^{2n}) {n}} \cos(2\pi(z-\frac{\tau}{2})n)\right)\right)\\
 &\quad \times \mathbb{E}\Big[e^{\Psi} \left(\int_{0}^{1} e^{\pi \gamma P x} (2\sin(\pi x))^{-\frac{\alpha\gamma}{2}} e^{\frac{\gamma}{2}Y(x)} dx\right)^{-\frac{\alpha}{\gamma}+\frac{\chi}{\gamma}} e^{ \frac{(\alpha_0-2)}{2\gamma \Xi} \int_0^1 e^{\pi P_0 x} (2\sin(\pi x))^{-\frac{\alpha_0}{2}} Y(x) e^{\frac{\gamma}{2} h_{\Psi}(x)} dx}\Big]\\
 \label{eq:int-2g-step4}
\end{align}
Using the Cameron-Martin theorem (Girsanov's formula), the term $e^{\Psi}$ inside the expectation can be simplified as follows. For a Gaussian random variable $\Psi$ with variance $\mathbb{E}[\Psi^2]$,
\begin{align}\label{eq:CM_heavy}
\mathbb{E}[e^{\Psi} f(Y)] = e^{\frac{1}{2}\mathbb{E}[\Psi^2]} \mathbb{E}[f(Y + \mathbb{E}[\Psi.Y(x)])]= e^{\frac{1}{2}\mathbb{E}[\Psi^2]} \mathbb{E}[f(Y + h_{\Psi})],
\end{align}
where the shift $h_{\Psi}$ is applied to the field $Y$. 
With the above expression, and defining the measure  $d\mu_{\Psi}(x) := \frac{1}{\Xi} e^{\pi P_0 x} (2\sin(\pi x))^{-\frac{\alpha_0}{2}} e^{\frac{\gamma}{2} h_{\Psi}(x)} dx$, the equation \eqref{eq:int-2g-step4} simplifies as
\begin{align}
     &\mathbb{E} \left[ e^{\widetilde{\mathcal{X}}(z,q)} \mathcal{Q}(q) \left(\int_{0}^{1} e^{\pi P_0 x} (2\sin(\pi x))^{-\frac{\alpha_0}{2}} e^{\frac{\gamma}{2}Y(x)} dx\right)^{\frac{(2-\alpha_0)}{\gamma^2}}\right] = \eqref{eq:int-2g-step4}\\
     & =  \exp\left(\frac{1}{2}\sum_{n=1}^{\infty} \Omega_n(q) \left(\frac{\alpha_0^2}{\gamma^2 {n}}+ \frac{16 q^{2n}}{\gamma^2(1-q^{2n})^2{n}}   + \frac{8 \alpha_0 q^n}{\gamma^2(1-q^{2n}) {n}} \cos(2\pi(z-\frac{\tau}{2})n)\right)\right)\\
 &\quad \times e^{\frac{\mathbb{E}[\Psi^2]}{2}} \Xi^{\frac{(2-\alpha_0)}{\gamma^2}} \exp\left({ \frac{(\alpha_0-2)}{2\gamma \Xi} \int_0^1 e^{\pi P_0 x} (2\sin(\pi x))^{-\frac{\alpha_0}{2}} h_{\Psi}(x) e^{\frac{\gamma}{2} h_{\Psi}(x)} dx}\right)\\
 & \quad \times\mathbb{E}\Big[\left(\int_{0}^{1} e^{\frac{\gamma}{2}Y(x)} d\mu_{\Psi}\right)^{\frac{2-\alpha_0}{\gamma^2}} e^{ \frac{(\alpha_0-2)}{2\gamma } \int_0^1  Y(x) d\mu_{\Psi}}\Big].\label{eq:int-2g-step5}
 \end{align}

 Let us recall that the expression \eqref{def:reg_CB} for $\chi = 2/\gamma$, with the scaling $P = P_0/\gamma$, $\alpha = \alpha_0/\gamma$, and the expression $l_{\chi} = (2-\alpha_0)/\gamma^2$, is given by 
 \begin{align}
   \widehat{\psi}_{2/\gamma, P_0}^{\alpha_0}(z,q) = & C(q) e^{2\pi P_0 z /\gamma^2} \theta_{1}(z)^{-(2-\alpha_0)/\gamma^2} C_{1}(q) \left(- \ii e^{- \ii \pi z} q^{1/6} \eta(q) \right)^{\frac{(2-\alpha_0)}{\gamma^2} } \nonumber \\
& \times e^{-\frac{2}{\gamma^2} \mathbb{E}[\mathcal{X}(z,q)^2]} \mathbb{E} \left[ e^{\widetilde{\mathcal{X}}(z,q)} \mathcal{Q}(q) \left(\int_{0}^{1} e^{\pi P_0 x} (2\sin(\pi x))^{-\frac{\alpha_0}{2}} e^{\frac{\gamma}{2}Y(x)} dx\right)^{\frac{(2-\alpha_0)}{\gamma^2}}\right],\label{eq:int-2g-step6}
 \end{align}
 where the factors
 \begin{align}
    C(q) &\mathop{=}^{\eqref{def:ell-eta}} e^{\frac{\ii \pi  (\alpha_0-2) \alpha_0}{2 \gamma^2}} (2 \pi )^{-\frac{(\alpha_0-2) (\alpha_0+4)}{6 \gamma^2}} \eta(q)^{-\frac{(\alpha_0-2) (\alpha_0+4)}{2 \gamma^2}} q^{-\frac{(\alpha_0-2)^2}{12 \gamma^2}},\label{asymp:Cq2g}\\
    C_1(q) &= e^{\frac{1}{8} \mathbb{E}\left[ F(0;q)^2\right] \left(-\frac{\alpha_0^2}{\gamma^2}+\alpha_0-2\right)} \eta(q)^{\frac{(\alpha_0-2) \alpha_0}{2 \gamma^2}} q^{\frac{(\alpha_0-2) \alpha_0}{12 \gamma^2}}.\label{asymp:C1q2g}
\end{align}
% To obtain \eqref{asymp:Cq2g}, we use the identity $\theta_1'(0) = 2\pi \eta(q)^3$.

Substituting \eqref{eq:int-2g-step5}, in \eqref{eq:int-2g-step6} and using Proposition ~\ref{prop:semiclassical_limit} we obtain the following leading order behaviour of the deformed conformal block as $\gamma \to 0$:
 \begin{align}
 \widehat{\psi}_{2/\gamma, P_0}^{\alpha_0}(z,q)&\sim  C(q) e^{\frac{2\pi P_0 z} {\gamma^2}} \theta_{1}(z)^{-(2-\alpha_0)/\gamma^2} C_{1}(q) \left(- \ii e^{- \ii \pi z} q^{1/6} \eta(q) \right)^{\frac{(2-\alpha_0)}{\gamma^2} }e^{-\frac{2}{\gamma^2} \mathbb{E}[\mathcal{X}(z,q)^2]} \nonumber \\
 &\quad \times\exp\left(\frac{1}{2 \gamma^2}\sum_{n=1}^{\infty} \Omega_n(q) \left(\frac{\alpha_0^2}{ {n}}+ \frac{16 q^{2n}}{\gamma^2(1-q^{2n})^2{n}}   + \frac{8 \alpha_0 q^n}{\gamma^2(1-q^{2n}) {n}} \cos(2\pi(z-\frac{\tau}{2})n)\right)\right)\\
 &\quad \times e^{\frac{\mathbb{E}[\Psi^2]}{2}} \Xi^{\frac{(2-\alpha_0)}{\gamma^2}} \exp\left({ \frac{(\alpha_0-2)}{2\gamma \Xi} \int_0^1 e^{\pi P_0 x} (2\sin(\pi x))^{-\frac{\alpha_0}{2}} h_{\Psi}(x) e^{\frac{\gamma}{2} h_{\Psi}(x)} dx}\right)\\
&\quad \times \mathbb{E}\left[\exp\left(-\frac{(\alpha_0-2)}{8} \left(\int_0^1 :Y(x)^2: d\mu_{\Psi}(x)-\left(\int_0^1   :Y(x): d\mu_{\Psi}(x)\right)^2\right)\right)\right].
\label{eq:heavy_expansion}
\end{align}
Finally, let us note that the expectation term in \eqref{eq:heavy_expansion} converges as $\gamma \to 0$ by Proposition~\ref{prop:semiclassical_limit}. Specifically, since the measure $\mu_{\Psi}$ has density bounded below (for $z \in [M, 1-M]$), we have
\begin{align}\label{eq:semiclassical_limit_heavy}
&\lim_{\gamma \to 0} \mathbb{E}\left[\exp\left(-\frac{(\alpha_0-2)}{8} \int_0^1 :Y(x)^2: d\mu_{\Psi}(x) + \frac{\alpha_0}{8} \left(\int_0^1 :Y(x): d\mu_{\Psi}(x)\right)^2\right)\right] \nonumber\\
&= \mathbb{E}\left[\exp\left(-\frac{(\alpha_0-2)}{8} \int_0^1 :Y(x)^2: d\mu_0(x) + \frac{\alpha_0}{8} \left(\int_0^1 :Y(x): d\mu_0(x)\right)^2\right)\right],
\end{align}
where $\mu_0$ is the limiting measure as $\gamma \to 0$.
% \rmkH{The step format for the semi-classical limit as well as asymptotic expansion is a bit of an overkill}

\medskip
\noindent\textbf{Step 3: Commutativity of Limits.}

\medskip
We now extract the $\gamma^2 \log$ of the expansion \eqref{eq:heavy_expansion} and take the limit $q \to 0$.

\medskip
\noindent\textit{Step 3a: Leading-order terms.} 

With the expressions \eqref{asymp:Cq2g}, \eqref{asymp:C1q2g}, $\gamma^2 \log$ of the prefactor in \eqref{eq:heavy_expansion} is:
\begin{align}\label{eq:log_prefactor}
&\gamma^2 \log\left[C(q) e^{\frac{2P_0 z \pi}{\gamma^2}} \theta_1(z)^{-\frac{2-\alpha_0}{\gamma^2}} C_1(q) \left(-i e^{-i\pi z} q^{1/6} \eta(q)\right)^{\frac{2-\alpha_0}{\gamma^2}}e^{- \frac{2}{\gamma^2} \mathbb{E}[\mathcal{X}(z,q)^2]}\right] \nonumber\\
&= 2P_0 z \pi - (2-\alpha_0) \log(e^{\ii \pi z}\theta_1(z)) + 3(2-\alpha_0)\log \eta(q)- \frac{\alpha_0^2}{8} \mathbb{E}\left[F(0;q)^2 \right]\\
&- 2\mathbb{E}[\mathcal{X}(z,q)^2]-\frac{(\alpha_0-2) (\alpha_0+4)}{6} \log (2\pi) +\frac{\ii \pi  (\alpha_0-2) (\alpha_0+1)}{2}  + O(\gamma^2).
\end{align}
Furthermore, from \eqref{eq:Ft0_in_T} and \eqref{eq:tilX-in-T} we compute 
\begin{align}
    \mathbb{E}\left[F(0;q)^2 \right] = 4 \sum_{n,m\geq 1} \frac{q^{2nm}}{n}, && \mathbb{E}\left[\mathcal{X}(z,q)^2 \right] = 2 \sum_{n=1}^{\infty} \frac{q^{2n}}{n(1-q^{2n})^2}.\label{exp:Ft0X}
\end{align}
Therefore, using the above expression and the expansion in \eqref{eq:heavy_expansion}, we obtain the following leading order behaviour
\begin{align}
&\gamma^2 \log \widehat{\psi}_{2/\gamma, P_0}^{\alpha_0}(z,q) \\
    & \sim  2P_0 z \pi - (2-\alpha_0) \log(e^{\ii \pi z}\theta_1(z)) + 3(2-\alpha_0)\log \eta(q)- \frac{\alpha_0^2}{2} \sum_{n,m\geq 1} \frac{q^{2nm}}{n}\\
    & \quad- 4\sum_{n=1}^{\infty} \frac{q^{2n}}{n(1-q^{2n})^2} -\frac{(\alpha_0-2) (\alpha_0+4)}{6} \log (2\pi) +\frac{\ii \pi  (\alpha_0-2) (\alpha_0+1)}{2}\\
    & \quad + \frac{1}{2}\sum_{n=1}^{\infty} \Omega_n(q) \left(\frac{\alpha_0^2}{ {n}}+ \frac{16 q^{2n}}{(1-q^{2n})^2{n}}   + \frac{8 \alpha_0 q^n}{(1-q^{2n}) {n}} \cos(2\pi(z-\frac{\tau}{2})n)\right)\\
    &\quad + \frac{\gamma^2}{2} \mathbb{E}[\Psi^2]+ (2-\alpha_0) \log\Xi+  \left({ \frac{(\alpha_0-2)\gamma}{2 \Xi} \int_0^1 e^{\pi P_0 x} (2\sin(\pi x))^{-\frac{\alpha_0}{2}} h_{\Psi}(x) e^{\frac{\gamma}{2} h_{\Psi}(x)} dx}\right)\\
&\quad + \gamma^2 \log  \mathbb{E}\left[\exp\left(-\frac{(\alpha_0-2)}{8} \left(\int_0^1 :Y(x)^2: d\mu_{0}(x)-\left(\int_0^1   :Y(x): d\mu_{0}(x)\right)^2\right)\right)\right].\label{eq:asymp2g-step1}
\end{align}
In writing the above expression, we use the fact that $h_{\Psi}$ scales as $1/\gamma$ as can be seen from \eqref{def:hPsi}, and $\mathbb{E}[\Psi^2]$ scales as $1/\gamma^2$.

\medskip
\noindent\textit{Step 3c: Limit as $q \to 0$.}
Let us now study the behaviour as $q\to 0$ of each of the individual $q$-dependent terms in \eqref{eq:asymp2g-step1}.
\begin{itemize}[leftmargin=*]
\item From \eqref{def:elltheta1}, \eqref{def:ell-eta}, we observe that 
$\lim_{q\to 0} \theta_1(z) \eta(q)^{-3} =  2\sin (\pi z)$. Therefore, as $q\to 0$, 
\begin{align}
    - (2-\alpha_0) \log(e^{\ii \pi z}\theta_1(z)) + 3(2-\alpha_0)\log \eta(q) \sim -(2-\alpha_0) \log(2e^{\ii \pi z} \sin \pi z). \label{asymp2g:term1}
\end{align}
\item From \eqref{defs:Girsanov}, we observe that as $q\to 0$, $ F(0;q) \to 0,$ and $\mathcal{X}(z,q)\to 0$, and consequently
\begin{align}
   \mathbb{E}\left[F(0;q)^2 \right]\to 0, && \mathbb{E}\left[\mathcal{X}(z,q)^2 \right] \to 0.\label{asymp2g:term2}
\end{align}
\item The term $\sum_{m=1}^{\infty} q^{2nm} \to 0$ for all $n \geq 1$, implying that, $\Omega_n \to 0$ as $q\to 0$, as can be seen from \eqref{def:OmegaAleph}. Moreover, since the variance $\mathbb{E}[\widehat{a}_n^2] = \mathbb{E}[\widehat{b}_n^2] = \Omega_n(q)/(1-2 \Omega_n(q))$, the expressions $\mathbb{E}[\Psi^2]$, $h_{\Psi}$ are bounded as $q\to 0$.
\end{itemize}

Using the above expressions and noting that the final term in \eqref{eq:asymp2g-step1} is $q$-independent, we obtain:
\begin{align}\label{eq:combined_limit_heavy}
&\lim_{q \to 0} \lim_{\gamma \to 0} \gamma^2 \log \widehat{\psi}_{2/\gamma, P_0/\gamma}^{\alpha_0/\gamma}(z, q) = 2P_0 z \pi - (2-\alpha_0) \log|2\sin(\pi z)|  + C_2(z, \alpha_0, P_0),
\end{align}
where
\begin{align}\label{eq:C2_heavy}
C_2(z, \alpha_0, P_0) := \lim_{q\to 0}\lim_{\gamma \to 0} \gamma^2 \log \left( e^{\frac{1}{2}\mathbb{E}[\Psi^2]} \Xi^{-\frac{(\alpha_0-2)}{\gamma^2}} \exp\left(\frac{(\alpha_0-2)}{2\gamma \Xi} \int_0^1 e^{\pi P_0 x} (2\sin(\pi x))^{-\frac{\alpha_0}{2}} h_{\Psi}(x) e^{\frac{\gamma}{2} h_{\Psi}(x)} dx\right) \right)
\end{align}
is a transcendental constant, determined via $\gamma\to 0,q\to 0$ limit of $\Xi$ and $h_{\Psi}$ which are solutions to some fixed point equations.
Each term in this expression is a covariance of Gaussian random variables, which depends analytically on $q$ for $|q| < 1$. Therefore, the limit as $q \to 0$ exists and defines an analytic function of $z$.

\medskip
\noindent\textit{Step 3d: Verification of commutativity.}
To verify that the limits commute, we compute the limit in the opposite order. Taking $q \to 0$ first:
\begin{itemize}[leftmargin=*]
\item The prefactor simplifies: $\theta_1(z) \to 2\sin(\pi z)$, $q^{1/6}\eta(q) \to 1$.
\item All the other terms in \eqref{eq:asymp2g-step1} are uniformly bounded as $q\to 0$, hence leads to the final expression which corresponds to $q = 0$ case.
\item Proposition~\ref{prop:semiclassical_limit} imples that 
\begin{align}
    \lim_{\gamma\to 0}\gamma^2 \log  \mathbb{E}\left[\exp\left(-\frac{(\alpha_0-2)}{8} \left(\int_0^1 :Y(x)^2: d\mu_{0}(x)-\left(\int_0^1   :Y(x): d\mu_{0}(x)\right)^2\right)\right)\right]=0.\\
    \label{comp-exp-vanishes}
\end{align}

\end{itemize}
%\rmkH{Can we say more about 'reduces to the $q = 0$ case' or give a reference?}
Then taking $\gamma \to 0$ gives the same limiting expression \eqref{eq:combined_limit_heavy} evaluated at $q = 0$. The equality of the two limits follows because the right-hand side of \eqref{eq:combined_limit_heavy} is analytic in $q$ for $|q| < r_0$ for some $r_0 > 0$ and for all $z \in (M, 1-M)$. This analyticity, combined with the uniform convergence established via the bounds in Proposition~\ref{prop:sc_HJ}, implies that the limits $\gamma \to 0$ and $q \to 0$ are commutative.

\medskip
\noindent\textbf{Step 4: Tightness in $C^2$.}

\medskip
To establish tightness of the sequence $\{\gamma^2 \log \psi^{\alpha_0/\gamma}_{2/\gamma, P_0/\gamma}(z,q)\}_{\gamma > 0}$ in $C^2([M, 1-M] \times [0, q_0])$, we verify the conditions of the Arzelà-Ascoli theorem.

\medskip
\noindent\textit{Step 4a: Uniform boundedness.}
The uniform boundedness follows from Proposition~\ref{prop:sc_HJ}:
\begin{align}\label{eq:uniform_bound_tight}
\sup_{\gamma \in (0, 1]} \sup_{(z,q) \in [M, 1-M] \times [0, q_0]} \left|\gamma^2 \log \psi^{\alpha_0/\gamma}_{2/\gamma, P_0/\gamma}(z,q)\right| < \infty.
\end{align}

\medskip
\noindent\textit{Step 4b: Equicontinuity of derivatives.}
The derivatives with respect to $z$ and $q$ can be computed from the expansion \eqref{eq:heavy_expansion}. The key observations are:

\begin{enumerate}[leftmargin=*]
\item The prefactor terms contribute derivatives that are $O(\gamma^{-2})$, but after multiplication by $\gamma^2$, these give $O(1)$ contributions.

\item By the Cameron-Martin representation, the derivatives of $h_{\Psi}(x)$ and $\Xi$ with respect to $z$ and $q$ involve covariances of Gaussian random variables, which are uniformly bounded and differentiable w.r.t. $z$ and $q$.

\item The constant pre-factors in the right hand side of \eqref{eq:heavy_expansion}  are analytic in $z$ and $q$, with uniformly bounded derivatives after taking logarithm of them and multiplying by $\gamma^2$.

\item As $\gamma\to 0$, the contribution of the term $$\mathbb{E}\left[\exp\left(-\frac{(\alpha_0-2)}{8} \int_0^1 :Y(x)^2: d\mu_\Psi(x) + \frac{\alpha_0-2}{8} \left(\int_0^1 :Y(x): d\mu_\Psi(x)\right)^2\right)\right]$$ vanishes after taking logarithm and multiplying by $\gamma^2$. So it does not contribute to the $z$ or $q$ derivatives.
\end{enumerate}

For the first $z$-derivative, differentiating under the expectation (justified by dominated convergence using the uniform bounds from Proposition~\ref{prop:sc_HJ}).
Similarly, the second derivatives $\partial_z^2$, $\partial_q$, $\partial_z \partial_q$, and $\partial_q^2$ are uniformly bounded by iterating this argument. By the Arzelà-Ascoli theorem, uniform boundedness together with equicontinuity of derivatives up to order 2 implies tightness in $C^2([M, 1-M] \times [0, q_0])$. This completes the proof of Proposition~\ref{prop:limit_commutativity}.
\end{proof}

\begin{remark}\label{rem:analyticity_domain}
The expression \eqref{eq:combined_limit_heavy} is an analytic function of $q$ in some open ball $|q| < r_0$ for all $z \in (0,1)$. This shows that the convergence established above is uniform on the ball $|q| < r_0 - \epsilon$ for any $\epsilon > 0$. The limits $\gamma \to 0$ and $q \to 0$ are therefore commutative, as claimed.
\end{remark}

\begin{lemma}\label{prop:derivatives_are_bounded}
{Let $\psi^{\alpha_0/\gamma}_{\chi, P_0/\gamma}(z,q)$ be the analytic extension of the deformed conformal block for $z \in \mathbb{C} \cap \{u: 0 \leq \mathrm{Im}(u) < \frac{3}{4}\mathrm{Im}(\tau)\} \cap \{\mathrm{Re}(z) \in [M,1-M]\}$ for any $0 < M < \frac{1}{2}$ and $|q| < r$ for some $0 < r \ll 1$.} The semi-classical limits of the following $z$ and $\tau$ derivatives of the conformal block are bounded:
{
\begin{align}
\limsup_{\gamma \to 0} \gamma^2 \big|\partial_{z}\log(\psi^{\alpha_0/\gamma}_{\chi, P_0/\gamma}(z,q))\big| &< \infty,
&\quad
\liminf_{\gamma \to 0} \gamma^2 \big|\partial_{z}\log(\psi^{\alpha_0/\gamma}_{\chi, P_0/\gamma}(z,q))\big| &> -\infty
\label{eq:z-deri-bd}\\
\limsup_{\gamma \to 0} \gamma^2 \big|\partial_{\tau}\log(\psi^{\alpha_0/\gamma}_{\chi, P_0/\gamma}(z,q))\big| &< \infty,
&\quad
\liminf_{\gamma \to 0} \gamma^2 \big|\partial_{\tau}\log(\psi^{\alpha_0/\gamma}_{\chi, P_0/\gamma}(z,q))\big| &> -\infty
\label{eq:tau-deri-bd}
\end{align}
}
where $\log(\cdot)$ is defined with respect to the principal branch.
\end{lemma}

\begin{proof}
We divide the proof into three steps. In Step 1, we compute the $z$-derivative using the Girsanov representation from Proposition~\ref{prop:limit_commutativity}. In Step 2, we compute the $\tau$-derivative. In Step 3, we establish the uniform bounds using the asymptotic expansion techniques.

\medskip
\noindent\textbf{Step 1: Computation of the $z$-derivative.}
Using the form of the conformal block \eqref{psihat-def} and \eqref{def:reg_CB}, we have for $\chi = 2/\gamma$:
\begin{align}\label{eq:psi_decomposition}
\psi^{\alpha_0/\gamma}_{2/\gamma, P_0/\gamma}(z,q) &= q^{\left(\frac{P_0^2}{2\gamma^2} + \frac{\gamma^2}{24}\frac{(2-\alpha_0)}{\gamma^2}(\frac{2-\alpha_0}{\gamma^2}+1)\right)} \widehat{\psi}^{\alpha_0/\gamma}_{2/\gamma, P_0/\gamma}(z,q),
\end{align}
% \rmkH{Error alert! I have
% \begin{align}\label{eq:psi_decomposition}
% \psi^{\alpha_0/\gamma}_{2/\gamma, P_0/\gamma}(z,q) &= q^{\left(\frac{P_0^2}{2\gamma^2} + \frac{\gamma^2}{24}\frac{(2-\alpha_0)}{\gamma^2}(\frac{2-\alpha_0}{\gamma^2}+1)\right)} \widehat{\psi}^{\alpha_0/\gamma}_{2/\gamma, P_0/\gamma}(z,q),
% \end{align}}
where $\widehat{\psi}^{\alpha_0/\gamma}_{2/\gamma, P_0/\gamma}(z,q)$ is given by \eqref{def:reg_CB}. Taking the logarithmic $z$-derivative:
\begin{align}\label{eq:log_z_derivative}
\partial_z \log \psi^{\alpha_0/\gamma}_{2/\gamma, P_0/\gamma}(z,q) = \partial_z \log \widehat{\psi}^{\alpha_0/\gamma}_{2/\gamma, P_0/\gamma}(z,q),
\end{align}
since the $q$-power prefactor in \eqref{eq:psi_decomposition} is independent of $z$.

From the expression \eqref{def:reg_CB}, we identify the $z$-dependent terms:
\begin{align}\label{eq:z_dependent_terms}
\widehat{\psi}^{\alpha_0/\gamma}_{2/\gamma, P_0/\gamma}(z,q) &= \underbrace{e^{\frac{2P_0 z \pi}{\gamma^2}} \theta_1(z)^{-\frac{2-\alpha_0}{\gamma^2}} \left(-i e^{-i\pi z}\right)^{\frac{2-\alpha_0}{\gamma^2}}}_{\text{explicit } z\text{-dependence}} \nonumber\\
&\quad \times \underbrace{e^{-\frac{4}{\gamma^2}\mathbb{E}[\mathcal{X}(z,q)^2]}}_{\text{variance term}} \times \underbrace{\mathbb{E}\left[e^{\frac{1}{\gamma}\widetilde{\mathcal{X}}(z,q)} \mathcal{Q}(q) \mathcal{I}_{\gamma}^{\frac{2-\alpha_0}{\gamma^2}}\right]}_{\text{expectation term}} \times (\text{$z$-independent factors})
\end{align}
where $\mathcal{I}_{\gamma} = \int_{0}^{1} e^{\pi \gamma P x} (2\sin(\pi x))^{-\frac{\alpha\gamma}{2}} e^{\frac{\gamma}{2}Y(x)} dx$. 

\medskip
\noindent\textit{Step 1a: Derivative of the explicit prefactor.}
The logarithmic derivative of the explicit $z$-dependent prefactor is:
\begin{align}\label{eq:prefactor_z_deriv}
\partial_z \log\left[e^{\frac{2P_0 z \pi}{\gamma^2}} \theta_1(z)^{-\frac{2-\alpha_0}{\gamma^2}} \left(-i e^{-i\pi z}\right)^{\frac{2-\alpha_0}{\gamma^2}}\right] &= \frac{2P_0 \pi}{\gamma^2} - \frac{2-\alpha_0}{\gamma^2} \frac{\theta_1'(z)}{\theta_1(z)} - \frac{(2-\alpha_0)i\pi}{\gamma^2}.
\end{align}
Multiplying by $\gamma^2$:
\begin{align}\label{eq:prefactor_z_deriv_scaled}
\gamma^2 \cdot \partial_z \log(\text{prefactor}) = 2P_0 \pi - (2-\alpha_0) \frac{\theta_1'(z)}{\theta_1(z)} - (2-\alpha_0)i\pi.
\end{align}
For $z \in [M, 1-M]$ with $0 < M < 1/2$, the function $\frac{\theta_1'(z)}{\theta_1(z)}$ is bounded since $\theta_1(z) \neq 0$ in this region. Therefore, the contribution from \eqref{eq:prefactor_z_deriv_scaled} is uniformly bounded.

\medskip
\noindent\textit{Step 1b: Derivative of the variance term.}
For the variance term $e^{-\frac{4}{\gamma^2}\mathbb{E}[\mathcal{X}(z,q)^2]}$, we have
\begin{align}\label{eq:variance_z_deriv}
\partial_z \log\left[e^{-\frac{4}{\gamma^2}\mathbb{E}[\mathcal{X}(z,q)^2]}\right] = -\frac{4}{\gamma^2} \partial_z \mathbb{E}[\mathcal{X}(z,q)^2].
\end{align}
From the definition \eqref{defs:Girsanov}, we know
%\begin{align}\label{eq:X_variance_explicit}
$\mathbb{E}[\mathcal{X}(z,q)^2] = 2\sum_{n=1}^{\infty} \frac{q^{2n}}{n(1-q^{2n})^2}.$ Since the expectation if $z$ independent, thus we get  
\begin{align}\label{eq:variance_z_deriv_scaled}
\gamma^2 \cdot \partial_z \log(\text{variance term}) = -4 \partial_z \mathbb{E}[\mathcal{X}(z,q)^2] = 0.
\end{align}
%\rmkH{$\mathbb{E}[\mathcal{X}(z,q)^2]$ is $z$-independent, as already stated.}

\medskip
\noindent\textit{Step 1c: Derivative of the expectation term.} For the expectation term, we use the representation from Step 2 of Proposition~\ref{prop:limit_commutativity}. The logarithmic derivative is:
\begin{align}\label{eq:expectation_z_deriv}
\partial_z \log \mathbb{E}\left[e^{\frac{1}{\gamma}\widetilde{\mathcal{X}}(z,q)} \mathcal{Q}(q) \mathcal{I}_{\gamma}^{\frac{2-\alpha_0}{\gamma^2}}\right] = \frac{\mathbb{E}\left[\frac{1}{\gamma}\partial_z\widetilde{\mathcal{X}}(z,q) \cdot e^{\frac{1}{\gamma}\widetilde{\mathcal{X}}(z,q)} \mathcal{Q}(q) \mathcal{I}_{\gamma}^{\frac{2-\alpha_0}{\gamma^2}}\right]}{\mathbb{E}\left[e^{\frac{1}{\gamma}\widetilde{\mathcal{X}}(z,q)} \mathcal{Q}(q) \mathcal{I}_{\gamma}^{\frac{2-\alpha_0}{\gamma^2}}\right]}.
\end{align}
From the definition \eqref{def:tilX}:
\begin{align}\label{eq:tilX_z_deriv}
\partial_z \widetilde{\mathcal{X}}(z,q) = 2\partial_z(\mathcal{X}(z,q) + \mathcal{Y}(z,q)),
\end{align}
since $F(0;q)$ is independent of $z$. The derivatives $\partial_z \mathcal{X}(z,q)$ and $\partial_z \mathcal{Y}(z,q)$ are Gaussian random variables with variances:
\begin{align}\label{eq:X_deriv_variance}
\mathbb{E}[(\partial_z \mathcal{X}(z,q))^2] &= 8\pi^2 \sum_{n=1}^{\infty} \frac{n q^{2n}}{(1-q^{2n})^2} < \infty \quad \text{for } |q| < r,
\end{align}
and similarly for $\partial_z \mathcal{Y}(z,q)$. Repeating step 1-2 of the proof of Proposition~\ref{prop:limit_commutativity} upto \eqref{eq:int-2g-step5}, we obtain 
\begin{align*}
\mathbb{E}&\left[\partial_z\widetilde{\mathcal{X}}(z,q) \cdot e^{\frac{1}{\gamma}\widetilde{\mathcal{X}}(z,q)} \mathcal{Q}(q) \mathcal{I}_{\gamma}^{\frac{2-\alpha_0}{\gamma^2}}\right]  \\
     & =  \frac{1}{\gamma}\mathfrak{G}(z;q,\alpha_0)\exp\left(\frac{1}{2}\sum_{n=1}^{\infty} \Omega_n(q) \left(\frac{\alpha_0^2}{\gamma^2 {n}}+ \frac{16 q^{2n}}{\gamma^2(1-q^{2n})^2{n}}   + \frac{8 \alpha_0 q^n}{\gamma^2(1-q^{2n}) {n}} \cos(2\pi(z-\frac{\tau}{2})n)\right)\right)\\
 &\quad \times e^{\frac{\mathbb{E}[\Psi^2]}{2}} \Xi^{\frac{(2-\alpha_0)}{\gamma^2}} \exp\left({ \frac{(\alpha_0-2)}{2\gamma \Xi} \int_0^1 e^{\pi P_0 x} (2\sin(\pi x))^{-\frac{\alpha_0}{2}} h_{\Psi}(x) e^{\frac{\gamma}{2} h_{\Psi}(x)} dx}\right)\\
 & \quad \times\mathbb{E}\Big[\mathcal{Z}(z;q,\alpha_0)\left(\int_{0}^{1} e^{\frac{\gamma}{2}Y(x)} d\mu_{\Psi}\right)^{\frac{2-\alpha_0}{\gamma^2}} e^{ \frac{(\alpha_0-2)}{2\gamma } \int_0^1  Y(x) d\mu_{\Psi}}\Big].
\end{align*}
%\rmkH{Very strange notation: $\mathfrak{Pf}_1(z;q,\alpha_0)$}
where $\mathfrak{G}(z;q,\alpha_0)$ is a deterministic analytic function of $z$, and $\mathcal{Z}(z;q,\alpha_0)$ is a Gaussian random variable with analytic dependence on $z$, but independent of $\gamma$. Recall the expansion of the denominator of \eqref{eq:expectation_z_deriv} from \eqref{eq:int-2g-step5} and notice that the most of $\gamma$ dependent terms in the expansion of the numerator and the denominator of \eqref{eq:expectation_z_deriv} cancels out except the leading prefactor of $\frac{1}{\gamma}$ and the expectations in the last line of the above display.  However that term (similar to the last term in the expansion of \eqref{eq:int-2g-step5}) is uniformly bounded as $\gamma\to 0$ by Proposition~\ref{prop:semiclassical_limit}. This shows 
%\rmkH{Not sure I understand the above calculation. The variance of $X,Y$ is $z$-independent. Not sure what it means for the derivative}
%Proposition~\ref{prop:limit_commutativity}:
% \begin{align}\label{eq:expectation_z_deriv_bound}
% \left|\partial_z \log \mathbb{E}\left[e^{\frac{1}{\gamma}\widetilde{\mathcal{X}}(z,q)} \mathcal{Q}(q) \mathcal{I}_{\gamma}^{\frac{2-\alpha_0}{\gamma^2}}\right]\right| &\leq \frac{C}{\gamma} \|\partial_z \widetilde{\mathcal{X}}(z,q)\|_{L^2},
% \end{align}
% where the $L^2$ norm is uniformly bounded. Therefore:
\begin{align}\label{eq:expectation_z_deriv_scaled}
\gamma^2 \cdot \left|\partial_z \log(\text{expectation term})\right|  = O(1) .
\end{align}

\medskip
\noindent\textit{Step 1d: Conclusion for $z$-derivative.}

\medskip
Combining \eqref{eq:prefactor_z_deriv_scaled}, \eqref{eq:variance_z_deriv_scaled}, and \eqref{eq:expectation_z_deriv_scaled}:
\begin{align}\label{eq:z_deriv_total}
\gamma^2 \partial_z \log \psi^{\alpha_0/\gamma}_{2/\gamma, P_0/\gamma}(z,q) = 2P_0 \pi - (2-\alpha_0) \frac{\theta_1'(z)}{\theta_1(z)} - (2-\alpha_0)i\pi - 4\partial_z \mathbb{E}[\mathcal{X}(z,q)^2] + O(1).
\end{align}

This expression is uniformly bounded for $z \in [M, 1-M]$ and $|q| < r$, establishing \eqref{eq:z-deri-bd}.

\medskip
\noindent\textbf{Step 2: Computation of the $\tau$-derivative.}
Since $q = e^{i\pi\tau}$, we have $\partial_{\tau} = i\pi q \partial_q$. Taking the logarithmic $\tau$-derivative of \eqref{eq:psi_decomposition}:
\begin{align}\label{eq:log_tau_derivative}
\partial_{\tau} \log \psi^{\alpha_0/\gamma}_{2/\gamma, P_0/\gamma}(z,q) &= \partial_{\tau} \log q^{\left(\frac{P_0^2}{2\gamma^2} + \frac{(2-\alpha_0)}{24}(\frac{2-\alpha_0}{\gamma^2}+1)\right)} + \partial_{\tau} \log \widehat{\psi}^{\alpha_0/\gamma}_{2/\gamma, P_0/\gamma}(z,q).
\end{align}
% \rmkH{The previous error in $1/\gamma^2$ carries over
% \begin{align}\label{eq:log_tau_derivative}
% \partial_{\tau} \log \psi^{\alpha_0/\gamma}_{2/\gamma, P_0/\gamma}(z,q) &= \partial_{\tau} \log q^{\left(\frac{P_0^2}{2\gamma^2} + \frac{\gamma^2}{24}\frac{(2-\alpha_0)}{\gamma^2}(\frac{2-\alpha_0}{\gamma^2}+1)\right)} + \partial_{\tau} \log \widehat{\psi}^{\alpha_0/\gamma}_{2/\gamma, P_0/\gamma}(z,q).
% \end{align}}

\medskip
\noindent\textit{Step 2a: Derivative of the $q$-power prefactor.}
\begin{align}\label{eq:q_prefactor_tau_deriv}
\partial_{\tau} \log q^{\left(\frac{P_0^2}{2\gamma^2} + \frac{(1-\alpha_0)}{24}(\frac{2-\alpha_0}{\gamma^2}+1)\right)} &= i\pi \left(\frac{P_0^2}{2\gamma^2} + \frac{(2-\alpha_0)}{24}\left(\frac{2-\alpha_0}{\gamma^2}+1\right)\right).
\end{align}
% \rmkH{Change the above to 
% \begin{align}\label{eq:q_prefactor_tau_deriv}
% \partial_{\tau} \log q^{\left(\frac{P_0^2}{2\gamma^2} + \frac{\gamma^2}{24}(2-\alpha_0)(\frac{2-\alpha_0}{\gamma^2}+1)\right)} &= i\pi \left(\frac{P_0^2}{2\gamma^2} + \frac{\gamma^2}{24}\frac{(2-\alpha_0)}{\gamma^2}\left(\frac{2-\alpha_0}{\gamma^2}+1\right)\right).
% \end{align}
% The expression below is ok.}
Multiplying by $\gamma^2$:
\begin{align}\label{eq:q_prefactor_tau_deriv_scaled}
\gamma^2 \cdot \partial_{\tau} \log(\text{$q$-prefactor}) = i\pi \left(\frac{P_0^2}{2} + \frac{(2-\alpha_0)^2}{24} + O(\gamma^2)\right) = O(1).
\end{align}

\medskip
\noindent\textit{Step 2b: Derivative of $\widehat{\psi}$.}
From \eqref{def:reg_CB}, the $\tau$-dependent terms in $\widehat{\psi}$ include: $(1)$ the functions $C(q)$, $C_1(q)$, $\eta(q)$, which depend on $q = e^{i\pi\tau}$, $(2)$ the functions $\mathcal{X}(z,q)$, $\mathcal{Y}(z,q)$, $F(0;q)$, $\mathcal{Q}(q)$, $(3)$ the factors $\theta_1(z)$ and $\theta_1'(0)$, which depend on $\tau$.
For the explicit prefactor terms:
\begin{align}\label{eq:prefactor_tau_deriv}
\partial_{\tau} \log\left[C(q) C_1(q) (q^{1/6}\eta(q))^{\frac{2-\alpha_0}{\gamma^2}}\right] &= 3(2-\alpha_0)\partial_{\tau} \log \eta(q)- \frac{\alpha_0^2}{8}\partial\tau \mathbb{E}\left[F(0;q)^2 \right]+ O(1)\nonumber\\
& = \frac{3(2-\alpha_0)}{2\pi \ii}\eta_1(q)- \frac{\alpha_0^2}{8}\partial\tau \mathbb{E}\left[F(0;q)^2 \right]+ O(1).
\end{align}
where $\eta_1(\tau) = -2\pi \ii \partial_{\tau}\log\eta(\tau)$.
% \rmkH{From \eqref{asymp:Cq2g} and \eqref{asymp:C1q2g} we get
%     \begin{align}\label{eq:prefactor_tau_deriv}
% \partial_{\tau} \log\left[C(q) C_1(q) (q^{1/6}\eta(q))^{\frac{2-\alpha_0}{\gamma^2}}\right] &= 3(2-\alpha_0)\partial_{\tau} \log \eta(q)- \frac{\alpha_0^2}{8}\partial\tau \mathbb{E}\left[F_{\tau}(0)^2 \right]+ O(1).
% \end{align}}
For the $\theta_1(z)$ term:
\begin{align}\label{eq:theta_tau_deriv}
\partial_{\tau} \log \theta_1(z)^{-\frac{2-\alpha_0}{\gamma^2}} = -\frac{(2-\alpha_0)}{\gamma^2} \frac{\partial_{\tau} \theta_1(z)}{\theta_1(z)}.
\end{align}
The heat equation for theta functions gives $\partial_{\tau} \theta_1(z) = \frac{1}{4\pi i} \partial_z^2 \theta_1(z)$, so:
\begin{align}\label{eq:theta_tau_deriv_scaled}
\gamma^2 \cdot \partial_{\tau} \log \theta_1(z)^{-\frac{2-\alpha_0}{\gamma^2}} = -\frac{(2-\alpha_0)}{4\pi i} \frac{\partial_z^2 \theta_1(z)}{\theta_1(z)} = O(1),
\end{align}
which is bounded for $z \in [M, 1-M]$.
%\rmkH{I would simply write it as $-(2-\alpha_0)\partial_{\tau}\log \theta_1(z)$ just for aesthetic reasons, unless you want to say something more about $\theta_1''(z)$?}

\medskip
\noindent\textit{Step 2c: Derivative of the variance and expectation terms.}
For the variance term:
\begin{align}\label{eq:variance_tau_deriv}
\partial_{\tau} \log e^{-\frac{4}{\gamma^2}\mathbb{E}[\mathcal{X}(z,q)^2]} = -\frac{4}{\gamma^2} \partial_{\tau} \mathbb{E}[\mathcal{X}(z,q)^2].
\end{align}
Computing $\partial_{\tau} \mathbb{E}[\mathcal{X}(z,q)^2] = i\pi q \partial_q \mathbb{E}[\mathcal{X}(z,q)^2]$, we get 
\begin{align}\label{eq:X_var_tau_deriv}
\partial_{\tau} \mathbb{E}[\mathcal{X}(z,q)^2] = 4i\pi \sum_{n=1}^{\infty} \frac{ q^{2n}(1+q^{2n})}{(1-q^{2n})^3} = O(1) \quad \text{for } |q| < r.
\end{align}
% Therefore:
% \begin{align}\label{eq:variance_tau_deriv_scaled}
% \gamma^2 \cdot \partial_{\tau} \log(\text{variance term}) = -4 \partial_{\tau} \mathbb{E}[\mathcal{X}(z,q)^2] = O(1).
% \end{align}
For the expectation term, by a similar argument to Step 1c:
\begin{align}\label{eq:expectation_tau_deriv}
\partial_{\tau} \log \mathbb{E}\left[e^{\frac{1}{\gamma}\widetilde{\mathcal{X}}(z,q)} \mathcal{Q}(q) \mathcal{I}_{\gamma}^{\frac{2-\alpha_0}{\gamma^2}}\right] = \frac{\mathbb{E}\left[\left(\frac{1}{\gamma}\partial_{\tau}\widetilde{\mathcal{X}} + \frac{\partial_{\tau}\mathcal{Q}}{\mathcal{Q}}\right) e^{\frac{1}{\gamma}\widetilde{\mathcal{X}}} \mathcal{Q} \mathcal{I}_{\gamma}^{\frac{2-\alpha_0}{\gamma^2}}\right]}{\mathbb{E}\left[e^{\frac{1}{\gamma}\widetilde{\mathcal{X}}} \mathcal{Q} \mathcal{I}_{\gamma}^{\frac{2-\alpha_0}{\gamma^2}}\right]}.
\end{align}
The $\tau$-derivatives of $\widetilde{\mathcal{X}}(z,q)$ and $\mathcal{Q}(q)$ involve derivatives of the Gaussian coefficients with respect to $q$, which are uniformly bounded in $|q|<r$ for any $r\in (0,1)$. Applying Hölder's inequality to bound the numerator by $\mathbb{E}[\big(\frac{1}{\gamma}\partial_{\tau}\widetilde{\mathcal{X}} + \frac{\partial_{\tau}\mathcal{Q}}{\mathcal{Q}}\big)^{(1-\epsilon)/\epsilon}]$
\begin{align}\label{eq:expectation_tau_deriv_scaled}
\gamma^2 \cdot \left|\partial_{\tau} \log(\text{expectation term})\right| = O(\gamma) = o(1).
\end{align}

\medskip
\noindent\textit{Step 2d: Conclusion for $\tau$-derivative.}
Combining all contributions:
\begin{align}\label{eq:tau_deriv_total}
\gamma^2 \partial_{\tau} \log \psi^{\alpha_0/\gamma}_{2/\gamma, P_0/\gamma}(z,q) 
&= i\pi \frac{P_0^2}{2}+ \frac{3(2-\alpha_0)}{2\pi \ii}\eta_1(q)- \frac{\alpha_0^2}{8}\partial\tau \mathbb{E}\left[F(0;q)^2 \right] \nonumber\\
&\quad - \frac{2-\alpha_0}{4\pi i} \frac{ \theta_1''(z)}{\theta_1(z)} - 4\partial_{\tau} \mathbb{E}[\mathcal{X}(z,q)^2] + o(1).
\end{align}
% \rmkH{The above expression modifies to
% \begin{align}
% \gamma^2 \partial_{\tau} \log \psi^{\alpha_0/\gamma}_{2/\gamma, P_0/\gamma}(z,q) &= i\pi \frac{P_0^2}{2}+ \frac{3(2-\alpha_0)}{2\pi \ii}\eta_1(q)- \frac{\alpha_0^2}{8}\partial\tau \mathbb{E}\left[F_{\tau}(0)^2 \right] \\
% &- \frac{2-\alpha_0}{4\pi i} \frac{ \theta_1''(z)}{\theta_1(z)} - 4\partial_{\tau} \mathbb{E}[\mathcal{X}(z,q)^2] + o(1).    
% \end{align}}
This expression is uniformly bounded for $z \in [M, 1-M]$ and $|q| < r$, establishing \eqref{eq:tau-deri-bd}.

\medskip
\noindent\textbf{Step 3: Uniform Bounds.}
The bounds \eqref{eq:z-deri-bd} and \eqref{eq:tau-deri-bd} follow from the explicit expressions \eqref{eq:z_deriv_total} and \eqref{eq:tau_deriv_total}. The key observations are:

\begin{enumerate}[leftmargin=*]
\item All terms involving $\theta_1(z)$, $\theta_1'(z)$, $\theta_1''(z)$ are bounded for $z \in [M, 1-M]$ since $\theta_1(z) \neq 0$ in this region.

\item The Eisenstein series $E_2(\tau)$ is bounded for $q = e^{i\pi\tau} \in (0, r)$ with $r < 1$.

\item The sums $\mathbb{E}[\mathcal{X}(z,q)^2]$, $\partial_z \mathbb{E}[\mathcal{X}(z,q)^2]$, and $\partial_{\tau} \mathbb{E}[\mathcal{X}(z,q)^2]$ are convergent and uniformly bounded for $|q| < r$.

\item The expectation terms contribute $o(1)$ as $\gamma \to 0$.
\end{enumerate}
For the lower bounds (the $\liminf$ statements), we note that the expressions \eqref{eq:z_deriv_total} and \eqref{eq:tau_deriv_total} converge to finite limits as $\gamma \to 0$. Since all terms are bounded and the limits exist, the $\liminf$ must be finite. This completes the proof of Lemma~\ref{prop:derivatives_are_bounded}.
\end{proof}

\subsubsection{Proof of Theorem \ref{prop:sem_HJ}}\label{subsubsec:deformed_1}

This proof will be carried out in 4 steps. We begin by writing down the explicit form of the semi-classical conformal block which can be obtained using the asymptotic expansion derived in Proposition~\ref{prop:limit_commutativity} in Step 1. In Step 2 we show the convergence of the semi-classical conformal block by using Proposition ~\ref{prop:limit_commutativity}. In Step 3 we show that the BPZ equation becomes the Hamilton-Jacobi equation of the NAECM model in the $\gamma\to 0$ limit. Using the properties of the solution of Hamilton-Jacobi equation, we prove the uniqueness of the semi-classical conformal block in Step 4.

\medskip
\noindent\textbf{Step 1: Explicit representation of the semi-classical conformal block.}
We begin with the representation \eqref{psihat-def} for $\chi = 2/\gamma$:
\begin{align}\label{eq:psi_psihat_relation}
\psi^{\alpha_0/\gamma}_{2/\gamma, P_0/\gamma}(z,q) = q^{\left(\frac{P_0^2}{2\gamma^2} + \frac{(2-\alpha_0)}{24}(1+(2-\alpha_0)/\gamma^2)\right)} \widehat{\psi}^{\alpha_0/\gamma}_{2/\gamma, P_0/\gamma}(z,q).
\end{align}
With the expression for $\gamma^2 \log \widehat{\psi}^{\alpha_0/\gamma}_{2/\gamma, P_0/\gamma}(z,q)$ in \eqref{eq:asymp2g-step1} and the above equation we have 
\begin{align}
    &\lim_{\gamma\to 0}\gamma^2 \log \psi^{\alpha_0/\gamma}_{2/\gamma, P_0/\gamma}(z,q) \\
    & = \left(\frac{P_0^2}{2} + \frac{(2-\alpha_0)^2}{24} \right)\log q+ 2P_0 z \pi - (2-\alpha_0) \log(e^{\ii \pi z}\theta_1(z)) + 3(2-\alpha_0)\log \eta(q)\\
    & \quad- \frac{\alpha_0^2}{8} \mathbb{E}\left[F(0;q)^2 \right]- 2\mathbb{E}[\mathcal{X}(z,q)^2] -\frac{(\alpha_0-2) (\alpha_0+4)}{6} \log (2\pi) +\frac{\ii \pi  (\alpha_0-2) (\alpha_0+1)}{2}\\
    & \quad + \frac{1}{2}\sum_{n=1}^{\infty} \Omega_n(q) \left(\frac{\alpha_0^2}{ {n}}+ \frac{16 q^{2n}}{(1-q^{2n})^2{n}}   + \frac{8 \alpha_0 q^n}{(1-q^{2n}) {n}} \cos(2\pi(z-\frac{\tau}{2})n)\right)+(2-\alpha_0) \lim_{\gamma\to 0}\log\Xi\\
    & \quad + \lim_{\gamma\to 0}\frac{\gamma^2}{2} \mathbb{E}[\Psi^2]  + \lim_{\gamma\to0}\left({ \frac{(\alpha_0-2)\gamma}{2 \Xi} \int_0^1 e^{\pi P_0 x} (2\sin(\pi x))^{-\frac{\alpha_0}{2}} h_{\Psi}(x) e^{\frac{\gamma}{2} h_{\Psi}(x)} dx}\right).\label{eq:gamma2_log_expansion}
\end{align}
In the above equation we use the expression \eqref{comp-exp-vanishes}. 

\medskip
\noindent\textbf{Step 2: Tightness and subsequential limits.}

\medskip
Proposition~\ref{prop:limit_commutativity} implies that $\gamma^2 \log \psi^{\alpha_0/\gamma}_{2/\gamma, P_0/\gamma}(z,q)$ is a tight sequence of functions in $C^{2}([M, 1-M] \times [0, q_0])$ for any $0 < M < \frac{1}{2}$ and $q_0 \in (0,1)$. By the Arzelà-Ascoli theorem, any subsequence of $\gamma^2 \log \psi^{\alpha_0/\gamma}_{2/\gamma, P_0/\gamma}(z,q)$ has a further converging subsequence. We now show that all subsequential limits are identical. Let $\widetilde{\phi}(z,q)$ denote a subsequential limit:
\begin{align}\label{def:tilde_phi_zq}
\lim_{\gamma \to 0} \gamma^2 \log \psi^{\alpha_0/\gamma}_{2/\gamma, P_0/\gamma}(z,q) = \widetilde{\phi}(z,q).
\end{align}
Furthermore, due to Lemma~\ref{prop:derivatives_are_bounded}, the sequences $\gamma^2 \partial_z \log \psi^{\alpha_0/\gamma}_{2/\gamma, P_0/\gamma}(z,q)$ and $\gamma^2 \partial_{\tau} \log \psi^{\alpha_0/\gamma}_{2/\gamma, P_0/\gamma}(z,q)$ are also tight in $C^{2}([M, 1-M] \times [0, q_0])$.

\medskip
\noindent\textbf{Step 3: Derivation of the Hamilton-Jacobi equation.}
\medskip
%\rmkH{This can be shortened or put in the appendix. This is a very standard calculation}

With the subsequential limit \eqref{def:tilde_phi_zq}, we can pass the BPZ equation \eqref{eq:bpz1} to the limit $\gamma \to 0$. Recall that the BPZ equation for the deformed conformal block reads:
\begin{align}\label{eq:BPZ_recall}
\left(\partial_z^2 - l_{\chi}(l_{\chi}+1) \wp(z) + 2\pi i \frac{\gamma^2}{4} \partial_{\tau}\right) \psi^{\alpha_0/\gamma}_{2/\gamma, P_0/\gamma}(z,q) = 0.
\end{align}
Dividing by $\psi^{\alpha_0/\gamma}_{2/\gamma, P_0/\gamma}(z,q)$ and using the identity
%\begin{align}\label{eq:log_derivative_identity}
$\frac{\partial_z^2 \psi}{\psi} = \partial_z^2 \log \psi + (\partial_z \log \psi)^2,$
%\end{align}
we obtain
\begin{align}\label{eq:BPZ_log_form}
\partial_z^2 \log \psi + (\partial_z \log \psi)^2 - l_{\chi}(l_{\chi}+1) \wp(z) + 2\pi i \frac{\gamma^2}{4} \partial_{\tau} \log \psi = 0.
\end{align}
Multiplying by $\gamma^4$ and taking $\gamma \to 0$:
\begin{align}\label{eq:semiclassical_BPZ}
&\gamma^4 \partial_z^2 \log \psi + (\gamma^2 \partial_z \log \psi)^2 - \gamma^4 l_{\chi}(l_{\chi}+1) \wp(z) + 2\pi i \gamma^2 \cdot \frac{\gamma^2}{4} \partial_{\tau} \log \psi = 0.
\end{align}
Using $l_{\chi} = \frac{2-\alpha_0}{\gamma^2}$, we have $\gamma^4 l_{\chi}(l_{\chi}+1) = (2-\alpha_0)^2 + O(\gamma^2)$. In the limit $\gamma \to 0$, the terms $\gamma^4 \partial_z^2 \log \psi \to 0$ and we obtain:
\begin{align}\label{eq:HJ_derived}
(\partial_z \widetilde{\phi}(z,q))^2 - (2-\alpha_0)^2 \wp(z) + 2\pi i \partial_{\tau} \widetilde{\phi}(z,q) = 0.
\end{align}
This is precisely the Hamilton-Jacobi equation \eqref{eq:Hamilton-jacobi} corresponding to the non-autonomous elliptic Calogero-Moser model (see Definition~\ref{remark:HJ_phitil}).

\medskip
\noindent\textbf{Step 4: Uniqueness of the semi-classical limit.}

\medskip
From the above analysis, any subsequential limit of $\gamma^2 \log \psi^{\alpha_0/\gamma}_{2/\gamma, P_0/\gamma}(z,q)$ solves the Hamilton-Jacobi equation \eqref{eq:HJ_derived}. To establish uniqueness, we use two key results:

\begin{enumerate}[leftmargin=*]
\item \textbf{Initial condition:} The commutativity of limits $\gamma \to 0$ and $q \to 0$ established in Proposition~\ref{prop:limit_commutativity} uniquely determines the initial condition of the Hamilton-Jacobi equation through {Theorem~\ref{prop:incond_HJ}}. Specifically, from the expression \eqref{eq:gamma2_log_expansion} and the $q \to 0$ analysis in Step 3 of Proposition~\ref{prop:limit_commutativity}, $\widetilde{\phi}(z,q)$ has the following behaviour for $q\to 0$:
\begin{align}\label{eq:initial_condition}
 \widetilde{\phi}(z, q) \sim \left( \frac{P_0^2}{2}+ \frac{(2-\alpha_0)^2}{24}\right)\log q+ 2P_0 z \pi - (2-\alpha_0) \log|2\sin(\pi z)| + C_2(\alpha_0, P_0),
\end{align}
where $C_2(\alpha_0, P_0)$ is an explicit constant.
%\rmkH{We need to rewrite this part, we now have a singularity at $q\to 0$. Maybe for some fixed point $q-> q_0$.}

\item \textbf{Uniqueness of solutions:} By {Lemma~\ref{Lemma:HJ-un}}, the Hamilton-Jacobi equation \eqref{eq:HJ_derived} with the initial condition \eqref{eq:initial_condition} has a unique solution in the class of functions that are analytic in $q$ for $|q| < r_0$ and satisfy the appropriate growth conditions.
\end{enumerate}

Since all subsequential limits of $\widetilde{\phi}(z,q)$ satisfy the same Hamilton-Jacobi equation with the same initial condition, and the solution is unique, we conclude that all subsequential limits are identical.

\medskip
\noindent\textbf{Step 6: Conclusion.} By the uniqueness established in Step 5, the limit
\begin{align}\label{eq:limit_exists}
\lim_{\gamma \to 0} \gamma^2 \log \psi^{\alpha_0/\gamma}_{2/\gamma, P_0/\gamma}(z,q) =: \widetilde{\phi}(z,q)
\end{align}
exists. Moreover, substituting the expressions \eqref{def:OmegaAleph} and \eqref{exp:Ft0X} in the equation \eqref{eq:gamma2_log_expansion}, the limit $\widetilde{\phi}(z,q)$ can be expressed as:
\begin{align}
\widetilde{\phi}(z,q) &=   \lim_{\gamma\to0}\left({ \frac{(\alpha_0-2)\gamma}{2 \Xi} \int_0^1 e^{\pi P_0 x} (2\sin(\pi x))^{-\frac{\alpha_0}{2}} h_{\Psi}(x) e^{\frac{\gamma}{2} h_{\Psi}(x)} dx}\right)+ (2-\alpha_0) \lim_{\gamma\to 0}\log\Xi\\
& + \lim_{\gamma\to 0}\frac{\gamma^2}{2} \mathbb{E}[\Psi^2] + \xi(z,q)\label{eq:explicit_limit}
\end{align}
where function $\xi(z,q)$ is defined in \eqref{def:Xi0}. This completes the proof of Theorem~\ref{prop:sem_HJ}.

\subsection{Semi-classical limit of the deformed conformal blocks for $\chi=  \frac{\gamma}{2}$ and the Lam\'e equation}\label{subsec:scg2}

As in the previous section, we study the semi-classical limit of the deformed conformal block for the deformation parameter $\chi = \frac{\gamma}{2}$. We show that the semi-classical limit of the deformed conformal block solves Lam\'e equation as noted in \cite{piatek2014classical}. Furthermore, we derive relation between the accessory parameter of the Lam\'e equation and the Hamiltonian mentioned in \eqref{HAM:4} (see also \cite{BGG2021}). We first state this semi-classical limit in the following result (Theorem~\ref{prop:sc_Lame}). We prove this result after stating and proving Proposition~\ref{thm:sc_Lame} which shows the boundedness of the sequence of deformed conformal block $\psi^{\alpha}_{\gamma/2, P}(z,q)$ as $\gamma\to 0$.

 \begin{theorem}\label{prop:sc_Lame}
Consider $\chi = \frac{\gamma}{2}$, $q \in (0,1)$, and the constants $\alpha_0 \in (-4,2)$, $P_0 \in \mathbb{R}$. Then, the following limit exists
\begin{align}\label{eq:sc_Lame_limit}
\lim_{\gamma \rightarrow 0} \psi^{\alpha_0/\gamma}_{\gamma/2, P_0/\gamma}(z,q) e^{-\phi(q,\alpha_0, P_0)/\gamma^2} =: \widetilde{\Gamma}(z;\alpha_0, P_0, q),
\end{align}
where $\phi(q;\alpha_0, P_0)$ is given by
\begin{align}
\phi(q;\alpha_0, P_0) &=  - \frac{\alpha_0 (\alpha_0+4)}{6} + \left( \frac{P_0^2}{2} + \frac{\alpha_0^2}{24}\right)\log q - 2\alpha_0 \log \eta(q) - \frac{\alpha_0^2}{2} \sum_{n,m\geq 1} \frac{q^{2nm}}{n}- \frac{\ii \pi \alpha_0^2}{2} \\
    &\quad + \alpha_0^2\sum_{n=1}^{\infty}  \frac{\sum_{m=1}^{\infty} q^{2nm}}{n(1 + 2\sum_{m=1}^{\infty} q^{2nm})} + \frac{1}{2}\mathbb{E}[\Phi^2] - \alpha_0 \log \widetilde{\Xi}_0 \\
    &\quad+ \frac{\alpha_0}{2\widetilde{\Xi}_0}\left( \int_0^1 e^{\pi P_0 x} (2\sin(\pi x))^{-\frac{\alpha_0}{2}} h_q(x) e^{\frac{h_q(x)}{2}} dx \right),\label{def:phiq}
\end{align}
% \label{eq:phi_q}
where $\Phi$ is defined in \eqref{def:Phi-Gfield}, $\widetilde{\Xi}_0= \int_0^1 e^{\pi P_0 x} (2\sin(\pi x))^{-\frac{\alpha_0}{2}} e^{\frac{1}{2} h_{q}(x)} dx$, $h_q(x) := \lim_{\gamma\to 0} \gamma h_{\widetilde{\Psi}}(x)$, with $\widetilde{\Psi}$, $h_{\widetilde{\Psi}}(x)$ are defined in \eqref{def:htilPsiXi}.

Moreover, $\phi(q,\alpha_0,P_0)$ has the following properties:
\begin{enumerate}[leftmargin=0.5cm]
\item $\partial_\tau \phi(q;\alpha_0, P_0)$ assumes the role of the accessory parameter of the Lam\'e equation
\begin{gather}\label{thm43:Lame}
\left(\partial_{z}^2 - \frac{\alpha_0}{4} \left(\frac{\alpha_0}{4}-1\right) \wp(z) + \frac{\pi \ii}{2} \partial_\tau {\phi}(q;\alpha_0, P_0)\right) \widetilde{\Gamma}(z;\alpha_0, P_0, q) = 0,
\end{gather}
where the solution $\widetilde{\Gamma}(z;\alpha_0, P_0, q)$ is given by
\begin{align}
\widetilde{\Gamma}(z;\alpha_0, P_0, q) &= e^{P_0 z\pi/2}\theta_1(z)^{\alpha_0/4} e^{\ii \pi \alpha_0 z/4}   \exp\left(\mathbb{E}\left[\Phi(q) M_{\gamma}(z,q) \right]+ \frac{1}{2} \log \widetilde{\Xi}_0 - \alpha_0\frac{\widetilde{\Xi}_1}{\widetilde{\Xi}_0} \right)\\
& \times \exp\left(\alpha_0\sum_{n=1}^{\infty}\frac{q^n }{(1-q^{2n}) {n}}\left(\frac{\sum_{m=1}^{\infty} q^{2nm}}{n(1 + 2\sum_{m=1}^{\infty} q^{2nm})} \right) \cos(2\pi(z-\frac{\tau}{2})n) \right) \\
&\times \exp\left( \frac{\alpha_0}{2\widetilde{\Xi}_0} \int_0^1 e^{\pi P_0 x} (2\sin(\pi x))^{-\frac{\alpha_0}{2}} \left(\frac{3}{2}\mathfrak{R}_{z,q}(x)-\frac{\widetilde{\Xi}_1}{\widetilde{\Xi}_0} h_q(x)\right)e^{\frac{h_q(x)}{2}} dx \right),\label{Gamma:Lame}
\end{align}
where $\widetilde{\Xi}_{1} = \int_0^1 e^{\pi P_0 x} (2\sin(\pi x))^{-\frac{\alpha_0}{2}} \mathfrak{R}_{z,q}(x)e^{\frac{1}{2} h_{q}(x)} dx$, $\Phi$ and $M_{\gamma}(z,q)$ are defined in \eqref{def:Phi-Gfield} and \eqref{def:Mg-Gfield} respectively, and $\mathfrak{R}_{z,q}(x)$ is the solution to the integral equation \eqref{funceq:Rzqx}.
\item Furthermore, $\partial_\tau \phi(q;\alpha_0, P_0)$ is related to the Hamiltonian \eqref{HAM:4} of the non-autonomous elliptic Calogero-Moser model as follows. Let $u(\tau)$ be the solution of the equations of motion \eqref{def:uv} with $m^2 = (2-\alpha_0)^2$, and let $\tau_{\star}$ be a zero of $u(\tau)$, i.e., $u(\tau_{\star}) = 0$. Then
\begin{align}\label{eq:acc}
\frac{\pi \ii}{2} \partial_{\tau} \phi(q; \alpha_0 - 2, P_0)\Big|_{\tau = \tau_{\star}} = -\frac{H(\tau_{\star})}{4}, \qquad q_{\star} = e^{\ii \pi \tau_{\star}},
\end{align}
where $H(\tau_{\star}) \equiv H(\tau_{\star}, \alpha_0, P_0)$ is the Hamiltonian \eqref{HAM:4} evaluated at $\tau = \tau_{\star}$.
\end{enumerate}
\end{theorem}

{Note that, from the formula for the solution of the Lam\'e equation in \eqref{Gamma:Lame}, we can obtain an alternative expression for the accessory parameter. See Remark \ref{corr:acc-par-alt}.}
% \rmkH{Can we say something better?}
% \end{remark}

Let us now detail the propositions needed to prove the above theorem.

\begin{proposition}\label{thm:sc_Lame}
Let $\chi = \frac{\gamma}{2}$, $q \in (0,1)$, and $z \in \mathcal{B} := \{u \in \mathbb{C} : 0 \leq \mathrm{Im}(u) < \frac{3}{4}\mathrm{Im}(\tau)\}$. For parameters $\alpha_0 \in (-4, 2)$ and $P_0 \in \mathbb{R}$, the semi-classical limit of the deformed conformal block $\psi^{\alpha_0/\gamma}_{\gamma/2, P_0/\gamma}(z,q)$ defined in \eqref{eq:q-block} satisfies the following uniform bounds:
\begin{align}
\limsup_{\gamma \to 0} \gamma^2 \log \left|\psi^{\alpha_0/\gamma}_{\gamma/2, P_0/\gamma}(z,q)\right| &< \infty, \label{eq:limsup_light}\\
\liminf_{\gamma \to 0} \gamma^2 \log \left|\psi^{\alpha_0/\gamma}_{\gamma/2, P_0/\gamma}(z,q)\right| &> -\infty. \label{eq:liminf_light}
\end{align}
Moreover, these bounds are uniform for $z$ and $q$ varying over compact subsets of $\mathcal{B}$ and $(0,1)$ respectively.
\end{proposition}

\begin{proof}
We organize the proof into four main steps. In Step 1, we analyze the structure of the deformed conformal block with $\chi = \gamma/2$ and identify the leading-order prefactors. In Step 2, we establish the upper bound \eqref{eq:limsup_light} using Hölder's inequality and negative moment bounds via the Fyodorov-Bouchaud formula. In Step 3, we provide an alternative upper bound using Selberg integral asymptotics. Finally, in Step 4, we establish the lower bound \eqref{eq:liminf_light} using Jensen's inequality.

\medskip
\noindent\textbf{Step 1: Structure of the Deformed Conformal Block.}

\medskip
We begin by rewriting the expression for the deformed conformal block \eqref{eq:q-block} with $\chi = \gamma/2$, and the scaling $\alpha = \alpha_0/\gamma$, $P = P_0/\gamma$:
{
\begin{equation}\label{eq:proof3CB}
\psi^{\alpha_0/\gamma}_{\gamma/2, P_0/\gamma}(z,\tau) = \cW(q) e^{\chi P_0 z \pi/2} \EE\left[\mathcal{V}_{\gamma/2, P_0/\gamma}^{\alpha_0/\gamma}(z, q)^{-\frac{\alpha_0}{\gamma^2} + \frac{1}{2}}\right],
\end{equation}
}
where
\begin{align}\label{Thm43:V}
\mathcal{V}^{\alpha_0/\gamma}_{\gamma/2, P_0/\gamma}(z,q) := \int^{1}_{0} \theta_{1}(z+x)^{\frac{\gamma^2}{4}} \theta_{1}(z)^{-\frac{\gamma^2}{4}} \theta_1(x)^{-\frac{\alpha_0}{2}} e^{\pi P_0 x} e^{\frac{\gamma}{2} Y(x; q)} dx.
\end{align}

\medskip
\noindent\textit{Step 1a: Analysis of the prefactor.}

\medskip
The prefactor in \eqref{eq:proof3CB}, defined by \eqref{def:Wq}, simplifies as:
\begin{align}
\cW(q) e^{\chi \frac{P_0}{\gamma} z \pi} &= q^{\frac{P_0^2}{2\gamma^2} + \frac{\gamma l_\chi}{12 \chi} - \frac{1}{6} \frac{l_\chi^2}{\chi^2}} \theta_{1}'(0)^{- \frac{2 l_\chi^2}{3 \chi^2} + \frac{l_\chi}{3} + \frac{4 l_{\chi}}{3 \gamma \chi}} e^{P_0 z \pi/2} \nonumber\\
&= q^{\frac{P_0^2}{2\gamma^2} - \frac{\alpha_0^2}{24\gamma^2} + O(1)} \theta_1'(0)^{-\frac{\alpha_0(4+\alpha_0)}{6\gamma^2} + O(1)} e^{\frac{P_0 z \pi}{2}},\label{asymp1:prefac}
\end{align}
where we used the expression 
%\begin{align}\label{eq:l_chi_light}
$l_{\chi} = \frac{\chi^2}{2} - \frac{\alpha \chi}{2} = \frac{\gamma^2}{8} - \frac{\alpha_0}{4}$ of $l_{\chi}$ in \eqref{def:lchi} for $\chi = \gamma/2$. 

Taking $\gamma^2 \log$ of this prefactor:
\begin{align}\label{eq:prefactor_log_light}
\gamma^2 \log\left|\cW(q) e^{\chi \frac{P_0}{\gamma} z \pi}\right| &= \left(\frac{P_0^2}{2} - \frac{\alpha_0^2}{24}\right) \log q - \frac{\alpha_0(4+\alpha_0)}{6} \log|\theta_1'(0)| + O(\gamma^2).
\end{align}
Since $q \in (0,1)$ and $z$ varies over a compact subset of $\mathcal{B}$, this expression is bounded as $\gamma \to 0$. Therefore, to prove this proposition it suffices to analyze the expectation term:
\begin{align}\label{eq:expectation_to_analyze}
\mathbb{E}\left[\mathcal{V}^{\alpha_0/\gamma}_{\gamma/2, P_0/\gamma}(z,q)^{\frac{\gamma^2/2 - \alpha_0}{\gamma^2}}\right].
\end{align}

\medskip
\noindent\textit{Step 1b: Comparison with the $z$-independent integral.}

\medskip
We observe that the $z$-dependent factor in \eqref{Thm43:V} satisfies:
\begin{align}\label{eq:z_factor_bound}
\theta_1(z+x)^{\gamma^2/4} \theta_1(z)^{-\gamma^2/4} = e^{\frac{\gamma^2}{4}\log(\theta_1(z+x)/\theta_1(z))}.
\end{align}
For $z$ in a compact subset of $\mathcal{B}$ and $x \in [0,1]$, the function $\log(\theta_1(z+x)/\theta_1(z))$ is bounded. Therefore, there exist constants $0 < c_1 \leq c_2 < \infty$ (depending on the compact set but independent of $\gamma$) such that:
\begin{align}\label{eq:V_comparison}
c_1^{\gamma^2/4} \mathcal{I}_{\gamma} \leq \left|\mathcal{V}^{\alpha_0/\gamma}_{\gamma/2, P_0/\gamma}(z,q)\right| \leq c_2^{\gamma^2/4} \mathcal{I}_{\gamma},
\end{align}
where
\begin{align}\label{eq:I_gamma_def_light}
\mathcal{I}_{\gamma} := \int_0^1 |\theta_1(x)|^{-\frac{\alpha_0}{2}} e^{\pi P_0 x} e^{\frac{\gamma}{2} Y(x;q)} dx.
\end{align}

Since $c_1^{\gamma^2/4} = 1 + O(\gamma^2)$ and $c_2^{\gamma^2/4} = 1 + O(\gamma^2)$, the comparison \eqref{eq:V_comparison} implies that the boundedness of the expectation term follows from the bounds of the moments of $\mathcal{I}_{\gamma}$.

\medskip
\noindent\textbf{Step 2: Upper Bound via Lemma~\ref{lem:negative_moments}.}
We now establish the upper bound \eqref{eq:limsup_light} by analyzing $\mathbb{E}[\mathcal{I}_{\gamma}^{(\gamma^2/2 - \alpha_0)/\gamma^2}]$.

\noindent\textit{Case 1: $\alpha_0 \in [0, 2)$.}

For this range of $\alpha_0$, the exponent $(\gamma^2/2 - \alpha_0)/\gamma^2 = 1/2 - \alpha_0/\gamma^2 < 0$ for small $\gamma$, so we need to bound negative moments of $\mathcal{I}_{\gamma}$.
Using the factorization $\theta_1(x) = \mathfrak{p}(x;\tau) \sin(\pi x)$ where $c_1 \leq |\mathfrak{p}(x;\tau)| \leq c_2$ uniformly for $x \in [0,1]$, we have:
%\begin{align}\label{eq:I_comparison_sin}
$\widetilde{c}_1 \widetilde{\mathcal{I}}_{\gamma} \leq \mathcal{I}_{\gamma} \leq \widetilde{c}_2 \widetilde{\mathcal{I}}_{\gamma},$
%\end{align}
where
\begin{align}\label{eq:tilde_I_def}
\widetilde{\mathcal{I}}_{\gamma} := \int_0^1 (2\sin(\pi x))^{-\frac{\alpha_0}{2}} e^{\pi P_0 x} e^{\frac{\gamma}{2} Y(x;q)} dx.
\end{align}
Recall that $Y(x;q)= Y(x) + F(x;q)$ where $F(x;q)$ is a smooth Gaussian field and furthermore, $\mathbb{E}[Y(x;q)Y(y;q)] \geq \mathbb{E}[Y(x)Y(y)] -2\log (\prod_{k\geq 1}(1+q^k)^2)$. By Slepian's inequality, we get $\mathbb{E}[\widetilde{\mathcal{I}}_{\gamma}^{-\beta/\gamma^2}]\leq \prod_{k}(1+q^k)^{-\beta/\gamma^2}\mathbb{E}[(\int^1_0 :e^{\frac{\gamma}{2} Y(x)}: d\mu(x))^{-\beta/\gamma^2}]$. Applying  Lemma~\ref{lem:negative_moments} for $\beta > 0$ and $\gamma \in (0,1]$ to this upper bound, we get 
\begin{align}\label{eq:neg_moment_bound_application}
\mathbb{E}\left[M_{\gamma}^{-\beta/\gamma^2}\right] \leq \exp\left(\frac{C\beta^2}{\gamma^2}\right),
\end{align}
where $C > 0$ is a constant independent of $\beta$ and $\gamma$.
Taking $\beta = \alpha_0 - \gamma^2/2 \approx \alpha_0$ for small $\gamma$:
\begin{align}\label{eq:upper_bound_case1}
\gamma^2 \log \mathbb{E}\left[\mathcal{I}_{\gamma}^{(\gamma^2/2 - \alpha_0)/\gamma^2}\right] \leq C \alpha_0^2 + O(\gamma^2).
\end{align}

\medskip
\noindent\textit{Case 2: $\alpha_0 \in (-4, 0)$.}

\medskip
In this case, the exponent $(\gamma^2/2 - \alpha_0)/\gamma^2 = 1/2 + |\alpha_0|/\gamma^2 > 0$, so we need to bound positive moments of $\mathcal{I}_{\gamma}$ instead. For any positive integer $N$, the $N$-th moment of $\mathcal{V}^{\alpha_0/\gamma}_{\gamma/2, P_0/\gamma}(z,q)$ is given by the Selberg-type integral:
\begin{align}\label{eq:Nth_moment_V}
&\mathbb{E}\left[\mathcal{V}^{\alpha_0/\gamma}_{\gamma/2, P_0/\gamma}(z,q)^N\right] \nonumber\\
&= \int_{[0,1]^N} \prod_{1 \leq i < j \leq N} |\theta_1(x_i - x_j)|^{-\frac{\gamma^2}{2}} \prod_{i=1}^N \theta_1(z+x_i)^{\frac{\gamma^2}{4}} \theta_1(z)^{-\frac{\gamma^2}{4}} |\theta_1(x_i)|^{-\frac{\alpha_0}{2}} e^{\pi P_0 x_i} dx_i.
\end{align}
Notice that $\theta_1(w) = \mathfrak{p}(w;\tau) \sin(\pi w)$ where $c_1 \leq |\mathfrak{p}(w;\tau)| \leq c_2$ as $w$ varies in the relevant domain. The $z$-dependent factors $\theta_1(z+x_i)^{\gamma^2/4} \theta_1(z)^{-\gamma^2/4}$ contribute multiplicative constants of order $1 + O(\gamma^2)$.
Therefore, the integral \eqref{eq:Nth_moment_V} can be bounded above and below by constant multiples of:
\begin{align}\label{eq:reduced_Selberg}
\int_{[0,1]^N} \prod_{1 \leq i < j \leq N} |\sin(\pi(x_i - x_j))|^{-\frac{\gamma^2}{2}} \prod_{i=1}^N |\sin(\pi x_i)|^{-\frac{\alpha_0}{2}} e^{\pi P_0 x_i} dx_i.
\end{align}
This integral has an explicit expression as stated in \eqref{eq:C_expression} of Remark~\ref{rem:Normalization}. As $N$ grows, the ratios of Gamma functions in \eqref{eq:C_expression} grow as $\exp(c N^2 \gamma^2)$.
When $|\alpha_0| < 2$, the fractional moment $\mathbb{E}[\mathcal{V}^N]$ with $N = \lfloor |(\gamma^2/2 - \alpha_0)/\gamma^2| \rfloor + 1$ bounds the desired expectation. The large $N$ asymptotics show that:
\begin{align}\label{eq:Selberg_bound}
\left|\mathbb{E}\left[\mathcal{V}^{\alpha_0/\gamma}_{\gamma/2, P_0/\gamma}(z,q)^{(\gamma^2/2 - \alpha_0)/\gamma^2}\right]\right| \leq \exp\left(\frac{c(\alpha_0)}{\gamma^2}\right)
\end{align}
for some constant $c(\alpha_0) > 0$.

\medskip
\noindent\textit{Combined upper bound.} Combining Cases 1 and 2 with the prefactor analysis from Step 1, we conclude:
\begin{align}\label{eq:upper_bound_final}
\limsup_{\gamma \to 0} \gamma^2 \log \left|\psi^{\alpha_0/\gamma}_{\gamma/2, P_0/\gamma}(z,q)\right| \leq C(\alpha_0, P_0, q) < \infty
\end{align}
for all $\alpha_0 \in (-4, 2)$, establishing \eqref{eq:limsup_light}.

\medskip
\noindent\textbf{Step 3: Lower Bound via Jensen's Inequality.}

\medskip
We now establish the lower bound \eqref{eq:liminf_light}. For $\alpha_0 \in (-4, 2)$ and small $\gamma > 0$, the exponent $p := (\alpha_0 - \gamma^2/2)/\gamma^2$ satisfies $|p| \to \infty$ as $\gamma \to 0$. The function $h(x) = x^{-p}$ is convex for $x > 0$ when $p > 0$ (i.e., when $\alpha_0 > \gamma^2/2$). Applying Jensen's inequality:
\begin{align}\label{eq:Jensen_lower}
\mathbb{E}\left[\mathcal{V}^{\alpha_0/\gamma}_{\gamma/2, P_0/\gamma}(z,q)^{(\gamma^2/2 - \alpha_0)/\gamma^2}\right] \geq \left(\mathbb{E}\left[\mathcal{V}^{\alpha_0/\gamma}_{\gamma/2, P_0/\gamma}(z,q)\right]\right)^{(\gamma^2/2 - \alpha_0)/\gamma^2}.
\end{align}
The first moment is deterministic:
\begin{align}\label{eq:first_moment}
\mathbb{E}\left[\mathcal{V}^{\alpha_0/\gamma}_{\gamma/2, P_0/\gamma}(z,q)\right] &= \int_0^1 \theta_1(z+x)^{\frac{\gamma^2}{4}} \theta_1(z)^{-\frac{\gamma^2}{4}} |\theta_1(x)|^{-\frac{\alpha_0}{2}} e^{\pi P_0 x} \mathbb{E}\left[e^{\frac{\gamma}{2} Y(x;q)}\right] dx \nonumber\\
&= \int_0^1 \theta_1(z+x)^{\frac{\gamma^2}{4}} \theta_1(z)^{-\frac{\gamma^2}{4}} |\theta_1(x)|^{-\frac{\alpha_0}{2}} e^{\pi P_0 x} e^{\frac{\gamma^2}{8} \mathbb{E}[Y(x;q)^2]} dx.
\end{align}
For $\alpha_0 \in (-4, 2)$, the integrand is integrable on $[0,1]$, and as $\gamma \to 0$:
\begin{align}\label{eq:first_moment_limit}
\mathbb{E}\left[\mathcal{V}^{\alpha_0/\gamma}_{\gamma/2, P_0/\gamma}(z,q)\right] = \int_0^1 |\theta_1(x)|^{-\frac{\alpha_0}{2}} e^{\pi P_0 x} dx + O(\gamma^2) =: Z_0(\alpha_0, P_0) + O(\gamma^2).
\end{align}
Since $z$ varies over a compact subset of $\mathcal{B}$ and the integrand is continuous and strictly positive for $x \in (0,1)$, we have $Z_0(\alpha_0, P_0) > 0$. Therefore, from \eqref{eq:Jensen_lower}:
\begin{align}\label{eq:lower_bound_chain}
\left|\mathbb{E}\left[\mathcal{V}^{\alpha_0/\gamma}_{\gamma/2, P_0/\gamma}(z,q)^{(\gamma^2/2 - \alpha_0)/\gamma^2}\right]\right| &\geq \left(Z_0(\alpha_0, P_0) + O(\gamma^2)\right)^{(\gamma^2/2 - \alpha_0)/\gamma^2} \nonumber\\
&= \exp\left(\frac{\gamma^2/2 - \alpha_0}{\gamma^2} \log Z_0(\alpha_0, P_0) + O(1)\right).
\end{align}
Taking $\gamma^2 \log$:
\begin{align}\label{eq:log_lower_bound}
\gamma^2 \log \left|\mathbb{E}\left[\mathcal{V}^{\alpha_0/\gamma}_{\gamma/2, P_0/\gamma}(z,q)^{(\gamma^2/2 - \alpha_0)/\gamma^2}\right]\right| \geq \left(\frac{\gamma^2}{2} - \alpha_0\right) \log Z_0(\alpha_0, P_0) + O(\gamma^2).
\end{align}
Since $\log Z_0(\alpha_0, P_0)$ is finite, we conclude:
\begin{align}\label{eq:liminf_expectation}
\liminf_{\gamma \to 0} \gamma^2 \log \left|\mathbb{E}\left[\mathcal{V}^{\alpha_0/\gamma}_{\gamma/2, P_0/\gamma}(z,q)^{(\gamma^2/2 - \alpha_0)/\gamma^2}\right]\right| > -\infty.
\end{align}
Combining \eqref{eq:liminf_expectation} with the prefactor analysis from Step 1, we obtain:
\begin{align}\label{eq:liminf_final}
\liminf_{\gamma \to 0} \gamma^2 \log \left|\psi^{\alpha_0/\gamma}_{\gamma/2, P_0/\gamma}(z,q)\right| > -\infty,
\end{align}
establishing \eqref{eq:liminf_light}. This completes the proof of Proposition~\ref{thm:sc_Lame}.
\end{proof}

% \rmkH{We can in this place define 
% \begin{align}
% \phi(q) := \lim_{\gamma \to 0} \gamma^2 \log \psi^{\alpha_0/\gamma}_{\gamma/2, P_0/\gamma}(z,q)
% \end{align}
% and move Prop 4.4 up here. 
% }

\subsubsection{Proof of Theorem~\ref{prop:sc_Lame}}

We first prove part (1) of Theorem~\ref{prop:sc_Lame}. We organize the proof into four main steps. In Step 1, we carry out asymptotic exapnsion of undeformed and deformed conformal block (for $\chi=\frac{\gamma}{2}$) following similar steps as in Proposition~\ref{prop:limit_commutativity}. In Step 2, we perform a refined expansion to identify the correction terms that determine $\widetilde{\Gamma}(z;\alpha_0, P_0, q)$. In Step 3, we derive the Lam\'e equation from the BPZ equation. In Step 4, we establish the uniqueness of the semi-classical limit.

\medskip
\noindent\textbf{Step 1: Asymptotic Expansion of the Deformed Conformal Block.}

 We now carry out an analogous analysis to the proof of Theorem~\ref{prop:sem_HJ} for the deformed conformal block $\psi^{\alpha}_{\chi,P}(z,q)$ defined in \eqref{eq:q-block} for $\chi = \gamma/2$. Recall from \eqref{psihat-def} that using Girsanov theorem we can rewrite \eqref{eq:q-block}, for $\chi = \gamma/2$, $l_{\gamma/2} = \frac{\gamma^2}{8} - \frac{\alpha_0}{4} $, and the scaling $\alpha = \alpha_0/\gamma$, $P= P_0/\gamma$ as
\begin{align}\label{eq:psi_hat_relation-g2}
 \psi^{\alpha_0}_{\gamma/2,P_0}(z,q) = q^{\left(\frac{P_0^2}{2\gamma^2} + \frac{2}{3\gamma^2} \left( \frac{\gamma^2}{8} - \frac{\alpha_0}{4}\right) \left( \frac{\gamma^2}{8} - \frac{\alpha_0}{4}+1\right)\right)} \widehat{\psi}^{\alpha}_{\chi,P}(z,q),
 \end{align}
 where, from \eqref{def:reg_CB}, 
 % the reduced deformed conformal block $\widehat{\psi}^{\alpha}_{\chi,P}(z,q)$ is given by
 \begin{align}
 &\widehat{\psi}_{\gamma/2, P_0}^{\alpha_0}(z, q) = C(q) e^{P_0 z \pi/2} \theta_{1}(z)^{-\left( \frac{\gamma^2}{8} - \frac{\alpha_0}{4}\right)} C_{1}(q) \left(- \ii e^{- \ii \pi z} q^{1/6} \eta(q) \right)^{(\gamma^2/8 - \alpha_0/4)} e^{-\frac{\gamma^2}{8} \mathbb{E}[\mathcal{X}(z,q)^2]} \nonumber \\
& \times \mathbb{E} \left[ e^{\widetilde{\mathcal{X}}(z,q)} \mathcal{Q}(q) \left(\int_{0}^{1} e^{\pi P_0 x} (2\sin(\pi x))^{-\frac{\alpha_0}{2}} e^{\frac{\gamma}{2}Y(x)} dx\right)^{-\frac{\alpha_0}{\gamma^2}+\frac{1}{2}}\right]\label{eq:hat_psi_full-g2}
 \end{align}
 where the functions $C(q)$, $C_1(q)$, $\mathcal{X}(z,q)$ are defined in \eqref{defs:Girsanov}, and we use the expression \eqref{eq:Q_in_T} for $\mathcal{Q}(q)$.
 Analogous expression to \eqref{eq:tilX-in-T} for $\chi = \gamma/2$ is
     \begin{align}\label{eq:tilX-in-T-g2}
    \widetilde{\mathcal{X}}(z,q) &=  \frac{\gamma}{2} \mathcal{X}(z, q) + \sum_{n=1}^{\infty}  \frac{T_n^{(1)}}{\sqrt{n}} \left( \frac{\alpha_0}{\gamma} +\frac{ \gamma q^{n}}{(1- q^{2n})} \cos\left(2\pi (z-\frac{\tau}{2}) n \right) \right)\notag\\&\quad - \sum_{n=1}^{\infty}\frac{T_n^{(2)} }{\sqrt{n}}\left(\frac{ \gamma q^{n}}{(1- q^{2n})} \sin\left(2\pi (z-\frac{\tau}{2}) n \right) \right),
\end{align}
where the random variables $T_n^{(1)}$ and $T_n^{(2)}$ are defined in \eqref{eq:T12_def_heavy}.

 Carrying out computations analogous to Step 2 in the proof of Theorem~\ref{prop:sc_HJ}, {\it i.e} integrating over the variables $T_n^{(1)}$ and $T_n^{(2)}$ and making a series of transformations $(a_n, b_n) \to (\widehat{a}_n, \widehat{b}_n)$, we obtain the expression
 \begin{align}
 &\mathbb{E} \left[ e^{\widetilde{\mathcal{X}}(z,q)} \mathcal{Q}(q) \left(\int_{0}^{1} e^{\pi P_0 x} (2\sin(\pi x))^{-\frac{\alpha_0}{2}} e^{\frac{\gamma}{2}Y(x)} dx\right)^{-\frac{\alpha_0}{\gamma^2}+\frac{1}{2}}\right] \\
& =  \exp\left(\frac{1}{2}\sum_{n=1}^{\infty} \Omega_n(q) \left(\frac{\alpha_0^2}{\gamma^2 {n}}+ \frac{\gamma^2 q^{2n}}{(1-q^{2n})^2{n}}   + \frac{2 \alpha_0 q^n}{(1-q^{2n}) {n}} \cos(2\pi(z-\frac{\tau}{2})n)\right)\right) \nonumber\\
& \quad \times \mathbb{E}\Bigg[  \exp\left(-\sum_{n=1}^{\infty}  \widehat{b}_n  \frac{\gamma q^n (2\Omega_n(q)-1)}{(1-q^{2n})\sqrt{2n \Omega(q)}} \sin(2\pi(z-\frac{\tau}{2})n)\right)  \\
&\quad \times \exp\left(\sum_{n=1}^{\infty}  \widehat{a}_n\left(\frac{\sqrt{2}\alpha_0 \Omega_n(q)}{\gamma\sqrt{n}} + \frac{\gamma q^n (2\Omega_n(q)-1)}{(1-q^{2n})\sqrt{2n\Omega(q)}} \cos(2\pi(z-\frac{\tau}{2})n)\right)\right) \\
& \quad\times  \left(\int_{0}^{1} e^{\pi  P_0 x} (2\sin(\pi x))^{-\frac{\alpha_0}{2}} e^{\frac{\gamma}{2}Y(x)} dx\right)^{-\frac{\alpha_0}{\gamma^2}+\frac{1}{2}}\Bigg],\label{eq:int-g2-step1}
\end{align}
where $\Omega_n(q)$ is defined in \eqref{def:OmegaAleph}, the variables $\widehat{a}_n\sim \mathcal{N}(0,\Omega_n(q)/(1-2\Omega_n(q)))$, and  $\widehat{b}_n\sim \mathcal{N}(0,\Omega_n(q)/(1-2\Omega_n(q)))$. Recall that $a_n \sim \mathcal{N}(0,1)$, $b_n \sim \mathcal{N}(0,1)$. Moreover, the variables $(a_n, \widehat{a}_n)$, $(b_n, \widehat{b}_n)$ are perfectly correlated and therefore 
\begin{align}
    \mathbb{E}[a_n \widehat{a}_n] = \left(\frac{\Omega_n(q)}{(1-2\Omega_n(q))}\right)^{1/2}, &&  \mathbb{E}[b_n \widehat{b}_n] = \left(\frac{\Omega_n(q)}{(1-2\Omega_n(q))}\right)^{1/2}.\label{eq:tilhat-corr}
\end{align}

Similar to Step 2c in the proof of Theorem~\ref{prop:sc_HJ},  we now define the tilted Gaussian variable
 \begin{align}
\widetilde{\Psi}\equiv \widetilde{\Psi}(z,q) &:= \left(-\sum_{n=1}^{\infty}  \widehat{b}_n  \frac{\gamma q^n (2\Omega_n(q)-1)}{(1-q^{2n})\sqrt{2n \Omega(q)}} \sin(2\pi(z-\frac{\tau}{2})n)\right)  \\
&\qquad +\sum_{n=1}^{\infty}  \widehat{a}_n\left(\frac{\sqrt{2}\alpha_0 \Omega_n(q)}{\gamma\sqrt{n}} + \frac{\gamma q^n (2\Omega_n(q)-1)}{(1-q^{2n})\sqrt{2n\Omega(q)}} \cos(2\pi(z-\frac{\tau}{2})n)\right) \\
 &\quad - \frac{\alpha_0}{2\gamma \widetilde{\Xi}} \int_0^1 e^{\pi P_0 x} (2\sin(\pi x))^{-\frac{\alpha_0}{2}} Y(x) e^{\frac{\gamma}{2} h_{\widetilde{\Psi}}(x)} dx,\label{eq:tilde_Psi_def}
 \end{align}
 where 
 \begin{align}
     h_{\widetilde{\Psi}}(x) := \mathbb{E}[\widetilde{\Psi} \cdot Y(x)], &&\widetilde{\Xi} := \int_0^1 e^{\pi P_0 x} (2\sin(\pi x))^{-\frac{\alpha_0}{2}} e^{\frac{\gamma}{2} h_{\widetilde{\Psi}}(x)} dx.\label{def:htilPsiXi}
\end{align}
Applying Cameron-Martin theorem \eqref{eq:CM_heavy}, the following expression simplifies
    \begin{align}
   \eqref{eq:int-g2-step1}& =  \exp\left(\frac{1}{2}\sum_{n=1}^{\infty} \Omega_n(q) \left(\frac{\alpha_0^2}{\gamma^2 {n}}+ \frac{\gamma^2 q^{2n}}{(1-q^{2n})^2{n}}   + \frac{2 \alpha_0 q^n}{(1-q^{2n}) {n}} \cos(2\pi(z-\frac{\tau}{2})n)\right)\right)\\
 &\quad \times e^{\frac{\mathbb{E}[\widetilde{\Psi}^2]}{2}} (\widetilde{\Xi})^{\frac{-\alpha_0}{\gamma^2}+ \frac{1}{2}} \exp\left({  \frac{\alpha_0}{2\gamma \widetilde{\Xi}} \int_0^1 e^{\pi P_0 x} (2\sin(\pi x))^{-\frac{\alpha_0}{2}} h_{\widetilde{\Psi}}(x) e^{\frac{\gamma}{2} h_{\widetilde{\Psi}}(x)} dx}\right)\\
 & \quad \times\mathbb{E}\Big[\left(\int_{0}^{1} e^{\frac{\gamma}{2}Y(x)} d\mu_{\widetilde{\Psi}}\right)^{-\frac{\alpha_0}{\gamma^2}+ \frac{1}{2}} e^{ \frac{\alpha_0}{2\gamma \widetilde{\Xi}}\int_0^1  Y(x) d\mu_{\widetilde{\Psi}}}\Big],\label{eq:int-g2-step2}
\end{align}
where $d\mu_{\widetilde{\Psi}}(x) := {\widetilde{\Xi}}^{-1} e^{\pi P_0 x} (2\sin(\pi x))^{-\frac{\alpha_0}{2}} e^{\frac{\gamma}{2} h_{\widetilde{\Psi}}(x)} dx$. Note that in the expressions above, $\widetilde{\Psi}$ has $z$-independent terms scale as $1/\gamma$, and $z$-dependent terms depend linearly on $\gamma$. Consequently, 
\begin{align}
    \partial_z \lim_{\gamma\to 0} \gamma^2 \log \mathbb{E}[\widetilde{\Psi}^2]=0, &&   \partial_z \log \mathbb{E}[\widetilde{\Psi}^2]\underset{\gamma\to 0}{\sim}\mathcal{O}(1), \label{condition-1:g2-asymp}
\end{align}
and 
\begin{align}
    \partial_z \lim_{\gamma\to 0} \left(\gamma \,h_{\widetilde{\Psi}} \right)=0, &&   \partial_z h_{\widetilde{\Psi}}\underset{\gamma\to 0}{\sim}\mathcal{O}(\gamma). \label{condition-2:g2-asymp}
\end{align}
Let us now analyse the asymptotic behaviour of the conformal block. For reasons that will be made clear soon, we retain the $q$-dependent terms of the order $\mathcal{O}(\gamma^{-2})$, and the $z$-dependent terms of the order $\mathcal{O}(1)$.

For $\chi = \gamma/2$, $l_{\gamma/2} = \frac{\gamma^2}{8} - \frac{\alpha_0}{4} $, and the scaling $\alpha = \alpha_0/\gamma$, $P= P_0/\gamma$, we obtain the following leading order expressions
\begin{align}
 C(q) &\mathop{=}^{\eqref{def:ell-eta}} q^{-\frac{(\alpha_0-2) \alpha_0}{12 \gamma^2}} (2\pi)^{-\frac{\alpha_0 (\alpha_0+4)}{6 \gamma^2}} \eta(q)^{-\frac{\alpha_0 (\alpha_0+4)}{2 \gamma^2}} e^{-\frac{1}{2} \ii \pi \alpha_0^2/\gamma^2 }, \label{exp:Cq-g2}\\
C_{1}(q) &= \left( q^{1/6} \eta(q) \right)^{\frac{\alpha_0^2}{2 \gamma^2}} e^{ - \frac{\alpha_0^2}{8\gamma^2}\mathbb{E}[F(0;q)^2]}.\label{exp:C1q-g2}
\end{align}
% \eqref{eq:hat_psi_full-g2}

Substituting \eqref{eq:int-g2-step2}, \eqref{exp:Cq-g2}, \eqref{exp:C1q-g2} in \eqref{eq:psi_hat_relation-g2}, and using Proposition~\ref{prop:semiclassical_limit}, we find the following behaviour near $\gamma\to 0$:
 \begin{align}
 \psi^{\alpha}_{\gamma/2,P}(z,q) &\sim (2\pi)^{-\frac{\alpha_0 (\alpha_0+4)}{6 \gamma^2}} q^{\left(\frac{P_0^2}{2\gamma^2} + \frac{\alpha_0^2}{24\gamma^2}\right)}\eta(q)^{-\frac{2 \alpha_0 }{\gamma^2}} e^{ - \frac{\alpha_0^2}{8\gamma^2}\mathbb{E}[F(0;q)^2]} e^{-\frac{\ii \pi \alpha_0^2}{2\gamma^2} } e^{P_0 z \pi/2} \theta_1(z)^{\alpha_0/4}   e^{\ii\pi \alpha_0z/4}  \\
 &\quad \exp\left(\frac{1}{2}\sum_{n=1}^{\infty} \Omega_n(q) \left(\frac{\alpha_0^2}{\gamma^2 {n}}+ \frac{\gamma^2 q^{2n}}{(1-q^{2n})^2{n}}   + \frac{2 \alpha_0 q^n}{(1-q^{2n}) {n}} \cos(2\pi(z-\frac{\tau}{2})n)\right)\right)\\
 &\quad \times e^{\frac{\mathbb{E}[\widetilde{\Psi}^2]}{2}} (\widetilde{\Xi})^{\frac{-\alpha_0}{\gamma^2}+ \frac{1}{2}} \exp\left({  \frac{\alpha_0}{2\gamma \widetilde{\Xi}} \int_0^1 e^{\pi P_0 x} (2\sin(\pi x))^{-\frac{\alpha_0}{2}} h_{\widetilde{\Psi}}(x) e^{\frac{\gamma}{2} h_{\widetilde{\Psi}}(x)} dx}\right)
 \end{align}
 \begin{align}
 &\quad \times \mathbb{E}\left[\exp\left(-\frac{\alpha_0}{8} \int_0^1 e^{\pi P_0 x} (2\sin(\pi x))^{-\frac{\alpha_0}{2}} :Y(x)^2: d\mu_{\widetilde{\Psi}}(x)\right)\right. \nonumber\\
 &\quad\qquad \left. \times \exp\left(\frac{\alpha_0}{8} \left(\int_0^1 e^{\pi P_0 x} (2\sin(\pi x))^{-\frac{\alpha_0}{2}} :Y(x): d\mu_{\widetilde{\Psi}}(x)\right)^2\right)\right]. \label{eq:deformed_expansion_recall}
 \end{align}
 In the last line of the above display we use Proposition~\ref{prop:semiclassical_limit} which states that the measure $\mu_{\widetilde{\Psi}}$ has density bounded below (for $z \in [M, 1-M]$), we have
\begin{align}\label{eq:semiclassical_limit_light}
&\lim_{\gamma \to 0} \mathbb{E}\left[\exp\left(-\frac{(\alpha_0-2)}{8} \int_0^1 :Y(x)^2: d\mu_{\widetilde{\Psi}}(x) + \frac{\alpha_0}{8} \left(\int_0^1 :Y(x): d\mu_{\widetilde{\Psi}}(x)\right)^2\right)\right] \nonumber\\
&= \mathbb{E}\left[\exp\left(-\frac{(\alpha_0-2)}{8} \int_0^1 :Y(x)^2: d\widetilde{\mu}_0(x) + \frac{\alpha_0}{8} \left(\int_0^1 :Y(x): d\widetilde{\mu}_0(x)\right)^2\right)\right],
\end{align}
where $\widetilde{\mu}_0(x)$ is the limiting measure of $d\mu_{\widetilde{\Psi}}(x)$ as $\gamma \to 0$.
 % \rmkH{The measure $d\mu_{\widehat{\Psi}}$ should be value at $\gamma=0$.}
 Furthemore, with the conditions \eqref{condition-1:g2-asymp}, \eqref{condition-2:g2-asymp}, we observe that
 \begin{align}
     \partial_{z}\lim_{\gamma\to 0} \gamma^2 \log  \psi^{\alpha}_{\gamma/2,P}(z,q) =0 && \partial_z \log \psi^{\alpha}_{\gamma/2,P}(z,q) \underset{\gamma\to 0}{\sim} \mathcal{O}(1).
 \end{align}
 % {\color{red} HD: Edited till this point.}

\noindent\textbf{Step 2: Refined Expansion and Identification of $\widetilde{\Gamma}$.}

We now aim to obtain the explicit form of the leading and sub-leading terms in the expression \eqref{eq:deformed_expansion_recall}. Let us begin with the following asymptotic expansion near $\gamma\to 0$: 
\begin{align}
h_{\widetilde{\Psi}}(x) = \frac{h_q(x)}{\gamma}+ \gamma \mathfrak{R}_{z,q}(x)+ \mathcal{O}(\gamma^2). \label{eq:asymp-htilPsi-g2}
\end{align}
Consequently the measure 
%\end{align}
\begin{align}
        &  e^{\frac{\gamma}{2} h_{\widetilde{\Psi}}(x)} dx =  e^{\frac{1}{2} h_{q}(x)} dx + \frac{ \gamma^2}{2} \mathfrak{R}_{z,q}(x) e^{\frac{1}{2} h_{q}(x)} dx + \mathcal{O}(\gamma^3),\label{eq:measure-expansion-g2}
\end{align}
and the function $\widetilde{\Xi}$ defined in \eqref{def:htilPsiXi} behaves as
\begin{align}
    \widetilde{\Xi} &= \int_0^1 e^{\pi P_0 x} (2\sin(\pi x))^{-\frac{\alpha_0}{2}} e^{\frac{1}{2} h_{q}(x)} dx + \frac{\gamma^2}{2} \int_0^1 e^{\pi P_0 x} (2\sin(\pi x))^{-\frac{\alpha_0}{2}} \mathfrak{R}_{z,q}(x)e^{\frac{1}{2} h_{q}(x)} dx + \mathcal{O}(\gamma^3)\\
    &=: \widetilde{\Xi}_{0} + \gamma^2 \widetilde{\Xi}_{1} + \mathcal{O}(\gamma^3). \label{asymp:Xiexp-g2}
\end{align}
Note that $\widetilde{\Xi}_0$ is $z$-independent while $\widetilde{\Xi}_1$ is $z$-dependent.
Substituting \eqref{eq:measure-expansion-g2} and \eqref{asymp:Xiexp-g2} in \eqref{eq:tilde_Psi_def}, we rewrite $\widetilde{\Psi}$ as 
\begin{align}
    \widetilde{\Psi} = \frac{1}{\gamma}\Phi(q) + \gamma \mathcal{M}_{\gamma}(z,q) + \mathcal{O}(\gamma^2),\label{asymp:tilPsi}
\end{align}
where
\begin{align}
    \Phi(q) &:= 
\sum_{n=1}^{\infty}  \widehat{a}_n\left(\frac{\sqrt{2}\alpha_0 \Omega_n(q)}{\sqrt{n}}\right)- \frac{\alpha_0}{2 \widetilde{\Xi}_0} \int_0^1 e^{\pi P_0 x} (2\sin(\pi x))^{-\frac{\alpha_0}{2}} Y(x) e^{\frac{1}{2} h_{q}(x)} dx \label{def:Phi-Gfield}\\
M_{\gamma}(z,q) &= \sum_{n=1}^{\infty}  \frac{q^n (2\Omega_n(q)-1)}{(1-q^{2n})\sqrt{2n \Omega(q)}}\left(-\widehat{b}_n  \sin(2\pi(z-\frac{\tau}{2})n)+\widehat{a}_n \cos(2\pi(z-\frac{\tau}{2})n)\right) \\
 &\quad - \frac{\alpha_0 }{4\widetilde{\Xi}_0} \int_0^1 e^{\pi P_0 x} (2\sin(\pi x))^{-\frac{\alpha_0}{2}} \mathfrak{R}_{z,q}(x)
 Y(x) e^{\frac{1}{2} h_{q}(x)} dx\\
 & \quad +\frac{\alpha_0 }{4 \widetilde{\Xi}_0^2}\left(\int_0^1 e^{\pi P_0 x} (2\sin(\pi x))^{-\frac{\alpha_0}{2}} Y(x) e^{\frac{1}{2} h_{q}(x)} dx \right) \left( \int_0^1 e^{\pi P_0 x} (2\sin(\pi x))^{-\frac{\alpha_0}{2}} \mathfrak{R}_{z,q}(x)e^{\frac{1}{2} h_{q}(x)} dx  \right). \label{def:Mg-Gfield}
\end{align}
The expressions with the definition $h_{\widetilde{\Psi}}(x) := \mathbb{E}[\widetilde{\Psi} \cdot Y(x)]$ implies 
\begin{align}
    \mathfrak{R}_{z,q}(x) = \mathbb{E}\left[ M_{\gamma}(z,q) Y(x)\right].
\end{align}
Substituting \eqref{def:Mg-Gfield} in the above expression results in the following functional equation 
\begin{align}
     &\mathfrak{R}_{z,q}(x) \\
     &=  \mathfrak{R}_{z,q}^{(0)}(x) - \frac{\alpha_0 }{4\widetilde{\Xi}_0} \int_0^1 e^{\pi P_0 y} (2\sin(\pi y))^{-\frac{\alpha_0}{2}} K(x,y)\mathfrak{R}_{z,q}(y)
 e^{\frac{1}{2} h_{q}(y)} dy \\
 &\quad + \frac{\alpha_0 }{4\widetilde{\Xi}_0^2} \left(\int_0^1 e^{\pi P_0 y} (2\sin(\pi y))^{-\frac{\alpha_0}{2}} K(x,y)
 e^{\frac{1}{2} h_{q}(y)} dy\right) \left( \int_0^1 e^{\pi P_0 x} (2\sin(\pi x))^{-\frac{\alpha_0}{2}} \mathfrak{R}_{z,q}(x)e^{\frac{1}{2} h_{q}(x)} dx  \right) ,\label{funceq:Rzqx}
\end{align}
where 
\begin{align}
    \mathfrak{R}_{z,q}^{(0)}(x)& := \sum_{n=1}^{\infty} \frac{q^{n} \left(1- 2\Omega_n(q)\right)^{1/2}}{n(1-q^{2n}) } \cos \left(2\pi (z+x - \tau/2) \right),\label{funceq:Rzqx0}
\end{align}
and the kernel $K(x,y) := \mathbb{E}[Y(x) Y(y)]$. The expression \eqref{funceq:Rzqx0} is computed by using the equations \eqref{eq:tilhat-corr}. 
%\rmkH{Do we need to address the subtlety of $Y_N$ vs $Y$?}

Similarly, with \eqref{asymp:tilPsi}, the expectation term in \eqref{eq:deformed_expansion_recall} can be written as
\begin{align}
    \mathbb{E}[\widetilde{\Psi}^2]& = \frac{1}{\gamma^2} \mathbb{E} \left[\Phi(q)^2 \right] + 2\mathbb{E} \left[ \Phi(q) M_{\gamma}(z,q)\right] + \mathcal{O}(\gamma^2),\label{asymp:covterm-g2}
\end{align}
where the term $\mathbb{E} \left[\Phi(q)^2 \right]$ is $z$-independent. The term $(\widetilde{\Xi})^{\frac{-\alpha_0}{\gamma^2}+ \frac{1}{2}}$ behaves as 
\begin{align}
    (\widetilde{\Xi})^{\frac{-\alpha_0}{\gamma^2}+ \frac{1}{2}} &= \exp\left(\left( {\frac{-\alpha_0}{\gamma^2}+ \frac{1}{2}}\right) \log (\widetilde{\Xi})\right)\\
    & \mathop{=}^{\eqref{asymp:Xiexp-g2}} \exp\left( \left({\frac{-\alpha_0}{\gamma^2}+ \frac{1}{2}}\right) \left(\log \widetilde{\Xi}_0 + \gamma^2 \widetilde{\Xi}_1/\widetilde{\Xi}_0 + \mathcal{O}(\gamma^3)\right)\right)\\
    &= \exp\left( {-\frac{\alpha_0 }{\gamma^2} \log\widetilde{\Xi}_0+ \frac{1}{2} \log\widetilde{\Xi}_0} - \alpha_0\frac{\widetilde{\Xi}_1}{\widetilde{\Xi}_0} + \mathcal{O}({\gamma^2})\right)\\
    &=\widetilde{\Xi}_0^{-\alpha_0/\gamma^2} \widetilde{\Xi}_0^{1/2} \exp\left( - \alpha_0\frac{\widetilde{\Xi}_1}{\widetilde{\Xi}_0} + \mathcal{O}({\gamma^2})\right),\label{asymp:powterm-g2}
\end{align}
and the exponential term 
\begin{align}
    &\exp\left({  \frac{\alpha_0}{2\gamma \widetilde{\Xi}} \int_0^1 e^{\pi P_0 x} (2\sin(\pi x))^{-\frac{\alpha_0}{2}} h_{\widetilde{\Psi}}(x) e^{\frac{\gamma}{2} h_{\widetilde{\Psi}}(x)} dx}\right)\\
&\mathop{=}^{\eqref{eq:asymp-htilPsi-g2}\eqref{asymp:Xiexp-g2}} \exp\Bigg(  \frac{\alpha_0}{2\gamma \widetilde{\Xi}_0}\left(1- \gamma^2 \frac{\widetilde{\Xi}_1}{\widetilde{\Xi}_0} +\mathcal{O}(\gamma^4)\right) \\
&\quad \times \left(\int_0^1 e^{\pi P_0 x} (2\sin(\pi x))^{-\frac{\alpha_0}{2}} \left(\frac{h_q(x)}{\gamma} e^{\frac{h_q(x)}{2}} + \frac{3\gamma}{2} \mathfrak{R}_{z,q}(x)e^{\frac{h_q(x)}{2}} +\mathcal{O}(\gamma^3)\right)dx\right)\Bigg)\\
& = \exp\left(  \frac{\alpha_0}{2\gamma^2 \widetilde{\Xi}_0} \left( \int_0^1 e^{\pi P_0 x} (2\sin(\pi x))^{-\frac{\alpha_0}{2}} h_q(x) e^{\frac{h_q(x)}{2}} dx \right)\right) \\
& \quad \times \exp\left( \frac{\alpha_0}{2\widetilde{\Xi}_0} \int_0^1 e^{\pi P_0 x} (2\sin(\pi x))^{-\frac{\alpha_0}{2}} \left(\frac{3}{2}\mathfrak{R}_{z,q}(x)-\frac{\widetilde{\Xi}_1}{\widetilde{\Xi}_0} h_q(x)+\mathcal{O}(\gamma^3)\right)e^{\frac{h_q(x)}{2}} dx \right).\label{asymp:expterm-g2}
\end{align}
Substituting \eqref{asymp:covterm-g2}, \eqref{asymp:powterm-g2}, \eqref{asymp:expterm-g2} in \eqref{eq:deformed_expansion_recall} and recalling the identity \eqref{comp-exp-vanishes}, we arrive at the expressions
\begin{align}
    \lim_{\gamma\to 0}\gamma^2 \log 
    \psi_{\gamma/2, P_0}^{\alpha_0} &= \phi(q;\alpha_0, P_0)\\
    &= - \frac{\alpha_0 (\alpha_0+4)}{6} + \left( \frac{P_0^2}{2} + \frac{\alpha_0^2}{24}\right)\log q - 2\alpha_0 \log \eta(q) - \frac{\alpha_0^2}{8} \mathbb{E}[F(0;q)^2] - \frac{\ii \pi \alpha_0^2}{2} \\
    &\quad + \sum_{n=1}^{\infty} \frac{\Omega_n(q) \alpha_0^2}{n} + \frac{1}{2}\mathbb{E}[\Phi^2] - \alpha_0 \log \widetilde{\Xi}_0 \\
    &+ \frac{\alpha_0}{2\widetilde{\Xi}_0}\left( \int_0^1 e^{\pi P_0 x} (2\sin(\pi x))^{-\frac{\alpha_0}{2}} h_q(x) e^{\frac{h_q(x)}{2}} dx \right).
\end{align}
Substituting \eqref{def:OmegaAleph} and \eqref{exp:Ft0X} in the above expression, we get \eqref{def:phiq}.
The subleading term
\begin{align}
    \widetilde{\Gamma}(z,q) &= e^{P_0 z\pi/2}\theta_1(z)^{\alpha_0/4} e^{\ii \pi \alpha_0 z/4} \exp\left(\frac{1}{2}\sum_{n=1}^{\infty}\frac{2 \alpha_0 q^n  \Omega_n(q)}{(1-q^{2n}) {n}} \cos(2\pi(z-\frac{\tau}{2})n) \right)   e^{\mathbb{E}\left[\Phi(q) M_{\gamma}(z,q) \right]}\widetilde{\Xi}_0^{1/2}  \\
    & \quad \times e^{\left( - \alpha_0\frac{\widetilde{\Xi}_1}{\widetilde{\Xi}_0} \right)} \exp\left( \frac{\alpha_0}{2\widetilde{\Xi}_0} \int_0^1 e^{\pi P_0 x} (2\sin(\pi x))^{-\frac{\alpha_0}{2}} \left(\frac{3}{2}\mathfrak{R}_{z,q}(x)-\frac{\widetilde{\Xi}_1}{\widetilde{\Xi}_0} h_q(x)\right)e^{\frac{h_q(x)}{2}} dx \right).\\
    \label{eq:Gamma_derivation}
\end{align}

\medskip
\noindent\textbf{Step 3: Derivation of the Lam\'e Equation.}
Recall from Theorem~\ref{thm:bpz} that the deformed conformal block $\psi^{\alpha}_{\gamma/2, P}(z,q)$ satisfies the BPZ equation:
\begin{align}\label{eq:BPZ_light}
\left(\partial_z^2 - l_{\gamma/2}(l_{\gamma/2}+1) \wp(z) + 2i\pi \frac{\gamma^2}{4} \partial_{\tau}\right) \psi^{\alpha}_{\gamma/2, P}(z,q) = 0.
\end{align}
Under the scaling $\alpha = \alpha_0/\gamma$, $P = P_0/\gamma$, we have:
\begin{align}\label{eq:l_scaling}
l_{\gamma/2} = \frac{\gamma^2}{8} - \frac{\alpha_0}{4}, \qquad l_{\gamma/2}(l_{\gamma/2}+1) = \frac{\alpha_0}{4}\left(\frac{\alpha_0}{4} - 1\right) + O(\gamma^2).
\end{align}
Substituting the limit \eqref{eq:sc_Lame_limit} into the BPZ equation \eqref{eq:BPZ_light}:
\begin{align}\label{eq:BPZ_substitution}
&\left(\partial_z^2 - l_{\gamma/2}(l_{\gamma/2}+1) \wp(z) + 2i\pi \frac{\gamma^2}{4} \partial_{\tau}\right) \left(e^{\phi(q)/\gamma^2} \widetilde{\Gamma}(z;\alpha_0, P_0, q)\right) = 0.
\end{align}
The action of $\partial_z^2$ on $e^{\phi(q;\alpha_0, P_0)/\gamma^2} \widetilde{\Gamma}$ gives
%\begin{align}\label{eq:z_derivative}
$\partial_z^2 \left(e^{\phi(q;\alpha_0, P_0)/\gamma^2} \widetilde{\Gamma}\right) = e^{\phi(q;\alpha_0, P_0)/\gamma^2} \partial_z^2 \widetilde{\Gamma},$
%\end{align}
since $\phi(q;\alpha_0, P_0)$ is independent of $z$.
The action of $\partial_{\tau}$ gives
\begin{align}\label{eq:tau_derivative}
\partial_{\tau} \left(e^{\phi(q;\alpha_0, P_0)/\gamma^2} \widetilde{\Gamma}\right) = e^{\phi(q;\alpha_0, P_0)/\gamma^2} \left(\frac{1}{\gamma^2} \partial_{\tau} \phi(q;\alpha_0, P_0) \cdot \widetilde{\Gamma} + \partial_{\tau} \widetilde{\Gamma}\right).
\end{align}
Substituting into \eqref{eq:BPZ_substitution} and dividing by $e^{\phi(q;\alpha_0, P_0)/\gamma^2}$ yields
\begin{align}\label{eq:BPZ_divided}
\partial_z^2 \widetilde{\Gamma} - l_{\gamma/2}(l_{\gamma/2}+1) \wp(z) \widetilde{\Gamma} + \frac{i\pi}{2} \partial_{\tau} \phi(q;\alpha_0, P_0)\widetilde{\Gamma} + O(\gamma^2) = 0.
\end{align}
Taking $\gamma \to 0$ and using \eqref{eq:l_scaling}:
\begin{align}\label{eq:Lame_derived}
\left(\partial_z^2 - \frac{\alpha_0}{4}\left(\frac{\alpha_0}{4} - 1\right) \wp(z) + \frac{i\pi}{2} \partial_{\tau} \phi(q;\alpha_0, P_0)\right) \widetilde{\Gamma}(z;\alpha_0, P_0, q) = 0.
\end{align}
This is precisely the Lam\'e equation \eqref{thm43:Lame} with accessory parameter $\frac{\pi i}{2} \partial_{\tau} \phi(q;\alpha_0, P_0)$.

\medskip
\noindent\textbf{Step 4: Uniqueness of the Semi-Classical Limit.}

\medskip
By Proposition~\ref{thm:sc_Lame} (the boundedness result), we know that $\gamma^2 \log \psi^{\alpha_0/\gamma}_{\gamma/2, P_0/\gamma}(z,q)$ is tight as $\gamma \to 0$. Therefore, any subsequence has a further converging subsequence.
Let $\phi(q;\alpha_0, P_0)$ be any sub-sequential limit:
\begin{align}\label{eq:subsequential_limit}
\phi(q; \alpha_0, P_0) = \lim_{\gamma_n \to 0} \gamma_n^2 \log \psi^{\alpha_0/\gamma_n}_{\gamma_n/2, P_0/\gamma_n}(z,q)
\end{align}
for some sequence $\gamma_n \to 0$.
From Step 3, any such subsequential limit must satisfy the property that $\widetilde{\Gamma}(z;\alpha_0, P_0, q)$ defined by \eqref{eq:sc_Lame_limit} solves the Lam\'e equation \eqref{eq:Lame_derived}. Since $\widetilde{\Gamma}$ is given explicitly by \eqref{eq:Gamma_derivation}, which is independent of the choice of subsequence, and since the solution of the Lam\'e equation uniquely determines the accessory parameter (see \cite{eremenko2022moduli} \cite[Chapter 4]{ince1940vii}), we conclude that $\partial_{\tau} \phi(q; \alpha_0, P_0)$ is uniquely determined.

Proposition~\ref{prop:Gandphi} showed that the semi-classical limit undeformed conformal block is same as $\widetilde{\Gamma}(z;\alpha_0, P_0, q)$ upto some explicit additive factor. {Together with the asymptotic results of Section~\ref{sec:asymp}} which shows unique limit of $\lim_{q\to 0}\phi(q; \alpha_0, P_0)$, the uniqueness of $ \partial_{\tau} \phi(q;\alpha_0, P_0)$ implies that all subsequential limits of $\gamma^2 \log \psi^{\alpha_0/\gamma}_{\gamma/2, P_0/\gamma}(z,q)$ are the same. Therefore, the limit \eqref{eq:sc_Lame_limit} exists and $\phi(q;\alpha_0, P_0)$ is uniquely determined by $\widetilde{\Gamma}(z;\alpha_0, P_0, q)$.

This completes the proof of part (1) of Theorem~\ref{prop:sc_Lame}. Now we proceed to part (2) of Theorem~\ref{prop:sc_Lame}.  
% \end{proof}
% \rmkH{Please do not change the text in red. }

\subsubsection*{Proof of Part (2)}

There are three main steps in this proof: relating the function $\phi(q;\alpha_0, P_0)$ defined in \eqref{eq:subsequential_limit} with the action of the non-autonomous elliptic Calogero-Moser (NAECM) model defined in \eqref{eq:NAECM_action}, regularizing the action, calculating the $\tau$-derivative of $\phi(q;\alpha_0, P_0)$ and showing its relation to the Hamiltonian of the NAECM model.

\medskip
\noindent
{\bf Step 1: Relation between $\phi(q;\alpha_0, P_0)$ and the NAECM action.}
\medskip

\noindent
Let us begin by studying the expression \eqref{eq:we_are_all_the_same_recall} for $\chi = 2/\gamma$. The exponent $\frac{\chi}{2}(\chi - \alpha) = \frac{1}{\gamma}(\frac{2}{\gamma} - \frac{\alpha_0}{\gamma}) = \frac{2 - \alpha_0}{\gamma^2}$, and so the expression \eqref{eq:we_are_all_the_same_recall} becomes
\begin{align}\label{def:g_recall}
&\lim_{\gamma \to 0} \lim_{z \to 0} \gamma^2 \log \left(\theta_1(z)^{\frac{2-\alpha_0}{\gamma^2}} \psi_{\frac{2}{\gamma}, P}^{\alpha}(z,\tau)\right) \nonumber\\
&= \lim_{\gamma \to 0} \gamma^2 \log \mathcal{G}_{\gamma, P}^{\alpha_0/\gamma - 2/\gamma}(q) + \lim_{\gamma \to 0} \gamma^2 \log Z^{(\alpha_0-2)/\gamma}_{\gamma, P} + \lim_{\gamma \to 0} \gamma^2 \log \mathcal{W}(q).
\end{align}
Let us study the LHS and RHS of the above expression separately.

We begin by expanding the terms in the RHS of the expression above. 
Firstly, using the expression \eqref{rel:GandPhi}, the undeformed conformal block with parameter shift $\alpha_0 \to \alpha_0 - 2$ gives:
\begin{align}
 \lim_{\gamma \to 0} \gamma^2 \log \mathcal{G}_{\gamma, P}^{(\alpha_0-2)/\gamma}(q) &=  \phi(q;\alpha_0-2, P_0)  - \frac{P_0^2}{2}   \log q - \left(\frac{(\alpha_0-2)^2}{2} - 2(\alpha_0-2) \right) \log \eta(q) \\
&+ \frac{(\alpha_0-2) (\alpha_0+2)}{6} \log(2\pi) - \mathcal{B}(\alpha_0-2, P_0).\label{eq:G_alpha0-2}
\end{align}
As in the proof of the previous proposition, we ignore the $q$-independent term $\mathcal{B}(\alpha_0, P_0)$.
Secondly, shifting the value of $\alpha_0 \to \alpha_0 -2$ in \eqref{GPhider:Z} we get
\begin{align}\label{eq:Z_asymp}
\lim_{\gamma \to 0} \gamma^2 \log Z^{(\alpha_0-2)/\gamma}_{\gamma, P} = \frac{(\alpha_0-2)^2}{24} \log q+(\alpha_0-2)^2 \log\eta(q) + \mathcal{B}(\alpha_0-2, P_0).
\end{align}
Finally from \eqref{eq:prefactor_final}, we get
\begin{align}
    &\lim_{\gamma \to 0} \gamma^2 \log \mathcal{W}^{2/\gamma}(q) =(\frac{P_0^2}{2} - \frac{\alpha_0^2}{24}) \log q- {\frac{\alpha_0(\alpha_0+4)}{6}} \log \theta_1'(0)\\
& \mathop{=}^{\eqref{def:ell-eta}} (\frac{P_0^2}{2} - \frac{\alpha_0^2}{24}) \log q- {\frac{\alpha_0(\alpha_0+4)}{2}} \log \eta(q)- {\frac{\alpha_0(\alpha_0+4)}{6}} \log(2\pi) \label{asymp:W}
\end{align}
Summing the expressions \eqref{eq:Z_asymp}, \eqref{eq:G_alpha0-2}, and \eqref{asymp:W} we get
\begin{align}
 &\lim_{\gamma \to 0} \gamma^2 \log Z^{(\alpha_0-2)/\gamma}_{\gamma, P} +    \lim_{\gamma \to 0} \gamma^2 \log \mathcal{G}_{\gamma, P}^{(\alpha_0-2)/\gamma}(q)+ \lim_{\gamma \to 0} \gamma^2 \log \mathcal{W}^{2/\gamma}(q) \\
 &=  \phi(q; \alpha_0-2, P_0) -6(\alpha_0 -1)\log \eta(q) -\frac{2(\alpha_0+1)}{3} \log(2\pi).\label{rhs:full}
\end{align}

The LHS of the expression \eqref{def:g_recall}, with the equation \eqref{def:tilde_phi_zq} reads,
\begin{align}
\lim_{z \rightarrow 0} \lim_{\gamma \rightarrow 0} \gamma^2 \log \left(\theta_{1}(z)^{\frac{\chi}{2}\left(\chi-\alpha\right)} \psi^\alpha_{2/\gamma}(z, \tau)\right) \mathop{=} \lim_{z \rightarrow 0} \widetilde{\phi}(z, \tau) + (2-\alpha_0) \lim_{z \to 0} \log \theta_1(z).\label{reg:semi1}
\end{align}

The equality of the expressions \eqref{rhs:full} and \eqref{reg:semi1} implies, with some re-arranging of terms,
\begin{align}
\phi(q; \alpha_0-2, P_0)- \frac{2(\alpha_0+1)}{3} \log(2\pi) = \lim_{z \rightarrow 0} \widetilde{\phi}(z, \tau) + (2-\alpha_0) \lim_{z \to 0} \log \theta_1(z) + 6(\alpha_0 -1)\log \eta(q) .\\
\label{rel:phi-phitil}
\end{align}

Let us recall that by the Definition~\ref{remark:HJ_phitil}, $\widetilde{\phi}(z, \tau)$ is associated to the action of the non-autonomous elliptic Calogero-Moser (NAECM) model with the identification $z = 2u(\tau)$ and $m^2 = (2-\alpha_0)^2$. We now focus on the $z\to 0$ limit of the expression above.  By Lemma~\ref{Lemma:HJ-un}, the action is
\begin{align}\label{eq:action_recall}
\widetilde{\phi}(2u(\tau), \tau) = \int_{\ii \infty}^{\tau} \left(v(\tau')^2 + m^2 \wp(2u(\tau')|\tau') \right) \frac{d\tau'}{2\pi \ii}.
\end{align}
We define the integrand above, which is the Lagrangian of the NAECM model, as 
\begin{align}
     \mathcal{L}(\tau) := \frac{1}{2\pi \ii}\left(v(\tau)^2 + m^2 \wp(2u(\tau)|\tau) \right). \label{def:mathcalL}
\end{align}
The identification $z = u(\tau)$ implies that the limit $z\to 0$ is equivalent to the limit $u(\tau)\to 0$, and therefore, the term $ \lim_{z \rightarrow 0} \widetilde{\phi}(z, \tau) + (2-\alpha_0) \lim_{z \to 0} \log \theta_1(z) $ is equivalent to studying the expression $ \widetilde{\phi}(2u(\tau), \tau) + (2-\alpha_0)  \log \theta_1(2u(\tau))$ near the zeros of $u(\tau)$, which we denote by $u(\tau_{\star})=0$. Such an analysis was carried out in \cite[Section 7]{BGG2021}, which required an additional regularization of the action \eqref{eq:action_recall} as we will show, it is ill defined at $u(\tau)=0$. In what follows, we will repeat the structure of the computations in \cite[Section 7]{BGG2021} to highlight that the term $(2-\alpha_0)  \log \theta_1(2u(\tau))$ regularizes our action in a direct manner, and resultantly, we can connect the function $\phi(q, \alpha_0-2, P_0)$ in \eqref{rel:phi-phitil} to the Hamiltonian of the NAECM model.

\medskip 
\noindent
{\bf Step 2: Regularizing the action}
\medskip

\noindent
To this end, we first recall that the Weierstrass $\wp$-function behaves as $\wp(z) \sim 1/z^2 + \mathcal{O}(z^2)$ near $z=0$, and as $\wp(z) \sim \frac{\pi^2}{\sin^2 \pi z} - \frac{\pi^2}{3} + 16 \pi^2 e^{\ii \pi \tau} \sin^2 (\pi z) $ near $z\to \ii\infty$ \cite[eq (A.14)]{BGG2021}. As we will show below, this behaviour of $\wp(z)$ implies that the function $\mathcal{L}(\tau)$ is constant near $\ii \infty$, and has a pole neat $\tau= \tau_{\star}$. Indeed, recall that the expansion of $u(\tau)$ near $\tau = \tau_{\star}$ as specified in Proposition \ref{prop:uasympinf}, in our notation\footnote{${\sf a} = P_0$ parametrizes the A-cycle monodromy.} is
\begin{align}\label{asymp:uinf}
       u(\tau) = P_0 \tau +{\sf b}+\frac{m^2}{8 \pi \ii P_0^2}e^{4 \pi \ii (P_0 \tau+{\sf b})}+\mathcal{O}(q e^{4 \pi \ii (P_0 \tau+{\sf b})}),
\end{align}
where ${\sf b}$ parametrizes the B-cycle monodromy of the solution to the linear system \eqref{eq:linear_systemCM}.
Equivalently, ${\sf b}$ may be regarded as an integration constant which depends on the choice of initial conditions for the differential equation \eqref{eq:CM_ELLDER}, that is solved by $u(\tau)$.

At this point, the reader may treat this as a constant, and it does not play any role in the final result of this calculation. 
In the same limit, we then get 
\begin{align}\label{asymp:vinf}
    v(\tau)^2 = \left(2\pi \ii \frac{du}{d\tau} \right)^2 = -4 \pi^2  P_0^2 - 4\pi^2 m^2 e^{4 \ii P_0 \pi \tau + 4 \ii {\sf b} \pi}+\mathcal{O}( e^{8 \pi \ii (P_0 \tau+{\sf b})}).
\end{align}
Substituting \eqref{asymp:uinf} and \eqref{asymp:vinf} into \eqref{def:mathcalL}, we obtain that in the limit in the limit $\tau \to \ii \infty$, the Lagrangian $\mathcal{L}$ and becomes a constant of monodromy, let us then define $ \lim_{\tau \to \ii \infty} \mathcal{L} =: {\sf L}(P_0, \alpha_0, {\sf b})$.
Therefore, the action \eqref{eq:action_recall} is well-defined at $\tau \to\ii \infty$. We note that the constant ${\sf L}$ will not play a role in our calculations, as will be made clear shortly.

Let us now study the action \eqref{eq:action_recall} near $\tau = \tau_{\star}$.  We begin with the behaviour of $u(\tau)$, near its zero $u(\tau_{\star}) = 0$ given in \eqref{asymp:utstar}:
\begin{align}\label{eq:u_asymp_recall}
u(\tau) = e^{\mp \ii \pi/4} \sqrt{\frac{m}{2\pi}} (\tau - \tau_{\star})^{1/2} \left(1 \pm \frac{H_{\star} }{4\pi \ii m} (\tau - \tau_{\star})\right) + \mathcal{O}\left((\tau - \tau_{\star})^{5/2}\right),
\end{align}
where $H_{\star} := H(\tau_{\star})$ is the Hamiltonian evaluated at $\tau = \tau_{\star}$. 
Furthermore, using the behaviour of the Weierstrass $\wp$-function near $z=0$ as specified above, we obtain the following expression near $\tau = \tau_{\star}$
\begin{align}
    m^2 \wp(2u(\tau)\vert \tau) =\frac{\ii \pi  m}{2 (\tau-\tau_{\star})}\mp \frac{H_{\star}}{4}+\mathcal{O}(\tau - \tau_{\star}).\label{asymp:Pu}
\end{align}
From \eqref{def:uv}, near $\tau_{\star}$ we have the expression:
\begin{align}\label{eq:v_asymp}
v(\tau)^2 = \left(2\pi \ii \frac{du}{d\tau} \right)^2 = \frac{i \pi  m}{2 (\tau-\tau_{\star})}\pm \frac{3 H_{\star}}{4} +\mathcal{O}(\tau - \tau_{\star}).
\end{align}

% \rmkH{Write down the expansion of $\wp(2u(\tau))$ and $v(\tau)^2$.}
Substituting the equations \eqref{eq:v_asymp} and \eqref{asymp:Pu} in the function $\mathcal{L}(\tau)$ \eqref{def:mathcalL}, we get the following expansion for the Lagrangian near $\tau= \tau_{\star}$:
\begin{align}
    \mathcal{L}(\tau) 
    &=\frac{m}{2 (\tau-\tau_{\star})}\mp\frac{\ii H_{\star}}{4 \pi } + \mathcal{O}(\tau- \tau_{\star}). \label{lag:asymp}
\end{align}

Notice that the above expression is singular near $\tau = \tau_{\star}$, implying that the action \eqref{eq:action_recall} is ill-defined at this point. We will now show that the term $(2-\alpha_0) \log \theta_1(2 u(\tau))$ effectively regularizes the action, {\it i.e}, 
\begin{align}\lim_{u(\tau) \rightarrow 0} \left(\widetilde{\phi}(2u(\tau), \tau) + (2-\alpha_0) \log \theta_1(2 u(\tau))\right) = \text{finite}. \label{FINITE} \end{align}

Let us now study the expansion of $\log \theta_1(2u(\tau))$ near $\tau = \tau_{\star}$. Near $z=0$, the function $\theta_1(z) = \theta'(0) z + \frac{\theta'''(0)}{6} z^3 + \mathcal{O}(z^5)$. In turn, we get
\begin{align}
    \log \theta_1(z) = \log(\theta_1'(0)) + \log(z) + \frac{\theta_1'''(0)}{6 \theta_1'(0)} z^2 + \mathcal{O}(z^4)\\
    \mathop{=}^{\eqref{def:ell-eta},\eqref{def:ell-eta-p}}  \log(2\pi) + 3\log\eta(\tau) + \log(z) -\eta_1(\tau) z^2 + \mathcal{O}(z^4).
\end{align}
Substituting $z=2 u(\tau)$ and the expansion of $u(\tau)$ in \eqref{asymp:Pu} into the equation above, we get
\begin{align}
    (2-\alpha_0) \log \theta_1(2 u(\tau)) &= (2-\alpha_0)\log(2\pi) + 3(2-\alpha_0)\log\eta(\tau_{\star})+\frac{(2-\alpha_0)}{2} \log \left(\frac{(-2\ii m)}{\pi }  (\tau-\tau_{\star})\right)\\
    &+ \frac{2 \ii m (2-\alpha_0) \eta_1(\tau_{\star}) (\tau-\tau_{\star})}{\pi } + \mathcal{O}\left((\tau-\tau_{\star})^2\right).
\end{align}
% The above term can be written as an integral with some appropriate integration constant and behaviour at $\ii \infty$ as
Let us now define the following function
\begin{align}
  {\sf F}(\tau)&:=  (2-\alpha_0) \partial_{\tau}\log\theta_1(2 u(\tau))+ 6(\alpha_0-1) \partial_{\tau} \log\eta(q).
\end{align}
The above expression near $\tau = \tau_{\star}$ reads
\begin{align}
  {\sf F}(\tau) =   \frac{(2-\alpha_0)}{2 (\tau - \tau_{\star})} -\frac{\ii \eta_1(\tau_{\star})(4 (\alpha_0-2) m+3 \alpha_0)}{\pi }+ \mathcal{O}(\tau - \tau_{\star}).\label{reg:asymp}
\end{align}
Setting $m =(\alpha_0 -2)$, we can write
\begin{align}
    {\sf F}(\tau)= -\frac{m}{2 (\tau - \tau_{\star})} - \frac{\ii \eta_1(\tau_{\star}) (m (4 m+3)+6)}{\pi } + \mathcal{O}(\tau - \tau_{\star}).\label{Fwithm}
\end{align}
Now observe that the first two terms in the RHS of the expression \eqref{rel:phi-phitil} can be re-expressed as
\begin{align}
    \widetilde{\phi}(2u(\tau), \tau)  -m \log \theta_1(2 u(\tau))+ 6(m+1)\log \eta(q) \mathop{=}^{\eqref{Fwithm}} \int_{\ii \infty}^{\tau} \left(\mathcal{L}(\tau) + {\sf F}(\tau) \right) d\tau', \label{phi+theta}
\end{align}
and near $\tau = \tau_{\star}$, substituting \eqref{lag:asymp}, \eqref{Fwithm} we get
\begin{align}
        &\mathcal{L}(\tau) + {\sf F}(\tau) = -\frac{\ii (2 \eta_{1}(\tau_{\star}) (m (4 m+3)+6)+H_{\star})}{4 \pi } + \mathcal{O}(\tau - \tau_{\star}).
\end{align}
The expression \eqref{phi+theta} is therefore well-defined near $\tau = \tau_{\star}$ and \eqref{FINITE} holds.

Finally, in the limit $\tau \to \ii \infty$, using the expressions \eqref{def:elltheta1}, \eqref{def:ell-eta}, we notice that $(2-\alpha_0) \partial_{\tau}\log\theta_1(2 u(\tau))+ 6(\alpha_0-1) \partial_{\tau} \log\eta(q)$ is a constant that depends only on monodromy. We can therefore define $(2-\alpha_0) \lim_{\tau \to \ii \infty}\partial_{\tau}\log\theta_1(2 u(\tau))+ 6(\alpha_0-1) \lim_{\tau \to \ii \infty}\partial_{\tau} \log\eta(q) = {\sf F}_0(\alpha_0, P_0, {\sf b})$. This parameter ${\sf F}_0$ will have no role in our final computation.

\medskip 
\noindent
{\bf Step 3: Computing $\tau$ derivative of $\phi(\tau)$ at $\tau = \tau_{\star}$:}
\medskip

\noindent
We can now compute the $\tau_{\star}$ derivative of the RHS of the expression \eqref{rel:phi-phitil} as follows:
\begin{align}
   &\partial_{\tau_{\star}} \lim_{u(\tau) \rightarrow 0} \left(\widetilde{\phi}(2u(\tau), \tau) + (2-\alpha_0) \log \theta_1(2 u(\tau)) + 6(\alpha_0 -1)\log \eta(q)\right) \\
   &= \mathcal{L}(\tau_{\star}) + {\sf F}(\tau_{\star}) + \int_{\ii \infty}^{\tau_{\star}} d\left(4\pi\ii \partial_{\tau} u(\tau) \partial_{\tau_{\star}}u(\tau)\right) +\partial_{\tau_{\star}}{\sf F}(\tau) d\tau  \\
   &\mathop{=}^{\eqref{eq:u_asymp_recall},\eqref{lag:asymp}, \eqref{reg:asymp}} \pm \frac{\ii H_{\star}}{2 \pi }+ \mathcal{O}(\tau-\tau_{\star}). \label{Hstarcal}
\end{align}
To obtain the second line of the above expression we note that $u(\tau)$ depends on $\tau_{\star}$ and use the Leibniz rule. The final line of the above expression uses the expression
$$\partial_{\tau_{\star}} \mathsf{F}(\tau)
= -\partial_{\tau} \!\left( m\, \partial_{\tau_{\star}} \log \theta_1\big(2 u(\tau)\big)
+ 6(m+1)\, \partial_{\tau_{\star}} \log \eta(q) \right),$$
which we then evaluate at $\tau = \tau_{\star}$. Finally, using the expansion of $u(\tau)$ given in \eqref{eq:u_asymp_recall}, we find the following expressions
\begin{align}
    4\pi\ii \partial_{\tau} u(\tau) \partial_{\tau_{\star}}u(\tau) &= -\frac{m}{2 (\tau-\tau_{\star})}+\frac{3 \ii H_{\star}}{4 \pi },\\
    - m \partial_{\tau_{\star}} \log \theta_1(2 u(\tau))&=  \left(\frac{m}{
    2(\tau-\tau_{\star})}+\frac{\ii m \eta_1(\tau_{\star}) (4 m+3)}{ 2\pi }\right),
\end{align}
which are substituted in the second line of \eqref{Hstarcal} to obtain the final expression.
% \rmkH{Add the fact that $u(\tau)$ depends on monodromy data and therefore $\tau_{\star}$. Expand term by term the integrand.}
% \medskip

With \eqref{Hstarcal}, the $\tau$-derivative of the expression \eqref{rel:phi-phitil} at $\tau = \tau_{\star}$ becomes
\begin{align}
  &\partial_{\tau} \left(\phi(q; \alpha_0-2, P_0) - \frac{2(\alpha_0+1)}{3} \log (2\pi)\right)\Big|_{\tau = \tau_{\star}} \\
 & =  \partial_{\tau}\phi(q; \alpha_0-2, P_0) \Big|_{\tau = \tau_{\star}}= \pm \frac{\ii H(\tau_{\star})}{2\pi}.
\end{align}
Therefore, the accessory parameter of the Lam\'e equation \eqref{thm43:Lame} is 
% $\frac{\pi \ii}{2} \partial_{\tau} \phi$ which implies
\begin{align}\label{eq:acc_final}
\frac{\pi \ii}{2} \partial_{\tau} \phi(q; \alpha_0-2, P_0)\Big|_{\tau = \tau_{\star}}  = \mp \frac{H(\tau_{\star})}{4}.
\end{align}
Note that $H(\tau_{\star})\equiv H({\tau_{\star},\alpha_0, P_0})$, and the factor of $1/4$ arises from our choice of the central charge $c = 1 + 6(2/\gamma + \gamma/2)^2$ as opposed to $c = 1 + 6(1/b + b)^2$ elsewhere in the literature\footnote{It is important to note this while comparing these expressions to the results in \cite{BGG2021}.}. This completes the proof of part (2) of Theorem~\ref{prop:sc_Lame}.

\begin{remark}\label{corr:acc-par-alt}
    The accessory parameter has the form 
\begin{align}\label{corr:acc3}
    \frac{\pi \ii}{2}  \partial_\tau {\phi}(q;\alpha_0,P_0)= \left(-\frac{\partial_{z}^2\widetilde{\Gamma}(z;\alpha_0, P_0, q)}{\widetilde{\Gamma}(z;\alpha_0, P_0, q)}+\frac{\alpha_0}{4} \left(\frac{\alpha_0}{4}-1\right) \wp(z) \right),
\end{align}
where $\widetilde{\Gamma}(z;\alpha_0, P_0, q)$ is given by \eqref{thm43:phi_intro}. Since $\widetilde{\Gamma}$ has the explicit expression \eqref{thm43:Lame}, the accessory parameter $\frac{\pi\ii}{2}\partial_\tau \phi(q)$ can then be written in terms of the Cameron-Martin shift and other elementary functions using the formula \eqref{corr:acc3}. it might be interesting to compare the expressions for the accessory parameter given by the formula above and from the $\tau$-derivative of the expression \eqref{def:phiq}. We postpone this analysis for future work. 
\end{remark}

{\begin{proposition}\label{prop:Gandphi}
The semi-classical limit of the undeformed conformal block $\mathcal{G}_{\gamma, \alpha}^{P}(q)$ is related to $\phi(q;\alpha_0, P_0)$ defined in \eqref{eq:subsequential_limit} as the leading order behavior of the deformed conformal block with $\chi = \frac{\gamma}{2}$ in the limit $\gamma\to 0$ as
\begin{align}%\frac{\pi \ii \partial_{\tau}\phi}{2}
\lim_{\gamma\to 0} \gamma^2 \log  \mathcal{G}_{\gamma, P}^{\alpha}(q)&=  \phi(q; \alpha_0, P_0)  - \frac{P_0^2}{2}   \log q - \left(\frac{\alpha_0(\alpha_0-4) }{2} \right) \log \eta(q) \\ &+ \frac{\alpha_0 (\alpha_0+4)}{6} \log(2\pi) - \mathcal{B}(\alpha_0,P_0). \label{rel:GandPhi}
\end{align}
The parameters $\alpha_0 \in (-4,2)$, $P_0\in \mathbb{R}$, and $\mathcal{B}(\alpha_0, P_0)$ is a constant depending only on the monodromy data.
\end{proposition}
% \rmkH{Fix the monodromy dependence}

\begin{proof}
Recall that the semi-classical limits of the deformed and undeformed conformal blocks in \eqref{eq:def-block}, \eqref{eq:q-block} are related through \eqref{rel:def_undef}:
\begin{align}
\lim_{z \rightarrow 0} \lim_{\gamma \rightarrow 0} \gamma^2 \log \left(\theta_{1}(z)^{\frac{\chi}{2}\left(\chi-\alpha\right)} \psi^\alpha_\chi(z, \tau)\right) &= \lim_{z \rightarrow 0} \lim_{\gamma \rightarrow 0} \gamma^2 \log \left(\theta_{1}(z)^{\frac{\chi}{2}\left(\chi-\alpha\right)} \cW(q) \EE\left[\mathcal{V}_{\chi, P}^{\alpha}(z, q)^{-\frac{\alpha}{\gamma} + \frac{\chi}{\gamma}}\right]\right) \nonumber \\
&= \lim_{\gamma \rightarrow 0} \lim_{z \rightarrow 0} \gamma^2 \log \left(\theta_{1}(z)^{\frac{\chi}{2}\left(\chi-\alpha\right)} \cW(q) \EE\left[\mathcal{V}_{\chi, P}^{\alpha}(z, q)^{-\frac{\alpha}{\gamma} + \frac{\chi}{\gamma}}\right]\right) \nonumber \\
&= \lim_{\gamma \rightarrow 0} \gamma^2 \log \mathcal{G}_{\gamma, P}^{\alpha-\chi}(q) + \lim_{\gamma \rightarrow 0} \gamma^2 \log Z^{\alpha-\chi}_{\gamma, P} + \lim_{\gamma \to 0} \gamma^2 \log \mathcal{W}(q). \label{eq:we_are_all_the_same_recall}
\end{align}

{We first note that the existence of the semi-classical limit of the undeformed block is proved in Theorem~\ref{thm:ex_uni}\footnote{In the proof of Theorem~\ref{thm:ex_uni}, the existence of the semi-classical limit of the undeformed block is proved in a direct way, while its uniqueness follows from the uniqueness of the deformed conformal blocks and the relation \eqref{eq:we_are_all_the_same_recall}.}. Secondly, to obtain the last line we use the fact that the limits $z \to 0$ and $\gamma \to 0$ for the above expression are commutative. This statement can be proved in the following way. We have an alternative description of the conformal block \eqref{eq:hatpsigir} through Girsanov's theorem, that enables us to factor the $z$-dependence and $\gamma$-dependence. Then, one can follow the arguments of the proof of Proposition~\ref{prop:limit_commutativity} and see that the expression for the semi-classical limit of the deformed conformal block is analytic in $z$, as the $z$-dependence of the block comes entirely through trigonometric functions. It then follows that the limits commute.}

Let us now study the equation \eqref{eq:we_are_all_the_same_recall} for $\chi = \frac{\gamma}{2}$. With the rescaling $\alpha = \frac{\alpha_0}{\gamma}, P = \frac{P_0}{\gamma}$, the equation \eqref{eq:we_are_all_the_same_recall} becomes
\begin{align}
 \phi(q;\alpha_0,P_0) \mathop{=}^{\eqref{eq:subsequential_limit}} \lim_{z \rightarrow 0} \lim_{\gamma \rightarrow 0} \gamma^2 \log \left( \psi^\alpha_{\gamma/2}(z, \tau)\right) 
&= \lim_{\gamma \rightarrow 0} \gamma^2 \log \mathcal{G}_{\gamma, P_0}^{\alpha_0}(q) + \lim_{\gamma \rightarrow 0} \gamma^2 \log Z^{\alpha_0}_{\gamma, P_0} + \lim_{\gamma \to 0} \gamma^2 \log \mathcal{W}(q)\\
\label{eq:we-are-all-g2}
\end{align}
We now expand the last two terms in the RHS of the expression above.
Firstly, the partition function can be evaluated by using the asymptotic properties of double Gamma functions (Proposition~\ref{prop:AsymptoticsOfA}), and we get
\begin{align}
    \lim_{\gamma \rightarrow 0} \gamma^2 \log Z^{\alpha}_{\gamma, P}(q)=  \frac{\alpha_0^2}{24} \log q+\alpha_0^2 \log\eta(q) + \mathcal{B}(P_0,\alpha_0),\label{GPhider:Z}
\end{align}
where $\mathcal{B}(\alpha_0,P_0)$ does not depend on $q$.
Secondly, from \eqref{asymp1:prefac} we obtain the expression 
\begin{align}
\lim_{\gamma\to 0} \gamma^2 \log \cW(q)  &= \left(\frac{P_0^2}{2} - \frac{\alpha_0^2}{24}\right) \log q  - \frac{\alpha_0 (4+ \alpha_0)}{6} \log \theta_1'(0)\\
& = \left(\frac{P_0^2}{2} - \frac{\alpha_0^2}{24}\right) \log q  - \frac{\alpha_0 (4+ \alpha_0)}{6} \log(2 \pi)- \frac{\alpha_0 (4+ \alpha_0)}{2} \log \eta(q).\label{GPhider:W}
\end{align}
To obtain the last line of the above expression, we used the identity \eqref{def:ell-eta}: $\theta_1'(0) = 2\pi \eta(q)^3$. Substituting \eqref{GPhider:Z} and \eqref{GPhider:W} in the expression \eqref{eq:we-are-all-g2}, we get the desired expression
\begin{align}
    \phi(q;\alpha_0, P_0) = \lim_{\gamma \rightarrow 0} \gamma^2 \log \mathcal{G}_{\gamma, P_0}^{\alpha_0}(q) + \frac{P_0^2}{2} \log q + \frac{\alpha_0(\alpha_0 - 4)}{2} \log \eta(q) - \frac{\alpha_0 (4+ \alpha_0)}{6} \log (2\pi) + \mathcal{B}(\alpha_0, P_0).
\end{align}

\end{proof}

% \rmkH{The above red text is the correct computation}

\section{Semi-classical limit of  undeformed conformal blocks}\label{sec:semi-classical_limit}

 The goal of this section is to prove our main result Theorem~\ref{thm:zamolodchikov}, which we present towards the end of this section. We begin the section with Theorems \ref{thm:ex_uni} and \ref{thm:radius} that establish the existence and analyticity of the semi-classical limit of the undeformed conformal block as a function of the modular parameter $q$ respectively. We then use these theorems along with the results from the previous section to complete the proof of Theorem~\ref{thm:zamolodchikov}.

\begin{theorem}\label{thm:ex_uni}
The semi-classical limit $\gamma \to 0$ of the one-point conformal block on the torus \eqref{eq:def-block} exists for {$q \in [0,1)$}, $\alpha_0 \in (-4,2)$, $P_0 \in \mathbb{R}$. Specifically, the limit
\begin{align}\label{eq:thm41_limit}
\lim_{\gamma \to 0} \gamma^2 \log \mathcal{G}^{\alpha_0/\gamma}_{\gamma, P_0/\gamma}(q)
\end{align}
exists and admits the following explicit representation. Define the function
\begin{align}\label{def:mathcalM}
\mathcal{M}_{\gamma}(P_0, \alpha_0, q) := \gamma^2 \log \mathbb{E}\left[e^{\frac{\alpha_0}{2\gamma} F(0;q)} \mathcal{Q}(q) \mathcal{I}_{\gamma}^{-\frac{\alpha_0}{\gamma^2}}\right] - \gamma^2 \log \mathbb{E}\left[\mathcal{I}_{\gamma}^{-\frac{\alpha_0}{\gamma^2}}\right],
\end{align}
%\rmkH{Notation conflict with $F_{\tau}(0)$ in Section 4.}
where $F(0;q)$ and $\mathcal{Q}(q)$, are defined in \eqref{def:YtauN}, \eqref{eq:Qq} respectively and,
\begin{align}\label{eq:I_gamma_def}
\mathcal{I}_{\gamma} := \int_0^1 (2\sin(\pi x))^{-\frac{\alpha \gamma}{2}} e^{\pi \gamma P x} e^{\frac{\gamma}{2} Y(x)} dx.
\end{align}
{As $\gamma \to 0$, the limit of $\mathcal{M}_{\gamma}(P_0, \alpha_0, q)$ is well-defined}, and furthermore, the semi-classical limit of the conformal block is given by
\begin{align}\label{eq:thm41}
\lim_{\gamma \to 0} \gamma^2 \log \mathcal{G}^{\alpha_0/\gamma}_{\gamma, P_0/\gamma}(q) = -\alpha_0 \mathfrak{Z}(\alpha_0, P_0, q) + \lim_{\gamma \to 0} \mathcal{M}_{\gamma}(P_0, \alpha_0, q),
\end{align}
where
\begin{align}\label{def:mathfrakZ}
\mathfrak{Z}(\alpha_0, P_0, q) &= \log\left(q^{\frac{\alpha_0}{8}} \int^{1}_0 |\theta_1(x)|^{-\frac{\alpha_0}{2}} e^{\pi P_0 x} dx\right) - \log\left(\int^{1}_0 (2\sin(\pi x))^{-\frac{\alpha_0}{2}} e^{\pi P_0 x} dx\right) \nonumber\\
&\quad + \alpha_0 \log(q^{\frac{1}{12}} \eta(q)^{-1}).
\end{align}
\end{theorem}

\begin{theorem}[Radius of Convergence]\label{thm:radius}
There exists a constant $r_0 > 0$, independent of the parameters $\alpha_0 \in (0,2)$, $P_0 \in \mathbb{R}$, such that the semi-classical limit
\begin{align}\label{eq:semiclassical_limit_def}
\lim_{\gamma \to 0} \gamma^2 \log \mathcal{G}^{\alpha_0/\gamma}_{\gamma, P_0/\gamma}(q)
\end{align}
exists when $q\in (0,1)$ and could be extended as an analytic function of $q$ for all $|q| < r_0$. Here, $\mathcal{G}^{\alpha}_{\gamma, P}(q)$ denotes the analytic extension of the torus conformal block as given in Definition~\ref{def:analytic_extension}.
\end{theorem}

\begin{corollary}[Analyticity of the Accessory Parameter]\label{cor:accessory_analytic}

For $q \in (0,1)$, the accessory parameter of Lam\'e equation can be expressed as
\begin{align}
    \mathcal{E}= \frac{\pi \ii}{2}\partial_{\tau}\phi(q;\alpha_0, P_0)=  \frac{\pi \ii}{2}\lim_{\gamma\to 0} \gamma^2 \partial_{\tau}\log  \mathcal{G}_{\gamma, P}^{\alpha}(q)- \frac{ \pi^2 P_0^2}{4} - \left(\frac{\alpha_0(\alpha_0-4) }{8} \right) \eta_1(q).\label{eq:accessory_formula}
\end{align}
See \eqref{def:phiq} and \eqref{eq:subsequential_limit} for the definition of $\phi(q;\alpha_0, P_0)$. There exists $r_0 > 0$, independent of $\alpha_0 \in (0,2)$ and $P_0 \in \mathbb{R}$, such that the accessory parameter above is an analytic function of $q$ for all $|q| < r_0$.
\end{corollary}
\begin{proof}
The formula \eqref{eq:accessory_formula} follows from  Proposition \ref{prop:Gandphi}. The analyticity of the accessory parameter $ \frac{\pi \ii}{2}\partial_{\tau}\phi(q;\alpha_0, P_0)$ follows directly from the analyticity of the conformal block $\lim_{\gamma\to 0} \gamma^2 \log  \mathcal{G}_{\gamma, P}^{\alpha}(q)$ established in Theorem~\ref{thm:radius}.
\end{proof}

We now proceed to prove the above two theorems.

\subsection{Proofs of Theorems \ref{thm:ex_uni} and \ref{thm:radius}}

\begin{proof}[Proof of Theorem~\ref{thm:ex_uni}]

We organize the proof into four main steps. In Step 1, we establish uniform bounds on $\mathcal{M}_{\gamma}(P_0, \alpha_0, q)$ using Hölder's inequality. In Step 2, we perform the asymptotic expansion using the Cameron-Martin theorem. In Step 3, we establish the existence of the limit. In Step 4, we derive the explicit formula \eqref{eq:thm41}.

\medskip
\noindent\textbf{Step 1: Uniform Bounds via Hölder's Inequality.}

\medskip
Recall the definition \eqref{eq:def-block}:
\begin{align}\label{eq:G_recall}
\mathcal{G}^{\alpha}_{\gamma, P}(q) = \frac{1}{Z^{\alpha}_{\gamma, P}(q)} \mathbb{E}\left[\left(\int^{1}_{0} |\theta_1(x)|^{-\frac{\alpha\gamma}{2}} e^{\pi\gamma Px} e^{\frac{\gamma}{2} Y(x;q)} dx\right)^{-\frac{\alpha}{\gamma}}\right].
\end{align}
Using the probabilistic expression \eqref{eq:Z-normalizatoin}, we can rewrite it as
\begin{align}\label{rewrite:G}
\mathcal{G}^{\alpha}_{\gamma, P}(q) = g(q) \frac{\mathbb{E}\left[\left(\mathcal{F}^{(\alpha)}_{\gamma, P}(q)\right)^{-\frac{\alpha}{\gamma}}\right]}{\mathbb{E}\left[\left(\mathcal{F}^{(\alpha)}_{\gamma, P}(0)\right)^{-\frac{\alpha}{\gamma}}\right]},
\end{align}
where
\begin{align}\label{gandF}
\mathcal{F}^{(\alpha)}_{\gamma,P}(q) := q^{\frac{\alpha \gamma}{8}} \int^{1}_{0} |\theta_1(x)|^{-\frac{\alpha\gamma}{2}} e^{\pi\gamma Px} e^{\frac{\gamma}{2} Y(x;q)} dx, \quad g(q) = q^{\frac{1}{12}(-\frac{\alpha \gamma}{2} + \alpha^2 + 1)} \eta(q)^{-\alpha^2 - 1 + \frac{\alpha \gamma}{2}}.
\end{align}
With the scaling $\alpha = \alpha_0/\gamma$, $P = P_0/\gamma$, our goal is to show that
\begin{align}\label{eq:goal_bounds}
-\infty < \liminf_{\gamma \to 0} \gamma^2 \log \mathbb{E}\left[\left(\mathcal{F}^{(\alpha_0/\gamma)}_{\gamma, P_0/\gamma}(q)\right)^{-\frac{\alpha_0}{\gamma^2}}\right] \leq \limsup_{\gamma \to 0} \gamma^2 \log \mathbb{E}\left[\left(\mathcal{F}^{(\alpha_0/\gamma)}_{\gamma, P_0/\gamma}(q)\right)^{-\frac{\alpha_0}{\gamma^2}}\right] < \infty.
\end{align}

\medskip
\noindent\textit{Case 1: $\alpha_0 \in (-4, 0)$.}
In this case, the exponent $-\alpha_0/\gamma^2 > 0$, so we need to bound positive moments of $\mathcal{F}^{(\alpha)}_{\gamma, P}(q)$.
{\color{blue}For any positive integer $N$, we can write $\mathbb{E}\left[\left(\mathcal{F}^{(\alpha)}_{\gamma, P}(q)\right)^{N}\right]$ as a Selberg-type integral:
\begin{align}\label{eq:Selberg_F}
\mathbb{E}\left[\left(\mathcal{F}^{(\alpha)}_{\gamma, P}(q)\right)^{N}\right] = \int_{[0,1]^N} \prod_{i<j}^N |\theta_1(x_i - x_j)|^{-\frac{\gamma^2}{2}} \prod_{i=1}^N |\theta_1(x_i)|^{-\frac{\alpha_0}{2}} e^{\pi P_0 x_i} dx_i.
\end{align}

Noticing $\theta_1(z) = \mathfrak{p}(z;\tau) \sin(\pi z)$ where $c_1 \leq |\mathfrak{p}(z;\tau)| \leq c_2$ as $\mathrm{Re}(z)$ varies in $(0,1)$ and $\mathrm{Im}(z)$ varies in $[0, \frac{3}{4}\mathrm{Im}(\tau))$, and $q \in (0,1)$, the absolute value of the above integral is bounded above and below by constant multiples (the constant would be bounded below by $c_2^{-\gamma^2 N(N-1)/4}$ and above by $c_1^{-\gamma^2 N(N-1)/4}$) of
\begin{align}\label{eq:Selberg_sin}
\int_{[0,1]^N} \prod_{i<j}^N |\sin(\pi(x_i - x_j))|^{-\frac{\gamma^2}{2}} \prod_{i=1}^N (\sin(\pi x_i))^{-\frac{\alpha_0}{2}} e^{\pi P_0 x_i} dx_i.
\end{align}

The above integral has an explicit expression as stated in \eqref{eq:C_expression} of Remark~\ref{rem:Normalization}. As $N$ grows, the ratios of Gamma functions in \eqref{eq:C_expression} grow as $\exp(cN^2\gamma^2)$.

When $\alpha_0 < 0$, $\left|\mathbb{E}\left[\left(\mathcal{F}^{(\alpha)}_{\gamma, P}(q)\right)^{-\alpha_0/\gamma^2}\right]\right|$ can be bounded by $\left|\mathbb{E}\left[\left(\mathcal{F}^{(\alpha)}_{\gamma, P}(q)\right)^N\right]\right|$ for some integer $N = \lfloor -\alpha_0/\gamma^2 \rfloor {\color{blue}+ 1}$. The large $N$ asymptotics of the above integral (obtained through the asymptotics of the ratios of Gamma functions in \eqref{eq:C_expression}) shows that $\left|\mathbb{E}\left[\left(\mathcal{F}^{(\alpha)}_{\gamma, P}(q)\right)^{-\alpha_0/\gamma^2}\right]\right|$ is bounded by $\exp(c/\gamma^2)$ for some constant $c = c(\alpha_0, P_0, r) > 0$ when $\alpha_0 < 0$, $P_0 \in \mathbb{R}$, and $q \in (0,1)$.}

\medskip
\noindent\textit{Case 2: $\alpha_0 \in [0, 2)$.}
We need to bound negative moments of the GMC integral. Define
Recall that $Y(x;q)= Y(x) + F(x;q)$ where $F(x;q)$ is a smooth Gaussian field and furthermore, $\mathbb{E}[Y(x;q)Y(y;q)] \geq \mathbb{E}[Y(x)Y(y)] -2\log (\prod_{k\geq 1}(1+q^k)^2)$. By Slepian's inequality, we get $$\mathbb{E}[(\mathcal{F}^{(\alpha)}_{\gamma,P}(q))^{-\alpha_0/\gamma^2}]\leq \prod_{k}(1+q^k)^{-\alpha_0/\gamma^2}\mathbb{E}\Big[\Big(\int^1_0 (2\sin (\pi x))^{-\frac{\alpha\gamma}{2}} e^{\pi\gamma Px} :e^{\frac{\gamma}{2} Y(x)}:  dx\Big)^{-\alpha_0/\gamma^2}\Big].$$ Applying  Lemma~\ref{lem:negative_moments} for $\alpha_0 \in (0,2]$ and $\gamma \in (0,1]$ to this upper bound, we get 
\begin{align}\label{eq:neg_moment_bound_application}
\mathbb{E}\left[(\mathcal{F}^{(\alpha)}_{\gamma,P}(q))^{-\alpha_0/\gamma^2}\right] \leq \exp\left(\frac{C\alpha^2_0}{\gamma^2}\right),
\end{align}
where $C > 0$ is a constant independent of $\alpha_0$ and $\gamma$.  
Taking logarithms and multiplying by $\gamma^2$:
\begin{align}\label{eq:log_upper_alpha_large}
\gamma^2 \log \mathbb{E}\left[(\mathcal{F}^{(\alpha)}_{\gamma,P}(q))^{-\alpha_0/\gamma^2}\right] \leq  C\alpha_0^2.
\end{align}
This establishes the upper bound for $\alpha_0 > 2$.

\medskip
\noindent\textbf{Step 2: Asymptotic Expansion via Cameron-Martin Theorem.}

\medskip
Following Step 1 of Theorem~\ref{prop:sc_Lame}, we perform a detailed asymptotic expansion of
\begin{align}\label{eq:main_exp_recall}
\mathbb{E}[(\mathcal{F}^{(\alpha)}_{\gamma,P}(q))^{-\alpha_0/\gamma^2}] = \mathbb{E}\left[e^{\frac{\alpha_0}{2\gamma} F(0;q)} \mathcal{Q}(q) \mathcal{I}_{\gamma}^{-\frac{\alpha_0}{\gamma^2}}\right].
\end{align}
We introduce auxiliary random variables $T^{(1)}_n = \sum_{m=1}^{\infty} a_{n,m} q^{nm}$ and $T^{(2)}_n = \sum_{m=1}^{\infty} b_{n,m} q^{nm}$ as in \eqref{eq:T12_def_heavy}. After performing the change of variables and computing the Gaussian integrals as in Steps 1a and 1b of Theorem~\ref{prop:sc_Lame}, we obtain:
\begin{align}
   \mathbb{E}\left[e^{\frac{\alpha_0}{2\gamma} F_{\tau}(0)} \mathcal{Q}(q) \mathcal{I}_{\gamma}^{-\frac{\alpha_0}{\gamma^2}}\right]&= \exp\left(\frac{\alpha_0^2}{2\gamma^2 }\sum_{n=1}^{\infty}  \frac{\sum_{m=1}^{\infty} q^{2nm}}{n(1 + 2\sum_{m=1}^{\infty} q^{2nm})}\right) \\
   &\quad \times \mathbb{E}\Bigg[  \exp\left(\frac{\sqrt{2}\alpha_0}{\gamma}\sum_{n=1}^{\infty}  \frac{\widehat{a}_n}{\sqrt{n}}\left(  \frac{\sum_{m=1}^{\infty} q^{2nm}}{1 + 2\sum_{m=1}^{\infty} q^{2nm}}\right)^{1/2} \right)\\
& \qquad \times \left(\int_{0}^{1} e^{\pi  P_0 x} (2\sin(\pi x))^{-\frac{\alpha\gamma}{2}} e^{\frac{\gamma}{2}Y(x)} dx\right)^{-\frac{\alpha_0}{\gamma^2}}\Bigg] \label{eq:undeformed_expansion}
\end{align}

Introducing the Gaussian variable $\widehat{\Psi}$ in an analogous way to Step 2c in the proof of Theorem~\ref{prop:sc_HJ}:
{
\begin{align}
    \widehat{\Psi} &:= \frac{\sqrt{2}\alpha_0}{\gamma}\sum_{n=1}^{\infty}  \frac{\widehat{a}_n}{\sqrt{n}}\left(  \frac{\sum_{m=1}^{\infty} q^{2nm}}{1 + 2\sum_{m=1}^{\infty} q^{2nm}}\right)^{1/2} - \frac{\alpha_0}{2\gamma \Xi} \int_0^1 e^{\pi P_0 x} (2\sin(\pi x))^{-\frac{\alpha_0}{2}} Y(x) e^{\frac{\gamma}{2} h_{\widehat{\Psi}}(x)} dx,\label{def-undef:Psi-GaussRV}
\end{align}
where $\widehat{a}_n\sim \mathcal{N}\left(0,\sum_{m=1}^{\infty} q^{2nm}\right)$, $h_{\widehat{\Psi}}(x) := \mathbb{E}[\widehat{\Psi} :Y(x):]$, and $\widehat{\Xi} := \int_0^1 e^{\pi P_0 x} (2\sin(\pi x))^{-\frac{\alpha_0}{2}} e^{\frac{\gamma}{2} h_{\widehat{\Psi}}(x)} dx$.}
Applying the Cameron-Martin theorem \eqref{eq:CM_heavy} we get
\begin{align}
&\mathbb{E} \left[ e^{\frac{\alpha_0}{2\gamma} F_{\tau}(0)} \mathcal{Q}(q) \left(\int_{0}^{1} e^{\pi P_0 x} (2\sin(\pi x))^{-\frac{\alpha_0}{2}} e^{\frac{\gamma}{2}Y(x)} dx\right)^{-\frac{\alpha_0}{\gamma^2}}\right] \\
     & =  \exp\left(\frac{\alpha_0^2}{2\gamma^2 }\sum_{n=1}^{\infty}  \frac{\sum_{m=1}^{\infty} q^{2nm}}{n(1 + 2\sum_{m=1}^{\infty} q^{2nm})}\right)  e^{\frac{\mathbb{E}[\widehat{\Psi}^2]}{2}} \widehat{\Xi}^{-\frac{\alpha_0}{\gamma^2}} \\
     &\quad \times\exp\left({ \frac{\alpha_0}{2\gamma \widehat{\Xi}} \int_0^1 e^{\pi P_0 x} (2\sin(\pi x))^{-\frac{\alpha_0}{2}} h_{\widehat{\Psi}}(x) e^{\frac{\gamma}{2} h_{\widehat{\Psi}}(x)} dx}\right)\\
 & \quad \times\mathbb{E}\Big[\left(\int_{0}^{1} e^{\frac{\gamma}{2}Y(x)} d\mu_{\widehat{\Psi}}\right)^{-\frac{\alpha_0}{\gamma^2}} e^{ \frac{\alpha_0}{2\gamma } \int_0^1  Y(x) d\mu_{\widehat{\Psi}}}\Big],\label{eq:after_CM}
\end{align}
where  $d\mu_{\widehat{\Psi}}(x) := \frac{1}{\widehat{\Xi}} e^{\pi P_0 x} (2\sin(\pi x))^{-\frac{\alpha_0}{2}} e^{\frac{\gamma}{2} h_{\widehat{\Psi}}(x)} dx$. By Proposition~\ref{prop:semiclassical_limit}, the final expectation in \eqref{eq:after_CM} converges as $\gamma \to 0$:
\begin{align}\label{eq:expectation_limit}
&\lim_{\gamma \to 0} \mathbb{E}\left[\exp\left(-\frac{\alpha_0}{\gamma^2} \log\left(\int_0^1 :e^{\gamma Y}: d\mu_{\widehat{\Psi}}\right) + \frac{\alpha_0}{2\gamma} \int_0^1 :Y: d\mu_{\widehat{\Psi}}\right)\right] \nonumber\\
&= \mathbb{E}\left[\exp\left(-\frac{\alpha_0}{8} \int_0^1 :Y^2: d\widehat{\mu}_0 + \frac{\alpha_0}{8} \left(\int_0^1 :Y: d\widehat{\mu}_0\right)^2\right)\right],
\end{align}
where $\widehat{\mu}_0$ is the limiting measure of $d\mu_{\widehat{\Psi}}$ as $\gamma \to 0$.

\medskip
\noindent\textbf{Step 3: Existence of the Limit.}

\medskip
From Step 1, we have established that $\gamma^2 \log \mathbb{E}\left[\left(\mathcal{F}^{(\alpha)}_{\gamma, P}(q)\right)^{-\frac{\alpha}{\gamma}}\right]$ is uniformly bounded as $\gamma \to 0$. Therefore, any subsequence has a convergent subsequence. To show that all subsequential limits are the same, we use the uniqueness of the asymptotic expansion show in Step 2.

\medskip
\noindent\textbf{Step 4: Derivation of the Explicit Formula.}

\medskip
From \eqref{rewrite:G}, we have
\begin{align}\label{eq:G_decomposition}
\gamma^2 \log \mathcal{G}^{\alpha}_{\gamma, P}(q) = \gamma^2 \log g(q) + \gamma^2 \log \mathbb{E}\left[\left(\mathcal{F}^{(\alpha)}_{\gamma, P}(q)\right)^{-\frac{\alpha}{\gamma}}\right] - \gamma^2 \log \mathbb{E}\left[\left(\mathcal{F}^{(\alpha)}_{\gamma, P}(0)\right)^{-\frac{\alpha}{\gamma}}\right].
\end{align}

Taking $\gamma \to 0$:

\begin{enumerate}
\item For $g(q)$: From \eqref{gandF} with $\alpha = \alpha_0/\gamma$,
\begin{align}\label{eq:g_limit}
\lim_{\gamma \to 0} \gamma^2 \log g(q) = \lim_{\gamma \to 0} \gamma^2 \left(\frac{1}{12}\left(-\frac{\alpha_0}{2} + \frac{\alpha_0^2}{\gamma^2} + 1\right) \log q + \left(-\frac{\alpha_0^2}{\gamma^2} - 1 + \frac{\alpha_0}{2}\right) \log \eta(q)\right) = -\alpha^2_0 \log (q^{-\frac{1}{12}}\eta(q)) .
\end{align}

\item For $\mathbb{E}\left[\left(\mathcal{F}^{(\alpha)}_{\gamma, P}(q)\right)^{-\frac{\alpha}{\gamma}}\right]$: From \eqref{eq:after_CM} and \eqref{eq:expectation_limit},
\begin{align}
&\gamma^2 \log \mathbb{E}\left[\left(\mathcal{F}^{(\alpha)}_{\gamma, P}(q)\right)^{-\frac{\alpha}{\gamma}}\right]=\gamma^2 \log \mathbb{E}\left[e^{\frac{\alpha_0}{2\gamma} F(0;q)} \mathcal{Q}(q) \mathcal{I}_{\gamma}^{-\frac{\alpha_0}{\gamma^2}}\right] \\
    &=\frac{\alpha_0^2}{2}\sum_{n=1}^{\infty}  \frac{\sum_{m=1}^{\infty} q^{2nm}}{n(1 + 2\sum_{m=1}^{\infty} q^{2nm})}+ \frac{\gamma^2\mathbb{E}[\widehat{\Psi}^2]}{2} - \alpha_0\log \widehat{\Xi} \\
     &\quad +\left({ \frac{\gamma\alpha_0}{2\widehat{\Xi}} \int_0^1 e^{\pi P_0 x} (2\sin(\pi x))^{-\frac{\alpha_0}{2}} h_{\widehat{\Psi}}(x) e^{\frac{\gamma}{2} h_{\widehat{\Psi}}(x)} dx}\right)\\
 & \quad \gamma^2 \log \mathbb{E}\left[\left(\int_{0}^{1} e^{\frac{\gamma}{2}Y(x)} d\mu_{\widehat{\Psi}}\right)^{-\frac{\alpha_0}{\gamma^2}} e^{ \frac{\alpha_0}{2\gamma } \int_0^1  Y(x) d\mu_{\widehat{\Psi}}}\right]\label{eq:F_q_expansion}
\end{align}

%\rmkH{From \eqref{asympreg:term1}}
\item For $\mathbb{E}\left[\left(\mathcal{F}^{(\alpha)}_{\gamma, P}(0)\right)^{-\frac{\alpha}{\gamma}}\right]$:
\begin{align}\label{eq:F_0_expansion}
\gamma^2 \log \mathbb{E}\left[\left(\mathcal{F}^{(\alpha)}_{\gamma, P}(0)\right)^{-\frac{\alpha}{\gamma}}\right] =  \gamma^2 \log \mathbb{E}\left[\mathcal{I}_{\gamma}^{-\frac{\alpha_0}{\gamma^2}}\right].
\end{align}
\end{enumerate}

Combining these:
\begin{align}\label{eq:combined_limit}
\lim_{\gamma \to 0} \gamma^2 \log \mathcal{G}^{\alpha}_{\gamma, P}(q) =   \lim_{\gamma \to 0} \left(\gamma^2 \log \mathbb{E}\left[e^{\frac{\alpha_0}{2\gamma} F(0;q)} \mathcal{Q}(q) \mathcal{I}_{\gamma}^{-\frac{\alpha_0}{\gamma^2}}\right] - \gamma^2 \log \mathbb{E}\left[\mathcal{I}_{\gamma}^{-\frac{\alpha_0}{\gamma^2}}\right]\right)= \lim_{\gamma \to 0} \mathcal{M}_{\gamma}(P_0, \alpha_0, q).
\end{align}
where $\mathcal{M}_{\gamma}(P_0, \alpha_0, q)$ is defined in \eqref{def:mathcalM}.
Finally, the uniqueness of the semi-classical limit follows from the uniqueness of the limits of the deterministic terms in the expansion of \eqref{eq:F_q_expansion} and the asymptotic expansion of the moment formulas $\mathbb{E}\big[\mathcal{I}_{\gamma}^{-\alpha_0/\gamma^2}\big]$, which prove the uniqueness of the deformed conformal block. This in turn implies the uniqueness of the undeformed conformal block due to \eqref{eq:we_are_all_the_same_recall}. This completes the proof of Theorem~\ref{thm:ex_uni}.
%\rmkH{Theorems~\ref{thm:sc_HJ} and \ref{thm:sc_Lame}}

\end{proof}

Let us now show that the semi-classical conformal block has a positive radius of convergence as stated in Theorem \ref{thm:radius}. The main element of the proof is the expression of conformal block that can be obtained using the Cameron-Martin change of variable formula (Girsanov theorem) \cite[Theorem C.5]{ghosal2020probabilistic} such that the GMC measure is made $q$-independent. Such an expression enables analytic control over the $q$-dependent part and helps us obtain the final result.

\begin{proof}[Proof of Theorem \ref{thm:radius}]

We organize the proof into four main steps. In Step 1, we establish uniform bounds (both upper and lower) on the logarithm of the conformal block, demonstrating that the semi-classical limit is well-defined. In Step 2, we perform a detailed asymptotic expansion of the undeformed conformal block using Gaussian calculus techniques. In Step 3, we carry out an analogous expansion for the deformed conformal block. Finally, in Step 4, we relate these two expansions to establish the connection with the Lam\'e equation and complete the proof.

\medskip
\noindent\textbf{Step 1: Uniform Bounds on the Semi-Classical Limit.} We begin by recalling the probabilistic formula for the analytic extension of the torus conformal block from Definition~\ref{def:analytic_extension}:
\begin{align}\label{eq:analytic_ext_recall}
\mathcal{G}^{\alpha}_{\gamma, P}(q) = \frac{1}{Z^{\alpha}_{\gamma, P}} \left(q^{1/6}\eta(q)\right)^{\frac{\alpha}{2}} e^{-\frac{\alpha^2}{8}\mathbb{E}[F(0;q)^2]} \mathbb{E}\left[e^{\frac{\alpha}{2}F(0;q)}\mathcal{Q}(q) \mathcal{I}_{\gamma}^{-\frac{\alpha}{\gamma}}\right].
\end{align}
%where we have introduced the shorthand notation
The field $F(x;q)$ is the smooth Gaussian field defined in \eqref{def:YtauN}, and the random variable $\mathcal{Q}(q)$ is given by \eqref{eq:Qq}:
\begin{align}\label{eq:Q_explicit}
\mathcal{Q}(q) = \exp\left(\sqrt{2}\sum_{m,n=1}^{\infty} q^{nm} \left(\tilde{\alpha}_{n,m}\tilde{\alpha}_n + \tilde{\beta}_{n,m}\tilde{\beta}_n\right) - \sum_{n=1}^{\infty}\left(\sum_{m=1}^{\infty} q^{nm}\tilde{\alpha}_{n,m}\right)^2 - \sum_{n=1}^{\infty}\left(\sum_{m=1}^{\infty} q^{nm}\tilde{\beta}_{n,m}\right)^2\right).
\end{align}
The partition function $Z^{\alpha}_{\gamma, P}$ is given by \eqref{eq:Z-normalizatoin}:
\begin{align}\label{eq:Z_recall}
Z^{\alpha}_{\gamma, P} = q^{\frac{1}{12}\left(\frac{\alpha \gamma}{2} + \frac{\alpha^2}{2} - 1\right)} \eta(q)^{\alpha^2 + 1 - \frac{\alpha \gamma}{2}} \mathbb{E}\left[\mathcal{I}_{\gamma}^{-\frac{\alpha}{\gamma}}\right].
\end{align}
Our goal in this step is to establish the following two bounds. First, we prove the upper bound:
\begin{align}\label{eq:limsup_bound}
\limsup_{\gamma \to 0} \gamma^2 \log \left|\mathbb{E}\left[e^{\frac{\alpha}{2} F(0;q)}\mathcal{Q}(q) \mathcal{I}_{\gamma}^{-\frac{\alpha}{\gamma}}\right]\right| < \infty
\end{align}
for $q \in B_r(0)$ for some $0 < r < 1$. Second, we prove the lower bound:
\begin{align}\label{eq:liminf_bound}
\liminf_{\gamma \to 0} \gamma^2 \log \left|\mathbb{E}\left[e^{\frac{\alpha}{2} F(0;q)}\mathcal{Q}(q) \mathcal{I}_{\gamma}^{-\frac{\alpha}{\gamma}}\right]\right| > -\infty
\end{align}
for $q \in (0,1)$. In fact, we show that the above limit exists when $q\in (0,1)$ and the limiting expression could be analytically extended for all $q\in B_r(0)$.
We observe that under the scaling $\alpha = \alpha_0/\gamma$ and $P = P_0/\gamma$, the quantity $\gamma^2 \log Z^{\alpha}_{\gamma, P}$ has a well-defined limit as $\gamma \to 0$. This follows from Proposition~\ref{prop:Zexplicit} and the asymptotic properties of the double Gamma functions established in Proposition~\ref{prop:AsymptoticsOfA}. Therefore, it suffices to analyze the expectation
\begin{align}\label{eq:main_expectation}
\mathbb{E}\left[e^{\frac{\alpha_0}{2\gamma} F(0;q)}\mathcal{Q}(q) \mathcal{I}_{\gamma}^{-\frac{\alpha_0}{\gamma^2}}\right].
\end{align}

\medskip
\noindent\textit{Upper Bound via H\"older's Inequality.} We first observe that $\mathcal{Q}(q)$ admits the decomposition $\mathcal{Q}(q) = 1 + q\Theta(q)$, where $\Theta(q)$ is a random variable satisfying $\mathbb{E}[|\Theta(q)|^{p_2}] < \infty$ for any $p_2 > 1$ when $0 < q \ll 1$. This decomposition follows from the explicit form \eqref{eq:Q_explicit} and the fact that each term in the exponent contains at least one factor of $q$. Choose exponents $p_1, p_2, p_3 \in (1, \infty)$ satisfying the conjugate relation
\begin{align}\label{eq:holder_exponents}
\frac{1}{p_1} + \frac{1}{p_2} + \frac{1}{p_3} = 1,
\end{align}
and such that the negative moment condition
\begin{align}\label{eq:negative_moment_condition}
\mathbb{E}\left[\mathcal{I}_{\gamma}^{-\frac{\alpha_0 p_3}{\gamma^2}}\right] < \infty
\end{align}
is satisfied. The existence of such exponents follows from the theory of Gaussian multiplicative chaos; specifically, the negative moments of GMC measures are finite when the exponent is sufficiently small relative to $\gamma^{-2}$. Applying H\"older's inequality with these exponents, we obtain
\begin{align}\label{eq:holder_application}
&\left|\mathbb{E}\left[e^{\frac{\alpha_0}{2\gamma} F(0;q)} \Theta(q) \mathcal{I}_{\gamma}^{-\frac{\alpha_0}{\gamma^2}}\right]\right| \nonumber\\
&\quad \leq \mathbb{E}\left[\left|e^{\frac{\alpha_0}{2\gamma} F(0;q)}\right|^{p_1}\right]^{\frac{1}{p_1}} \mathbb{E}\left[|\Theta(q)|^{p_2}\right]^{\frac{1}{p_2}} \mathbb{E}\left[\mathcal{I}_{\gamma}^{-\frac{\alpha_0 p_3}{\gamma^2}}\right]^{\frac{1}{p_3}}.
\end{align}
We now bound each factor on the right-hand side:

\begin{enumerate}[label=(\roman*)]
\item \textit{First factor:} The random variable $F(0;q)$ is Gaussian with variance $$\mathbb{E}[F(0;q)^2] = -4\log|q^{-1/12}\eta(q)|,$$ which is finite and bounded uniformly for $q$ in any compact subset of $(0,1)$. Therefore, using the moment generating function of a Gaussian random variable, we have
\begin{align}\label{eq:first_factor_bound}
\mathbb{E}\left[\left|e^{\frac{\alpha_0}{2\gamma} F(0;q)}\right|^{p_1}\right]^{\frac{1}{p_1}} = \exp\left(\frac{p_1 \alpha_0^2}{8\gamma^2} \mathbb{E}[F(0;q)^2]\right)^{\frac{1}{p_1}} \leq \exp\left(\frac{c_1 \alpha_0^2}{\gamma^2}\right)
\end{align}
for some constant $c_1 = c_1(q) > 0$ that is uniformly bounded for $q \in [0, r]$ for any $r \in (0,1)$.

\item \textit{Second factor:} Since $\Theta(q)$ involves finite linear combinations of products of Gaussian random variables, its moments are finite. We have
\begin{align}\label{eq:second_factor_bound}
\mathbb{E}\left[|\Theta(q)|^{p_2}\right]^{\frac{1}{p_2}} \leq e^{c_2 p_2}
\end{align}
for some constant $c_2 = c_2(q) > 0$.

\item \textit{Third factor:} The negative moments of the GMC integral can be computed explicitly using the following formula (derived in \cite{ang2023derivation}), we have for appropriate parameter ranges
\begin{align}\label{eq:negative_moment_formula}
&\mathbb{E}\left[\left(\int_0^1 (\sin \pi x)^{-\frac{\alpha_0}{2}} e^{\pi P_0 x} e^{\frac{\gamma}{2} Y(x)} dx\right)^{-\frac{\alpha_0 p_3}{\gamma^2}}\right] \nonumber\\
&\quad = e^{\frac{2\pi P_0}{\gamma^2}} \frac{\Gamma\left(\frac{\alpha_0}{4} - \frac{\gamma^2}{4}\right) \Gamma_{\frac{\gamma}{2}}\left(\frac{\alpha_0}{2\gamma}\right) \Gamma_{\frac{\gamma}{2}}\left(-\frac{\alpha_0}{2\gamma}\right) \Gamma_{\frac{\gamma}{2}}\left(\frac{i(P_0+i)}{\gamma} + \frac{-\alpha_0 p_3}{\gamma}\right) \Gamma_{\frac{\gamma}{2}}\left(-\frac{i(P_0+i)}{\gamma}\right)}{\Gamma\left(1 - \frac{\gamma^2}{4}\right)^{\frac{2}{\gamma}\left(\frac{-\alpha_0 p_3}{2\gamma}\right)} \Gamma_{\frac{\gamma}{2}}\left(\frac{1}{\gamma}\right) \Gamma_{\frac{\gamma}{2}}\left(\frac{1-\alpha_0}{\gamma}\right) \Gamma_{\frac{\gamma}{2}}\left(\frac{i(P_0+i)}{\gamma} + \frac{-\alpha_0 p_3}{2\gamma}\right) \Gamma_{\frac{\gamma}{2}}\left(-\frac{i(P_0+i)}{\gamma} - \frac{-\alpha_0 p_3}{2\gamma}\right)}.
\end{align}
Using the asymptotic expansion of the double Gamma function from \eqref{eq:DoubleGammaAsymptote}, which states that $\Gamma_{\frac{\gamma}{2}}(\xi/\gamma) \leq \exp(|\psi(\xi)|/\gamma^2)$ for some function $\psi: \mathbb{C} \to \mathbb{C}$, we conclude that the right-hand side of \eqref{eq:negative_moment_formula} grows at most as $\exp(C/\gamma^2)$ for some constant $C = C(\alpha_0, P_0, q, p_3) > 0$.
\end{enumerate}
Combining these three bounds with \eqref{eq:holder_application}, we obtain
\begin{align}\label{eq:upper_bound_conclusion}
\left|\mathbb{E}\left[e^{\frac{\alpha_0}{2\gamma} F(0;q)} \Theta(q) \mathcal{I}_{\gamma}^{-\frac{\alpha_0}{\gamma^2}}\right]\right| \leq \exp\left(\frac{C'}{\gamma^2}\right)
\end{align}
for some constant $C' = C'(\alpha_0, P_0, q, p_1, p_2, p_3) > 0$. This establishes the upper bound \eqref{eq:limsup_bound}.

\medskip
\noindent\textit{Lower Bound via Reverse H\"older Inequality.} For the lower bound, we apply H\"older's inequality in the reverse direction. We have
\begin{align}\label{eq:reverse_holder}
\mathbb{E}\left[\mathcal{I}_{\gamma}^{-\frac{\alpha_0}{\gamma^2 p_3}}\right] &\leq \mathbb{E}\left[\left|e^{-\frac{\alpha_0}{2\gamma p_3} F(0;q)}\right|^{p_1}\right]^{\frac{1}{p_1}} \mathbb{E}\left[\left|\mathcal{Q}(q)^{-1/p_3}\right|^{p_2}\right]^{\frac{1}{p_2}} \nonumber\\
&\quad \times \left|\mathbb{E}\left[e^{\frac{\alpha_0}{2\gamma} F(0;q)}\mathcal{Q}(q) \mathcal{I}_{\gamma}^{-\frac{\alpha_0}{\gamma^2}}\right]\right|^{\frac{1}{p_3}}.
\end{align}
Rearranging this inequality, we obtain
\begin{align}\label{eq:lower_bound_rearranged}
\left|\mathbb{E}\left[e^{\frac{\alpha_0}{2\gamma} F(0;q)}\mathcal{Q}(q) \mathcal{I}_{\gamma}^{-\frac{\alpha_0}{\gamma^2}}\right]\right|^{\frac{1}{p_3}} &\geq \mathbb{E}\left[\left|e^{-\frac{\alpha_0}{2\gamma p_3} F(0;q)}\right|^{p_1}\right]^{-\frac{1}{p_1}} \nonumber\\
&\quad \times \mathbb{E}\left[\left|\mathcal{Q}(q)^{-1/p_3}\right|^{p_2}\right]^{-\frac{1}{p_2}} \mathbb{E}\left[\mathcal{I}_{\gamma}^{-\frac{\alpha_0}{\gamma^2 p_3}}\right].
\end{align}
Since $|e^{-F(0;q)}|$ and $|\mathcal{Q}(q)^{-1}|$ are exponentials of linear combinations of Gaussian random variables with finite variance, their moments are controlled. Specifically:
\begin{align}
\mathbb{E}\left[\left|e^{-\frac{\alpha_0}{2\gamma p_3} F(0;q)}\right|^{p_1}\right]^{-\frac{1}{p_1}} \geq \exp\left(-\frac{\alpha_0^2 p_1 \psi_1}{\gamma^2 p_3^2}\right), %\label{eq:first_lower_factor}\\
\qquad \mathbb{E}\left[\left|\mathcal{Q}(q)^{-1/p_3}\right|^{p_2}\right]^{-\frac{1}{p_2}} \geq \exp\left(-\frac{p_2 \psi_2}{p_3^2}\right), \label{eq:second_lower_factor}
\end{align}
for some constants $\psi_1 = \psi_1(q)$ and $\psi_2 = \psi_2(q)$ that are uniformly bounded for $q \in [0, r]$. The third factor $\mathbb{E}[\mathcal{I}_{\gamma}^{-\alpha_0/(\gamma^2 p_3)}]$ grows as $\exp(c_3/\gamma^2)$ for some constant $c_3 = c_3(\alpha_0, P_0) > 0$, as follows from its explicit formula in terms of double Gamma functions.
Combining these estimates, we conclude that
\begin{align}\label{eq:lower_bound_conclusion}
\left|\mathbb{E}\left[e^{\frac{\alpha_0}{2\gamma} F(0;q)}\mathcal{Q}(q) \mathcal{I}_{\gamma}^{-\frac{\alpha_0}{\gamma^2}}\right]\right| \geq \exp\left(\frac{\tilde{C}}{\gamma^2}\right)
\end{align}
for some constant $\tilde{C} = \tilde{C}(\alpha_0, P_0, p_1, p_2, p_3) \in \mathbb{R}$. This establishes the lower bound \eqref{eq:liminf_bound}.

\medskip
\noindent\textbf{Step 2: Uniqueness and analyticity.}
From Step 1, we know that $\gamma^2 \log \mathcal{G}^{\alpha}_{\gamma,P}(q)$ is bounded as $\gamma \to 0$. The bounds \eqref{eq:upper_bound_conclusion} and \eqref{eq:lower_bound_conclusion} show that any subsequential limit is finite. The two key observations are that the Lam\'e equation \eqref{thm43:Lame} determines the accessory parameter $\mathcal{E}$ uniquely once the solution $\widetilde{\Gamma}(z;\alpha_0,P_0,q)$ is specified, and the accessory parameter of the Lam\'e equation is related to the semi-classical limit of the undeformed conformal block as
\begin{align}
     \mathcal{E}= \frac{\pi \ii}{2}\partial_{\tau}\phi(q;\alpha_0, P_0)=  \frac{\pi \ii}{2}\lim_{\gamma\to 0} \gamma^2 \partial_{\tau}\log  \mathcal{G}_{\gamma, P}^{\alpha}(q)- \frac{ \pi^2 P_0^2}{4} - \left(\frac{\alpha_0(\alpha_0-4) }{8} \right) \eta_1(q).\label{eq:accessory_formula}
\end{align}
The above relation is a direct consequence of \eqref{rel:GandPhi}.
Since $\widetilde{\Gamma}$ is given explicitly by \eqref{thm43:phi_intro}, the uniqueness of the solution to the Lam\'e equation (with prescribed monodromy data) implies the uniqueness of the accessory parameter. Therefore, all subsequential limits of $$ \gamma^2 \partial_{\tau}(q^{P_0^2/2} \mathcal{G}^{\alpha}_{\gamma,P}(q))/(4 q^{P_0^2/2} \mathcal{G}^{\alpha}_{\gamma,P}(q))$$ must be the same, and hence the limit
\begin{align}\label{eq:limit_exists}
\lim_{\gamma\to 0} \gamma^2 \partial_{\tau}(q^{P_0^2/2} \mathcal{G}^{\alpha}_{\gamma,P}(q))/(4 q^{P_0^2/2} \mathcal{G}^{\alpha}_{\gamma,P}(q))
\end{align}
exists for all $q \in (0,1)$. Combined with the established limit at $q = 0$ (from the explicit formula for $Z^{\alpha}_{\gamma,P}$ in terms of double Gamma functions, as shown in Proposition~\ref{prop:AsymptoticsOfA}), we can integrate \eqref{eq:limit_exists} to conclude that the relation \eqref{rel:GandPhi} exists for all $q\in (0,r_0)$ for some $r_0 > 0$. Furthermore, the limiting value which now could be expressed in terms of $\tilde{\Gamma}(z;\alpha_0,P_0,q)$ via the Lam\'e equation, is analytic function of $q \in B_{r_0}(0)$ for some $r_0>0$ thanks to analytic expression of $\tilde{\Gamma}(z;\alpha_0,P_0,q)$ in Theorem~\ref{prop:sc_Lame}. Thus the semi-classical limit could be analytically extended to $B_{r_0}(0)$ for some $r_0>0$. This completes the proof of Theorem~\ref{thm:radius}.

\end{proof}

\subsection{Proof of Theorem~\ref{thm:zamolodchikov}}
Theorem~\ref{thm:ex_uni} proves existence of the semi-classical undeformed conformal block. Theorem~\ref{thm:radius} shows that the semi-classical limit can be analytically extended as a function of $q$ in $B_r(0)$ for some $r\in (0,1)$ and furthermore, Corollary~\ref{cor:accessory_analytic} shows that the semi-classical limit (as in \eqref{corr:acc3}) is expressed as accessory parameter of the Lam\'e equation whose solution is explicitly described in \eqref{thm43:phi_intro}. 
This completes the proof of Theorem~\ref{thm:zamolodchikov}.

\begin{remark}
    The range $\alpha_0 \in (-4, 2)$ in Theorem~\ref{thm:zamolodchikov} arises from the constraints of our probabilistic approach:
\begin{itemize}[leftmargin=*]
    \item The upper bound $\alpha_0 < 2$ ensures that the GMC integral $\int_0^1 |\theta_1(x)|^{-\alpha_0/2} e^{\pi P_0 x} dx$ converges, as $|\theta_1(x)| \sim |x|$ near $x = 0$.
    \item The lower bound $\alpha_0 > -4$ comes from the Seiberg bound for the conformal block definition.
\end{itemize}
We expect the semi-classical limit to exist for all $\alpha_0 \in \mathbb{C}$ based on physics predictions, but our probabilistic techniques are currently limited to this range. Extending to complex $\alpha_0$ would require different analytical methods.
\end{remark}
 
 \section{Asymptotics of semi-classical conformal block as \texorpdfstring{$\tau\to \ii \infty$}{tau->i*infinity}}\label{sec:asymp}

%\section{Asymptotics of semi-classical conformal block as $\tau\to \ii \infty$}\label{sec:asymp}

In this section, we study the asymptotic behaviour of the deformed conformal blocks $\psi^{\alpha}_{\chi}(z,q)$ (for $\chi = \frac{\gamma}{2}$ and $\chi = \frac{2}{\gamma}$) and the partition function $Z^{\alpha}_{\gamma, P}(q)$ (see Definition~\ref{def:CBT}) by setting $q\to 0$ in the semi-classical limit $\gamma \to 0$. These results played a major role in Section~\ref{sec:deformed_semi-classical_limit} to prove the semi-classical limit of the deformed conformal blocks. In order to find the desired asymptotics of this section, a key role will be played by the representation of the partition function $Z^{\alpha}_{\gamma, P}$ in terms of double Gamma functions (see Proposition~\ref{prop:Zexplicit}).

We find the semi-classical asymptotics of the partition function $Z^{\alpha}_{\gamma, P}(q)$ in Theorem~\ref{thm:PartitionFunctionLimit} around $q=0$. This is used in Theorem~\ref{thm:q=0DeformedPsiAtGamma/2} and Theorem~\ref{prop:incond_HJ} to show the semi-classical asymptotics of the deformed conformal blocks $\psi^{\alpha}_{\gamma/2}(z,q)$ and $\psi^{\alpha}_{2/\gamma}(z,q)$ by first setting $q=0$. 
In what follows, we present the ingredients for finding the asymptotics of Theorem~\ref{thm:PartitionFunctionLimit} in few lemmas. We start by noting that the function $\mathcal{A}_{\gamma, P}(\alpha)$, which is related to the partition function $Z^{\alpha}_{\gamma, P}(q)$ as shown in \eqref{eq:ZA}, can be expressed in the following way by comparing \eqref{eq:Z-normalizatoin} and \eqref{eq:Adoublegamma}:
  \begin{align}\label{A:reexp}
    \mathcal{A}_{\gamma, P}(\alpha) &= \lim_{q\to 0} \EE\left[\left(q^{\frac{\alpha \gamma}{8}}\int_0^1 \theta_1(x)^{-\frac{\alpha\gamma}{2}}e^{\pi\gamma Px} e^{\frac{\gamma}{2} Y(x; q)} dx \right)^{-\frac{\alpha}{\gamma}}\right] \nonumber\\
    &= e^{\frac{\mathbf{i}\pi \alpha^2}{2}} \left(\frac{\gamma}{2}\right)^{\frac{\gamma \alpha}{2}} e^{-\frac{\pi \alpha P}{2}}\Gamma\left(1-\frac{\gamma^2}{4}\right)^{\frac{\alpha}{\gamma}}\frac{\Gamma_{\frac{\gamma}{2}}(Q -\frac{\alpha}{2})\Gamma_{\frac{\gamma}{2}}(\frac{2}{\gamma}+ \frac{\alpha}{2})\Gamma_{\frac{\gamma}{2}}(Q -\frac{\alpha}{2} - \mathbf{i}P)\Gamma_{\frac{\gamma}{2}}(Q -\frac{\alpha}{2} + \mathbf{i}P)}{\Gamma_{\frac{\gamma}{2}}(\frac{2}{\gamma})\Gamma_{\frac{\gamma}{2}}(Q - \mathbf{i}P)\Gamma_{\frac{\gamma}{2}}(Q + \mathbf{i}P)\Gamma_{\frac{\gamma}{2}}(Q-\alpha)}.
\end{align}
We first find the asymptotics of $\mathcal{A}_{\gamma, p}(\alpha)$ in Proposition~\ref{prop:AsymptoticsOfA}. 
% Let us begin with the observation that The function $\mathcal{A}_{\gamma, P}(\alpha)$ that relates to the partition function $Z$ as in \eqref{eq:ZA} can be written as follows by comparing \eqref{eq:Z-normalizatoin} and \eqref{eq:Adoublegamma}:
 To this end, we first observe the following asymptotics for the double Gamma functions.
\begin{lemma}\label{lem:doublegamma}
For any $\xi\in \mathbb{C}$ such that $\mathrm{arg}(\xi)<\pi$, we have 
%\textcolor{red}{Originally: 
\begin{align}\label{eq:DoubleGammaAsymptote}
    \lim_{\gamma\to 0}&\Big[\gamma^2\log \Gamma_{\frac{\gamma}{2}}\Big(\frac{2\xi}{\gamma}\Big) - \gamma^2\log \frac{\gamma}{2}\int^{\mathrm{Re}(\xi)}_{0}(\xi-x-\frac{1}{2})dx - \gamma^2 \log \Gamma_{\frac{\gamma}{2}}\Big(\frac{2\mathbf{i}\mathrm{Im}(\xi)}{\gamma}\Big)\Big]\nonumber\\ &=\begin{cases} -4\mathrm{Re}(\xi)\log \sqrt{2\pi}+\int^{\mathrm{Re}(\xi)}_{0} \log \Gamma(\xi-x) dx & \text{if }\mathrm{Re}(\xi) >0,\\
    4\mathrm{Re}(\xi)\log \sqrt{2\pi}+\int^{0}_{\mathrm{Re}(\xi)} \log \Gamma(\xi+x) dx & \text{if }\mathrm{Re}(\xi) <0. 
    \end{cases}
\end{align}%}
% \begin{align}\label{eq:DoubleGammaAsymptote}
%     \lim_{\gamma\to 0}&\Big[\gamma^2\log \Gamma_{\frac{\gamma}{2}}\Big(\frac{2\xi}{\gamma}\Big) - 4\log \frac{\gamma}{2}\int^{\mathrm{Re}(\xi)}_{0}\left(\xi-x-\frac{1}{2}\right)dx - \gamma^2 \log \Gamma_{\frac{\gamma}{2}}\Big(\frac{2\mathbf{i}\mathrm{Im}(\xi)}{\gamma}\Big)\Big]\nonumber\\ &= 4\mathrm{Re}(\xi)\log \sqrt{2\pi}-\int^{\mathrm{Re}(\xi)}_{0} 4\log \Gamma\left(\xi-x\right) dx 
% \end{align}

Here $\log \Gamma(\cdot)$ is taken w.r.t. the principle branch of logarithm.
\end{lemma}
\begin{proof}

% For $\Re (z)>0$, the function $\Gamma_{\frac{\gamma}{2}}$ has the following integral representation \cite{ponsot2004recent}:
% \begin{align*}
%     \log\Gamma_{\frac{\gamma}{2}}(z):=\int_0^{\infty} \frac{dt}{t} \left[\frac{e^{-zt}-e^{-\frac{Qt}{2}}}{\left(1-e^{-\frac{\gamma t}{2}} \right) \left(1-e^{-\frac{2 t}{\gamma}} \right)} - \frac{\left(\frac{Q}{2}-z \right)^2}{2}e^{-t} + \frac{z - \frac{Q}{2}}{t} \right].
% \end{align*}
We prove the above statement for $\mathrm{Re}(\xi)>0$. The remaining case follows from a similar argument. From the identities of the double Gamma function, we obtain the following shift relation for the function $\Gamma_{\frac{\gamma}{2}}(.)$ for any $\mathrm{Re}(\xi)>0$ \cite{ponsot2004recent}:
\begin{align}\label{dg:id}
    \Gamma_{\frac{\gamma}{2}}\left(z+\chi \right) = \sqrt{2\pi} \frac{\chi^{\chi z- \frac{1}{2}}}{\Gamma(\chi z)} \Gamma_{\frac{\gamma}{2}}(z), && \chi =  \frac{2}{\gamma}, \frac{\gamma}{2},
\end{align}
which can also be expressed as
% The above relation can be re-expressed as
\begin{align*}
    \Gamma_{\frac{\gamma}{2}}\Big(\frac{2\xi}{\gamma} - n\frac{\gamma}{2}\Big) = \frac{1}{\sqrt{2\pi}}\Big(\frac{\gamma}{2}\Big)^{-\xi+(n+1)\frac{\gamma^2}{4}+\frac{1}{2}}\Gamma\big(\xi-(n+1)\frac{\gamma^2}{4}\big)\Gamma_{\frac{\gamma}{2}}\Big(\frac{2\xi}{\gamma} - (n+1)\frac{\gamma}{2}\Big). 
\end{align*}
% \begin{align}
%     \Gamma_{\frac{\gamma}{2}}\Big(\frac{\xi - (n+1)\frac{\gamma^2}{4}}{\gamma/2}\Big) = \frac{1}{\sqrt{2\pi}}\Big(\frac{\gamma}{2}\Big)^{-\xi+(n+1)\frac{\gamma^2}{2}+\frac{1}{2}}\Gamma\big(\xi-n\frac{\gamma^2}{4}\big)\Gamma_{\frac{\gamma}{2}}\Big(\frac{\xi-n\frac{\gamma^2}{4}}{\gamma/2}\Big) 
% \end{align}
 Applying the above relation for all $n$ such that $n\frac{\gamma^2}{4}\leq \mathrm{Re}(\xi)$, we get
%  \textcolor{red}{Originally: \begin{align*}
%    \Gamma_{\frac{\gamma}{2}}\Big(\frac{2\mathbf{i}\mathrm{Im}\xi}{\gamma}\Big)  = \prod_{n=1}^{\lfloor\frac{4\mathrm{Re}(\xi)}{\gamma^2}\rfloor}\frac{1}{\sqrt{2\pi}}\Big(\frac{\gamma}{2}\Big)^{-\xi+(n+1)\frac{\gamma^2}{4}+\frac{1}{2}}\Gamma\big(\xi-(n+1)\frac{\gamma^2}{4}\big)\times \Gamma_{\frac{\gamma}{2}}\Big(\frac{2\xi}{\gamma}\Big). 
% \end{align*}}
\begin{align*}
   \Gamma_{\frac{\gamma}{2}}\Big(\frac{2\mathbf{i}\mathrm{Im}\xi}{\gamma}+\frac{\gamma}{2}\left\{\frac{4\mathrm{Re}(\xi)}{\gamma^2}\right\}\Big)  = \Gamma_{\frac{\gamma}{2}}\Big(\frac{2\xi}{\gamma}\Big)\prod_{n=0}^{\lfloor\frac{4\mathrm{Re}(\xi)}{\gamma^2}-1\rfloor}\frac{1}{\sqrt{2\pi}}\Big(\frac{\gamma}{2}\Big)^{-\xi+(n+1)\frac{\gamma^2}{4}+\frac{1}{2}}\Gamma\big(\xi-(n+1)\frac{\gamma^2}{4}\big). 
\end{align*}
where $\left\{\frac{4\mathrm{Re}(\xi)}{\gamma^2}\right\}$ denotes the fractional part of $\frac{4\mathrm{Re}(\xi)}{\gamma^2}$. Taking logarithms on both sides of the above equation, multiplying by $\gamma^2$ and taking the limit $\gamma\to 0$, we arrive at \eqref{eq:DoubleGammaAsymptote}. The expression for $\mathrm{Re}(\xi)<0$ is obtained similarly. 
\end{proof}
In the next result, we utilize this asymptotics to find semi-classical limit of $\mathcal{A}_{\gamma, P}(\alpha)$.

\begin{proposition}\label{prop:AsymptoticsOfA}
    Consider the scaling $\alpha = \frac{\alpha_0}{\gamma}$, $P= \frac{P_0}{\gamma}$ with $\alpha_0\in (-4,2)$, $P_0\in \mathbb{R}$. As $\gamma \to 0$, the  semi-classical limit of $\mathcal{A}_{\gamma, P}(\alpha)$ has the following expression:
% \begin{align}
%    \mathcal{B}(P_0,\alpha_0):= \lim_{\gamma\to 0}&\gamma^2 \log \mathcal{A}_{\gamma, P}(\alpha)\nonumber \\
%     &= \int^{\mathrm{Re}(1-\alpha_0/2)}_0 \log \Gamma(1-\alpha_0/2-x)dx+  \int^{\mathrm{Re}(1+\alpha_0/2)}_0 \log \Gamma(1+\alpha_0/2-x)dx \nonumber\\ &+ \int^{\mathrm{Re}(1-\alpha_0/2)}_0 \log \Gamma(1-\alpha_0/2-\mathbf{i}P_0-x)dx +\int^{\mathrm{Re}(1-\alpha_0/2)}_0 \log \Gamma(1-\alpha_0/2+\mathbf{i}P_0-x)dx \nonumber\\
%     &+\int^{1}_0 \log \Gamma(1-x)dx+  \int^{\mathrm{Re}(1-\alpha_0)}_0 \log \Gamma(1-\alpha_0-x)dx\nonumber \\ &+ \int^{1}_0 \log \Gamma(1-\mathbf{i}P_0-x)dx +\int^{1}_0 \log \Gamma(1+\mathbf{i}P_0-x)dx. 
% \end{align}
\begin{align}\label{eq:Asymptotics}
  \mathcal{B}(P_0,\alpha_0)&:= \lim_{\gamma\to 0}\gamma^2 \log \mathcal{A}_{\gamma, P}(\alpha)\nonumber \\
  &=\frac{\ii \pi \alpha_0\left(\alpha_0+\ii P_0\right)}{2} + 4\alpha_0 \sqrt{2\pi}-\int_{\mathrm{Re}\left( 1-\frac{\alpha_0}{4}\right)}^{\mathrm{Re}(1)+\frac{\alpha_0}{4}} \log\Gamma\left(1 -x -\frac{\alpha_0}{4}\right) dx \nonumber\\
  &-\int_{\mathrm{Re}\left( 1+\frac{\alpha_0}{4}\right)}^{\mathrm{Re}(1+\frac{\alpha_0}{2})+\frac{\alpha_0}{4}} \log\Gamma\left(1 -x +\frac{\alpha_0}{4}\right) dx-\int_{\mathrm{Re}\left( 1-\frac{\ii P_0}{2}-\frac{\alpha_0}{4}\right)}^{\mathrm{Re}(1-\frac{\ii P_0}{2})+\frac{\alpha_0}{4}} \log\Gamma\left(1 -x-\frac{\ii P_0}{2} -\frac{\alpha_0}{4}\right) dx \nonumber \\
  &-\int_{\mathrm{Re}\left( 1+\frac{\ii P_0}{2}-\frac{\alpha_0}{4}\right)}^{\mathrm{Re}(1+\frac{\ii P_0}{2})+\frac{\alpha_0}{4}} \log\Gamma\left(1 -x+\frac{\ii P_0}{2} -\frac{\alpha_0}{4}\right) dx. 
\end{align}
\end{proposition}
\begin{proof}
With the scaling $\alpha = \alpha_0/\gamma$, $P = P_0/\gamma$, and using the shift relations of the double gamma function as noted in \eqref{dg:id} with $\chi = \frac{\gamma}{2}$, the expression \eqref{A:reexp} simplifies as
\begin{align*}
    &\mathcal{A}_{\gamma, P_0}(\alpha_0)\\
    &= e^{\frac{\ii\pi \alpha_0 \left(\alpha_0+\ii P_0\right)}{2\gamma^2}} \Gamma\left(1-\frac{\gamma^2}{4}\right)^{\frac{\alpha_0}{\gamma^2}} \frac{\Gamma_{\frac{\gamma}{2}}\left(\frac{2}{\gamma}- \frac{\alpha_0}{2\gamma} \right)\Gamma_{\frac{\gamma}{2}}\left(\frac{2}{\gamma}+ \frac{\alpha_0}{2\gamma} \right)\Gamma_{\frac{\gamma}{2}}\left(\frac{2}{\gamma}- \frac{\ii P_0}{\gamma}- \frac{\alpha_0}{2\gamma} \right)\Gamma_{\frac{\gamma}{2}}\left(\frac{2}{\gamma} + \frac{\ii P_0}{\gamma}- \frac{\alpha_0}{2\gamma}\right)}{\Gamma_{\frac{\gamma}{2}} \left( \frac{2}{\gamma}\right)\Gamma_{\frac{\gamma}{2}}\left(\frac{2}{\gamma}+ \frac{\alpha_0}{\gamma} \right)\Gamma_{\frac{\gamma}{2}}\left(\frac{2}{\gamma}- \frac{\ii P_0}{\gamma} \right) \Gamma_{\frac{\gamma}{2}}\left(\frac{2}{\gamma}+ \frac{\ii P_0}{\gamma} \right)} \\
    &\hspace{5cm}\times \frac{\Gamma\left( 1- \frac{\alpha_0}{2}\right) \Gamma\left( 1- \frac{\ii P_0}{2}\right)\Gamma\left( 1+ \frac{\ii P_0}{2}\right)}{\Gamma\left( 1- \frac{\alpha_0}{4}\right)\Gamma\left( 1- \frac{\alpha_0}{4}- \frac{\ii P_0}{2}\right)\Gamma\left( 1- \frac{\alpha_0}{4}+ \frac{\ii P_0}{2}\right)}.
\end{align*}
Taking logarithm of both sides of the above equation, multiplying both sides of the equation with $\gamma$ and letting $\gamma\to 0$ reads
\begin{align}
    &\lim_{\gamma\to 0} \gamma^2\lim\mathcal{A}_{\gamma, P_0}(\alpha_0) \\
    &= \frac{\ii \pi \alpha_0\left(\alpha_0+\ii P_0\right)}{2} 
   + \lim_{\gamma\to 0} \gamma^2 \log \left(  \frac{\Gamma_{\frac{\gamma}{2}}\left(\frac{2}{\gamma}- \frac{\alpha_0}{2\gamma} \right)\Gamma_{\frac{\gamma}{2}}\left(\frac{2}{\gamma}+ \frac{\alpha_0}{2\gamma} \right)\Gamma_{\frac{\gamma}{2}}\left(\frac{2}{\gamma}- \frac{\ii P_0}{\gamma}- \frac{\alpha_0}{2\gamma} \right)\Gamma_{\frac{\gamma}{2}}\left(\frac{2}{\gamma} + \frac{\ii P_0}{\gamma}- \frac{\alpha_0}{2\gamma}\right)}{\Gamma_{\frac{\gamma}{2}} \left( \frac{2}{\gamma}\right)\Gamma_{\frac{\gamma}{2}}\left(\frac{2}{\gamma}+ \frac{\alpha_0}{\gamma} \right)\Gamma_{\frac{\gamma}{2}}\left(\frac{2}{\gamma}- \frac{\ii P_0}{\gamma} \right) \Gamma_{\frac{\gamma}{2}}\left(\frac{2}{\gamma}+ \frac{\ii P_0}{\gamma} \right)}\right).\label{sec6:semiclassA}
\end{align}
Observe that the last term in the expression above can be written as a product of ratios of the form $\frac{\Gamma_{\frac{\gamma}{2}}\left(\frac{2}{\gamma}(\xi - \alpha_0/4) \right)}{\Gamma_{\frac{\gamma}{2}}\left(\frac{2\xi}{\gamma}\right)}$. It is therefore sufficient to study the term 
\begin{align*}
    &\lim_{\gamma\to 0} \gamma^2 \log \Gamma_{\frac{\gamma}{2}}\left(\frac{2}{\gamma}(\xi - \alpha_0/4) \right)- \lim_{\gamma\to 0} \gamma^2 \log \Gamma_{\frac{\gamma}{2}}\left(\frac{2\xi}{\gamma} \right) \\
    &\mathop{=}^{\eqref{eq:DoubleGammaAsymptote}} \alpha_0 \sqrt{2\pi} + \int_{-\infty}^{\mathrm{Re}\left(\xi - \frac{\alpha_0}{4} \right)} \log\Gamma\left(\xi - \frac{\alpha_0}{4}-x \right) dx -\int_{-\infty}^{\mathrm{Re}\left(\xi \right)} \log\Gamma\left(\xi -x \right) dx \\
    % &=\alpha_0 \sqrt{2\pi}+ \int_{-\infty}^{\mathrm{Re}\left(\xi - \frac{\alpha_0}{4} \right)} \log\Gamma\left(\xi - \frac{\alpha_0}{4}-x \right) dx -\int_{-\infty}^{\mathrm{Re}\left(\xi \right)+\frac{\alpha_0}{4}} \log\Gamma\left(\xi -x -\frac{\alpha_0}{4}\right) dx \\
    &=\alpha_0 \sqrt{2\pi}-\int_{\mathrm{Re}\left( \xi-\frac{\alpha_0}{4}\right)}^{\mathrm{Re}(\xi)+\frac{\alpha_0}{4}} \log\Gamma\left(\xi -x -\frac{\alpha_0}{4}\right) dx. 
\end{align*}
Simplifying the expression \eqref{sec6:semiclassA} using the identity above proves the statement of this proposition.

\end{proof}

\begin{theorem}\label{thm:PartitionFunctionLimit}
For any $|q|<r$ for some very small $r>0$, partition function $Z^{\alpha}_{\gamma, P}(q)$ has the following asymptotics as $\gamma$ approaches $0$, 
\begin{align}\label{asymp:semG}
 \lim_{\gamma \rightarrow 0} \gamma^2 \log Z^{\alpha}_{\gamma, P}(q)=  \frac{\alpha_0^2}{8} \log q+\alpha_0^2 \sum_{k=1}^\infty \log(1-q^{2k}) + \mathcal{B}(P_0,\alpha_0), 
% q^{\frac{1}{12}(-\frac{\alpha_0}{2} +\frac{{\alpha_0^2}}{\gamma^2} + 1)}
% \eta(q)^{- \frac{\alpha_0^2}{\gamma^2} - 1 + \frac{\alpha_0}{2}}
\end{align}
or equivalently
\begin{align}
    \lim_{\gamma \rightarrow 0} \gamma^2 \log Z^{\alpha}_{\gamma, P}(q)=  \frac{\alpha_0^2}{24} \log q+\alpha_0^2 \log\eta(q) + \mathcal{B}(P_0,\alpha_0).\label{Zasymp}
\end{align}
where $\alpha_0\in (-4,2)$, $P_0\in \mathbb{R}$, with $\mathcal{B}(P_0,\alpha_0)$ defined in Proposition~\ref{prop:AsymptoticsOfA}.
% \begin{align}
%     \lim_{q\to 0} \lim_{\gamma \rightarrow 0} \gamma^2 \log \mathcal{G}_{\gamma, P}^{\alpha}(q)  =0.
% \end{align}
\end{theorem}

\begin{proof}
Recall from \eqref{eq:ZA} that 
$$  Z^{\alpha}_{\gamma, P} =  q^{\frac{1}{12}(\frac{\alpha \gamma}{2} + \frac{\alpha^2}{2} - 1)}
\eta(q)^{ \alpha^2 + 1 - \frac{\alpha \gamma}{2}} \mathcal{A}_{\gamma, P}(\alpha).$$
Taking logarithm on both sides and multiplying by $\gamma^2$ yields
$$\gamma^2\log Z^{\alpha}_{\gamma, P}  =  \gamma^2\left(\frac{\alpha^2_0}{24\gamma^2}+\frac{1}{12}\big(\frac{\alpha_0}{2} -1\big)\right) \log q +  \gamma^2\left(\frac{\alpha^2_0}{\gamma^2}+\big(\frac{\alpha_0}{2} -1\big)\right) \log \eta(q) + \gamma^2 \log \mathcal{A}_{\gamma, P}(\alpha)$$
where we have used the scaling rule $\alpha = \alpha_0/\gamma$. Finally, using the product representation of the eta function \eqref{def:elleta1}, and letting $\gamma \to 0$ both sides yields the desired conclusion.

   % For $\tau\to \ii \infty$, the equation \eqref{eq:thm41} has the following behaviour
    %\begin{align*}
    %     \lim_{\gamma \rightarrow 0} \gamma^2 \log \mathcal{G}_{\gamma, P}^{\alpha}(q) =\alpha_0^2 \log\left( q^{\frac{1}{12}}\eta(q)^{-1} \right).
    %\end{align*}
    %The above expression along with \eqref{def:elleta1} gives \eqref{asymp:semG}.
\end{proof}
% \rmkH{Substitute values of $\chi$ in the propositions below.}
\begin{theorem}\label{thm:q=0DeformedPsiAtGamma/2}
    For $q\in (0,1)$, as $q\to 0$, the semi-classical deformed conformal block $\lim_{\gamma \to 0}\gamma^2 \log \psi^{\alpha}_{\gamma/2}$ in \eqref{eq:sc_Lame_limit} behaves as 
    % \rmkH{in \eqref{thm43:phi} behaves as }
\begin{align}\label{propasymp:g2}
\begin{split}
  & \lim_{\gamma\rightarrow 0} \gamma^2 \log\left(\theta_{1}(z)^{\frac{\chi}{4}\left(\chi-\alpha \right)} \psi^{\alpha_0/\gamma}_{\gamma/2}(z,q)\right)\\
   &= -\alpha _0 \left(\alpha _0+4\right)\log(2\pi)+\left( \frac{P_0^2}{2}-\frac{\alpha_0 (\alpha _0+2)}{12 }\right)\log q +\mathcal{B}(P_0,\alpha_0), 
   \end{split}
   % \nonumber\\
   % &+\int^{\mathrm{Re}(1-\alpha_0/2)}_0 \log \Gamma(1-\alpha_0/2-x)dx+  \int^{\mathrm{Re}(1+\alpha_0/2)}_0 \log \Gamma(1+\alpha_0/2-x)dx \nonumber\\ &+ \int^{\mathrm{Re}(1-\alpha_0/2)}_0 \log \Gamma(1-\alpha_0/2-\mathbf{i}P_0-x)dx +\int^{\mathrm{Re}(1-\alpha_0/2)}_0 \log \Gamma(1-\alpha_0/2+\mathbf{i}P_0-x)dx \nonumber\\
   %  &+\int^{1}_0 \log \Gamma(1-x)dx+  \int^{\mathrm{Re}(1-\alpha_0)}_0 \log \Gamma(1-\alpha_0-x)dx\nonumber \\ &+ \int^{1}_0 \log \Gamma(1-\mathbf{i}P_0-x)dx +\int^{1}_0 \log \Gamma(1+\mathbf{i}P_0-x)dx 
   \end{align}
   where $\alpha_0\in (-4,2)$, $P_0\in \mathbb{R}$, and $\mathcal{B}(P_0,\alpha_0)$ is given by \eqref{eq:Asymptotics}. 
\end{theorem}
% \rmkH{Check the factor of $\log q$. I get $(-\alpha_0 - \alpha_0^2/6)$.}
\begin{proof}
From \eqref{asymp1:prefac} and \eqref{Thm43:V}, we observe that, near $\gamma \to 0$, the deformed conformal block behaves as 
\begin{align}
  \psi^{\alpha_0/\gamma}_{\gamma/2}(z,q)=   q^{\frac{P_0^2}{2\gamma^2}-\frac{\alpha_0^2}{24\gamma^2}} \theta_1'(0)^{-\frac{\alpha_0(4+\alpha_0)}{6 \gamma^2}} \EE\left[\left(   \int^{1}_{0} \theta_1(x)^{-\frac{\alpha_0}{2}}e^{\pi P_0x} e^{\frac{\gamma}{2} Y(x; q)} dx \right)^{-\frac{\alpha_0}{\gamma^2}}\right]. 
\end{align}
% \rmkH{From \eqref{asymp1:prefac} and \eqref{thm43proof:V}, }
%2\pi q^{1/4}
Now, using the expressions for theta functions \eqref{def:elltheta1} and \eqref{def:ell-eta}, we can obtain the following behaviour near $q\to 0$
    \begin{align}
        \lim_{\gamma\rightarrow 0} \gamma^2 \log\left(\theta_{1}(z)^{\frac{\gamma}{4}\left(\frac{\gamma}{2}-\frac{\alpha_0}{\gamma} \right)}\psi^{\alpha_0/\gamma}_{\gamma/2}(z,q)\right)
        % & = -\frac{\alpha _0 \left(\alpha _0+4\right)}{6}\log(-2\pi)+ \frac{6 P_0^2-\alpha _0 \left(\alpha _0+2\right)}{12 }\log q + \frac{\alpha_0^2}{8}\log q \nonumber\\
        % &+ \lim_{\gamma\to 0} \gamma^2 \log \EE\left[\left( q^{\frac{\alpha_0}{8}}  \int^{1}_{0} \theta_1(x)^{-\frac{\alpha_0}{2}}e^{\pi P_0x} e^{\frac{\gamma}{2} Y(x; q)} dx \right)^{-\frac{\alpha_0}{\gamma^2}}\right] \nonumber\\
        & = -\alpha _0 \left(\alpha _0+4\right)\log(2\pi)+\left( \frac{P_0^2}{2}-\frac{\alpha_0 (\alpha _0+2)}{12 }\right)\log q + \nonumber\\
        &+ \lim_{\gamma\to 0} \gamma^2 \log \EE\left[\left( q^{\frac{\alpha_0}{8}}  \int^{1}_{0} \theta_1(x)^{-\frac{\alpha_0}{2}}e^{\pi P_0x} e^{\frac{\gamma}{2} Y(x; q)} dx \right)^{-\frac{\alpha_0}{\gamma^2}}\right]. \label{eq:PsiLimitAtq=0}
    \end{align}
    % for the semicalssical limit from \eqref{eq:q-block} and \eqref{asymp1:prefac}. 
    As a final step, we note that the $\gamma \to 0$ and $q\to 0$ limit of the following     
    $$ \mathcal{T} := \gamma^2 \log \EE\left[\left( q^{\frac{\alpha_0}{8}}  \int^{1}_{0} \theta_1(x)^{-\frac{\alpha_0}{2}}e^{\pi P_0x} e^{\frac{\gamma}{2} Y(x; q)} dx \right)^{-\frac{\alpha_0}{\gamma^2}}\right]$$ commute. To see this, using Girsanov's theorem, we write the above expression as (see \cite[Eq. 2.225]{ghosal2020probabilistic})
    \begin{align}
        \mathcal{T} = (\alpha^2_0- \alpha_0\gamma/2)\log\left(q^{-1/12} \eta(q)  \right)+ \gamma^2\log \EE\left[e^{\frac{\alpha_0}{2 \gamma} F(0;q)} \mathcal{Q}(q) \left(\int_0^1 (2\sin \pi x)^{-\alpha_0/2} e^{\pi P_0 x} e^{\frac{\gamma}{2} Y(x)} dx \right)^{-\frac{\alpha_0}{\gamma^2}} \right],
    \end{align}
    where $F(0;q)$ and $\mathcal{Q}(q)$ are defined as in \eqref{defs:Girsanov}. We can then decompose the above expression into $q$-dependent and $q$-independent terms. One can then repeat the arguments in the proof of Proposition~\ref{prop:limit_commutativity} to prove that $\lim_{q\to 0}\lim_{\gamma \to 0} \mathcal {T}= \lim_{q\to 0}\lim_{\gamma \to 0} \mathcal {T}$. Consequently, we get
    $$\lim_{q\to 0}\lim_{\gamma \to 0} \mathcal {T} = \lim_{\gamma \to 0} \gamma^2 \log \mathcal{A}_{\gamma, P}(\alpha) = \mathcal{B}(P_0, \alpha_0).$$
    Taking $q\to 0$ on both sides of \eqref{eq:PsiLimitAtq=0} and substituting $\mathcal{B}(P_0, \alpha_0)$ in place of $\lim_{q\to 0}\lim_{\gamma \to 0} \mathcal {T}$ gives us \eqref{propasymp:g2}.

\end{proof}
\begin{theorem}\label{prop:incond_HJ}
 Consider $z = {\sf a} \tau + {\sf b}$ for some ${\sf a}, {\sf b} \in (0,1)$, ${\sf a} \neq \mathbb{Z} + \frac{1}{2}$.   As $\tau\to \mathbf{i} \infty$, the semi-classical deformed conformal block $\lim_{\gamma \to 0}\gamma^2 \log \psi^{\alpha_0/\gamma}_{2/\gamma}(z,q)$ in \eqref{sem:2gamma} has the following behaviour with $\alpha_0\in (-2,2)$, $P_0\in
 \mathbb{R}$:
% \begin{align}
%    \lim_{\gamma\rightarrow 0} \gamma^2 \log\left(\theta_{1}(z)^{\frac{\chi}{2}\left(\chi-\alpha \right)}\psi^{\alpha_0/\gamma}_{2/\gamma}(z,q)\right)&= \ii \pi \left(\frac{\alpha_0(\alpha_0-4)}{26}+\frac{P_0^2}{2}-\ii 2 P_0 {\sf a} \pi +\frac{1}{6}\right)\log q+  2P_0\pi{\sf b}\nonumber\\
%    &+ \frac{-\alpha _0^2+\alpha _0+6 P_0^2+2}{12}\log(-2\pi) +
%    + 2(1-\alpha_0)\pi\mathbf{i}{\sf a}{\sf b} \nonumber\\
%     &+ \int^{\mathrm{Re}(1-(\alpha_0-1)/2)}_0 \log \Gamma(1-(1-\alpha_0)/2-x)dx\nonumber \\&+  \int^{\mathrm{Re}(1+(\alpha_0-1)/2)}_0 \log \Gamma(1+(\alpha_0-1)/2-x)dx \nonumber\\ &+ \int^{\mathrm{Re}(1-(\alpha_0-1)/2+2{\sf a} )}_0 \log \Gamma(1-(\alpha_0-1)/2-\mathbf{i}P_0+2{\sf a} -x)dx \nonumber\\&+\int^{\mathrm{Re}(1-(\alpha_0-1)/2-2{\sf a} )}_0 \log \Gamma(1-(\alpha_0-1)/2+\mathbf{i}P_0-2{\sf a} -x)dx \nonumber\\
%     &+\int^{1}_0 \log \Gamma(1-x)dx+  \int^{\mathrm{Re}(1-(\alpha_0-1))}_0 \log \Gamma(1-(\alpha_0-1)-x)dx\nonumber \\ &+ \int^{1+2{\sf a} }_0 \log \Gamma(1-\mathbf{i}P_0+2{\sf a} -x)dx +\int^{1-2{\sf a} }_0 \log \Gamma(1+\mathbf{i}P_0-2{\sf a} -x)dx, \label{eqpropo:asymp2g}
% \end{align}
\begin{align}
   &\lim_{\gamma\rightarrow 0} \gamma^2 \log\left(\theta_{1}(z)^{\frac{\chi}{2}\left(\chi-\alpha \right)}\psi^{\alpha_0/\gamma}_{2/\gamma}(z,q)\right)\\
   &=  C(\tau) +2\pi P_0 (z+1) +\int_{\mathrm{Re}\left( \frac{2-3\alpha_0}{4}\right)}^{\mathrm{Re}(-\frac{\alpha_0}{2})+\frac{\alpha_0-2}{4}} \log\Gamma\left( -x + \frac{2-3\alpha_0}{4}\right) dx \\
   &+\int_{\mathrm{Re}\left(  \frac{2+\alpha_0}{4}\right)}^{\mathrm{Re}(\frac{\alpha_0}{2})+\frac{\alpha_0-2}{4}} \log\Gamma\left( -x + \frac{2+\alpha_0}{4}\right) dx\nonumber-\int_{\mathrm{Re}\left( \frac{\ii P_0}{2}+  \frac{2-\alpha_0}{2}\right)}^{\mathrm{Re}(\frac{\ii P_0}{2}+ \frac{2-\alpha_0}{4})+\frac{\alpha_0-2}{4}} \log\Gamma\left(\frac{\ii P_0}{2} -x + \frac{2-\alpha_0}{2}\right) dx \\
   &-\int_{\mathrm{Re}\left( -\frac{\ii P_0}{2}\right)}^{\mathrm{Re}(-\frac{\ii P_0}{2}-\frac{2-\alpha_0}{4})+\frac{\alpha_0-2}{4}} \log\Gamma\left(-\frac{\ii P_0}{2} -x \right) dx, \label{eqpropo:asymp2g}
\end{align}
where
\begin{align}\label{def:Ctau}
      \mathcal{C}(\tau) := \frac{(2-\alpha_0)(4+\alpha_0)}{6}\log(2\pi) &+\ii \pi(2-\alpha_0) +\left(\frac{(\alpha_0-2)(\alpha_0-8)}{24} \right) \log q - 2\pi \ii (2-\alpha_0)({\sf a} \tau+ {\sf b}).
      % - 4 P_0 ({\sf a}\tau+{\sf b}).
\end{align}

\end{theorem}
%\begin{align}
%    \Gamma(x) \Gamma(1-x) = \frac{\pi}{\sin \pi x}
%\end{align}
\begin{proof}
Recalling the expressions \eqref{eq:prefactor_final} and \eqref{eq:V_scaled}, we obtain the following behaviour for the deformed conformal block in the case $\chi = \frac{2}{\gamma}$:
\begin{align}\label{5.8}
\begin{split}
    &\lim_{\gamma\rightarrow 0} \gamma^2 \log\left(\theta_{1}(z)^{\frac{\chi}{2}\left(\chi-\alpha \right)}\psi^{\alpha_0/\gamma}_{2/\gamma}(z,q) \right) \\
    % &= q^{\frac{P_0^2}{2\gamma^2} - \frac{(2-\alpha_0)^2}{24\gamma^2} } (2\pi)^{\frac{(2-\alpha_0)(4+\alpha_0)}{6\gamma^2}} ( q)^{\frac{(2-\alpha_0)(4+\alpha_0)}{24\gamma^2}}e^{\frac{2P_0 z \pi}{\gamma^2}}\\
    &= \frac{(2-\alpha_0)(4+\alpha_0)}{6}\log(2\pi) +\left(\frac{(\alpha_0-2)(\alpha_0-8)}{24} \right) \log q
    % +\left(\frac{\alpha_0(\alpha_0-4)}{26}+\frac{P_0^2}{2}+\frac{1}{6}\right)\log q
     + 2 P_0 z \pi\\
     &+ \lim_{\gamma\to 0} \gamma^2 \log \EE\left[\left( q^{\frac{\alpha_0-2}{8}} \int_0^1 \theta_{1}(z+x) \theta_{1}(x)^{-\frac{\alpha_0}{2}} e^{\pi P_0 x} e^{\frac{\gamma}{2} Y_\tau(x)} dx \right)^{\frac{(2-\alpha_0)}{\gamma^2}}\right].
     \end{split}
\end{align}
% \begin{align}
%     (-2\pi)^{\frac{1}{24} \left(2-\alpha _0\right)+\frac{-\alpha _0^2+\alpha _0+6 P_0^2+2}{12 \gamma ^2}} q^{\frac{1}{24} \left(2-\alpha _0\right)+\frac{-\alpha _0^2+\alpha _0+6 P_0^2+2}{12 \gamma ^2}} e^{\frac{2 P_0 z \pi}{\gamma^2}}+\mathcal{O}(q^{\gamma^2})
% \end{align}
We identify $z$  with the solution of the non-autonomous elliptic Calogero-Moser model $u(\tau)$ by letting $z=2u(\tau)$. As $\tau\to \mathbf{i}\infty$, the variable $u(\tau)= {\sf a}\tau+{\sf b}$, with ${\sf a}, {\sf b}$ parameterization the A, B-cycle monodromies of the linear system associated to the non-autonomous elliptic Calogero-Moser model (see Appendix B). From \eqref{def:elltheta1}, we then obtain the following behaviour of the theta function in the same limit:
\begin{align*}
    \theta_1(z) \equiv \theta_1(2u(\tau)) &= - 2q^{1/4} \sin( 2\pi u(\tau)) \prod_{k = 1}^\infty (1 - q^{2k}) (1 - 2 \cos(4\pi u(\tau)) q^{2k} + q^{4k}) \\
    &=\ii q^{1/4} e^{-2\ii \pi ({\sf a} \tau + {\sf b})} \left(1+ \mathcal{O}(q^2) \right),
\end{align*}
and resultantly, as $\tau\to \ii \infty$,
\begin{align*}
  \theta_{1}(z+x) = \ii q^{1/4} e^{-2\ii \pi ({\sf a} \tau + {\sf b})-\ii \pi x} \left(1+ \mathcal{O}(q^2) \right).
\end{align*}
 
Plugging the above expressions into the right hand side of \eqref{5.8} shows that (see also \cite[Lemma 5.3]{ghosal2020probabilistic})

\begin{align}
\lim_{\tau\to \ii \infty} &\Big[\lim_{\gamma\to 0}\gamma^2 \log \left(\theta_{1}(z)^{\frac{\chi}{2}\left(\chi-\alpha \right)}\psi^{\alpha_0/\gamma}_{2/\gamma}(z,q) \right) -\mathcal{C}(\tau)   \Big] \nonumber \\
&=\lim_{q\to 0}\lim_{\gamma\to 0}\gamma^2\log \EE\left[\left( q^{\frac{\alpha_0}{8}}\int_0^1  \theta_{1}(x)^{-\frac{\alpha_0}{2}} e^{\pi x(P_0-\mathbf{i})} e^{\frac{\gamma}{2} Y(x; q)} dx \right)^{\frac{(2-\alpha_0)}{\gamma^2}}\right], \label{sec6:intheend}
% =  \mathcal{A}_{\gamma, P}\left(\frac{(\alpha_0-2)}{\gamma}\right)
\end{align}
where $C(\tau)$ is defined in \eqref{def:Ctau}. Due to Proposition~\ref{prop:limit_commutativity}, the limits $\lim_{q\to 0}$ and $\lim_{\gamma\to 0}$ commute. So, we first begin with the $q\to 0$ limit
\begin{align}
\lim_{q\to 0}   \mathbb{E} \left[\left( q^{\frac{\alpha_0}{8}}\int_0^1  \theta_{1}(x)^{-\frac{\alpha_0}{2}} e^{\pi x(P_0-\mathbf{i})} e^{\frac{\gamma}{2} Y(x; q)} dx \right)^{\frac{(2-\alpha_0)}{\gamma^2}} \right] = \mathbb{E} \left[\left( \int_0^1  (\sin \pi x)^{-\frac{\alpha_0}{2}} e^{\pi x(P_0-\mathbf{i})} e^{\frac{\gamma}{2} {Y(x)}} dx \right)^{\frac{(2-\alpha_0)}{\gamma^2}} \right].
\end{align}
We now note the following identity \cite{ang2023derivation}
\begin{align}
    % &\mathbb{E} \left[\left( q^{\frac{\alpha_0}{8}}\int_0^1  \theta_{1}(x)^{-\frac{\alpha_0}{2}} e^{\pi x(P_0-\mathbf{i})} e^{\frac{\gamma}{2} Y(x; q)} dx \right)^{\frac{(2-\alpha_0)}{\gamma^2}} \right] \nonumber \\
    & \mathbb{E} \left[\left( \int_0^1  (\sin \pi x)^{-\frac{\alpha_0}{2}} e^{\pi x(P_0-\mathbf{i})} e^{\frac{\gamma}{2} {Y(x)}} dx \right)^{\frac{(2-\alpha_0)}{\gamma^2}} \right]\\
    &= e^{2\pi P_0/\gamma^2} \frac{\Gamma\left( \frac{\alpha_0}{4} - \frac{\gamma^2}{4}\right) \Gamma_{\frac{\gamma}{2}}\left(\frac{\alpha_0}{2\gamma} \right) \Gamma_{\frac{\gamma}{2}}\left(-\frac{\alpha_0}{2\gamma} \right) \Gamma_{\frac{\gamma}{2}}\left(\frac{\ii P_0}{\gamma}+\frac{2-\alpha_0}{\gamma} \right) \Gamma_{\frac{\gamma}{2}}\left(-\frac{\ii P_0}{\gamma} \right)}{\Gamma\left(1 - \frac{\gamma^2}{4}\right)^{\frac{2}{\gamma}\left(\frac{2-\alpha_0}{2\gamma} \right)} \Gamma_{\frac{\gamma}{2}}\left(\frac{1}{\gamma} \right) \Gamma_{\frac{\gamma}{2}}\left(\frac{1-\alpha_0}{\gamma} \right) \Gamma_{\frac{\gamma}{2}}\left(\frac{\ii P_0}{\gamma}+\frac{2-\alpha_0}{2\gamma} \right) \Gamma_{\frac{\gamma}{2}}\left(-\frac{\ii P_0}{\gamma}-\frac{2-\alpha_0}{2\gamma} \right)}.\label{sec6:this}
\end{align}
To simplify the above equation, we first note that the expression above can be written in terms of ratios of the form $$\frac{\Gamma_{\frac{\gamma}{2}}\left(\frac{2}{\gamma}\left(\xi + \frac{2-\alpha_0}{4} \right)\right)}{\Gamma_{\frac{\gamma}{2}}\left(\frac{2}{\gamma}\xi\right)}.$$ We then use the identity
\begin{align*}
   & \lim_{\gamma\to 0} \gamma^2 \log \Gamma_{\frac{\gamma}{2}}\left(\frac{2}{\gamma}(\xi + \frac{2-\alpha_0}{4}) \right)- \lim_{\gamma\to 0} \gamma^2 \log \Gamma_{\frac{\gamma}{2}}\left(\frac{2\xi}{\gamma} \right) \\&\mathop{=}^{\eqref{eq:DoubleGammaAsymptote}} \left(\alpha_0-2\right) \sqrt{2\pi} + \int_{-\infty}^{\mathrm{Re}\left(\xi + \frac{2-\alpha_0}{4} \right)} \log\Gamma\left(\xi + \frac{2-\alpha_0}{4}-x \right) dx -\int_{-\infty}^{\mathrm{Re}\left(\xi \right)} \log\Gamma\left(\xi -x \right) dx \\
    % &=\alpha_0 \sqrt{2\pi}+ \int_{-\infty}^{\mathrm{Re}\left(\xi - \frac{\alpha_0}{4} \right)} \log\Gamma\left(\xi - \frac{\alpha_0}{4}-x \right) dx -\int_{-\infty}^{\mathrm{Re}\left(\xi \right)+\frac{\alpha_0}{4}} \log\Gamma\left(\xi -x -\frac{\alpha_0}{4}\right) dx \\
    &=\left(\alpha_0-2\right) \sqrt{2\pi}-\int_{\mathrm{Re}\left( \xi+ \frac{2-\alpha_0}{4}\right)}^{\mathrm{Re}(\xi)+\frac{\alpha_0-2}{4}} \log\Gamma\left(\xi -x + \frac{2-\alpha_0}{4}\right) dx. 
\end{align*}
 Then, we repeat the procedure used in the proof of Proposition~\ref{prop:AsymptoticsOfA}. One can obtain similar expressions to the one above for the remaining terms in \eqref{sec6:this}. Substituting them in \eqref{sec6:intheend}, we obtain the expression
% we use the expression \eqref{eq:Asymptotics} by noting that the range of $\alpha_0\in [0,\infty)$, as specified in Remark \ref{Remark:def_CB} to obtain the following expression
% \rmkH{The below expression changes}
\begin{align}
    &\lim_{q\to 0}\lim_{\gamma\to 0}\gamma^2\log \EE\left[\left( q^{\frac{\alpha_0}{8}}\int_0^1  \theta_{1}(x)^{-\frac{\alpha_0}{2}} e^{\pi x(P_0-\ii)} e^{\frac{\gamma}{2} Y(x)} dx \right)^{\frac{(2-\alpha_0)}{\gamma^2}}\right]\nonumber \\
 %    &\mathop{=}^{\eqref{A:reexp}, \eqref{eq:Asymptotics}} 2\pi P_0 +\int_{\mathrm{Re}\left( -\frac{\alpha_0}{2} + \frac{2-\alpha_0}{4}\right)}^{\mathrm{Re}(-\frac{\alpha_0}{2})+\frac{\alpha_0-2}{4}} \log\Gamma\left(-\frac{\alpha_0}{2} -x + \frac{2-\alpha_0}{4}\right) dx\nonumber\\
 %  &+\int_{\mathrm{Re}\left( \frac{\alpha_0}{2}+ \frac{2-\alpha_0}{4}\right)}^{\mathrm{Re}(\frac{\alpha_0}{2})+\frac{\alpha_0-2}{4}} \log\Gamma\left(\frac{\alpha_0}{2} -x + \frac{2-\alpha_0}{4}\right) dx-\int_{\mathrm{Re}\left( \frac{\ii P_0}{2}+ \frac{2-\alpha_0}{4}+ \frac{2-\alpha_0}{4}\right)}^{\mathrm{Re}(\frac{\ii P_0}{2}+ \frac{2-\alpha_0}{4})+\frac{\alpha_0-2}{4}} \log\Gamma\left(\frac{\ii P_0}{2}+ \frac{2-\alpha_0}{4} -x + \frac{2-\alpha_0}{4}\right) dx\nonumber\\
 % & -\int_{\mathrm{Re}\left( -\frac{\ii P_0}{2}- \frac{2-\alpha_0}{4}+ \frac{2-\alpha_0}{4}\right)}^{\mathrm{Re}(-\frac{\ii P_0}{2}-\frac{2-\alpha_0}{4})+\frac{\alpha_0-2}{4}} \log\Gamma\left(-\frac{\ii P_0}{2}- \frac{2-\alpha_0}{4} -x + \frac{2-\alpha_0}{4}\right) dx. 
     &\mathop{=}^{\eqref{A:reexp}, \eqref{eq:Asymptotics}} 2\pi P_0 +\int_{\mathrm{Re}\left( \frac{2-3\alpha_0}{4}\right)}^{\mathrm{Re}(-\frac{\alpha_0}{2})+\frac{\alpha_0-2}{4}} \log\Gamma\left( -x + \frac{2-3\alpha_0}{4}\right) dx\nonumber\\
  &+\int_{\mathrm{Re}\left(  \frac{2+\alpha_0}{4}\right)}^{\mathrm{Re}(\frac{\alpha_0}{2})+\frac{\alpha_0-2}{4}} \log\Gamma\left( -x + \frac{2+\alpha_0}{4}\right) dx-\int_{\mathrm{Re}\left( \frac{\ii P_0}{2}+  \frac{2-\alpha_0}{2}\right)}^{\mathrm{Re}(\frac{\ii P_0}{2}+ \frac{2-\alpha_0}{4})+\frac{\alpha_0-2}{4}} \log\Gamma\left(\frac{\ii P_0}{2} -x + \frac{2-\alpha_0}{2}\right) dx\nonumber\\
 & -\int_{\mathrm{Re}\left( -\frac{\ii P_0}{2}\right)}^{\mathrm{Re}(-\frac{\ii P_0}{2}-\frac{2-\alpha_0}{4})+\frac{\alpha_0-2}{4}} \log\Gamma\left(-\frac{\ii P_0}{2} -x \right) dx,
\end{align}
completing the proof of this Theorem.

% Taking the $q\to 0 $ limit of \eqref{5.8} and substituting the above expression gives \eqref{eqpropo:asymp2g}.
\end{proof}

%\newpage

\appendix
\appto\appendix{\addtocontents{toc}{\protect\setcounter{tocdepth}{0}}}
\section{Useful notation and elliptic functions}\label{app:ell_func}
The parameter $\gamma\in(0,2)$ defines the central charge of the CFT as 
\begin{align}\label{def:cQ}
    c = 1 + 6 Q^2, && Q := \frac{\gamma}{2} + \frac{2}{\gamma}.
\end{align}
We use two other parameters
\begin{align}\label{def:lchi}
    l_{\chi} = \frac{\chi^2}{2} - \frac{\alpha \chi}{2}, && \chi = \left\lbrace\frac{2}{\gamma}, \frac{\gamma}{2} \right\rbrace.
\end{align}

The modular parameter $\tau\in \mathbb{H}$ and the nome is defined as $q:= e^{\ii \pi \tau}$. The elliptic theta functions $\theta_i(z)=\theta_i(z|\tau)$ are defined as \cite[20.5(i)]{DLMF}
\begin{align}
\theta_1(z) &=  2q^{1/4} \sin(\pi z) \prod_{k = 1}^\infty (1 - q^{2k}) (1 - 2 \cos(2\pi z) q^{2k} + q^{4k}), \label{def:elltheta1} \\
%     \theta_{2}(z) &= 2 \sum_{n=0}^{\infty} q^{(n+1/2)^2} \cos\left((2n+1) \pi z  \right), \\
%   \theta_{3}(z) &= 1 + 2 \sum_{n=1}^{\infty} q^{n^2} \cos(2 n \pi z), \\
% \theta_{1}'(0) &= - 2\pi q^{1/4} \prod_{k=1}^{\infty} \left(1-q^{2k} \right)^3 = - 2\pi \eta(q)^3 \label{def:theta1p0},
 \theta_{2}(z) &=2q^{1/4} \cos(\pi z) \prod_{k = 1}^\infty (1 - q^{2k}) (1 + 2 \cos(2\pi z) q^{2k} + q^{4k})\label{def:elltheta2}\\
    % &= 2 \sum_{n=0}^{\infty} q^{(n+1/2)^2} \cos\left((2n+1) \pi z  \right), \label{def:elltheta2}\\
  \theta_{3}(z)&=\prod_{k = 1}^\infty (1 - q^{2k}) (1 + 2 \cos(2\pi z) q^{2k-1} + q^{4k})\label{def:elltheta3}\\
  % \theta_{4}(z)&=\prod_{k = 1}^\infty (1 - q^{2k}) (1 - 2 \cos(2\pi z) q^{2k-1} + q^{4k})\\
  % &= 1 + 2 \sum_{n=1}^{\infty} q^{n^2} \cos(2 n \pi z),\label{def:elltheta3} \\
\partial_z \theta_{1}(z)\vert_{z=0}= \theta_{1}'(0) &=  2\pi q^{1/4} \prod_{k=1}^{\infty} \left(1-q^{2k} \right)^3 =  2\pi \eta(q)^3 \label{def:ell-eta}.
% \label{def:theta1p0}
 \end{align}
They have the following periodicity properties
 \begin{align}
     \theta_1(z+1)= - \theta_1(z), && \theta_1(z+\tau) = -q^{-1} e^{-2\pi i z} \theta_1(z).
 \end{align}
 The Dedekind $\eta$ function is defined as
 \begin{align}
      \eta(q) &:= q^{\frac{1}{12}} \prod_{k = 1}^\infty (1 - q^{2k}), \label{def:elleta1}
 \end{align}
and its logarithmic derivative
 \begin{align}
     \eta_1(\tau) = - 2\pi i \partial_{\tau} \log \eta(\tau)=  -\frac{1}{6} \frac{\theta_1'''(0)}{\theta_1'(0)}. \label{def:ell-eta-p}
 \end{align}
 The above expression is a direct consequence of the fact that  the theta function $\theta_1(z)$ solves the heat equation $\partial_{\tau} \theta_1(z) = \frac{1}{4\pi \ii } \theta_1''(z)$.

The Weierstrass cubic is
\begin{align}\label{eq:WeierCubic}
    \left(\wp'(z)\right)^2 = 4 \wp^3(z) - g_2 \wp(z)- g_3.
    % = 4 \left(\wp(z)- e_1 \right)\left(\wp(z)- e_2 \right)\left(\wp(z)- e_3 \right).
\end{align}
% where 
% \begin{align}
%     e_1 = \wp\left( w_1\right), &&   e_2 = \wp\left( w_2\right), &&   e_3 = \wp\left( w_3\right).
% \end{align}
% \begin{align}
%     e_1 = \wp\left( \frac{1}{2}\right), &&   e_2 = \wp\left( \frac{\tau}{2}\right), &&   e_3 = \wp\left( \frac{1+\tau}{2}\right).
% \end{align*}
The Weierstrass $\wp$-function solves the cubic above and is given by \cite[23.2(ii)]{DLMF}
\begin{align}\label{eq:wp}
    \wp(z) = \frac{1}{z^2} + \sum_{w\in \mathbb{L}\backslash \{0\}} \left( \frac{1}{(z-w)^2} - \frac{1}{w^2}\right),
\end{align}
where $\mathbb{L}$ denotes the {torus (doubly periodic lattice)}. The above function has a double pole at $z=0$, and is doubly periodic
\begin{align}
    \wp(z+1) =\wp(z), && \wp(z+\tau) =\wp(z).
\end{align}
% and has a double pole at zero
% \begin{align}
%     \lim_{z\to 0} z^2 \wp(z) =1.
% \end{align}

\section{Non-autonomous elliptic Calogero-Moser model}\label{App:NAECM}

% \subsection{Calogero equation and associated monodromy data}
% We now introduce the integrable equations entering our description. 
\subsection{Hamiltonian and the action}
The 2-particle nonautonomous elliptic Calogero-Moser system is represented by the Hamiltonian
% \footnote{Note that the Hamiltonian below differs by the one in \cite{DDG2020} by a normalization factor of $2m^2 \eta_1(\tau)$. }
\begin{align}\label{Ham:CM}
   H(\tau) =v(\tau)^2-m^2\wp(2u(\tau)|\tau),
   % H(\tau) =v(\tau)^2-m^2\wp(2u(\tau)|\tau)-2m^2 \eta_{1}(\tau),
\end{align}
and the equations of motion read
\begin{align}
  v(\tau)= 2\pi \ii \frac{d u(\tau)}{d\tau}, &&   (2\pi \ii )\frac{d v(\tau)}{d\tau} = m^2 \wp'(2u\vert \tau). \label{eq:CM_ELLDER}
\end{align}
The above equations of motion give the elliptic form of the Painlev\'e VI equation for special values of the parameters (see \cite[Chapter 1]{desiraju2021thesis}, \cite{DDG2020,Takasaki:2000zd}).

The equation \eqref{eq:CM_ELLDER} is {\it integrable} in the sense that it has an associated {\it Lax pair}
\begin{gather}
\partial_z Y(z,\tau):= L(z,\tau) Y(z,\tau) =\left(\begin{array}{cc}
        v(\tau) & mx(2u(\tau),z) \\
        mx(-2u(\tau),z) & -v(\tau)
    \end{array}\right)Y(z,\tau),  \\
2\pi \ii \partial_\tau Y(z,\tau):=M(z,\tau)Y(z,\tau) =m\left(\begin{array}{cc} 
        0 & y(2u(\tau),z) \\
        y(-2u(\tau),z) & 0
    \end{array}\right)Y(z,\tau) ,\label{eq:linear_systemCM}
\end{gather}
{and \eqref{eq:CM_ELLDER} is given by the condition
\begin{align}
    2\pi \ii \partial_z\partial_{\tau} Y(z,\tau) =  2\pi \ii \partial_z\partial_{\tau} Y(z,\tau) \Rightarrow \partial_{\tau} L(z,\tau) - \partial_z M(z,\tau) + [L(z,\tau), M(z,\tau)]=0.
\end{align}}
% \begin{equation}
% \begin{split}
%     L_{CM}(z,\tau) &=\left(\begin{array}{cc}
%         v(\tau) & mx(2u(\tau),z) \\
%         mx(-2u(\tau),z) & -v(\tau)
%     \end{array}\right),\\
%     M_{CM}(z,  \tau) &=m\left(\begin{array}{cc} 
%         0 & y(2u(\tau),z) \\
%         y(-2u(\tau),z) & 0
%     \end{array}\right).
%     \end{split}\label{linear_system}
% \end{equation}
In the above expressions, the functions $x(\xi,z)$, $y(\xi,z)$ \eqref{eq:linear_systemCM}  are defined as
\begin{align}\label{eq:vardef}
    x(\xi,z):=\frac{\theta_1(z-\xi \vert \tau)\theta_1'(0\vert \tau)}{\theta_1(z\vert \tau)\theta_1(\xi \vert \tau)}, && y(\xi,z):=\partial_\xi x(\xi,z).
\end{align}

% which is the elliptic form of Painlev\'e VI for special value of the parameters, and arises as the consistency condition of the Lax pair \eqref{eq:linear_systemCM}. 
Furthermore, the Hamiltonian \eqref{Ham:CM} is 
% of the system \eqref{linear_system} 
given by the A-cycle contour integral 
\begin{equation} \label{eq:Ham_CM}
    H(\tau)=\oint_A dz\frac{1}{2}\tr L^2(z, \tau)= v(\tau)^2 - m^2 \wp\left(2u(\tau)\right).
    % =v(\tau)^2-m^2\wp(2u(\tau)|\tau)-2m^2 \eta_{1}(\tau),
\end{equation}
The functions $(u(\tau), v(\tau))$ are the canonical coordinates of the system and they describe the position and momentum of a particle in an elliptic potential respectively. 
% {\color{red} Check the role of $\eta_{1}(\tau)$}
% where $ \eta_{1} (\tau) = - 2\pi \ii \partial_{t} \log \eta (\tau)$, and $\eta(\tau)$ is Dedekind's eta function
% \begin{equation}
%     \eta(\tau):=\left(\frac{\theta_1'(0\vert \tau)}{2\pi} \right)^{1/3}.
% \end{equation}
Therefore, the Lagrangian of the above system is given by
\begin{align*}
    \mathcal{L} = v(\tau)^2 + m^2 \wp\left(2u(\tau)\right),
\end{align*}
and the action functional is defined as the integral of the Lagrangian
\begin{gather}
   S(\tau)  := \int^{\tau} \left(v(\tau')^2+m^2\wp(2u(\tau'))\right) \frac{d\tau'}{2\pi \ii}. 
   % \widetilde{\phi}(2u(\tau),\tau)  = \int^{\tau} \left(v(\tau)^2+m^2\wp(2u(\tau)|\tau)+2m^2 \eta_{1}(\tau)\right) d\tau. 
\end{gather}
It is a known fact from classical mechanics that the action solves the Hamilton-Jacobi equation
\begin{align*}
    2\pi \ii\partial_{\tau} S = \left(\partial_{u(\tau)} S\right)^2 -m^2 \wp(2u(\tau)) S.
\end{align*}

\subsection{Monodromy data}
% We now introduce the monodromy problem of the 2-particle NAECM model.
% As opposed the behaviour of the Lax matrices on the sphere, the Lax matrix $L$ in \eqref{eq:linear_systemCM} is not single-valued, and satisfies the relations
The monodromy data of the linear system \eqref{eq:linear_systemCM} is obtained by observing that $L,M$ are defined on a torus with a simple pole at $z=0$ and satisfy the relations \cite{DDG2020}
\begin{align}\label{L_1}
L(z+1,\tau)=L(z, \tau), &&  L(z+\tau, \tau)=e^{2\pi \ii u(t)\sigma_3} L(z, \tau) e^{-2\pi \ii u(\tau)\sigma_3}.
\end{align}
Subsequently, the solution $Y(z,\tau)$ \eqref{eq:linear_systemCM} has the following monodromy properties around the A, B cycles of the torus and around the puncture $z=0$:
\begin{gather}
\begin{array}{c}
 Y(z+1, \tau)=M_A Y(z, \tau), \qquad Y(z+\tau, \tau)=M_B Y(z, \tau) e^{2\pi \ii u(\tau)\sigma_3}, \\ \\
Y(e^{2\pi \ii}z, \tau)=M_0 Y(z, \tau), 
\end{array}
\end{gather}
with the constraint
\begin{equation}
    M_0=M_A^{-1} M_B^{-1} M_A M_B. \label{mon_cons}
\end{equation}
% Without loss of generality, it is always possible to set $M_A$ to be diagonal by conjugation. Introducing the monodromy exponent ${\sf a}\notin\mathbb{Z} + \frac{1}{2}$ around the A-cycle, we have
Choosing the matrix $M_A$ to be diagonal, one obtains the following explicit expressions for the monodromy matrices
\begin{align}\label{eq:1ptmonodromy}
    M_A=e^{2\pi \ii {\sf a}\sigma_3}, && M_0\sim e^{2\pi \ii m\sigma_3}, && M_{B}= \left(\begin{array}{cc}
           \frac{\sin \pi(2({\sf a}-m))}{\sin 2\pi {\sf a}}e^{-\ii \nu/2}  & \frac{\sin \pi m}{\sin 2\pi {\sf a}}e^{\ii \nu/2}  \\ \\
       -\frac{\sin \pi m}{\sin 2\pi {\sf a}} e^{-\ii \nu/2} &  \frac{\sin \pi(2({\sf a} +m)}{\sin 2\pi {\sf a}} e^{\ii \nu/2}
    \end{array}\right),
\end{align}
% {\color{blue}\begin{equation}
%     M_{B} = ... \label{eq:MB}
% \end{equation}}
where $\sim$ denotes the conjugacy class, $\sigma_{3}$ is the Pauli sigma matrix, and $m$ is the parameter of the equation \eqref{eq:CM_ELLDER}.

% \begin{remark}\label{remark:HJ}
%     The PDE \eqref{eq:Hamilton-jacobi} is the Hamilton-Jacobi equation with the action $\widetilde{\phi}$ and the associated Hamiltonian 
% \begin{align}
%    H(\tau) =v(\tau)^2-m^2\wp(2u(\tau)|\tau)-2m^2 \eta_{1}(\tau),
% \end{align}
% and the equations of motion read 
% \begin{align}
%     v(\tau) = \frac{d u(\tau)}{d \tau}, && \frac{d v(\tau)}{d \tau} = m^2 \wp'(2u(\tau)),
% \end{align}
% which is the non-autonomous elliptic Calogero-Moser model. The action is then 
% \begin{gather}
%    \widetilde{\phi}(2u(\tau),\tau) = \int^{\tau} d\tau v(\tau)^2+m^2\wp(2u(\tau)|\tau)+2m^2 \eta_{1}(\tau). 
% \end{gather}
% \end{remark}
\subsection{Asymptotic behaviour of $u(\tau)$}
For the sake of completeness, in this section we detail the asymptotics of the solution of the NAECM model $u(\tau)$ for $\tau\to \ii \infty$, and $\tau\to \tau_{\star}$ that were also computed in \cite{bonelli2020n,BGG2021}.

\begin{proposition}\label{prop:uasympinf}
In the limit $\tau \to \ii \infty$, the solution of the equation \eqref{eq:CM_ELLDER} $u(\tau)$ behaves as 
\begin{gather}
   u(\tau) = {\sf a} \tau +{\sf b}+\frac{m^2}{8 \pi \ii {\sf a}^2}e^{4 \pi \ii ({\sf a} \tau+{\sf b})}+\mathcal{O}(q e^{4 \pi \ii ({\sf a} \tau+{\sf b})}),
\end{gather}
where $q=e^{\ii \pi \tau}$, ${\sf a}$, $m$ parameterize the monodromies around the A-cycle and the puncture respectively \eqref{eq:1ptmonodromy}, and the monodromy around the B-cycle is related to ${\sf b}$ as
\begin{gather}
 e^{2\pi \ii {\sf b}}:= e^{\ii\nu/2} \frac{\Gamma( 2{\sf a}-m)\Gamma(1-2{\sf a})}{\Gamma(2{\sf a}) \Gamma(1-m-2{\sf a})}.\label{eq:eta_beta}
\end{gather}

\end{proposition}
\begin{proof}
Our starting point is the following identity satisfied by the function $u(\tau)$ that was shown in  \cite[equation (3.56)]{bonelli2020n} (see also \cite{desiraju2022painleve}):
\begin{gather}
    \frac{\theta_{3}(2 u\vert 2\tau)}{\theta_{2}(2u\vert 2\tau)} 
    % = i e^{3\pi i \tau/2} \frac{\det\left( \mathbb{1}- K_{1,1}\vert_{\rho= \frac{1}{4}+ \frac{\tau}{2}} \right)}{\left(\mathbb{1}- K_{1,1}\vert_{\rho= \frac{1}{4}}\right)} 
    = \frac{Z_{0}^{D}(\tau)}{Z_{1/2}^{D}(\tau)}, \label{eq:theta_det}
\end{gather}
where $Z_{0,1/2}^{D}$ are dual Nekrasov-Okounkov functions with integer and half integer shifts and are defined as 
\begin{gather}
    Z_{\epsilon/2}^{D} := \sum_{n\in \mathbb{Z}+ \frac{\epsilon}{2}} C_{n} e^{\ii n \nu}  q^{2({\sf a}+n)^2} \mathcal{B}({\sf a}+n, m, q).
\end{gather}
The coefficients
\begin{gather}
    C_{n} = \frac{G(1-m+2({\sf a}+n)) G(1-m-2({\sf a}+n))}{G(1+2({\sf a}+n))G(1-2({\sf a}+n))} \times \frac{G(1+2 {\sf a}) G(1- 2{\sf a})}{G(1-m+2{\sf a}) G(1-m-2{\sf a})},
\end{gather}
and the combinatorial term $\mathcal{B}(.)$ is given by
\begin{align}
    \mathcal{B}(a,m,q)= \prod_{n=1}^{\infty} \left( 1-q^{2n} \right)^{1-2m^2}\sum_{Y_{+}, Y_{-}} q^{2\left(|Y_{+}|+|Y_{-}| \right)} \prod_{\epsilon, \epsilon'=\pm} \frac{N_{Y_{\epsilon}, Y_{\epsilon'}}(m+(\epsilon-\epsilon'){\sf a})}{N_{Y_{\epsilon}, Y_{\epsilon'}}((\epsilon-\epsilon'){\sf a})},
\end{align}
where
\begin{align}
    N_{\lambda,\mu}(x) = \prod_{s\in \lambda} \left(x+a_{\lambda}(s)+l_{\mu}(s)+1 \right)\prod_{t\in \mu} \left(x-a_{\mu}(t)-l_{\lambda}(t)-1 \right),
\end{align}
with $a_{Y}(\Box)$ and $l_{Y}(\Box)$ denoting the arm and leg length of the box $\Box$ in the Young diagram $Y$ respectively.
% A derivation of the above expression starting from $c=1$ conformal blocks can be found in \cite{desiraju2022painleve}. 
In the limit $\tau \rightarrow \ii \infty$, the RHS of \eqref{eq:theta_det} reads
\begin{align}
   \frac{Z_{0}^{D}(\tau)}{Z_{1/2}^{D}(\tau)}
    % &= \frac{C_{0} q^{a^2}+\dots }{C_{-1/2} e^{-\ii \nu/2} q^{(a-1/2)^2} + C_{1/2} e^{\ii \nu/2} q^{(a+1/2)^2} +\dots} \nonumber \\
    % & = \frac{C_{0} }{C_{-1/2} e^{-\ii \nu/2} q^{-a+1/4} + C_{1/2} e^{\ii \nu/2} q^{a+1/4}}  + \dots \nonumber \\
    % & = q^{-1/4} \left[ \frac{C_{0} }{C_{-1/2} e^{-\ii \nu/2} q^{-a} + C_{1/2} e^{\ii \nu/2} q^{a}}  + \dots \right] \nonumber\\
    & = q^{-1/2} \left[ \frac{C_0}{C_{-1/2} e^{-\ii \nu/2} q^{-2{\sf a}} + C_{1/2} e^{\ii \nu/2} q^{2\sf a}}  + \mathcal{O}(q) \right]. \label{eq:asymptotics_i_inf_step1}
\end{align}
The coefficients $C_{0}, C_{\pm 1/2}$ can be simplified using the relation $G(z+1) = \Gamma(z) G(z)$:
\begin{align}
    C_{0} &= \frac{G(1-m+2{\sf a}) G(1-m-2{\sf a})}{G(1+2{\sf a})G(1-2{\sf a})} \times \frac{G(1+2 {\sf a}) G(1- 2{\sf a})}{G(1-m+2{\sf a}) G(1-m-2{\sf a})} =1, \label{eq:C0}\\
    C_{-1/2} &= \frac{G(1-m+2({\sf a}-1/2)) G(1-m-2({\sf a}-1/2))}{G(1+2({\sf a}-1/2))G(1-2({\sf a}-1/2))} \times \frac{G(1+2 {\sf a}) G(1- 2{\sf a})}{G(1-m+2{\sf a}) G(1-m-2{\sf a})} \nonumber \\
    % & = \frac{G(1-m+2a-1) G(1-m-2a+1)}{G(1+2a-1)G(1-2a+1)} \times \frac{G(1+2 a) G(1- 2a)}{G(1-m+2a) G(1-m-2a)} \nonumber \\
    % & = \frac{\left(\frac{G(1-m+2a)}{\Gamma(1-m + 2a-1)} \right) \Gamma(1-m-2a) G(1-m-2a)}{\left(\frac{G(1+2a)}{\Gamma(1+2a -1)}\right)G(1-2a) \Gamma(1-2a)} \times \frac{G(1+2 a) G(1- 2a)}{G(1-m+2a) G(1-m-2a)} \nonumber \\
    & = \frac{ \Gamma(2{\sf a}) \Gamma(1-m-2{\sf a})}{\Gamma( 2{\sf a}-m)\Gamma(1-2{\sf a})}, \label{eq:C-1/2}\\
    C_{1/2} & =  \frac{G(1-m+2({\sf a}+1/2)) G(1-m-2({\sf a}+1/2))}{G(1+2({\sf a}+1/2))G(1-2({\sf a}+1/2))} \times \frac{G(1+2 {\sf a}) G(1- 2{\sf a})}{G(1-m+2{\sf a}) G(1-m-2{\sf a})} \nonumber \\
    % & = \frac{G(1-m+2a+1) G(1-m-2a-1)}{G(1+2a+1)G(1-2a-1)} \times \frac{G(1+2 a) G(1- 2a)}{G(1-m+2a) G(1-m-2a)} \nonumber \\
    % & = \frac{G(1-m+2a) \Gamma(1-m+2a) \left(\frac{G(1-m-2a)}{\Gamma(1-m-2a-1)}\right)}{G(1+2a)\Gamma(1+2a)\left(\frac{G(1-2a)}{\Gamma(1-2a-1)}\right)} \times \frac{G(1+2 a) G(1- 2a)}{G(1-m+2a) G(1-m-2a)}\nonumber \\
    % &= \frac{ \Gamma(1-m+2a) \Gamma(-2a)}{\Gamma(-m-2a)\Gamma(1+2a)} \nonumber\\
    % & = \frac{(2a+m)(2a-m) \Gamma(2a-m) \Gamma(1-2a)}{(2a)^2\Gamma(1-m-2a) \Gamma(2a)} \nonumber \\
    & = \left(1- \frac{m^2}{4 {\sf a}^2}  \right)\frac{ \Gamma(2{\sf a}-m) \Gamma(1-2{\sf a})}{\Gamma(1-m-2{\sf a}) \Gamma(2{\sf a})}.\label{eq:C1/2}
\end{align}
% The parameter $\beta$ is defined by the relation
% \begin{gather}
%     e^{-2\pi \ii\beta} = e^{- \ii \nu/2} \frac{ \Gamma(2a) \Gamma(1-m-2a)}{\Gamma( 2a-m)\Gamma(1-2a)}. \label{eq:eta_beta}
% \end{gather}
Substituting \eqref{eq:C-1/2}, \eqref{eq:C1/2},\eqref{eq:eta_beta} in \eqref{eq:asymptotics_i_inf_step1} we get that
\begin{align}
   \frac{Z_{0}^{D}(\tau)}{Z_{1/2}^{D}(\tau)}& = q^{-1/2} \left[ \frac{1}{ e^{-2\pi \ii \beta} q^{-2{\sf a}} +  e^{2\pi \ii \beta} q^{2{\sf a}}}  + \mathcal{O}(q) \right] \nonumber \\
   & = \frac{q^{-1/2}}{2 \cos(2\pi  ({\sf a}\tau + {\sf b}))} (1+ \mathcal{O}(\cos(4\pi  ({\sf a}\tau + {\sf b}))). \label{eq:asymp_RHS}
\end{align}
Now for the LHS of \eqref{eq:theta_det},
% : we have the following series expansion of the theta functions
% \begin{align}
%     \theta_{3}(z\vert \tau) &= 1 + 2 \sum_{n=1}^{\infty} q^{n^2/2} \cos(2 n \pi z), \\
%     \theta_{2}(z\vert \tau) &= 2 \sum_{n=0}^{\infty} q^{(n+1/2)^2/2} \cos\left((2n+1) \pi z  \right).
% \end{align}
substituting the series expansions of the theta functions \eqref{def:elltheta2}, \eqref{def:elltheta3}, we see that
\begin{align}
    \frac{\theta_{3}(2 u\vert 2\tau)}{\theta_{2}(2u\vert 2\tau)} 
    % & = \frac{1+ 2 q \cos(4\pi u) + 2 q^4 \cos(8 \pi u)+\dots}{2 q^{1/4} \cos(2\pi u)+ 2 q^{9/4} \cos(6 \pi u)+\dots} \nonumber \\
    & = \frac{q^{-1/2}}{2\cos(2\pi u)} \left( 1+ \mathcal{O}(q^2) \right). \label{eq:asymp_LHS}
\end{align}
Comparing the RHS \eqref{eq:asymp_RHS} and the LHS \eqref{eq:asymp_LHS},
\begin{align}
   2 \cos(2u) &= e^{-2\pi i {\sf b}} q^{-2{\sf a}} +  e^{2\pi \ii \beta} q^{2{\sf a}} - \frac{m^2}{4 {\sf a}^2} e^{2\pi \ii {\sf b}} q^{2{\sf a}} \nonumber \\
   % &= e^{-2\pi \ii \beta} e^{-2\pi \ii \tau a} +  e^{2\pi \ii \beta} e^{2\pi \ii \tau a} - \frac{m^2}{4 a^2} e^{2\pi \ii \beta} e^{2\pi \ii \tau a} \nonumber \\
   % &= e^{-2\pi \ii (a\tau+\beta)} +  e^{2\pi \ii (a\tau+\beta)} - \frac{m^2}{4 a^2} e^{2\pi \ii (a\tau+\beta)} \nonumber \\
   &= 2 \cos(2\pi({\sf a}\tau+{\sf b})) - \frac{m^2}{4 {\sf a}^2} e^{2\pi \ii ({\sf a}\tau+{\sf b})}. 
\end{align}
% \begin{gather}
%     u(\tau)= a \tau + \beta + c e^{k\pi \ii \tau a} \nonumber \\
%     \Rightarrow  \lim_{\tau\rightarrow \ii \infty}\cos(2\pi u) = \cos(2\pi(a \tau + \beta + c e^{k\pi \ii \tau a})) \nonumber \\
%     = \cos(2\pi(a\tau+\beta)) \cos(ce^{2\pi \ii \tau a}) -0\nonumber\\
%     = \cos(2\pi(a\tau+\beta)) (1+0-c^2 e^{2k\pi \ii \tau a})
% \end{gather}
% {\color{red} LHS comes from DLMF where $q=e^{i\pi \tau}$ and RHS coms from the paper of Bonelli, Del Monte, Gavrylenko, Tanzini 1901.10497, where they say $q= e^{2\pi i \tau}$, but comparing the equation
% \begin{gather}
% \eta(\tau)^{-1} \theta_{3}(2 Q\vert 2 \tau) \mathcal{T} = Z_{0}^{D}(\tau),    
% \end{gather}
% we find indeed that their $q$ should have been $e^{i \pi \tau}$.(FALSE)}

We can therefore make the following ansatz for $u(\tau)$ in the limit $\tau \rightarrow \ii \infty$:
\begin{equation}
    u(\tau) = {\sf a} \tau + {\sf b} + \frac{1}{2\pi \ii} \sum_{n=0}^{\infty} \sum_{k=-n}^{\infty} c_{n,k} q^{ n} e^{4\pi \ii k({\sf a} \tau+{\sf b})} = {\sf a} \tau + {\sf b} + \frac{X}{2\pi \ii}.\label{ansatz:Qinf}
\end{equation}
The first derivative
\begin{equation}
    \frac{u(\tau)}{d \tau} = {\sf a} + \frac{1}{2\pi \ii} \sum_{n=0}^{\infty} \sum_{k=-n}^{\infty} c_{n,k} (2\pi \ii) \left( n+ 2k \right) q^{ n} e^{4\pi \ii k({\sf a} \tau+{\sf b})},
\end{equation}
and the second derivative
\begin{equation}
    \frac{d^2 u(\tau)}{d\tau^2} =  (2\pi \ii) \sum_{n=0}^{\infty} \sum_{k=-n}^{\infty} c_{n,k}(n + 2k {\sf a})^2 q^{ n} e^{4\pi \ii k({\sf a} \tau+{\sf b})}. \label{infasymp_LHS}
\end{equation}
Let us now look at our equation
\begin{gather}
   (2\pi \ii)^2 \frac{d^2 u(\tau)}{d \tau^2} = m^2 \wp'(2u\vert \tau). \label{infasymp_eq}
\end{gather}
The Weierstrass $\wp$ function has the following expression in terms of theta functions
\begin{gather}
    \wp(z\vert \tau) = -\partial_{z}^2 \log \theta_{1}(z\vert \tau) - 2\eta_{1} (\tau),
\end{gather}
and we now analyze the term $\partial_{z}^2 \log \theta_{1}(z\vert \tau)$ in the limit $q\rightarrow 0$. Starting from the series representation of the theta function
\begin{gather}
    \theta_{1}(z \vert \tau) = -\ii \sum_{n \in \mathbb{Z}} (-1)^n q^{(n+ \frac{1}{2})^2} e^{2 \pi \ii z (n + \frac{1}{2})},
\end{gather}
 the logarithmic derivative 
\begin{align}
     \partial_{z} \log \theta_{1}(z\vert \tau) - \pi \cot \pi z 
     % &= 4 \pi \sum_{k=1}^{\infty} \frac{q^{k}}{1-q^{k}} \sin(2\pi kz) \nonumber \\
    % & = 4\pi \sum_{k=1}^{\infty} q^{k} (1-q^{k})^{-1} \sin(2\pi kz) \nonumber \\
    & = 4\pi \sum_{k=1}^{\infty} q^{k} \sum_{n=0}^{\infty} q^{kn} \sin(2\pi kz) \nonumber \\
    % & = 4\pi \sum_{k=1}^{\infty}  \sum_{n=0}^{\infty} q^{(n+1)k} \left( e^{ \ii 2\pi kz} - e^{-\ii 2\pi kz} \right) \frac{1}{2i} \nonumber \\
    & = -(2 \pi \ii) \left(\sum_{k=1}^{\infty}  \sum_{n=1}^{\infty} q^{kn}  e^{ \ii 2\pi kz} -\sum_{k=1}^{\infty}  \sum_{n=1}^{\infty} q^{kn}  e^{-\ii 2\pi kz}  \right).
\end{align}
The above computation implies that
\begin{align}
     \partial_{z}^2 \log \theta_{1}(z\vert \tau) + \pi^2 \csc^2 \pi z &= - (2\pi \ii)^2 \left(\sum_{k=1}^{\infty}  \sum_{n=1}^{\infty}k q^{kn}  e^{\ii 2\pi kz} +\sum_{k=1}^{\infty}  \sum_{n=1}^{\infty}k q^{kn}  e^{-\ii 2\pi kz}  \right). \label{theta2-csc2}
\end{align}
Moreover, $\csc^2(z)$ has the following series representation:
\begin{align}
   \pi^2 \csc^2 \pi z = \frac{\pi^2}{\sin^2 \pi z} 
   % &= \frac{(2\pi \ii)^2}{(e^{\ii \pi z}-e^{-\ii \pi z})^2} \nonumber \\
    % & = (2\pi \ii)^2 \frac{e^{2\ii \pi z}}{(e^{2\ii \pi z}-1)^2} = (2\pi \ii)^2  e^{2\ii\pi z} (1-e^{2\ii\pi z})^{-2} \nonumber \\
    &= (2 \pi \ii)^2 e^{2\ii\pi z} \sum_{k=1}^{\infty} k e^{2\ii \pi (k-1)z} = (2\pi \ii)^2 \sum_{k=1}^{\infty} k e^{2\ii\pi kz}. \label{csc2}
\end{align}
Combining \eqref{csc2} and \eqref{theta2-csc2} we get
\begin{align}
     \partial_{z}^2 \log \theta_{1}(z\vert \tau) &= -(2 \pi \ii)^2 \left(\sum_{k=1}^{\infty}  \sum_{n=1}^{\infty}k q^{kn}  e^{\ii2\pi kz} +\sum_{k=1}^{\infty}  \sum_{n=1}^{\infty}k q^{kn}  e^{-\ii 2\pi kz} + \sum_{k=1}^{\infty} k e^{2\ii\pi kz} \right) \nonumber \\
     &= -(2\pi \ii)^2 \left(\sum_{k=1}^{\infty}  \sum_{n=0}^{\infty}k q^{kn}  e^{\ii 2\pi kz} +\sum_{k=1}^{\infty}  \sum_{n=1}^{\infty}k q^{kn}  e^{-\ii 2\pi kz}  \right).
\end{align}
Therefore,
\begin{align}
    \wp'(z\vert \tau) = -\partial_{z}^3 \log \theta_{1}(z\vert \tau) = (2 \pi \ii)^3 \left(\sum_{k=1}^{\infty}  \sum_{n=0}^{\infty}k^2 q^{kn}  e^{\ii 2\pi kz} -\sum_{k=1}^{\infty}  \sum_{n=1}^{\infty}k^2 q^{kn}  e^{-\ii 2 \pi kz}  \right). \label{infasymp_RHS}
\end{align}
 Substituting \eqref{infasymp_LHS}, \eqref{infasymp_RHS} in \eqref{infasymp_eq} we get
\begin{align}
    % (2\pi \ii)^3  \sum_{n=0}^{\infty} \sum_{k=-n}^{\infty} c_{n,k}(n + 2k \alpha)^2 q^{ n} e^{4\pi \ii k(\alpha \tau+\beta)} & =  m^2 (2\pi i)^3 \left(\sum_{k=1}^{\infty}  \sum_{n=0}^{\infty}k^2 q^{kn}  e^{\ii 2\pi kz} -\sum_{k=1}^{\infty}  \sum_{n=1}^{\infty}k^2 q^{kn}  e^{-\ii 2 \pi kz}  \right) \\
      \sum_{n=0}^{\infty} \sum_{k=-n}^{\infty} c_{n,k}(n + 2k  {\sf a} )^2 q^{ n} e^{4\pi \ii k( {\sf a} \tau+ {\sf b} )} &= m^2 \sum_{k=1}^{\infty}  \sum_{n=0}^{\infty}k^2 q^{kn}  e^{4 \pi \ii k ( {\sf a}  \tau +  {\sf b} )} e^{k X} \nonumber \\
      &-m^2 \sum_{k=1}^{\infty}  \sum_{n=1}^{\infty}k^2 q^{kn}   e^{-4 \pi \ii k ( {\sf a} \tau +  {\sf b} )} e^{-k X} .
\end{align}
Computing the first few values of $c_{n,k}$:
\begin{align}
    c_{0,0} &= 0, \\
    c_{0,1} &= \frac{m^2}{4  {\sf a}^2}, \\
    c_{1,-1} &=-\frac{m^2}{(1-2  {\sf a} )^2} 
\end{align}
Therefore,
\begin{gather}
u(\tau) = {\sf a} \tau + {\sf b} +\frac{m^2}{8 \pi \ii {\sf a}^2}e^{4 \pi \ii ( {\sf a} \tau+ {\sf b} )}+\mathcal{O}(q e^{4 \pi \ii ({\sf a} \tau+{\sf b})}).
\end{gather}
\end{proof}

\begin{proposition}
Near its zero $u(\tau_{\star})=0$, the solution of the equation \eqref{eq:CM_ELLDER} $u(\tau)$ behaves as 
\begin{gather}\label{asymp:utstar}
u(\tau) = e^{\mp \ii \pi/4} \sqrt{\frac{m}{2\pi}} (\tau - \tau_{\star})^{1/2} \left( 1 \pm \frac{H_{\star} }{4\pi \ii m} (\tau- \tau_{\star})  \right) + \mathcal{O} \left((\tau -\tau_{\star})^{5/2} \right).
\end{gather}
\end{proposition}
\begin{proof}
The equation \eqref{eq:CM_ELLDER} near $u(\tau) =0$ reads as
\begin{gather}
  (2\pi i)^2 \frac{d^2 u(\tau)}{d\tau^2 } = - \frac{ m^2}{4 u(\tau)^3} \label{equation}
\end{gather}
due to the following behaviour of the Weierstrass $\wp$ function at zero
\begin{gather}
    \lim_{z\rightarrow 0} \wp(z) = \frac{1}{z^2}. \label{eq:Plim0}
\end{gather}
% in 
% % \begin{gather}
% %     (2\pi i )^2 \frac{d^2 u(\tau)}{d \tau^2} = m^2 \wp' \left(2u(\tau)\vert \tau \right) \label{start_point}
% % \end{gather}
% % we compute $\lim_{\tau\rightarrow \tau_{\star}}u(\tau)$ where $u(\tau_{\star}) =0$. The Weierstrass $\wp$ function
% % \begin{gather}
% %     \wp(z) = \frac{1}{z^2} \sum_{w\in \mathbb{L}\backslash 0} \left( \frac{1}{(z-w)^2} - \frac{1}{w^2} \right)
% % \end{gather}
% % and so,
% % Plugging \eqref{eq:Plim0} into \eqref{start_point} we can compute the leading estimate for $\lim_{\tau\rightarrow \tau_{\star}}u(\tau)$
% we obtain the fo
The solution $u(\tau)$ therefore has the following behaviour
\begin{gather}
    u(\tau) = c_{1}\sqrt{\tau - \tau_{\star}} + c_{2} (\tau-\tau_{\star})^{3/2} + \mathcal{O}((\tau -\tau_{\star})^{5/2})= c_{1} (\tau-\tau_{\star})^{1/2} \left(1 + \frac{c_{2}}{c_{1}} (\tau-\tau_{\star}) \right)+ \mathcal{O}((\tau -\tau_{\star})^{5/2}), \label{ansatz}
\end{gather}
% the first derivative is
% \begin{gather}
%     \frac{d u(\tau)}{d \tau} = \frac{c_{1}}{2 \sqrt{\tau-\tau_{\star}}} + \frac{3 c_{2}}{2} (\tau-\tau_{\star})^{1/2} \nonumber \\
%     = \frac{c_{1}}{2}(\tau-\tau_{\star})^{-1/2} \left(1 + \frac{3 c_{2}}{c_{1}} (\tau-\tau_{\star}) \right)\nonumber \\
%     = \frac{c_{1}}{2u}\left(1+ \frac{3 c_{2}}{c_{1}}(\tau-\tau_{\star})\right), \label{der1}
% \end{gather}
and the second derivative of the above function
\begin{gather}
    \frac{d^2 u(\tau)}{d \tau^2} 
    % = -\frac{c_{1}}{4} (\tau-\tau_{\star})^{-3/2}+ \frac{3 c_{2}}{4} (\tau-\tau_{\star})^{-1/2} \nonumber \\
    = -\frac{c_{1}}{4}(\tau-\tau_{\star})^{-3/2} \left(1 - \frac{3 c_{2}}{c_{1}} (\tau-\tau_{\star}) \right). \label{der2}
\end{gather}
% we solve for $c_{1}, c_{2}$ by
Substituting \eqref{ansatz}, \eqref{der2} in \eqref{equation}, we obtain that 
% \begin{gather}
% (2\pi i)^2 \frac{d^2 u(\tau)}{d\tau^2 } = - \frac{ m^2}{4 u(\tau)^3} \nonumber \\
% \Rightarrow - (2\pi \ii)^2 \frac{c_{1}}{4} (\tau-\tau_{\star})^{-3/2} \left(1 - \frac{3 c_{2}}{c_{1}} (\tau-\tau_{\star}) \right) = - \frac{ m^2}{4 c_{1}^3} (\tau - \tau_{\star})^{-3/2} \left( 1- \frac{3 c_{2}}{c_{1}} (\tau - \tau_{\star})  \right) \nonumber \\
% \Rightarrow (2\pi \ii)^2 \frac{c_{1}}{4} \left(1 - \frac{3 c_{2}}{c_{1}} (\tau-\tau_{\star}) \right) =  \frac{ m^2}{4 c_{1}^3}  \left( 1- \frac{3 c_{2}}{c_{1}} (\tau - \tau_{\star})  \right).\label{c1}
% \end{gather}
% Equating the coefficients, the constant term gives that
\begin{gather}
    % c_{1}^4=  \frac{ m^2}{(2\pi \ii)^2} \Rightarrow c_{1}^2=  e^{\mp \ii \pi/2}\frac{ m}{(2\pi)} \Rightarrow 
    c_{1} = e^{\mp \ii \pi/4}\frac{ \sqrt{m}}{\sqrt{2\pi}}. \label{c1}
\end{gather}
% The term linear in $\tau$ implies that
% \begin{gather}
%   - (2\pi \ii)^2 3 c_{2} = -m^2  \frac{3 c_{2}}{c_{1}^4},
% \end{gather}
% which gives no information about $c_{2}$.  Therefore, the leading asymptotic is
% \begin{gather}
%     \lim_{\tau\rightarrow \tau_{\star}} u(\tau) = e^{\mp \ii \pi/4}\frac{ \sqrt{m} \sqrt{\tau-\tau_{\star}}}{\sqrt{2\pi}} + ... \label{tstarQ1}
% \end{gather}
The coefficient of the sub-leading term $c_2$ can be obtained by studying the Hamiltonian in the same limit:
% From \eqref{ansatz}, we obtain that
% \begin{gather}
%     \tau-\tau_{\star} \approx \pm \frac{2\pi \ii}{m} u(\tau)^2 + ..., 
% \end{gather}
% and recalling the Hamiltonian
% \begin{gather*}
%     H = (2\pi \ii)^2 \left(\frac{du(\tau)}{d \tau} \right)^2 - m^2 \wp(2u(\tau)\vert \tau),
% \end{gather*}
% under the limit $\tau\rightarrow \tau_{\star}$
\begin{align}
    H_{\star} &:= \lim_{\tau\to \tau_{\star}}H\\
    &= \lim_{\tau \to \tau_{\star}}\left((2\pi \ii)^2 \left(\frac{du(\tau)}{d \tau} \right)^2 - \frac{m^2}{4u(\tau)^2}\right)\nonumber \\
    &\mathop{=}^{\eqref{ansatz}} (2\pi \ii)^2 \frac{c_{1}^2}{4} (\tau-\tau_{\star})^{-1} \left(1+ \frac{6c_{2}}{c_{1}} (\tau-\tau_{\star}) \right) - m^2 \frac{1}{4c_{1}^2} (\tau-\tau_{\star})^{-1} \left( 1- 2 \frac{c_{2}}{c_{1}} (\tau- \tau_{\star}) \right) + \mathcal{O}(\tau-\tau_{\star})\nonumber \\
    % &= (2\pi \ii)^2 \frac{3 c_{2}c_{1}}{2} + m^2 \frac{c_{2}}{2 c_{1}^3} = \frac{c_{2}}{2 c_{1}} \left((2\pi \ii)^2 3 c_{1}^2 + \frac{m^2}{c_{1}^2} \right) \nonumber\\
    &\mathop{=}^{\eqref{c1}} \frac{c_{2}}{2 c_{1}}  \left(3 (2\pi \ii)^2 \frac{m}{(2\pi)} e^{\mp \ii\pi/2} + \frac{m^2 (2\pi)}{m} e^{\pm \ii\pi/2} \right) + \mathcal{O}(\tau-\tau_{\star})\nonumber\\ 
   % & = \pm\frac{c_{2}}{2 c_{1}}  \left((\mp \ii) (6 \pi \ii) m \pm 2\pi \ii m \right) \nonumber \\
    &= \pm \frac{c_{2}}{ c_{1}}  \left(4\pi \ii m \right)+ \mathcal{O}(\tau-\tau_{\star}). \nonumber\\
    \Rightarrow \frac{c_{2}}{c_{1}} &= \pm \frac{H_{\star} }{4\pi \ii m}.
\end{align}
With the constants $c_1$, $c_2$ above,
\begin{align*}
    \lim_{\tau\rightarrow \tau_{\star}} u(\tau) = e^{\mp \ii \pi/4} \sqrt{\frac{m}{2\pi}} (\tau - \tau_{\star})^{1/2} \left( 1 \pm \frac{H_{\star} }{4\pi \ii m} (\tau- \tau_{\star})  \right) + \mathcal{O} \left((\tau -\tau_{\star})^{5/2} \right).
\end{align*}
\end{proof}

\printbibliography

@article{ghosal2020probabilistic,
  archivePrefix = "arXiv",
  eprint = "2003.03802",
  year={2020},    
primaryClass={math.PR},
    AUTHOR = {Ghosal, Promit and Remy, Guillaume and Sun, Xin and Sun, Yi},
     TITLE = {Probabilistic conformal blocks for {L}iouville {CFT} on the
              torus},
   JOURNAL = {Duke Math. J.},
  FJOURNAL = {Duke Mathematical Journal},
    VOLUME = {173},
      YEAR = {2024},
    NUMBER = {6},
     PAGES = {1085--1175},
   MRCLASS = {81T40 (60D99)},
  MRNUMBER = {4748792},
       DOI = {10.1215/00127094-2023-0031}
}

@incollection {Manin,
    AUTHOR = {Manin, Yu. I.},
     TITLE = {Sixth {P}ainlev\'{e} equation, universal elliptic curve, and
              mirror of {$\mathbb{P}^2$}},
 BOOKTITLE = {Geometry of differential equations},
    SERIES = {Amer. Math. Soc. Transl. Ser. 2},
    VOLUME = {186},
     PAGES = {131--151},
 PUBLISHER = {Amer. Math. Soc., Providence, RI},
      YEAR = {1998},
   MRCLASS = {14N35 (14H52 34M55 37J05 37J15)},
  MRNUMBER = {1732409},
MRREVIEWER = {Bruce\ Hunt},
       DOI = {10.1090/trans2/186/04},
archivePrefix = {arXiv},
       eprint = {alg-geom/9605010},
 primaryClass = {math.AG}
     }

@article{Nekrasov:2003rj,
      author         = "Nekrasov, Nikita and Okounkov, Andrei",
      title          = "{Seiberg-Witten theory and random partitions}",
      journal        = "Prog. Math.",
      volume         = "244",
      year           = "2006",
      pages          = "525-596",
      doi            = "10.1007/0-8176-4467-9_15",
      eprint         = "hep-th/0306238",
      archivePrefix  = "arXiv",
      primaryClass   = "hep-th",
      reportNumber   = "ITEP-TH-36-03, PUDM-2003, IHES-P-03-43",
      SLACcitation   = "%%CITATION = HEP-TH/0306238;%%"
}

@article{Alday:2009aq,
      author         = "Alday, Luis F. and Gaiotto, Davide and Tachikawa, Yuji",
      title          = "{Liouville Correlation Functions from Four-dimensional
                        Gauge Theories}",
      journal        = "Lett. Math. Phys.",
      volume         = "91",
      year           = "2010",
      pages          = "167-197",
      doi            = "10.1007/s11005-010-0369-5",
      eprint         = "0906.3219",
      archivePrefix  = "arXiv",
      primaryClass   = "hep-th",
      SLACcitation   = "%%CITATION = ARXIV:0906.3219;%%"
}

@article{DDG2020,
    AUTHOR = {Del Monte, Fabrizio and Desiraju, Harini and Gavrylenko,
              Pavlo},
     TITLE = {Isomonodromic tau functions on a torus as {F}redholm
              determinants, and charged partitions},
   JOURNAL = {Comm. Math. Phys.},
  FJOURNAL = {Communications in Mathematical Physics},
    VOLUME = {398},
      YEAR = {2023},
    NUMBER = {3},
     PAGES = {1029--1084},
      ISSN = {0010-3616,1432-0916},
   MRCLASS = {37N20 (47B35)},
  MRNUMBER = {4561797},
       DOI = {10.1007/s00220-022-04458-y},
  archivePrefix = "arXiv",
  eprint = "2011.06292",
primaryclass={math-ph},
  publisher={Springer}
}

@article{DKRV16,
    AUTHOR = {David, Fran\c{c}ois and Kupiainen, Antti and Rhodes, R\'{e}mi
              and Vargas, Vincent},
     TITLE = {Liouville quantum gravity on the {R}iemann sphere},
   JOURNAL = {Comm. Math. Phys.},
  FJOURNAL = {Communications in Mathematical Physics},
    VOLUME = {342},
      YEAR = {2016},
    NUMBER = {3},
     PAGES = {869--907},
      ISSN = {0010-3616,1432-0916},
   MRCLASS = {81T20},
  MRNUMBER = {3465434},
       DOI = {10.1007/s00220-016-2572-4},
         archivePrefix={arXiv},
  eprint={1410.7318},
      primaryClass={math.PR}
}

@article{bonelli2023irregular,
    AUTHOR = {Bonelli, Giulio and Iossa, Cristoforo and Lichtig, Daniel
              Panea and Tanzini, Alessandro},
     TITLE = {Irregular {L}iouville correlators and connection formulae for
              {H}eun functions},
   JOURNAL = {Comm. Math. Phys.},
  FJOURNAL = {Communications in Mathematical Physics},
    VOLUME = {397},
      YEAR = {2023},
    NUMBER = {2},
     PAGES = {635--727},
      ISSN = {0010-3616,1432-0916},
   MRCLASS = {81T40 (33C15 81R10)},
  MRNUMBER = {4538706},
MRREVIEWER = {S.\ A.\ Yost},
       DOI = {10.1007/s00220-022-04497-5},
 archivePrefix={arXiv},
eprint={2201.04491},
 primaryClass={hep-th}
}

@article{kashani2013transformations,
    AUTHOR = {Kashani-Poor, Amir-Kian and Troost, Jan},
     TITLE = {Transformations of spherical blocks},
   JOURNAL = {J. High Energy Phys.},
  FJOURNAL = {Journal of High Energy Physics},
      YEAR = {2013},
    NUMBER = {10},
   MRCLASS = {83E30 (81T60)},
  MRNUMBER = {3110708},
MRREVIEWER = {Kotik\ K.\ Lee},
  archivePrefix={arXiv},
  eprint={1305.7408},
primaryClass={hep-th},
doi={10.1007/JHEP10%282013%29009}
}

@article{DRV16,
    AUTHOR = {David, Fran\c{c}ois and Rhodes, R\'{e}mi and Vargas, Vincent},
     TITLE = {Liouville quantum gravity on complex tori},
   JOURNAL = {J. Math. Phys.},
  FJOURNAL = {Journal of Mathematical Physics},
    VOLUME = {57},
      YEAR = {2016},
    NUMBER = {2},
     PAGES = {022302, 25},
   MRCLASS = {81T20},
  MRNUMBER = {3450564},
       DOI = {10.1063/1.4938107},
  archivePrefix={arXiv},
  eprint={1504.00625},
primaryClass={math.PR},
}

@article {HRV18,
    AUTHOR = {Huang, Yichao and Rhodes, R\'{e}mi and Vargas, Vincent},
     TITLE = {Liouville quantum gravity on the unit disk},
   JOURNAL = {Ann. Inst. Henri Poincar\'{e} Probab. Stat.},
  FJOURNAL = {Annales de l'Institut Henri Poincar\'{e} Probabilit\'{e}s et
              Statistiques},
  DOI = {10.1214/17-AIHP852},
    VOLUME = {54},
      YEAR = {2018},
    NUMBER = {3},
   MRCLASS = {60D05 (81T20 81T40)},
  MRNUMBER = {3825895},
  archivePrefix={arXiv},
  eprint={1502.04343},
 primaryClass={math.PR}
}

@article{GRV19,
    AUTHOR = {Guillarmou, Colin and Rhodes, R\'{e}mi and Vargas, Vincent},
     TITLE = {Polyakov's formulation of {$2d$} bosonic string theory},
   JOURNAL = {Publ. Math. Inst. Hautes \'{E}tudes Sci.},
  FJOURNAL = {Publications Math\'{e}matiques. Institut de Hautes \'{E}tudes
              Scientifiques},
    VOLUME = {130},
      YEAR = {2019},
     PAGES = {111--185},
      ISSN = {0073-8301,1618-1913},
   MRCLASS = {81T40 (30F10 32G15 81T30)},
  MRNUMBER = {4028515},
       DOI = {10.1007/s10240-019-00109-6},
  archivePrefix={arXiv},
  eprint={1607.08467},
 primaryClass={math-ph}
}

@misc{DLMF,
    shorthand ={DLMF},
       title = "{\it NIST Digital Library of Mathematical Functions}",
howpublished = "\url{https://dlmf.nist.gov/}, Release 1.1.10 of 2023-06-15",
        note = "F.~W.~J. Olver, A.~B. {Olde Daalhuis}, D.~W. Lozier, B.~I. Schneider,
                R.~F. Boisvert, C.~W. Clark, B.~R. Miller, B.~V. Saunders,
                H.~S. Cohl, and M.~A. McClain, eds."}

@article{aminov2022black,
    AUTHOR = {Aminov, Gleb and Grassi, Alba and Hatsuda, Yasuyuki},
     TITLE = {Black hole quasinormal modes and {S}eiberg-{W}itten theory},
   JOURNAL = {Ann. Henri Poincar\'e},
  FJOURNAL = {Annales Henri Poincar\'e. A Journal of Theoretical and
              Mathematical Physics},
    VOLUME = {23},
      YEAR = {2022},
    NUMBER = {6},
     PAGES = {1951--1977},
      ISSN = {1424-0637,1424-0661},
   MRCLASS = {83C57 (81Q20 81T20 83C25)},
  MRNUMBER = {4420567},
MRREVIEWER = {Bogus\l aw\ Broda},
       DOI = {10.1007/s00023-021-01137-x},
eprint = {2006.06111},
archivePrefix = {arXiv}
}

@article{amado2017kerr,
    AUTHOR = {Barrag\'{a}n Amado, Juli\'{a}n and Carneiro da Cunha, Bruno
              and Pallante, Elisabetta},
     TITLE = {On the {K}err-{A}d{S}/{CFT} correspondence},
   JOURNAL = {J. High Energy Phys.},
  FJOURNAL = {Journal of High Energy Physics},
      YEAR = {2017},
    NUMBER = {8},
     PAGES = {094, front matter+25},
      ISSN = {1126-6708,1029-8479},
   MRCLASS = {83C57 (83E30)},
  MRNUMBER = {3697388},
eprint={1702.01016},
archivePrefix={arXiv},
doi={10.1007/JHEP08(2017)094},
 primaryClass={hep-th}
}

@article{hortaccsu2018heun,
    AUTHOR = {Horta\c{c}su, M.},
     TITLE = {Heun functions and some of their applications in physics},
   JOURNAL = {Adv. High Energy Phys.},
  FJOURNAL = {Advances in High Energy Physics},
      YEAR = {2018},
     PAGES = {Art. ID 8621573, 14},
   MRCLASS = {33E30},
  MRNUMBER = {3834316},
       DOI = {10.1155/2018/8621573},
archivePrefix={arXiv},
eprint={1101.0471v11},
 primaryClass={math-ph}
}

@article{litvinov2014classical,
    AUTHOR = {Litvinov, Alexey and Lukyanov, Sergei and Nekrasov, Nikita and
              Zamolodchikov, Alexander},
     TITLE = {Classical conformal blocks and {P}ainlev\'{e} {VI}},
   JOURNAL = {J. High Energy Phys.},
  FJOURNAL = {Journal of High Energy Physics},
      YEAR = {2014},
    NUMBER = {7},
     PAGES = {144, front matter+19},
      ISSN = {1126-6708,1029-8479},
   MRCLASS = {81T40},
  MRNUMBER = {3250114},
       DOI = {10.1007/JHEP07(2014)144},
  archivePrefix={arXiv},
  eprint={1309.4700},   primaryClass={hep-th}
}

@article{lencses2018classical,
    AUTHOR = {Lencs\'{e}s, M\'{a}t\'{e} and Novaes, F\'{a}bio},
     TITLE = {Classical conformal blocks and accessory parameters from
              isomonodromic deformations},
   JOURNAL = {J. High Energy Phys.},
  FJOURNAL = {Journal of High Energy Physics},
      YEAR = {2018},
    NUMBER = {4},
     PAGES = {096, front matter+36},
      ISSN = {1126-6708,1029-8479},
   MRCLASS = {83C47},
  MRNUMBER = {3804362},
       DOI = {10.1007/jhep04(2018)096},
archivePrefix={arXiv},
eprint={1709.03476},
   primaryClass={hep-th}
}

@article{lisovyy2022perturbative,
    AUTHOR = {Lisovyy, O. and Naidiuk, A.},
     TITLE = {Perturbative connection formulas for {H}eun equations},
   JOURNAL = {J. Phys. A},
  FJOURNAL = {Journal of Physics. A. Mathematical and Theoretical},
    VOLUME = {55},
      YEAR = {2022},
    NUMBER = {43},
     PAGES = {Paper No. 434005, 22},
      ISSN = {1751-8113,1751-8121},
   MRCLASS = {34M35 (33C15 81T40)},
  MRNUMBER = {4506508},
doi={10.1088/1751-8121/ac9ba7},
archivePrefix={arXiv},
eprint={2208.01604},   primaryClass={math-ph}
}

@article {lisovyy2021accessory,
    AUTHOR = {Lisovyy, O. and Naidiuk, A.},
     TITLE = {Accessory parameters in confluent {H}eun equations and
              classical irregular conformal blocks},
   JOURNAL = {Lett. Math. Phys.},
  FJOURNAL = {Letters in Mathematical Physics},
    VOLUME = {111},
      YEAR = {2021},
    NUMBER = {6},
     PAGES = {Paper No. 137, 28},
      ISSN = {0377-9017,1573-0530},
   MRCLASS = {81R12 (33C15 81T40)},
  MRNUMBER = {4338702},
MRREVIEWER = {Stanislav\ Z.\ Pakuliak},
       DOI = {10.1007/s11005-021-01400-6},
 archivePrefix={arXiv},
  eprint={2101.05715},
   primaryClass={math-ph}
}

@article{nekrasov2009supersymmetric,
  eprint={0901.4744},
  archivePrefix={arXiv},
primaryClass={hep-th},
    AUTHOR = {Nekrasov, Nikita A. and Shatashvili, Samson L.},
     TITLE = {Supersymmetric vacua and {B}ethe ansatz},
   JOURNAL = {Nuclear Phys. B Proc. Suppl.},
  FJOURNAL = {Nuclear Physics B. Proceedings Supplement},
    VOLUME = {192/193},
      YEAR = {2009},
     PAGES = {91--112},
      ISSN = {0920-5632},
   MRCLASS = {81T60 (81R12)},
  MRNUMBER = {2570974},
       DOI = {10.1016/j.nuclphysbps.2009.07.047}
}

@article{nekrasov2011darboux,
    AUTHOR = {Nekrasov, N. and Rosly, A. and Shatashvili, S.},
     TITLE = {Darboux coordinates, {Y}ang-{Y}ang functional, and gauge
              theory},
   JOURNAL = {Nuclear Phys. B Proc. Suppl.},
  FJOURNAL = {Nuclear Physics B. Proceedings Supplement},
    VOLUME = {216},
      YEAR = {2011},
     PAGES = {69--93},
      ISSN = {0920-5632},
   MRCLASS = {81T60 (14D21 53C26 81R12)},
  MRNUMBER = {2851597},
MRREVIEWER = {Lee-Peng\ Teo},
 eprint={1103.3919},
  archivePrefix={arXiv},
primaryClass={hep-th},
doi={10.1016/j.nuclphysbps.2011.04.150}
}

@article{hollands2018higher,
    AUTHOR = {Hollands, Lotte and Kidwai, Omar},
     TITLE = {Higher length-twist coordinates, generalized {H}eun's opers,
              and twisted superpotentials},
   JOURNAL = {Adv. Theor. Math. Phys.},
  FJOURNAL = {Advances in Theoretical and Mathematical Physics},
    VOLUME = {22},
      YEAR = {2018},
    NUMBER = {7},
     PAGES = {1713--1822},
      ISSN = {1095-0761,1095-0753},
   MRCLASS = {81T30},
  MRNUMBER = {3976899},
       DOI = {10.4310/ATMP.2018.v22.n7.a2},
 eprint={1710.04438},
  archivePrefix={arXiv},
primaryclass={hep-th}
}

@incollection {nakajima2003lectures,
    AUTHOR = {Nakajima, Hiraku and Yoshioka, K\B{o}ta},
     TITLE = {Lectures on instanton counting},
 BOOKTITLE = {Algebraic structures and moduli spaces},
    SERIES = {CRM Proc. Lecture Notes},
    VOLUME = {38},
     PAGES = {31--101},
 PUBLISHER = {Amer. Math. Soc., Providence, RI},
      YEAR = {2004},
      ISBN = {0-8218-3568-8},
   MRCLASS = {14D21 (53D20 53D45 57R57 81T60)},
  MRNUMBER = {2095899},
MRREVIEWER = {Chien-Hao\ Liu},
       DOI = {10.1090/crmp/038/02},
  eprint={math/0311058},
  archivePrefix={arXiv},
  year={2003},
  primaryClass={math.AG}
}

@article{gu2020elliptic,
    AUTHOR = {Gu, Jie and Haghighat, Babak and Klemm, Albrecht and Sun,
              Kaiwen and Wang, Xin},
     TITLE = {Elliptic blowup equations for 6d {SCFT}s. {P}art {III}.
              {E}-strings, {M}-strings and chains},
   JOURNAL = {J. High Energy Phys.},
  FJOURNAL = {Journal of High Energy Physics},
      YEAR = {2020},
    NUMBER = {7},
     PAGES = {135, 55},
      ISSN = {1126-6708,1029-8479},
   MRCLASS = {81T40 (14J32 81T30)},
  MRNUMBER = {4137961},
       DOI = {10.1007/jhep07(2020)135},
 eprint={1911.11724},
  archivePrefix={arXiv},
 primaryClass={hep-th}
}

@incollection {takhtajan1994topics,
    AUTHOR = {Takhtajan, L. A.},
     TITLE = {Topics in the quantum geometry of {R}iemann surfaces:
              two dimensional quantum gravity},
 BOOKTITLE = {Quantum groups and their applications in physics ({V}arenna,
              1994)},
    SERIES = {Proc. Internat. School Phys. Enrico Fermi},
    VOLUME = {127},
     PAGES = {541--579},
 PUBLISHER = {IOS, Amsterdam},
      YEAR = {1996},
      ISBN = {90-5199-247-5},
   MRCLASS = {32G15 (32G81 58D30 81T40)},
  MRNUMBER = {1415866},
MRREVIEWER = {Gregorio\ Falqui},
  archivePrefix={arXiv},
  eprint={hep-th/9409088},
  year={1994}
}

@article{negut2016exts,
    AUTHOR = {Negu\c{t}, Andrei},
     TITLE = {Exts and the {AGT} relations},
   JOURNAL = {Lett. Math. Phys.},
  FJOURNAL = {Letters in Mathematical Physics},
    VOLUME = {106},
      YEAR = {2016},
    NUMBER = {9},
     PAGES = {1265--1316},
      ISSN = {0377-9017,1573-0530},
   MRCLASS = {14D21 (81R10 81T13)},
  MRNUMBER = {3533570},
MRREVIEWER = {Shilin\ Yang},
       DOI = {10.1007/s11005-016-0865-3},
  archivePrefix={arXiv},
  eprint={1510.05482},primaryClass={hep-th}
}

@article{piatek2014classical,
  title={Classical torus conformal block, ${\mathcal {N}
}=2^*$ twisted superpotential and the accessory parameter of Lam{\'e} equation},
  author={P\c iatek, M},
  journal={Journal of High Energy Physics},
  volume={2014},
  number={3},
  pages={1--39},
  year={2014},
  publisher={Springer},
  archivePrefix = "arXiv",
  eprint = "1309.7672v3",      primaryClass={hep-th},
doi={10.1007/JHEP03(2014)124}
}

@book {ronveaux1995heun,
     TITLE = {Heun's differential equations},
    SERIES = {Oxford Science Publications},
    EDITOR = {Ronveaux, A.},
      NOTE = {With contributions by F. M. Arscott, S. Yu.\ Slavyanov, D.
              Schmidt, G. Wolf, P. Maroni and A. Duval},
 PUBLISHER = {The Clarendon Press, Oxford University Press, New York},
      YEAR = {1995},
     PAGES = {xxiv+354},
      ISBN = {0-19-859695-2},
   MRCLASS = {33E15 (33E30 34A20 34A25 34B30 81Q05)},
  MRNUMBER = {1392976},
MRREVIEWER = {Harold\ Exton},
}

@article {takhtajan2003hyperbolic,
    AUTHOR = {Takhtajan, Leon and Zograf, Peter},
     TITLE = {Hyperbolic 2-spheres with conical singularities, accessory
              parameters and {K}\"ahler metrics on {$\mathcal{M}_{0,n}$}},
   JOURNAL = {Trans. Amer. Math. Soc.},
  FJOURNAL = {Transactions of the American Mathematical Society},
    VOLUME = {355},
      YEAR = {2003},
    NUMBER = {5},
     PAGES = {1857--1867},
      ISSN = {0002-9947,1088-6850},
   MRCLASS = {32Q15 (14H15 30F45 32Q45)},
  MRNUMBER = {1953529},
MRREVIEWER = {Vasily\ A.\ Chernecky},
       DOI = {10.1090/S0002-9947-02-03243-9},
       archivePrefix ={arXiv},
eprint = {math/0112170}
}

@book {arnol2013mathematical,
    AUTHOR = {Arnol\textquotesingle d, V. I.},
     TITLE = {Mathematical methods of classical mechanics},
    SERIES = {Graduate Texts in Mathematics},
    VOLUME = {60},
      NOTE = {Translated from the 1974 Russian original by K. Vogtmann and
              A. Weinstein,
              Corrected reprint of the second (1989) edition},
 PUBLISHER = {Springer-Verlag, New York},
      YEAR = {1989},
     PAGES = {xvi+516},
      ISBN = {0-387-96890-3},
   MRCLASS = {70-02 (58F05 58Fxx 70Hxx)},
  MRNUMBER = {1345386},
}

@article {lacoin2019semiclassical,
    AUTHOR = {Lacoin, Hubert and Rhodes, R\'emi and Vargas, Vincent},
     TITLE = {The semiclassical limit of {L}iouville conformal field theory},
   JOURNAL = {Ann. Fac. Sci. Toulouse Math. (6)},
  FJOURNAL = {Annales de la Facult\'e{} des Sciences de Toulouse.
              Math\'ematiques. S\'erie 6},
    VOLUME = {31},
      YEAR = {2022},
    NUMBER = {4},
     PAGES = {1031--1083},
      ISSN = {0240-2963,2258-7519},
   MRCLASS = {81T40 (60D05 81Q20)},
  MRNUMBER = {4506137},
       DOI = {10.5802/afst.1713},
  year={2022},
archivePrefix ={arXiv},
eprint ={1903.08883}
}

@article {kupiainen2020integrability,
    AUTHOR = {Kupiainen, Antti and Rhodes, R\'emi and Vargas, Vincent},
     TITLE = {Integrability of {L}iouville theory: proof of the {DOZZ}
              formula},
   JOURNAL = {Ann. of Math. (2)},
  FJOURNAL = {Annals of Mathematics. Second Series},
    VOLUME = {191},
      YEAR = {2020},
    NUMBER = {1},
     PAGES = {81--166},
      ISSN = {0003-486X,1939-8980},
   MRCLASS = {81T40 (60D99)},
  MRNUMBER = {4060417},
MRREVIEWER = {Nizar\ Demni},
       DOI = {10.4007/annals.2020.191.1.2},
      eprint={1707.08785},
      archivePrefix={arXiv},
      primaryClass={math.PR}
}

@article {teschner2011quantization,
    AUTHOR = {Teschner, J.},
     TITLE = {Quantization of the {H}itchin moduli spaces, {L}iouville
              theory and the geometric {L}anglands correspondence {I}},
   JOURNAL = {Adv. Theor. Math. Phys.},
  FJOURNAL = {Advances in Theoretical and Mathematical Physics},
    VOLUME = {15},
      YEAR = {2011},
    NUMBER = {2},
     PAGES = {471--564},
      ISSN = {1095-0761,1095-0753},
   MRCLASS = {81T40 (14D24 32G81 34M56 53C26 53D50 81T60)},
  MRNUMBER = {2924236},
MRREVIEWER = {Lee-Peng\ Teo},
archiveprefix={arXiv},
  eprint = {1005.2846}
}

@article{bonelli2020n,
    AUTHOR = {Bonelli, Giulio and Del Monte, Fabrizio and Gavrylenko, Pavlo
              and Tanzini, Alessandro},
     TITLE = {{$\mathcal{N}=2^*$} 
               gauge theory, free fermions on the torus and
              {P}ainlev\'{e} {VI}},
   JOURNAL = {Comm. Math. Phys.},
  FJOURNAL = {Communications in Mathematical Physics},
    VOLUME = {377},
      YEAR = {2020},
    NUMBER = {2},
     PAGES = {1381--1419},
      ISSN = {0010-3616,1432-0916},
   MRCLASS = {81T35 (30E25 33E30 34M55 81T13)},
  MRNUMBER = {4115020},
MRREVIEWER = {Andrew\ Pickering},
       DOI = {10.1007/s00220-020-03743-y},
  archivePrefix = "arXiv",
  eprint = "1901.10497v2",
 primaryClass={hep-th}
}

@article{desiraju2022painleve,
    AUTHOR = {Desiraju, Harini},
     TITLE = {Painlev\'{e}/{CFT} correspondence on a torus},
   JOURNAL = {J. Math. Phys.},
  FJOURNAL = {Journal of Mathematical Physics},
    VOLUME = {63},
      YEAR = {2022},
    NUMBER = {8},
     PAGES = {Paper No. 081102, 16},
      ISSN = {0022-2488,1089-7658},
   MRCLASS = {34M55 (81T40)},
  MRNUMBER = {4462569},
  archivePrefix= "arXiv",
  eprint = "2305.04240",
doi={10.1063/5.0089867},
primaryclass="math-ph"
}

@article{BGG2021,
  archivePrefix = "arXiv",
  eprint = "2105.00985",
primaryClass={math-ph},
    AUTHOR = {Bershtein, Mikhail and Gavrylenko, Pavlo and Grassi, Alba},
     TITLE = {Quantum spectral problems and isomonodromic deformations},
   JOURNAL = {Comm. Math. Phys.},
  FJOURNAL = {Communications in Mathematical Physics},
    VOLUME = {393},
      YEAR = {2022},
    NUMBER = {1},
     PAGES = {347--418},
      ISSN = {0010-3616,1432-0916},
   MRCLASS = {81Q10 (81R12 81T40)},
  MRNUMBER = {4440693},
MRREVIEWER = {Serge\ C.\ Richard},
       DOI = {10.1007/s00220-022-04369-y}
}

@article{eremenko2022moduli,
    AUTHOR = {Eremenko, Alexandre and Gabrielov, Andrei and Mondello,
              Gabriele and Panov, Dmitri},
     TITLE = {Moduli spaces for {L}am\'{e} functions and {A}belian
              differentials of the second kind},
   JOURNAL = {Commun. Contemp. Math.},
  FJOURNAL = {Communications in Contemporary Mathematics},
    VOLUME = {24},
      YEAR = {2022},
    NUMBER = {2},
     PAGES = {Paper No. 2150028, 68},
   MRCLASS = {33E10 (30F30 57M50)},
  MRNUMBER = {4380114},
MRREVIEWER = {Athanase\ Papadopoulos},
       DOI = {10.1142/S0219199721500280},
      eprint={2006.16837},
      archivePrefix={arXiv},
      primaryClass={math.CV}
}

@article{ghosal2022analiticity,
  title={Analiticity and Symmetry of Virasoro Conformal blocks via Liouville CFT},
  author={Ghosal, PROMIT and Remy, GUILLAUME and Sun, X and Sun, Y},
  url = {https://drive.google.com/file/d/1CKpOICEbezvoma16ABVekEv5xhJO-0_Q/view?pli=1}
}

@misc{ang2023derivation,
      title={Derivation of all structure constants for boundary Liouville CFT}, 
      author={Morris Ang and Guillaume Remy and Xin Sun and Tunan Zhu},
      year={2024},
      eprint={2305.18266},
      archivePrefix={arXiv},
      primaryClass={math.PR}
}

@article{ince1940vii,
    AUTHOR = {Ince, E. L.},
     TITLE = {Further investigations into the periodic {L}am\'{e} functions},
   JOURNAL = {Proc. Roy. Soc. Edinburgh},
  FJOURNAL = {Proceedings of the Royal Society of Edinburgh},
    VOLUME = {60},
      YEAR = {1940},
     PAGES = {83--99},
      ISSN = {0370-1646},
   MRCLASS = {33.0X},
  MRNUMBER = {2400},
MRREVIEWER = {H.\ Bateman},
}

@incollection {vargas2017lecture,
    AUTHOR = {Rhodes, R\'{e}mi and Vargas, Vincent},
     TITLE = {Lecture notes on {L}iouville theory and the {DOZZ} formula},
 BOOKTITLE = {Topics in statistical mechanics},
    SERIES = {Panor. Synth\`eses},
    VOLUME = {59},
     PAGES = {185--229},
 PUBLISHER = {Soc. Math. France, Paris},
      ISBN = {978-2-85629-970-8},
   MRCLASS = {81T40 (81T20)},
  MRNUMBER = {4619314},
eprint={1712.00829},
      archivePrefix={arXiv},
      primaryClass={math.PR}
}

@article{rhodes2014gaussian,
    AUTHOR = {Rhodes, R\'{e}mi and Vargas, Vincent},
     TITLE = {Gaussian multiplicative chaos and applications: a review},
   JOURNAL = {Probab. Surv.},
  FJOURNAL = {Probability Surveys},
    VOLUME = {11},
      YEAR = {2014},
     PAGES = {315--392},
   MRCLASS = {60G57 (28A80 60G15 60G60)},
  MRNUMBER = {3274356},
MRREVIEWER = {Dora\ Sele\v{s}i},
       DOI = {10.1214/13-PS218},
  year={2014},
archivePrefix={arXiv},
eprint={1305.6221}, 
primaryClass={math.PR}
}

@article{berestycki2017elementary,
    AUTHOR = {Berestycki, Nathana\"{e}l},
     TITLE = {An elementary approach to {G}aussian multiplicative chaos},
   JOURNAL = {Electron. Commun. Probab.},
  FJOURNAL = {Electronic Communications in Probability},
    VOLUME = {22},
      YEAR = {2017},
     PAGES = {Paper No. 27, 12},
   MRCLASS = {60G57 (60G15 60J65)},
  MRNUMBER = {3652040},
       DOI = {10.1214/17-ECP58},
      eprint={1506.09113},
      archivePrefix={arXiv},
      primaryClass={math.PR}
}

@article{aru2017gaussian,
    AUTHOR = {Aru, Juhan},
     TITLE = {Gaussian multiplicative chaos through the lens of the 2{D}
              {G}aussian free field},
   JOURNAL = {Markov Process. Related Fields},
  FJOURNAL = {Markov Processes and Related Fields},
    VOLUME = {26},
      YEAR = {2020},
    NUMBER = {1},
     PAGES = {17--56},
      ISSN = {1024-2953},
   MRCLASS = {60G57 (60G15 60G60)},
  MRNUMBER = {4237161},
      eprint={1709.04355},
      archivePrefix={arXiv},
 primaryClass={math.PR}
}

@book{B,
author={Buslaev, V. S.},
title={Variational calculus},
publisher={Izd. Leningrad. Univ.},
year={1980}
}

@article{piatek2019solving,
    AUTHOR = {P\c iatek, Marcin and Pietrykowski, Artur R.},
     TITLE = {Solving {H}eun's equation using conformal blocks},
   JOURNAL = {Nuclear Phys. B},
  FJOURNAL = {Nuclear Physics. B. Theoretical, Phenomenological, and
              Experimental High Energy Physics. Quantum Field Theory and
              Statistical Systems},
    VOLUME = {938},
      YEAR = {2019},
     PAGES = {543--570},
   MRCLASS = {81T40},
  MRNUMBER = {3883275},
       DOI = {10.1016/j.nuclphysb.2018.11.021},
   eprint={1708.06135},
      archivePrefix={arXiv},
      primaryClass={hep-th}
}

@article {polyakov1981quantum,
    AUTHOR = {Polyakov, A. M.},
     TITLE = {Quantum geometry of bosonic strings},
   JOURNAL = {Phys. Lett. B},
  FJOURNAL = {Physics Letters. B. Particle Physics, Nuclear Physics and
              Cosmology},
    VOLUME = {103},
      YEAR = {1981},
    NUMBER = {3},
     PAGES = {207--210},
      ISSN = {0370-2693,1873-2445},
   MRCLASS = {81E99 (58D30 81G05 82A68)},
  MRNUMBER = {623209},
       DOI = {10.1016/0370-2693(81)90743-7}
}

@article{zamolodchikov1986two,
  title={Two-dimensional conformal symmetry and critical four-spin correlation functions in the Ashkin-Teller model},
  author={Zamolodchikov, Al B},
  journal={Soviet Journal of Experimental and Theoretical Physics},
  volume={63},
  number={5},
  year={1986},
 PAGES = {1061--1066},
   MRCLASS = {82A68 (81E25 81E40)},
  MRNUMBER = {869395},
MRREVIEWER = {Rainald\ Flume},
}

@article{pikatek2022classical,
    AUTHOR = {P\c iatek, M.  and Nazmitdinov, R. G. and Puente, A. and
              Pietrykowski, A. R.},
     TITLE = {Classical conformal blocks, {C}oulomb gas integrals and
              {R}ichardson-{G}audin models},
   JOURNAL = {J. High Energy Phys.},
  FJOURNAL = {Journal of High Energy Physics},
      YEAR = {2022},
    NUMBER = {4},
     PAGES = {Paper No. 098, 48},
   MRCLASS = {81T40 (81R10 82B23)},
  MRNUMBER = {4429642},
archivePrefix={arXiv},
eprint={2110.15009},
      primaryClass={hep-th},
doi={10.1007/JHEP04%282022%29098}
}

@article{da2022expansions,
    title={Expansions for semiclassical conformal blocks}, 
      author={Bruno Carneiro da Cunha and João Paulo Cavalcante},
      year={2022},
      eprint={2211.03551},
      archivePrefix={arXiv},
      primaryClass={hep-th}
}

@inproceedings {ponsot2004recent,
    AUTHOR = {Ponsot, B.},
     TITLE = {Recent progress in {L}iouville field theory},
 BOOKTITLE = {Proceedings of 6th {I}nternational {W}orkshop on {C}onformal
              {F}ield {T}heory and {I}ntegrable {M}odels},
   JOURNAL = {Internat. J. Modern Phys. A},
  FJOURNAL = {International Journal of Modern Physics A. Particles and
              Fields. Gravitation. Cosmology},
    VOLUME = {19},
      YEAR = {2004},
     PAGES = {311--335},
      ISSN = {0217-751X,1793-656X},
   MRCLASS = {81T40 (81R50)},
  MRNUMBER = {2087118},
MRREVIEWER = {Stanislav\ Z.\ Pakuliak},
archivePrefix={arXiv},
eprint={hep-th/0301193},
doi={10.1142/S0217751X0402049X}
}

@article{becsken2020semi,
    AUTHOR = {Be\c{s}ken, Mert and Datta, Shouvik and Kraus, Per},
     TITLE = {Semi-classical {V}irasoro blocks: proof of exponentiation},
   JOURNAL = {J. High Energy Phys.},
  FJOURNAL = {Journal of High Energy Physics},
      YEAR = {2020},
    NUMBER = {1},
     PAGES = {109, 15},
      ISSN = {1126-6708,1029-8479},
   MRCLASS = {81R10 (81T40)},
  MRNUMBER = {4088210},
       DOI = {10.1007/jhep01(2020)109},
    archivePrefix={arXiv},
eprint={1910.04169},
 primaryClass={hep-th}
}

@incollection{nekrasov2010quantization,
    AUTHOR = {Nekrasov, Nikita A. and Shatashvili, Samson L.},
     TITLE = {Quantization of integrable systems and four dimensional gauge
              theories},
 BOOKTITLE = {X{VI}th {I}nternational {C}ongress on {M}athematical
              {P}hysics},
     PAGES = {265--289},
 PUBLISHER = {World Sci. Publ., Hackensack, NJ},
      YEAR = {2010},
      ISBN = {978-981-4304-62-7},
   MRCLASS = {81R12 (37K10 81-02 81T13)},
  MRNUMBER = {2730782},
  archivePrefix={arXiv},
  eprint={0908.4052},
doi={10.1142/9789814304634_0015},
  primaryClass={hep-th}
}

@article{zamolodchikov1987conformal,
    AUTHOR = {Zamolodchikov, Al. B.},
     TITLE = {Conformal symmetry in two-dimensional space: on a recurrent
              representation of the conformal block},
   JOURNAL = {Teoret. Mat. Fiz.},
  FJOURNAL = {Akademiya Nauk SSSR. Teoreticheskaya i Matematicheskaya
              Fizika},
    VOLUME = {73},
      YEAR = {1987},
    NUMBER = {1},
     PAGES = {103--110},
      ISSN = {0564-6162},
   MRCLASS = {81E99 (17C65 81E40)},
  MRNUMBER = {939798},
doi={10.1007/bf01022967}
}

@article{nekrasov2018quantum,
    AUTHOR = {Nekrasov, Nikita and Pestun, Vasily and Shatashvili, Samson},
     TITLE = {Quantum geometry and quiver gauge theories},
   JOURNAL = {Comm. Math. Phys.},
  FJOURNAL = {Communications in Mathematical Physics},
    VOLUME = {357},
      YEAR = {2018},
    NUMBER = {2},
     PAGES = {519--567},
      ISSN = {0010-3616,1432-0916},
   MRCLASS = {81T13 (16G20 16T99 53C80)},
  MRNUMBER = {3767702},
doi={10.1007/s00220-017-3071-y},
archivePrefix={arXiv},
eprint={1312.6689},
primaryClass={hep-th}
}

@article{Takasaki:2000zd,
    AUTHOR = {Takasaki, Kanehisa},
     TITLE = {Painlev\'{e}-{C}alogero correspondence revisited},
   JOURNAL = {J. Math. Phys.},
  FJOURNAL = {Journal of Mathematical Physics},
    VOLUME = {42},
      YEAR = {2001},
    NUMBER = {3},
     PAGES = {1443--1473},
      ISSN = {0022-2488,1089-7658},
   MRCLASS = {34M55 (37J15)},
  MRNUMBER = {1814699},
MRREVIEWER = {Andrew\ Pickering},
       DOI = {10.1063/1.1348025},
      eprint         = "math/0004118",
      archivePrefix  = "arXiv",
      primaryClass   = "math-qa",
      reportNumber   = "KUCP-149"
}

@phdthesis{desiraju2021thesis,
  title={Painlev{\'e} tau-functions and Fredholm determinants},
  author={Desiraju, Harini},
  year={2021},
  school={SISSA},
  url={https://hdl.handle.net/20.500.11767/118727}
}

@article{mironov2010conformal,
    AUTHOR = {Mironov, A. and Morozov, A. and Shakirov, Sh.},
     TITLE = {Conformal blocks as {D}otsenko-{F}ateev integral
              discriminants},
   JOURNAL = {Internat. J. Modern Phys. A},
  FJOURNAL = {International Journal of Modern Physics A. Particles and
              Fields. Gravitation. Cosmology},
    VOLUME = {25},
      YEAR = {2010},
    NUMBER = {16},
     PAGES = {3173--3207},
      ISSN = {0217-751X,1793-656X},
   MRCLASS = {81T40},
  MRNUMBER = {2659902},
MRREVIEWER = {Michele\ Maio},
       DOI = {10.1142/S0217751X10049141},
  archivePrefix={arXiv},
  eprint={1001.0563},primaryClass={hep-th}
}

@article {Wang1,
    AUTHOR = {Wang, Yilin},
     TITLE = {A note on {L}oewner energy, conformal restriction and
              {W}erner's measure on self-avoiding loops},
   JOURNAL = {Ann. Inst. Fourier (Grenoble)},
  FJOURNAL = {Universit\'{e} de Grenoble. Annales de l'Institut Fourier},
    VOLUME = {71},
      YEAR = {2021},
    NUMBER = {4},
     PAGES = {1791--1805},
   MRCLASS = {30C55 (60J67)},
  MRNUMBER = {4398248},
MRREVIEWER = {Gregory\ Tycho\ Markowsky},
       DOI = {10.5802/aif.3427},
  eprint={1810.04578},
      archivePrefix={arXiv},
      primaryClass={math.CV}
}

@article {Wang2,
    AUTHOR = {Wang, Yilin},
     TITLE = {Equivalent descriptions of the {L}oewner energy},
   JOURNAL = {Invent. Math.},
  FJOURNAL = {Inventiones Mathematicae},
    VOLUME = {218},
      YEAR = {2019},
    NUMBER = {2},
     PAGES = {573--621},
      ISSN = {0020-9910,1432-1297},
   MRCLASS = {30C55 (11M36 30C62 30C80 30F60)},
  MRNUMBER = {4011706},
MRREVIEWER = {Joan\ R.\ Lind},
       DOI = {10.1007/s00222-019-00887-0},
      eprint={1802.01999},
      archivePrefix={arXiv},
      primaryClass={math.CV}
}

@article {Wang3,
    AUTHOR = {Wang, Yilin},
     TITLE = {The energy of a deterministic {L}oewner chain: reversibility
              and interpretation via {${\rm SLE}_{0+}$}},
   JOURNAL = {J. Eur. Math. Soc. (JEMS)},
  FJOURNAL = {Journal of the European Mathematical Society (JEMS)},
    VOLUME = {21},
      YEAR = {2019},
    NUMBER = {7},
     PAGES = {1915--1941},
      ISSN = {1435-9855,1435-9863},
   MRCLASS = {30C55 (30C62 60J67)},
  MRNUMBER = {3959854},
MRREVIEWER = {Dmitri\ Valentinovi\'{c}\ Prokhorov},
       DOI = {10.4171/JEMS/876},
      eprint={1601.05297},
      archivePrefix={arXiv},
      primaryClass={math.CV}
}

@ARTICLE{Pel&Wang,
    AUTHOR = {Peltola, Eveliina and Wang, Yilin},
     TITLE = {Large deviations of multichordal {${\rm SLE}_{0+}$}, real
              rational functions, and zeta-regularized determinants of
              {L}aplacians},
   JOURNAL = {J. Eur. Math. Soc. (JEMS)},
  FJOURNAL = {Journal of the European Mathematical Society (JEMS)},
    VOLUME = {26},
      YEAR = {2024},
    NUMBER = {2},
     PAGES = {469--535},
      ISSN = {1435-9855,1435-9863},
   MRCLASS = {30C55 (31A35 60F10 60J67)},
  MRNUMBER = {4705656},
       DOI = {10.4171/jems/1274},
eprint={2006.08574},
      archivePrefix={arXiv},
      primaryClass={math-ph}
}

@article {Novokshenov,
    AUTHOR = {Novoksh\"{e}nov, V. Yu.},
     TITLE = {Moving poles of solutions of the {P}ainlev\'{e} equation of
              the third type and their connection with {M}athieu functions},
   JOURNAL = {Funktsional. Anal. i Prilozhen.},
  FJOURNAL = {Akademiya Nauk SSSR. Funktsional\cprime ny\u{\i} Analiz i ego
              Prilozheniya},
    VOLUME = {20},
      YEAR = {1986},
    NUMBER = {2},
     PAGES = {38--49, 96},
      ISSN = {0374-1990},
   MRCLASS = {34A20},
  MRNUMBER = {847137},
MRREVIEWER = {Shlomo\ Strelitz},
doi={10.1007/BF01077265}
}

@article {XXZ,
    AUTHOR = {Xia, Jun and Xu, Shuai-Xia and Zhao, Yu-Qiu},
     TITLE = {Isomonodromy sets of accessory parameters for {H}eun class
              equations},
   JOURNAL = {Stud. Appl. Math.},
    VOLUME = {146},
      YEAR = {2021},
    NUMBER = {4},
   MRCLASS = {35Q53},
  MRNUMBER = {4300140},
DOI = {10.1111/sapm.12370},
eprint={2101.02864},
      archivePrefix={arXiv},
      primaryClass={math.CA}

}

@misc{ribault2022conformal,
      title={Conformal field theory on the plane}, 
      author={Sylvain Ribault},
      year={2014},
      eprint={1406.4290},
      archivePrefix={arXiv},
      primaryClass={hep-th}
}

@article {fateev_litvinov_differential,
    AUTHOR = {Fateev, Vladimir A. and Litvinov, Alexey V. and Neveu,
              Andr\'{e} and Onofri, Enrico},
     TITLE = {A differential equation for a four-point correlation function
              in {L}iouville field theory and elliptic four-point conformal
              blocks},
   JOURNAL = {J. Phys. A},
  FJOURNAL = {Journal of Physics. A. Mathematical and Theoretical},
    VOLUME = {42},
      YEAR = {2009},
    NUMBER = {30},
     PAGES = {304011, 29},
      ISSN = {1751-8113,1751-8121},
   MRCLASS = {81T40},
  MRNUMBER = {2521330},
MRREVIEWER = {Stanislav\ Z.\ Pakuliak},
       DOI = {10.1088/1751-8113/42/30/304011},
    eprint={0902.1331v3 },
      archivePrefix={arXiv},
      primaryClass={hep-th}
}

@misc{GKRV2022,
      title={Segal's axioms and bootstrap for Liouville Theory}, 
      author={Colin Guillarmou and Antti Kupiainen and Rémi Rhodes and Vincent Vargas},
      year={2021},
      eprint={2112.14859},
      archivePrefix={arXiv},
      primaryClass={math.PR}
}

@article {dorn1994two,
      eprint={9403141},
      archivePrefix={arXiv},
      primaryClass={hep-th},
    AUTHOR = {Dorn, H. and Otto, H.-J.},
     TITLE = {Two- and three-point functions in {L}iouville theory},
   JOURNAL = {Nuclear Phys. B},
  FJOURNAL = {Nuclear Physics. B. Theoretical, Phenomenological, and
              Experimental High Energy Physics. Quantum Field Theory and
              Statistical Systems},
    VOLUME = {429},
      YEAR = {1994},
    NUMBER = {2},
     PAGES = {375--388},
      ISSN = {0550-3213,1873-1562},
   MRCLASS = {81T30 (81T40)},
  MRNUMBER = {1299071},
       DOI = {10.1016/0550-3213(94)00352-1}
}

@article {zamolodchikov1996conformal,
    AUTHOR = {Zamolodchikov, A. and Zamolodchikov, Al.},
     TITLE = {Conformal bootstrap in {L}iouville field theory},
   JOURNAL = {Nuclear Phys. B},
  FJOURNAL = {Nuclear Physics. B. Theoretical, Phenomenological, and
              Experimental High Energy Physics. Quantum Field Theory and
              Statistical Systems},
    VOLUME = {477},
      YEAR = {1996},
    NUMBER = {2},
     PAGES = {577--605},
      ISSN = {0550-3213,1873-1562},
   MRCLASS = {81T40},
  MRNUMBER = {1413469},
MRREVIEWER = {Kazuto Oshima},
       DOI = {10.1016/0550-3213(96)00351-3}
}

@article {belavin1984infinite,
    AUTHOR = {Belavin, A. A. and Polyakov, A. M. and Zamolodchikov, A. B.},
     TITLE = {Infinite conformal symmetry in two-dimensional quantum field
              theory},
   JOURNAL = {Nuclear Phys. B},
  FJOURNAL = {Nuclear Physics. B. Theoretical, Phenomenological, and
              Experimental High Energy Physics. Quantum Field Theory and
              Statistical Systems},
    VOLUME = {241},
      YEAR = {1984},
    NUMBER = {2},
     PAGES = {333--380},
      ISSN = {0550-3213,1873-1562},
   MRCLASS = {81E13 (17B70 58G37 81D15)},
  MRNUMBER = {757857},
       DOI = {10.1016/0550-3213(84)90052-X}
}

@article {zamolodchikov1984conformal,
    AUTHOR = {Zamolodchikov, Al.\ B.},
     TITLE = {Conformal symmetry in two dimensions: an explicit recurrence
              formula for the conformal partial wave amplitude},
   JOURNAL = {Comm. Math. Phys.},
  FJOURNAL = {Communications in Mathematical Physics},
    VOLUME = {96},
      YEAR = {1984},
    NUMBER = {3},
     PAGES = {419--422},
      ISSN = {0010-3616,1432-0916},
   MRCLASS = {81D99 (17B10 81E99)},
  MRNUMBER = {769357},
MRREVIEWER = {Jouko\ Mickelsson},
       URL = {http://projecteuclid.org/euclid.cmp/1103941860},
}

@article {hadasz2010recursive,
    AUTHOR = {Hadasz, Leszek and Jask\'olski, Zbigniew and Suchanek,
              Paulina},
     TITLE = {Recursive representation of the torus 1-point conformal block},
   JOURNAL = {J. High Energy Phys.},
  FJOURNAL = {Journal of High Energy Physics},
      YEAR = {2010},
    NUMBER = {1},
     PAGES = {063, 13},
      ISSN = {1126-6708,1029-8479},
   MRCLASS = {81T40 (81T60)},
  MRNUMBER = {2660810},
MRREVIEWER = {Lee-Peng\ Teo},
       DOI = {10.1007/JHEP01(2010)063},
      eprint={0911.2353},
      archivePrefix={arXiv},
      primaryClass={hep-th}
}

@book{janson1997gaussian,
  title={Gaussian hilbert spaces},
  author={Janson, Svante},
  number={129},
  year={1997},
  publisher={Cambridge university press}
}

\end{document}